\documentclass[11pt]{amsart}
\usepackage[margin=1in]{geometry}
\usepackage{amsfonts, float, mathtools}
\usepackage{amssymb}
\usepackage[dvips]{graphics}
\usepackage{epsfig}
\usepackage{euscript}
\usepackage{color}
\usepackage{cite}

\usepackage[all]{xy}
\usepackage{multirow}
\usepackage{tikz}

\usepackage{fancyvrb}

\usepackage[pdftex]{hyperref}

 \newtheorem{thm}{Theorem}[section]
 \newtheorem{cor}[thm]{Corollary}
 \newtheorem{lemma}[thm]{Lemma}
 \newtheorem{prop}[thm]{Proposition}
 
 \newtheorem{claim}[thm]{Claim}
 \theoremstyle{definition}
 \newtheorem{defn}[thm]{Definition}
  \newtheorem{remk}[thm]{Remark}
 \theoremstyle{remark}
 \newtheorem*{ex}{Example}
 \newtheorem{assumption}[thm]{Assumption}
 \numberwithin{equation}{section}
 
 \def\idtyty{{\mathchoice {\mathrm{1\mskip-4mu l}} {\mathrm{1\mskip-4mu l}} %
{\mathrm{1\mskip-4.5mu l}} {\mathrm{1\mskip-5mu l}}}}

 \def\idty{{\mathchoice {\mathrm{1\mskip-4mu l}} {\mathrm{1\mskip-4mu l}} %
{\mathrm{1\mskip-4.5mu l}} {\mathrm{1\mskip-5mu l}}}}

\newcommand{\bR}{{\mathbb R}}

\newcommand{\Ir}{{\mathbb Z}}
\newcommand{\bC}{{\mathbb C}}
\newcommand{\cD}{{\mathcal D}}
\newcommand{\bN}{{\mathbb N}}
\newcommand{\bZ}{{\mathbb Z}}
\newcommand{\cA}{{\mathcal A}}
\newcommand{\cB}{{\mathcal B}}
\newcommand{\cH}{{\mathcal H}}
\newcommand{\cI}{{\mathcal I}}
\newcommand{\cT}{{\mathcal T}}

\newcommand{\supp}{\operatorname{supp}}

\newcommand{\cP}{{\mathcal P}}
\newcommand{\cS}{{\mathcal S}}
\newcommand{\cE}{{\mathcal E}}
\newcommand{\bE}{{\mathbb E}}

\newcommand{\cAG}{\cA_\Gamma^{{\rm loc}, G}}
\newcommand{\gsE}{\overline{E}}

\newcommand{\caA}{{\mathcal A}}
\newcommand{\caB}{{\mathcal B}}

\newcommand{\caF}{{\mathcal F}}
\newcommand{\caG}{{\mathcal G}}

\newcommand{\caI}{{\mathcal I}}

\newcommand{\cK}{{\mathcal K}}

\newcommand{\caS}{{\mathcal S}}
\newcommand{\caT}{{\mathcal T}}

\newcommand{\bbE}{{\mathbb E}}
\newcommand{\bbF}{{\mathbb F}}

\newcommand{\tr}{\mathrm{tr}}
\newcommand{\Tr}{\mathrm{Tr}}
\newcommand{\gap}{\mathrm{gap}}

\newcommand{\braket}[2]{\left\langle #1 , #2\right\rangle}

\newcommand{\diam}{\mathrm{diam}}
\newcommand{\ran}{{\rm ran}}

\newcommand{\bra}[1]{\langle #1 \vert} 

\newcommand{\ket}[1]{\vert #1 \rangle}

\newcommand{\Cx}{\mathbb{C}}

\renewcommand{\H}{\mathcal{H}} 

\newcommand{\dom}{\mathop{\rm dom}}
\newcommand{\spec}{\mathop{\rm spec}}

\newcommand{\be}{\begin{equation}}
\newcommand{\ee}{\end{equation}}
\newcommand{\bea}{\begin{eqnarray}}
\newcommand{\eea}{\end{eqnarray}}
\newcommand{\beann}{\begin{eqnarray*}}
\newcommand{\eeann}{\end{eqnarray*}}

\newcommand{\Rl}{\bR}

\newcommand{\eq}[1]{(\ref{#1})}

\newcommand{\floor}[1]{\lfloor #1 \rfloor}

\title[Quasi-Locality Bounds II]{Quasi-Locality Bounds for Quantum Lattice Systems.
Part II. Perturbations of Frustration-Free Spin Models with Gapped Ground States.}

\author[B. Nachtergaele]{Bruno Nachtergaele}
\thanks{Based upon work supported by the National Science Foundation under DMS-1813149.}
\address{Department of Mathematics and Center for Quantum Mathematics and Physics\\
University of California, Davis\\
Davis, CA 95616, USA}
\email{bxn@math.ucdavis.edu}

\author[R. Sims]{Robert Sims}
\address{Department of Mathematics \\
University of Arizona\\
Tuscon, AZ 85721, USA}
\email{rsims@math.arizona.edu}

\author[A. Young]{Amanda Young}
\address{Munich Center for Quantum Science and Technology, and\\
Zentrum Mathematik, TU M\"{u}nchen\\
85747 Garching, Germany}
\email{young@ma.tum.de}

\begin{document}
\date{\today }
\begin{abstract}
We study the stability with respect to a broad class of perturbations of gapped ground state phases of quantum spin systems defined by frustration-free Hamiltonians. The core result of this work is a proof using the Bravyi-Hastings-Michalakis (BHM) strategy
that under a condition of Local Topological Quantum Order, the bulk gap is stable under perturbations that decay at long distances faster than 
a stretched exponential. Compared to previous work we expand the class of frustration-free quantum spin models that can be handled
to include models with more general boundary conditions, and models with discrete symmetry breaking. Detailed estimates allow us to formulate sufficient conditions for the validity of positive lower bounds for the gap that are uniform in the system size and that are explicit to some degree.
We provide a survey of the BHM strategy following the approach of Michalakis and Zwolak, with alterations introduced to accommodate 
more general than just periodic boundary conditions and more general lattices. We express the fundamental condition known as LTQO
by means of the notion of indistinguishability radius, which we introduce. Using the uniform finite-volume results we then proceed
to study the thermodynamic limit. We first study the case of a unique limiting ground state and then also consider models with spontaneous
breaking of a discrete symmetry. In the latter case, LTQO cannot hold for all local observables. However, for perturbations that preserve the symmetry, we show stability of the gap and the structure of the broken symmetry phases. We prove that the GNS Hamiltonian associated with each pure state has a non-zero spectral gap above the ground state.
\end{abstract}

\maketitle

\tableofcontents

\section{Introduction}\label{sec:intro}

\subsection{Stability of the ground state gap}

The main object of study in this paper is the gap above the ground state of Hamiltonians of the form
$$
H(s) = H + s V,
$$
where $H$ is a finite-range frustration free quantum spin Hamiltonian with a gap above its ground state, and $V$ is a perturbation described by an interaction $\Phi$ of which the decay at long distances is upper bounded by a stretched exponential. The goal is to prove a lower bound for the ground state gap for $H(s)$ for sufficiently small $s$ under a set of conditions on  $H$ and its ground states. The existence of a positive lower bound for $|s|<s_0$, for some $s_0 >0$, uniformly in the system size, is referred to as {\em stability} of the ground state gap. Good introductions to the mathematics of quantum spin systems can be found in \cite{bratteli:1997,naaijkens:2017,tasaki:2020}.

A gap above the ground state in the spectrum of a quantum many-body Hamiltonian is a signature property that has important implications for the 
physics of the system described by that Hamiltonian. For example, it is well known (and proven) that it quite generally implies exponential decay of correlations in the ground state \cite{nachtergaele:2006a,hastings:2006}. In one dimension, a non-vanishing gap for the infinite system 
implies the split property \cite{matsui:2010}, which in turn plays a crucial role in definition of a topological index for symmetry protected topological 
phases \cite{ogata:2020a,ogata:2019a,ogata:2019b}.
More generally, the presence of a spectral gap features as an assumption in the theories classifying topological phases of matter \cite{chen:2011,pollmann:2012,chen:2013,ogata:2016a,moon:2019a,moon:2020} and the derivation of the quantum Hall effect and similar properties \cite{hastings:2015,bachmann:2017b,bachmann:2020a,bachmann:2020}.

To prove existence of a gap and, in particular, to obtain a positive lower bound uniform in the system size, is in general a hard problem.
It was shown in \cite{cubitt:2015a,cubitt:2015b} that the question whether an arbitrary translation-invariant, frustration-free, nearest-neighbor two-dimensional quantum spin model has a gap above the ground state or not, is undecidable in the technical sense. That result, however, has only limited 
bearing on what we can learn mathematically for specific classes of systems. There are a number of examples in the literature of such systems for 
which the question has been settled 
\cite{bachmann:2015,bishop:2016a,gosset:2016,abdul-rahman:2020,lemm:2019b,lemm:2019,pomata:2019,lemm:2020a,pomata:2020,guo:2020}. 
For one-dimensional frustration-free systems arguments to prove a gap have been extended even further \cite{affleck:1988,knabe:1988,fannes:1992,nachtergaele:1996,koma:1997,matsui:1997,koma:2001,spitzer:2003,bravyi:2015,nachtergaele:2020}.
The stability results of Bravyi, Hastings, and Michalakis significantly amplify the class of models for which one can prove a spectral gap uniform in the system size \cite{bravyi:2010,bravyi:2011,michalakis:2013}. In this work, we employ the Bravyi-Hastings-Michalakis (BHM) strategy to expand the class of models for which a gap can be proved even further.

Early results on the stability of the ground state gap have typically been framed as perturbation theory for the ground states of quantum 
spin systems. These were usually focused on a specific model or a limited class of models \cite{matsui:1990,kennedy:1992,borgs:1996,datta:1996,datta:1996a,fernandez:2006,yarotsky:2006}.
The growing interest in topologically ordered ground states, however, raised the general stability as a crucial question for their possible 
experimental observability and utilization as quantum memory. To address {\em stability} it is important to look for an approach that allows 
for the widest possible class of perturbations. Mathematically, the perturbations are described by an interaction involving arbitrary $k$-body terms and that belongs to a suitable Banach space, the norm of which expresses interaction strength and the decay at long distances.

The Toric Code model \cite{kitaev:2006} was the first test case for proving this type of stability in the presence of topological order. 
It has a unique ground state on the infinite lattice $\Ir^2$ \cite{alicki:2007}, but finite systems have multiple ground states 
and the ground state degeneracy is strongly dependent on the boundary conditions. The first proof of stability for the Toric 
Code model is by Bravyi, Hastings and Michalakis  \cite{bravyi:2010}. Klich addressed the same question in  \cite{klich:2010}.
Bravyi and Hastings followed up with a streamlined proof in \cite{bravyi:2011}. Their proof applies to the class of frustration-free 
commuting Hamiltonians satisfying a natural topological order condition. The term `commuting' here refers to the fact that all terms 
in the Hamiltonian commute, which holds for the general class of quantum double models defined by Kitaev \cite{kitaev:2006,kitaev:2009},
and also for the Levin-Wen models \cite{levin:2005}.

To be able to cover more physically realistic models it was important to get rid of the restriction to {\em  commuting} Hamiltonians. This was achieved 
by Michalakis and Zwolak in \cite{michalakis:2013}. In that work the condition of {\em Local Topological Quantum Order} (LTQO) is introduced in
 essentially the same form we will use it here. Other more restrictive notions of stability were investigated in 
\cite{schuch:2011a,cirac:2013,szehr:2015,haah:2016}. In the latter, the LTQO condition is either automatically satisfied or expressed in a different 
way by the assumptions for the particular class of systems under consideration. In this work we will focus on the BHM strategy and we refine the LTQO condition to obtain extensions of the stability results in two directions, namely, (i) finite systems with other than periodic boundary conditions and 
(ii) systems in which discrete symmetry breaking occurs. 

A generalization we do not pursue in the this paper is the inclusion of models with unbounded on-site Hamiltonians of the type 
considered in \cite{yarotsky:2006,frohlich:2020} and unbounded interactions as in \cite{del-vecchio:2019}. 
Fr\"ohlich and Pizzo introduced a method that handles a class of unbounded one-dimensional lattice Hamiltonians with ease 
as long as the unperturbed ground state is unique and given by a product state. The latter restriction excludes 
non-trivial order, topological or otherwise, and, naturally, any version of an LTQO condition is automatically satisfied.
In \cite{del-vecchio:2019a}, Del Vecchio, Fr\"ohlich, Pizzo, and Rossi prove analyticity of the ground state energy density
for translation-invariant chains of the same class.

There also has been recent interest in stability for fermionic lattice systems. We outlined a BHM strategy for lattice fermions in \cite{nachtergaele:2018}, on which we plan to elaborate in a separate paper \cite{paper4}. Hastings sketched a related approach
for perturbations of quasi-free fermion systems in \cite{hastings:2019}, which was elaborated upon by Koma in \cite{koma:2020}. 
Another stability result for gapped quasi-free lattice fermion systems was proved in \cite{de-roeck:2019}.
The applicability of these results to topologically ordered systems with gapless boundary modes remains unclear, and we will address 
this issue in \cite{paper4}. 

We also mention that the stability question for {\em irreversible dynamics} with an exponentially clustering invariant state has been addressed in 
\cite{cubitt:2015}.

Statements about the thermodynamic limit are highly relevant for the classification of gapped ground state phases, including symmetry protected topologically order phases. For example, the topological indices introduced in \cite{ogata:2020a,ogata:2019a,ogata:2019b} are for infinite gapped systems. In this work we study infinite systems as limits of
sequences of finite systems. In this familiar approach, as an intermediate step, one studies a sequence of finite systems for which estimates uniform in the system size can be derived. This is described in more detail in the summary given in the next section. It is worth noting that this does not cover all cases of interest since our ability to carry this out depends on the existence of a uniformly positive gap for finite systems, which depends on knowing suitable boundary conditions for which there are no gapless boundary states that may obscure the existence of a bulk gap. Such boundary conditions
are not known or even known to exist in all cases. They may, in fact, not exist \cite{wang:2019}. For this reason we will present a direct approach to the bulk gap in the infinite system that bypasses this difficulty in a forthcoming paper \cite{paper3}.

\subsection{The Bravyi-Hastings-Michalakis strategy and main results}\label{sec:BHM}

In this section we will sketch the general approach to proving stability of spectral gaps for quantum spin systems introduced by Bravyi, 
Hastings, and Michalakis in \cite{bravyi:2010}, and streamlined and extended in \cite{bravyi:2011,michalakis:2013}, to which we will henceforth
refer to as the {\em BHM strategy}. In this general overview, we do not spell out the technical assumptions in detail, but focus instead on the overall 
structure of the main arguments and the qualitative role of the basic assumptions. This will also allow us to point out where the new contributions 
of this work are located in the overall scheme and we hope it will be helpful to the reader. We believe that the BHM strategy combined with 
the results in this paper and future enhancements will continue to extend its reach. Precise definitions, assumptions, and statements of the results  follow in later sections, see Section \ref{sec:summaries} for an outline.

The class of quantum spin models under consideration are defined on a lattice $\Gamma$, which we assume is a metric space that satisfies a regularity (finite-dimensionality) condition expressed by requiring that the cardinality of balls does not grow faster than a power of their radius.
Often one can take $\Gamma =\Ir^\nu$ with the lattice distance, as in done in the work of Bravyi, Hastings, and Michalakis (BHM)\cite{bravyi:2010,bravyi:2011,michalakis:2013}. The generalization to general regular $\Gamma$ is mostly straightforward and of significance only when we consider various boundary conditions and study the thermodynamic limit. BHM only consider finite systems with periodic boundary conditions.

The Hamiltonians can be written in terms of two interactions, $\eta$ and $\Phi$, that map each finite subset, $X$,
of the lattice to a self-adjoint observable, $\eta(X)$ and $\Phi(X)$, that is supported in $X$. Formally,
\be
H(s) = \sum_X \eta(X) + s\Phi(X),
\ee
where $s\in\Rl$ is the perturbation parameter. The unperturbed model, $H(0)$ defined by $\eta$ alone, is assumed to be finite-range, frustration-free, and with a gap in the spectrum above the ground state, uniformly in the system size and subject to suitable boundary conditions. 
Stability, the property we want to prove, means that there exists $s_0>0$, such that
for all $s$ with $|s|<s_0$, $H(s)$ has a gap above its ground state that is bounded below uniformly in the system-size. See Section \ref{sec:stab+sf} for a detailed discussion. 

Stability does not hold in general, of course. The BHM strategy assumes two essential conditions, one on the unperturbed model $H(0)$, and one
on the perturbation $\Phi$. Due to the frustration-freeness, the ground state space of the unperturbed model defined on a finite volume is given by the kernel of the Hamiltonian. This kernel need not be one-dimensional and its dimension may grow unbounded as a function of system size. A stable
gap above the ground state implies that the ground state splitting by arbitrary perturbations is `non-essential'. In particular, the distinct ground states should not be distinguished by any of the local perturbation terms since otherwise their energies would split in first order of perturbation theory. A suitable indistinguishability assumption is introduced by BHM which they refer to as topological order conditions. It is called
Local Topological Quantum Order (LTQO) in \cite{michalakis:2013}, which is the version we also use. To deal with more general boundary conditions we found it useful to formulate it in terms of an {\em indistinguishability radius} (see Definition \ref{LTQO_def}).

The second condition needed to prove stability is that $\Phi$ is sufficiently short-range. Roughly speaking, we require that the interaction strength decays at least as fast as a stretched exponential at long distances. Mathematically, this can be expressed by introducing a suitable Banach space
of allowed interactions $\Phi$. There is some freedom to choose $\eta$ and $\Phi$ given $H(s)$, and this freedom is useful for some arguments.
In particular, we will make use of the notion of {\em anchored interaction}. See Section \ref{sec:balled-int} and Appendix \ref{sec:est-trans-balled-ints} for a detailed discussion of norms on interactions and a proof that certain useful choices for the interactions have comparable norms. We note here that in \cite{michalakis:2013} the  main stability result is claimed to hold for models on the lattice $\Ir^\nu$, under perturbations that decay
as a power law with an exponent greater than $\nu + 2$. We have not been able to verify this claim and explained in \cite[Section VI.E.1, page 62]{nachtergaele:2019} why it may be erroneous.

The essence of the BHM strategy is a combination of something novel with something classical. The classical element is to use relative boundedness of a perturbation with respect to the unperturbed Hamiltonian to show that gaps in the spectrum remain open for small coupling
constants. In \cite{bravyi:2010,bravyi:2011} a relative norm bound is used. We use a relative form bound as in \cite{michalakis:2013}. See Section 
\ref{sec:gen_pert} for a discussion of gaps and relative form bounds. The new ingredient is the quasi-local `quasi-adiabatic' evolution, which is an $s$-dependent unitary transformation $U(s)$,  introduced and pioneered by Hastings \cite{hastings:2004,hastings:2004b,hastings:2005}. We studied this evolution in detail in \cite{bachmann:2012}, where we called it the {\em spectral flow}.
It has since been used in many other interesting applications \cite{bachmann:2012bf,bachmann:2014c,hastings:2015,bachmann:2015a,bachmann:2016,bachmann:2017b,bachmann:2017c,bachmann:2018,bachmann:2020a,cha:2020,bachmann:2020}. 
Its two main features are that it is quasi-local and that it exactly transforms the ground state spaces of $H(s)$ into each other. 
Quasi-locality is expressed with a Lieb-Robinson bound, which holds for the spectral flow in the same way as for physical dynamics 
generated by a short-range interaction \cite{lieb:1972}.

The BHM strategy is to first apply the quasi-adiabatic evolution to $H(s)$ and then prove a relative bound for the transformed Hamiltonian, which of course has the same spectrum as $H(s)$. 
Concretely, one defines $\tilde\Phi(x,n,s)$, for each site $x\in\Gamma$, such that
\be\label{unitary_transformation}
U^*(s) H(s) U(s) = H(0) + \sum_{x, n\geq 1} \tilde\Phi(x,n,s) + R(s) + E(s)\idty,
\ee
in which $R(s)$ is a remainder term that vanishes in the thermodynamic limit, $E(s)$ is a good approximation of the perturbed
ground state energy, and $\tilde\Phi(x,n,s)$ has the following properties: 
(i) $\tilde\Phi(x,n,s)$ is supported in the ball centered at $x$ with radius $n$; (ii) $\Vert \tilde\Phi(x,n,s)\Vert$ decays at least fast as 
as a stretched exponential; (iii) $\tilde\Phi(x,n,s)$ vanishes on the ground states of $H(0)$. These properties are proved in 
Sections \ref{sec:Step1} and
\ref{sec:LTQO2}, in which we use of the assumptions and the quasi-locality properties of $U(s)$. We refer to \cite{nachtergaele:2019} for a 
detailed analysis of the latter.

As shown in Section \ref{sec:form_bd} (Theorem \ref{thm:Step3General}) the properties of $\tilde\Phi$, and some technical assumptions we skip over here,  imply that the 
perturbation it defines satisfies a relative form bound with respect to $H(0)$ of the form
\be\label{formbd_beta}
\left| \sum_{x, n\geq 1} \langle\psi,\tilde\Phi(x,n,s) \psi\rangle\right| \leq |s| \beta \langle \psi, H(0)\psi\rangle, \quad \psi \in \cH,
\ee
where $\beta$ is a constant that depends on the unperturbed model, a suitable norm of the perturbing interaction $\Phi$, 
and on a choice of $\gamma\in(0,\gamma_0)$, where $\gamma_0$ is the gap of the unperturbed system. $\beta$ \emph{not} depend 
on the finite volume and $s$. As we explain in Section \ref{sec:Pert_Theory}, this implies stability of the gap in the spectrum above the 
ground state.

A topic not discussed in the work of Bravyi, Hastings, and Michalakis is the thermodynamic limit. In preparation for studying 
the thermodynamic limit, we investigate sequences of 
finite systems for which the estimates leading to the constant $\beta$  in \eq{formbd_beta} and the vanishing of the remainder term $R(s)$ in
\eq{unitary_transformation} hold uniformly for a sequence of finite systems in Section \ref{sec:uniform_sequences}. The thermodynamic limit
is then discussed in Section \ref{sec:automorphic-equivalence}. We show that under assumptions for which the thermodynamic limit of the dynamics exists, the ground states of the finite-volume systems converge to a unique pure ground state of the infinite system and this state
also satisfies LTQO for $|s|< s_0$ (Theorem \ref{thm:convergence_LTQO}). The lower bounds of the finite systems yields a lower bound for the 
gap in spectrum of the GNS Hamiltonian above the ground state. This is shown in Theorem \ref{thm:gap_ineq_TL} and Corollary \ref{cor:GNSgap}. 
The unitary evolution $U(s)$ leads to a strongly continuous co-cycle of automorphisms relating the ground states at different values of $s$ in
the interval $(-s_0,s_0)$. These are the automorphisms that implement the notion of {\em automorphic equivalence} introduced in 
\cite{bachmann:2012} and that appear in the definition of {\em gapped ground state phase} \cite{chen:2011,chen:2013,nachtergaele:2019}.

The gap of the GNS Hamiltonian is often referred to as the {\em bulk gap}. It is interesting to note that the applicability of our results 
includes cases 
where the  gap for finite systems with open boundary conditions tends to $0$ as the system size tends to infinity, while the gap with periodic 
boundary conditions is bounded below uniformly. In \cite{paper3} we extend this result to prove the stability of the bulk gap directly, regardless 
of the behavior with particular boundary conditions.

The last question we address in this paper is the situation in the presence of discrete symmetries. Spontaneous symmetry breaking in the
ground states is a common phenomenon and it is compatible with a non-vanishing gap above the ground states, as well as topological order
(the Goldstone theorem shows that continuous symmetry breaking, however, is not compatible with a gap \cite{landau:1981}). In Section
\ref{sec:TL_G} we discuss 
in some detail how the BHM strategy can be adapted to the situation with discrete symmetry breaking of three types: (S1) local symmetries such as 
spin flip, discrete spin rotation, time-reversal etc; (S2) breaking of lattice translations to a subgroup leading to periodic ground states; (S3) other lattice symmetries such as reflections and lattice rotations. In each of these cases we show that, if the unperturbed model has the symmetry and it is spontaneously broken in the ground states, then the spectral gap {\em and} the symmetry breaking are stable under perturbations that possess the symmetry. This is the content of Theorem \ref{thm:sym_stability}. Examples of breaking of each of the three types of discrete symmetry can be found in one-dimensional models with a finite set if pure matrix product ground states. In Section \ref{sec:MPSexamples} we show how the assumptions for the general results are satisfied for this class of examples.

Next, we give a synopsis of the remaining sections and the two appendices of this paper.

\subsection{Summaries of the sections}\label{sec:summaries}

\begin{description}
\item[Section \ref{sec:set-up}]
We introduce the main elements of the mathematical setting for quantum spin systems, the notion of interaction, and the property of frustration-freeness
in Section \ref{sec:QSS}. Given the ground state space of a system, we define the indistinguishability radius, which is used to express Local Topological Quantum Order, and discuss some examples. We review $F$-norms on spaces of interactions, Lieb-Robinson bounds, and basic quasi-local estimates. In Section \ref{sec:stab+sf} we give a precise definition of stability of the spectral gap. Section \ref{sec:spec_flow} reviews the Hastings generator and the spectral flow, which is the essential tool of the BHM strategy. In Section \ref{sec:balled-int} we discuss how general interactions can be rewritten in {\em anchored form}, which is often convenient and we discuss how this affects their norms.
\item[Section \ref{sec:Pert_Theory}]
In Section \ref{sec:gen_pert} we prove a general Level Repulsion Principle (Lemma \ref{lem:variational_principle_for_gaps}), which can also be seen as a variational principle for gaps. The proof of spectral gap stability uses relative form bounds and we explain in some detail how this proceeds in the remainder of Section \ref{sec:gen_pert}. Section \ref{sec:form_bd} gives a relative form bound for a special class of interactions
(Theorem \ref{thm:Step3General}). This requires a regularity assumption on the lattice and the notion of separating partitions to deal with general (non-commuting) Hamiltonians (Definition \ref{ass:sep_part}).
\item[Section \ref{sec:Step1}]
In Section~\ref{sec:Step1}, we begin to implement the basic BHM strategy. Rather than study the spectrum of
the perturbed Hamiltonian $H(s)$ directly, we consider the transformed Hamiltonian $\alpha_s(H(s))$ where
$\alpha_s( \cdot)$ is the spectral flow automorphism. The main goal of Section~\ref{sec:Step1} is to begin a 
decomposition of $\alpha_s(H(s))$ into a form suitable for an application of the general perturbation
theory results proven in Section~\ref{sec:Pert_Theory}, namely Theorem~\ref{thm:Pert_Est}. After introductory
preliminaries in Section~\ref{sec:intro_Step1}, this main goal is accomplished by proving two results: Proposition~\ref{prop:step1-commute}
in Section~\ref{sec:conv-lab} and Theorem~\ref{thm:Step1} in Section~\ref{sec:1ests}.  As detailed in
Section~\ref{sec:just-tech}, the bulk of the work in this section involves 
an appropriate choice of local decompositions and a familiarity with quasi-locality estimates; the foundations of which
we considered in \cite{nachtergaele:2019}.  
\item[Section \ref{sec:LTQO2}] In Section~\ref{sec:LTQO2}, we complete the decomposition procedure
we began in Section~\ref{sec:Step1}. It is here that the structure of the unperturbed ground state space, and
in particular, our assumptions on local topological quantum order play a crucial role.  
The main content is a proof of Theorem~\ref{thm:step2-(i)+(ii)} and Theorem~\ref{thm:rel_bound_cond}. Theorem~\ref{thm:step2-(i)+(ii)} establishes 
estimates on the remainder terms that arise from the decomposition of the transformed Hamiltonian, 
whereas Theorem~\ref{thm:rel_bound_cond} demonstrates that the remaining 
anchored interaction terms satisfy the necessary constraints so that the general form bound estimate, 
see Theorem~\ref{thm:Step3General} in Section~\ref{sec:form_bd}, is applicable. 
\item[Section \ref{sec:uniform_sequences}]
In Section~\ref{sec:uniform_sequences}, we set the stage for considerations of the thermodynamic limit. 
In fact, we consider assumptions for models, defined on an increasing and absorbing sequence of finite volumes,
which are sufficiently uniform so that the stability estimates hold uniformly in these finite volumes.
We characterize uniformity of the models with Definition~\ref{def:pert_model} which we refer to as {\it perturbation models}.
We characterize uniformity of the estimates for these models by Assumption~\ref{ass:uni_pert_model} which
defines {\it uniform perturbation models}. The main results of this section are Theorem~\ref{thm:uni_stab} and
Corollary~\ref{cor:zerod+e} which both demonstrate forms of stability for uniform perturbation models. 
In Section~\ref{sec:veri_ass}, we discuss some common cases where one can verify that our model
assumptions hold, and in Section~\ref{subsec:PBC}, we discuss cases where these stability arguments simplify,
for example, in situations where the models of interest have periodic boundary conditions.

\item[Section \ref{sec:automorphic-equivalence}] We consider spectral gap stability in the thermodynamic limit for uniform perturbation 
models for which limiting dynamics exists and the ground states are everywhere indistinguishable (see Definition~\ref{ass:bf}). We show in Theorem~\ref{thm:convergence_LTQO} and Corollary~\ref{cor:uniqueTL} that the perturbed models are also everywhere indistinguishable, 
and that their ground states converge to a unique, pure infinite volume ground state. We establish a criterion for which finite volume ground state indistinguishability implies a spectral gap of the associated GNS Hamiltonian in Theorem~\ref{thm:gap_ineq_TL} and
Corollary~\ref{cor:GNSgap}. 

\item[Section \ref{sec:TL_G}] We introduce two indistinguishability radii that can be used to prove stability for several cases of discrete symmetry breaking. For uniform finite volume stability, it is sufficient to consider the $G$-symmetric radius (see Definition~\ref{LTQO_def_G}). To recover the stability of the GNS gap, one needs the stronger $G$-broken radius from Assumption~\ref{assumption:LTQO_N}. We discuss in detail how to adjust the BHM strategy to prove stability for a model with a broken gauge symmetry which is proved in Theorem~\ref{thm:sym_stability},  and explain how to modify this argument to hold for cases of broken lattice symmetries in Section~\ref{sec:casesS2S3}. We conclude with Section~\ref{sec:MPSexamples} where we provide a class of examples with symmetry broken MPS ground states for which our methods apply.

\item[Appendix\ref{sec:est-trans-balled-ints}] In Appendix~\ref{sec:est-trans-balled-ints}, we provide some basic quasi-locality estimates
with particular emphasis on models defined by anchored interactions. More general results of this type are described in detail in \cite{nachtergaele:2019}.
The main result is Theorem~\ref{thm:trans-int-bd}, see also Corollary~\ref{cor:trans-int-F-dec}, which establishes a bound in $F$-norm on a quasi-locally transformed anchored interaction. Results of this specific type enter the analysis of Section~\ref{sec:Step1}.

\item[Appendix\ref{app:MPS}] We prove a lower bound on the indistinguishability radius for models with a unique infinite volume MPS ground state in Section~\ref{app:unique_MPS}. We then consider models with $N$-distinct MPS ground states, and show that the conditions from Assumption~\ref{assumption:LTQO_N} for stability in the case of symmetry breaking hold.
\end{description}

\section{General framework and auxiliary results} \label{sec:set-up}

\subsection{Introduction}

Our aim in this paper is to present the current status of stability results for quantum spin systems in considerable generality. This is not to say that all results are stated under the most general conditions available to date. Attempting to do that would produce and unreadable text and require an excessive amount of definitions and notations. Our emphasis is on general ideas that work for large classes of systems and we illustrate the application of these ideas by presenting detailed arguments that cover the known results and, in fact, allow us to provide a number of generalizations and new results that can be obtained using the same principles.

In support of this goal, we describe in this section the more or less standard mathematical framework for studying quantum spin systems and discuss the basic notions that feature in the stability properties of the spectral gap above the ground state(s). Some definitions generalize what has appeared in the literature so far and in some cases we found it useful to discuss relationships between different ways basic properties may be expressed. This is a bit more material than is strictly needed to read the rest of the paper, but we hope some readers will find it useful. 

\subsection{Quantum Spin Systems}\label{sec:QSS}

In this work, we study quantum spin models defined by uniformly bounded and frustration-free interactions. While the interactions defining both the initial system and the perturbation are static, the method of choice to study the spectrum of these Hamiltonians relies on auxiliary dynamics generated
by a time-dependent generator (the Hastings generator of the so-called spectral flow). Therefore, we consider both time-independent and
time-dependent interactions in our setup.

We consider quantum spin systems defined on a countable metric space $(\Gamma, \, d)$ that is \emph{$\nu$-regular}, meaning there is a non-negative integer $\nu$ and constant $\kappa >0$ such that for any $x\in \Gamma$,
\begin{equation}  \label{nu-reg}
|b_x(n)| \leq \kappa n^\nu,
\end{equation} 
where $b_x(n) = \{y\in \Gamma \, : \, d(x,y)\leq n \}$. If $\Gamma$ is a regular lattice, \eq{nu-reg} holds with $\nu$ the lattice dimension.
At every site $x\in\Gamma$, we associate a finite-dimensional Hilbert space $\cH_x=\bC^{n_x}$, and denote by $\cB(\cH_x)$ the algebra of all bounded linear operators. We use $\cP_0(\Gamma)$ to denote the set of all finite subsets of $\Gamma$, and for each $\Lambda \in \cP_0(\Gamma)$ we define the \emph{state space} and \emph{algebra of observables}, respectively, by
\begin{equation}
\cH_{\Lambda} = \bigotimes_{x\in\Lambda} \cH_x, \qquad \cA_\Lambda := \bigotimes_{x\in\Lambda} \cB(\cH_x) = \cB(\cH_\Lambda).
\end{equation}
Some results will in fact hold more generally for systems with infinite-dimensional state spaces and we will point this out where applicable.

For any two finite subsets $\Lambda_0 \subset \Lambda$, there is a natural embedding $\cA_{\Lambda_0} \hookrightarrow \cA_\Lambda$ via $A \mapsto A\otimes \idtyty_{\Lambda \setminus \Lambda_0}$ for all $A\in \cA_{\Lambda_0}$. With respect to this identification, the \emph{algebra of local observables} is defined by the inductive limit
\begin{equation}
\cA_\Gamma^{\rm loc} = \bigcup_{X\in\cP_0(\Gamma)} \cA_X,
\end{equation}
and the C$^*$-algebra of \emph{quasi-local observables}, denoted $\cA_\Gamma$, is given by the norm completion of $\cA_\Gamma^{\rm loc}$. 

A quantum spin model is defined in terms of an \emph{interaction} $\Phi$. In the time-independent case, this is a map $\Phi:\cP_0(\Gamma)\to\cA_{\Gamma}^{\rm loc}$ such that $\Phi(X)^* = \Phi(X) \in \cA_X$. The \emph{local Hamiltonian} associated to any $\Lambda \in \cP_0(\Gamma)$ is the sum of all interaction terms supported on $\Lambda$, i.e.
\begin{equation} \label{ti_ham}
H_\Lambda = \sum_{X\subseteq \Lambda} \Phi(X).
\end{equation}
An interaction $\Phi$ is \emph{uniformly bounded} if 
\begin{equation}
\sup_{X\in\cP_0(\Gamma)}\|\Phi(X)\| < \infty,
\end{equation}
and \emph{finite range} if there exists $R>0$ such that $\Phi(X) = 0$ for any finite $X$ with $\diam(X) > R$. The smallest such $R$ for which this holds is called the \emph{range} of the interaction. We will also consider interactions $\Phi$ that are \emph{frustration-free} and have \emph{local topological quantum order (LTQO)}. These are both properties on the ground states associated with the finite volume Hamiltonians. We describe these properties in detail as they will be key assumptions for the main results of this work. 

\subsubsection{Frustration-free Interactions} \label{sec:FF-ints} A frustration-free interaction is one where the ground states of any finite volume Hamiltonian $H_\Lambda$ simultaneously minimize the energy of all interaction terms $\Phi(X)$, $X\subseteq \Lambda$. Said differently, up to shifting each interaction term, $\Phi(X)$, by its ground state energy, we say that an interaction $\Phi:\cP_0(\Gamma)\to\cA_\Gamma^{\rm loc}$ is \emph{frustration-free} if the following two properties hold:
\begin{enumerate}
	\item $\Phi(X) \geq 0$ for all $X\in\cP_0(\Gamma)$.
	\item $\min\spec(H_\Lambda) = 0$ for all $\Lambda\in\cP_0(\Gamma)$.
\end{enumerate}
It follows immediately from the definition that the ground state space of $H_\Lambda$ is $\caG_\Lambda := \ker(H_\Lambda)$, and that $\psi\in \caG_\Lambda$ if and only if $\Phi(X) \psi = 0$ for all $X\subseteq \Lambda$, i.e.
\be\label{int_prop_ff}
\caG_\Lambda = \bigcap_{X\subseteq\Lambda} \ker(\Phi(X)).
\ee
Let $P_\Lambda$ denote the orthogonal projection onto $\caG_\Lambda$ for any $\Lambda \in \cP_0(\Gamma)$. By identifying $H_{\Lambda_0}\mapsto H_{\Lambda_0}\otimes\idtyty_{\Lambda\setminus \Lambda_0}\in \cA_{\Lambda}$ for $\Lambda_0\subseteq \Lambda$, the above equation implies that $\caG_\Lambda \subseteq \caG_{\Lambda_0}$. As a consequence, the associated ground state projections satisfy
\begin{equation}\label{ff_prop}
P_{\Lambda}P_{\Lambda_0} = P_{\Lambda_0}P_{\Lambda} = P_{\Lambda}.
\end{equation}
This ground state projection property is a key feature of frustration-free interactions and will frequently be used in our analysis.

\subsubsection{Local Topological Quantum Order} \label{sec:LTQO} 


A characteristic feature of topological order is the degeneracy of the ground state accompanied by the property that observables localized away from the boundary of the volume do not (or barely) distinguish between different ground states. In the absence of a boundary (for example finite volumes considered with periodic boundary conditions, say a torus) observables with support that is small with respect to the size of topologically non-trivial closed paths in the volume similarly cannot distinguish between different ground states. It is this feature that makes such systems candidates to serve as robust quantum memory: you can store information by selecting a particular ground state without the danger that local noise will erase that information.
The robustness of this property requires that there is a gap in the spectrum above the ground state that does not vanish with increasing system size.
In fact, the local indistinguishability of the ground states itself implies that local perturbations to the Hamiltonian will not affect the ground state energy, at least not up to high orders in perturbation theory.


This motivates the notion of local topological quantum order (LTQO), which describes this property of  local indistinguishability of the ground states
in a quantitative way. LTQO is a central condition for the stability results we present in this paper. The term LTQO was coined by Michalakis and Zwolak in \cite{michalakis:2013} but essentially the same property was first considered by Bravyi, Hastings, and Michalakis in  \cite{bravyi:2010}. In these and other subsequent works, the authors only consider Hamiltonians with periodic boundary conditions and define their topological order condition specifically for this situation. There are situations in which it is necessary or preferable to consider the ground state problem for models with other boundary conditions. Therefore, in this paper we introduce a more general LTQO condition built on the notion of an \emph{indistinguishability radius}.

As before, for any finite volume $\Lambda \in \cP_0(\Gamma)$, we denote by $P_\Lambda$ the orthogonal projection onto the ground state space of $H_\Lambda = \sum_{X\subseteq\Lambda}\Phi(X)$, and denote by $\omega_\Lambda: \cA_\Lambda \to \bC$ the ground state functional
\begin{equation} \label{ground_state_fun}
\omega_\Lambda(A) = \Tr[P_\Lambda A]/\Tr[P_\Lambda].
\end{equation}
We define the indistinguishability radius of a site $x\in \Lambda$ in terms of balls with respect to $\Lambda$. To differentiate these from the balls in $\Gamma$, we use the notation $b_x^\Lambda(n) = \{y\in \Lambda \, : d(x,y)\leq n \}$. Since $b_x^\Lambda(n) = b_x(n) \cap \Lambda$,
we necessarily have $|b_x^\Lambda(n)|\leq \kappa n^\nu$ for all $\Lambda \in \cP_0(\Gamma)$ and $\nu$-regular $\Gamma$.

\begin{defn}[Indistinguishability radius]\label{LTQO_def}
	Let $\Omega:\bR \to [0,\infty)$ be a non-increasing function. The \emph{indistinguishability radius} of $H_\Lambda$ at $x\in\Lambda$, denoted $r_x^\Omega(\Lambda)$, is the largest integer $r_x^\Omega(\Lambda)\leq \diam(\Lambda)$ such that for all integers $0 \leq k \leq n \leq r_x^\Omega(\Lambda)$ and all observables $A\in \cA_{b_x^\Lambda(k)}$,
	\begin{equation}\label{LTQO_length}
	\|P_{b_x^\Lambda(n)} A P_{b_x^\Lambda(n)} - \omega_\Lambda(A) P_{b_x^\Lambda(n)}\| \leq |b_x^\Lambda(k)| \|A\| \Omega(n-k).
	\end{equation}
\end{defn}	

Loosely speaking, a system is said to have the LTQO property if, for fixed $x\in\Gamma$,  the indistinguishability radius $r_x^\Omega(\Lambda)\to\infty$
as the system size increases and the distance of $x$ to the boundary of $\Lambda$ diverges. 

Several comments are in order. First, the set of indistinguishability radii is a property of the ground state space of the model and they obviously depend on the choice of the function $\Omega$. There is no {\em a priori} obvious optimal choice. Both $\Omega$ and the radii, $r_x^\Omega$, that appear in crucial estimates are derived from computing a good upper bound on the left hand side of \eqref{LTQO_length} for the system under consideration.  Typically,
one wants $\lim_{n\to\infty} \Omega(n) = 0$. The rate of this convergence is related to the vanishing dependence of local expectations on boundary conditions, see \eq{LTQO_length}. Therefore, as a second comment, we note that the indistinguishability radius $r_x^\Omega(\Lambda)$ depends not only on the volume $\Lambda$ but possibly also the choice of boundary conditions for the system. Finally, (\ref{LTQO_length}) shows that given any $A\in \cA_{b_x(k)}$ and $k<<n\leq r_x^\Omega(\Lambda)$, the matrix $P_{b_x^\Lambda(n)}A P_{b_x^\Lambda(n)}$ is approximately a multiple of $P_{b_x^\Lambda(n)}$. This property, generally referred to as LTQO, does not require that the model is defined by a frustration free interaction. In the case of frustration free models, however, we have the following proposition, first proved in \cite{michalakis:2013}, which will be used when we apply LTQO further on.

\begin{prop} \label{cor:LTQO} Let $H_\Lambda$ be a frustration-free Hamiltonian 
	$\Omega: \bR \to [0,\,\infty)$ be a non-increasing function, and $x\in\Lambda$. Then, for any $0<k \leq n \leq r_x^\Omega(\Lambda)$ and $A \in \caA_{b_x^\Lambda(k)}$ one has
	\begin{equation} \label{LTQO_cor_esta} 
	\left| \|AP_{b_x^\Lambda(n)}\| - \|AP_{\Lambda}\| \right| \leq \|A\| \sqrt{2|b_x^\Lambda(k)| \Omega(n-k)}.
	\end{equation}
\end{prop}

\begin{proof}
	Since $|a-b|^2 \leq |a^2-b^2|$ for any $a,b \geq 0$, first note that
	\begin{align}
	\left| \|AP_{b_x^\Lambda(n)}\| - \|AP_{\Lambda}\|\right|^2 
	\, & \leq  \,
	\left| \|AP_{b_x^\Lambda(n)}\|^2 - \|AP_{\Lambda}\|^2\right| \nonumber \\	
	\, & \leq \,
	\left|\|P_{b_x^\Lambda(n)}A^*AP_{b_x^\Lambda(n)}\| - \omega_\Lambda(A^*A)\right| 
	+ \left|\|P_{\Lambda}A^*AP_{\Lambda}\| - \omega_\Lambda(A^*A)\right|. \label{ApplyLTQOa}
	\end{align}
	The result follows from individually bounding the terms on the RHS of \eqref{ApplyLTQOa}. For any $k \leq n \leq r_x^\Omega(\Lambda)$, \eqref{LTQO_length} holds, and we can estimate the first term of \eqref{ApplyLTQOa} as follows:
	\begin{align}\label{LTQOBd1a}
	\left|\|P_{b_x^\Lambda(n)}A^*AP_{b_x^\Lambda(n)}\| - \omega_\Lambda(A^*A)\right| 
	\, & = \,  
	\left|\|P_{b_x^\Lambda(n)}A^*AP_{b_x^\Lambda(n)}\| - \omega_\Lambda(A^*A)\|P_{b_x^\Lambda(n)}\|\right| \nonumber \\
	\, &\leq \,
	\|P_{b_x^\Lambda(n)}A^*AP_{b_x^\Lambda(n)} - \omega_\Lambda(A^*A)P_{b_x^\Lambda(n)}\| \nonumber \\
	\, &\leq \, |b_x^\Lambda(k)|\|A\|^2\Omega(n-k).
	\end{align}
	For the second term of \eqref{ApplyLTQOa}, using the same argument as above we have
	\[ \left|\|P_{\Lambda}A^*AP_{\Lambda}\| - \omega_\Lambda(A^*A)\right| 
	\; \leq \; 
	\|P_{\Lambda}A^*AP_{\Lambda} - \omega_\Lambda(A^*A)P_{\Lambda}\|. \]
	To simplify notation, let $r=r_x^{\Omega}(\Lambda)$. It follows from the frustration-free property \eqref{ff_prop} that
	$P_{\Lambda} = P_{\Lambda}P_{b_x^\Lambda(r)}  = P_{b_x^\Lambda(r)}P_{\Lambda}$, and therefore
	\begin{align}
	\|P_{\Lambda}A^*AP_{\Lambda} - \omega(A^*A)P_{\Lambda}\|
	\, & \leq \,
	\| P_{b_x^\Lambda(r)}A^*A P_{b_x^\Lambda(r)} - \omega(A^*A) P_{b_x^\Lambda(r)}\| \nonumber \\
	\, & \leq \,
	|b_x(k)|\|A\|^2\Omega(r-k),
	\end{align}
	where we have again applied \eqref{LTQO_length}. Since $\Omega$ is non-increasing and $n \leq r$, the bound in (\ref{LTQO_cor_esta}) readily follows. 
\end{proof}

To illustrate the notion of indistinguishability length we now give  some examples of systems where there is a natural choice for $\Omega$ and for which good estimates of the indistinguishability radii can be given.

\begin{enumerate}
	
	\item LTQO by itself does not imply non-trivial topological order. Clearly, a system with a unique ground state that is not sensitive to boundary effects will have large indistinguishability radii but there will be no topological order of any kind. For example, consider a model with finite volume Hamiltonians $H_\Lambda$ that have a unique ground state given by a product vector $\bigotimes_{x\in\Lambda} \phi_x$, where $ \phi_x$ is independent of $\Lambda$. Then, one can take $\Omega\equiv 0$ and $r_x^\Omega(\Lambda)=\diam(\Lambda)$.
	
	\item Frustration-free spin chains with a unique translation-invariant matrix product ground state (e.g., the AKLT chain \cite{affleck:1988}) give an interesting class of examples that includes interesting cases of symmetry-protected topological order \cite{chen:2013,tasaki:2018,ogata:2019}. As is shown in Appendix \ref{app:MPS}, $\Omega$ can be taken to be of the form $\Omega(r)=Ce^{-r/\xi}$, where $\xi$ can be taken to be the correlation length 
	of the MPS state. 
	The indistinguishability radii depend on the boundary conditions as follows. If $\Lambda=[a,b]\subset\Ir$ is a finite interval with open boundary 
	conditions, one can show
	$$
	r^\Omega_x(\Lambda) \geq \min( |x-a|, |b-x| ) - c,
	$$
	for a suitable constant $c$, which depends on the model but not on $\Lambda$. For the model on a ring of $N$ sites, i.e., $\Lambda=\Ir/(N\Ir)$, with periodic or twisted periodic boundary conditions, we have, with the same $\Omega$,
	$$
	r^\Omega_x(\Lambda) \geq \floor{N/2}.
	$$
	
	\item The Toric Code model, the simplest example of the quantum double models introduced by Kitaev \cite{kitaev:2003,kitaev:2006}, was the system that inspired the original LTQO-type conditions introduced in \cite{bravyi:2010,bravyi:2011}. It can be defined on a square lattice ($\Gamma
	=\Ir^2$), with qubits on each edge ($\H_x=\Cx^2$ for all $x$ in the edge lattice, which is also a square lattice). The interactions are four-body terms associated with elementary squares (plaquettes) and stars (four edges meeting in a site). These interaction terms mutually commute, a situation often describe as a {\em commuting Hamiltonian}.
	
	For this model, one can take $\Omega$ to be the step function of the form
	$$
	\Omega(r)=\begin{cases} 2 & \mbox{ if } r\leq 2\\ 0 & \mbox{ if } r>2\end{cases}
	$$
	Again, precise estimates for the indistinguishability radii depend on the choice of boundary conditions. For the model defined on a torus
	$\Lambda = \Ir^2/(N_1\Ir\times N_2\Ir)$, one can show
	$$
	r^\Omega_x(\Lambda) \geq \min( N_1,N_2 ) - 2.
	$$
	In this case the indistinguishability radius, which does not depend on $x$, is essentially the code distance, meaning, the number of bits one has
	to modify to make an unrecoverable error. The case of general quantum double models was worked out in \cite{cui:2019}. 
	
	\item Levin-Wen models \cite{levin:2005} are another interesting class of two-dimensional models with commuting Hamiltonians (in the sense of the previous example). Their LTQO properties are similar to those of the Toric Code model and have been analyzed in \cite{qiu:2020}.
	
\end{enumerate}

It is easy to see that if a model has two or more ground states that can be distinguished by a local observable $A$, as is the case for the Ising model, the indistinguishability radius $r^\Omega_x(\Lambda)$ will be bounded or even vanish for any choice of $\Omega$ that tends to zero at infinity. In Section \ref{sec:TL_G}  we will consider spectral gap stability for models with discrete symmetry breaking. There we show that by using a symmetry restricted notion of the indistinguishability radius, which only requiries \eq{LTQO_length} for observables that satisfy a symmetry condition, one can also prove stability of the spectral gap in models with multiple (distinguishable) ground states.

%
%
%
%
%
%

\subsubsection{Decay of Interactions, Lieb-Robinson Bounds, and Quasi-locality} \label{sec:Fnorm}

Lieb-Robinson bounds for the Heisenberg dynamics of time (in)dependent interactions will play a key role in these stability results. In this section, we briefly review this topic. We first introduce the framework for time-independent interactions, and then discuss time-dependent interactions. We conclude with a statement of Lieb-Robinson bounds and a brief summary of quasi-local maps.

Lieb-Robinson bounds provide an upper bound for the speed of propagation of dynamically evolved observable through a quantum lattice system. This estimate is closely tied to the locality of the interaction in question, which we quantify using so-called $F$-functions. Given a countable metric space $(\Gamma,\,d)$, an \emph{F-function} $F: [0,\infty) \to (0,\infty)$ is a non-increasing function that satisfies the following two properties:
\begin{enumerate}
	\item[(i)] $F$ is uniformly-integrable, i.e.
	\begin{equation}
	\|F\| = \sup_{x\in\Gamma}\sum_{y\in\Gamma}F(d(x,y))<\infty.
	\end{equation} 
	\item[(ii)] $F$ has a finite convolution coefficient,
	\begin{equation}
	C_F := \sup_{x,\, y\in \Gamma} \sum_{z\in\Gamma} \frac{F(d(x,z))F(d(z,y))}{F(d(x,y))} <\infty.
	\end{equation}
\end{enumerate}
For a $\nu$-regular metric space $(\Gamma, d)$, any function of the form
\begin{equation} \label{poly-dec-F}
F(r) = \frac{1}{(1+r)^\zeta} \quad \text{for} \quad \zeta>\nu+1
\end{equation}
is an $F$-function with $C_F \leq 2^{\zeta+1}\|F\|$. If $\Gamma = \Ir^\nu$, one can take any $\zeta>\nu$. Given an $F$-function $F$ and any non-negative, non-decreasing, sub-additive function $g:[0,\infty) \to [0,\infty)$,  i.e. $g(r+s) \leq g(r)+g(s)$, the function
\begin{equation}
F_g(r) = e^{-g(r)}F(r)
\end{equation}
is also an $F$-function with $\|F_g\| \leq \|F\|$ and $C_{F_g}\leq C_F$. We refer to such functions as \emph{weighted $F$-functions}. The special case of
\be\label{stretch_exp_Ffun}
F(r) = \frac{e^{-ar^\theta}}{(1+r)^{\zeta}}, \;\;\text{with}\;\; a>0, \;\; 0<\theta\leq 1,\;\;\text{and}\;\; \zeta>\nu+1
\ee
are a particularly useful class of $F$-functions that will be frequently referenced in this work. 


We use $F$-functions to define decay classes of interactions in terms of $F$-norms. Given an $F$-function $F$ and an interaction $\Phi: \cP_0(\Gamma) \to \cA_{\Gamma}^{\rm loc}$, we say that $\Phi\in \cB_F$ if its \emph{$F$-norm} is finite, i.e.
\begin{equation}\label{ti_Fnorm}
\|\Phi\|_{F} : = \sup_{x,\, y\in \Gamma} \frac{1}{F(d(x,y))} 
\sum_{\substack{X\in \cP_0(\Gamma) \\ x, \, y\in X}} \|\Phi(X)\| < \infty.
\end{equation}
For example, if $\Phi$ is a uniformly bounded, finite range interaction on a $\nu$-regular metric space, one can easily check $\|\Phi\|_F<\infty$ for any exponentially decaying $F$-function $F(r) = e^{-ar}(1+r)^{-\zeta}$ with $\zeta>\nu+1$ and $a>0$.

From \eqref{ti_Fnorm}, it is useful to observe that for any $x,\, y\in \Gamma$,
\begin{equation}\label{decay_nb}
\sum_{\substack{X\in \cP_0(\Gamma) \\ x,\, y\in X}} \|\Phi(X)\| \leq \|\Phi\|_F F(d(x,y)),
\end{equation}
and in particular, $\|\Phi(X)\| \leq \|\Phi\|_F F(\diam(X))$ for any $X\in\cP_0(\Gamma)$ and $\Phi\in\cB_F$. 

In certain contexts, it becomes natural to consider the decay of an interaction term $\Phi(X)$ weighted against the size of its support, $|X|$. In this situation, the \emph{$m$-th moment $F$-norm} of the interaction is relevant. Given an integer $m\geq 0$ this is defined to be
\be\label{mom_Fnorm}
\|\Phi\|_{m,F} = \sup_{x,\, y\in \Gamma} \frac{1}{F(d(x,y))} 
\sum_{\substack{X\in \cP_0(\Gamma) \\ x, \, y\in X}} |X|^m\|\Phi(X)\|,
\ee
and we write $\Phi\in\cB_F^m$ when $\|\Phi\|_{m,F}<\infty$. With this notation, it is clear that $\|\Phi\|_{0,F} = \|\Phi\|_F$.

In our analysis, we also need to consider decay classes of continuous time-dependent interactions. Given an interval $I\subseteq \bR$ (possibly infinite),  we consider \emph{time-dependent interactions} $\Phi: \cP_0(\Gamma) \times I \to \cA_{\Gamma}^{\rm loc}$ for which
\begin{enumerate}
	\item[(i)] $\Phi(X, \, t)^* = \Phi(X,t) \in \cA_X$ for each $t\in I$ and $X\in\cP_0(\Gamma)$.
	\item[(ii)]$\Phi(X,t)$ is continuous in $t$ for all $X\in\cP_0(\Gamma)$. 
\end{enumerate}
We note that there is no ambiguity in the notion of continuity above (i.e. weak, strong, norm) since $\dim(\cH_X)<\infty$ for every $X\in\cP_0(\Gamma)$. The local Hamiltonians for a time-dependent interaction are defined analogously to the time-independent case, specifically
\begin{equation}
H_\Lambda(t) = \sum_{X\subseteq \Lambda} \Phi(X,t),\;\; \text{for all} \;\; \Lambda \in \cP_0(\Gamma)\;\; \text{and} \;\;t\in I.
\end{equation}
Furthermore, given an $F$-function $F$, we say that a time-dependent interaction $\Phi$ belongs to $\cB_F(I)$ if 
\begin{equation} \label{td_Fnorm}
\|\Phi\|_F(t) := \sup_{x,\, y\in \Gamma} \frac{1}{F(d(x,y))} \sum_{\substack{X\in \cP_0(\Gamma) \\ x, \, y\in X}} \|\Phi(X,t)\| < \infty
\end{equation}
and is locally bounded as a function of $t$.  In this case, $t \to \|\Phi\|_F (t)$ is measurable (as it is the supremum of a countable family of measurable functions) and hence locally integrable. As in the time-independent case, \eqref{td_Fnorm} implies that for all $t \in I$ and $x, y \in \Gamma$,
\begin{equation}
\sum_{\substack{X\in \cP_0(\Gamma) \\ x,\, y\in X}} \|\Phi(X,t)\| \leq \|\Phi\|_F(t) F(d(x,y)).
\end{equation}
The $m$-th moment $F$-norm, $\|\Phi\|_{m,F}(t)$, and decay class $\cB_F^m(I)$ for a time-dependent interaction are defined analogously, i.e. by substituting $\|\Phi(X,t)\|$ for $\|\Phi(X)\|$ in \eqref{mom_Fnorm}.

One can use the $F$-norm of an interaction to bound the speed of propagation of observables evolved under the Heisenberg dynamics. Given a finite volume $\Lambda\in \cP_0(\Gamma)$ and any $t,s\in I$, the Heisenberg dynamics $\tau_{t,s}^\Lambda: \cA_\Lambda \to \cA_\Lambda$ associated with a time-dependent interaction $\Phi$ is given by
\begin{equation} \label{fv_Heis_dyn}
\tau_{t,s}^\Lambda(A) = U_\Lambda(t,s)^* A U_\Lambda(t,s)
\end{equation}
where $U_\Lambda(t,s)$ is the solution to
\begin{equation}\label{td_diffeqn}
\frac{d}{dt} U_\Lambda(t,s) = -iH_\Lambda(t) U_\Lambda(t,s),\;\; U_\Lambda(s,s) = \idtyty.
\end{equation}
In the case of a time-independent interaction, the Heisenberg dynamics is defined analogously by replacing $H_\Lambda(t)$ with $H_\Lambda$ in \eqref{td_diffeqn} above.

In the situation that $\Phi\in \cB_F(I)$ for an $F$-function $F$ on $(\Gamma,\, d)$, one can prove the following well-known quasi-locality estimate on $\tau_{t,s}^\Lambda$. A proof of this result can be found, e.g. in \cite{nachtergaele:2019}.

\begin{thm}[Lieb-Robinson Bounds]\label{thm:LRBounds}
	Let $\Phi \in \cB_F(I)$. Then for any $\Lambda \in \cP_0(\Gamma)$ and $t,s\in I$, the Heisenberg dynamics $\tau_{t,s}^\Lambda$ satisfies the following bound: for any $A\in \cA_X$ and $B\in \cA_Y$ with $X\cap Y = \emptyset$ and $X\cup Y \subseteq \Lambda$, 
	\begin{equation}\label{LR_bound}
	\|[\tau_{t,s}^\Lambda(A), \, B]\| \leq \frac{2\|A\|\|B\|}{C_F}\left( e^{2C_F\int_{t_-}^{t_+}\|\Phi\|_F(r)dr}-1\right)\sum_{x\in X}\sum_{y\in Y} F(d(x,y)).
	\end{equation}
	where $t_- = \min\{t,s\}$ and $t_+ = \max\{t,s\}$.
\end{thm}

In the case that $\Phi$ is time-independent, Theorem~\ref{thm:LRBounds} holds with $|t-s|\|\Phi\|_F$ replacing the integral on the RHS of \eqref{LR_bound}. For a weighted $F$-function $F_g(r) = e^{-g(r)}F(r)$ and $\Phi\in\cB_{F_g}(I)$ such that $v_\Phi := 2C_{F_g}\sup_{t\in I}\|\Phi\|_{F_g}(t)<\infty$, one can use the uniform integrability of $F$ to show
\begin{equation}\label{velocity_bd}
\|[\tau_{t,s}^\Lambda(A), \, B]\| \leq \frac{2\|A\|\|B\|}{C_{F_g}}|X|\|F\|e^{v_\Phi|t-s|-g(d(X,Y))},
\end{equation}
which is decreasing in the distance between $X$ and $Y$. Here, the quantity $v_\Phi$ is known as the \emph{Lieb-Robinson} velocity, or, more appropriately, a bound on it. It is important to note that the norm bounds in both \eqref{LR_bound} and \eqref{velocity_bd} are independent of the volume $\Lambda$. The uniformity of these bounds in the system size plays a key role in our analysis.

We will also consider other maps with Lieb-Robinson type estimates, which we refer to as quasi-local maps. For any $\Lambda\in\cP_0(\Gamma)$, we say that a map $\cK_\Lambda : \cA_\Lambda \to \cA_\Lambda$ is \emph{quasi-local} if there is an integer $q\geq 0$, and non-increasing function $G:[0,\infty) \to (0,\infty)$, with $\lim_{x\to\infty} G(x) =0$, so that the norm bound
\begin{equation}\label{ql_map}
\|[\cK_\Lambda(A),B]\| \leq \|A\|\|B\||X|^q G(d(X,Y)),
\end{equation}
holds for any $A\in\cA_X$ and $B\in\cA_Y$ with $X,\, Y\subset \Lambda$. As in the case of the Heisenberg dynamics, it is often the case that there is a family of quasi-local maps $\{\cK_\Lambda \, : \, \Lambda \in \cP_0(\Gamma)\}$ for which $q$ and $G$ can be chosen uniform in $\Lambda$, a key property in applications. An example of such a map is the \emph{spectral flow} which we introduce in \ref{sec:spec_flow}. For a detailed, general analysis of quasi-local maps, see \cite[Section~5]{nachtergaele:2019}.

\subsection{Stability of the Spectral Gap} \label{sec:stab+sf}

The main focus of this work is spectral gap stability of a quantum spin model under local perturbations, for which the spectrum of the local Hamiltonian is the set of all eigenvalues for any fixed finite volume $\Lambda\in\cP_0(\Gamma)$. We will consider Hamiltonians that depend differentiably on a parameter $s$, and as such the eigenvalues can be written as continuous functions of the parameter. Without loss of generality, we assume the parameter
range is $s\in[0,1]$. Given such a differentiable family of Hamiltonians, $H_\Lambda(s)$, we consider a partition of the spectrum of $H_\Lambda(s)$ into two disjoint sets $\Sigma_1^\Lambda(s)$ and $ \Sigma_2^\Lambda(s)$ of the form 
\be \label{spec_partition}
\Sigma_1^\Lambda (s)= \spec(H_\Lambda(s))\cap I(s), \quad \Sigma_2^\Lambda(s) = \spec(H_\Lambda(s))\setminus \Sigma_1^\Lambda(s),
\ee
where $I(s)\subset \Rl$ is a closed interval with end points that depend smoothly on $s$.

As mentioned above, it is well known from perturbation theory (see \cite[Section 2.1]{kato:1995}) that the 
eigenvalues of the Hermitian matrix $H_\Lambda(s)$ are given by a family of continuous functions $\{\lambda^\Lambda_i(\cdot)\mid i=1,\ldots, \dim(\H_\Lambda)\}$ for which $\lambda^\Lambda_1(s)\leq \lambda^\Lambda_2(s) \leq \cdots$, for all $s\in [0,1]$. We are mainly interested the behavior of the gap above the ground state in $\spec(H_\Lambda(s))$, which we define as follows. Choosing
the partition of the form \eq{spec_partition} given by
\be\label{spec_sets}
\begin{array}{ll}
	\Sigma_1^\Lambda(s) &= \left\{\lambda^\Lambda_i(s)\in \spec(H_\Lambda(s))  : \lambda^\Lambda_i(0) = \lambda^\Lambda_1(0) \right\}\\
	\Sigma_2^\Lambda(s) &= \left\{\lambda^\Lambda_i(s) \in \spec(H_\Lambda(s)) : \lambda^\Lambda_i(0) > \lambda^\Lambda_1(0)  \right\},
\end{array}
\ee
we define the {\em ground state gap}, $\gap(H_\Lambda(s)) $ by:
\begin{equation} \label{fv_pert_gap}
\gap(H_\Lambda(s)) = \text{dist}(\Sigma_1^\Lambda(s), \, \Sigma_2^\Lambda(s)),
\end{equation}
where $\text{dist}(X, Y)  = \inf\{|x-y| \, : \, x\in X, \, y\in Y \}$ for any two non-empty sets $X, \, Y\subset \bR$. In the cases of particular interest to us, $H_\Lambda(0)$ is a finite-volume Hamiltonian of a frustration free interaction, which by definition has $\lambda_1(0) =\inf\spec H_\Lambda(0)=0$ and hence $\Sigma_1(0) = \{0\}$. From the continuity of the eigenvalues, it is clear that for any fixed $\Lambda$ and $0 < \gamma < \gap(H_\Lambda)$, 
\begin{equation} \label{def:s_lam_gam}
s^\Lambda_\gamma := \sup\{s'\in [0,1] \, : \, \gap(H_\Lambda(s))\geq \gamma \; \text{  for all } \; 0\leq s \leq s' \}>0.
\end{equation}
Our goal will be to obtain a useful lower bound for $s^\Lambda_\gamma$. Without loss of generality we may assume that $s^\Lambda_\gamma<1$
(as $s^\Lambda_\gamma=1$ is a useful lower bound). A visualization of $s^\Lambda_\gamma$ is given in Figure \ref{fig:eigsplitgap}. 

A sequence of finite volumes $\Lambda_n\subset \Gamma$ is said to be {\em absorbing} (for $\Gamma$), denoted $\Lambda_n \to \Gamma$, if for all $x\in \Gamma$, there exists an $n$ such that $x\in\Lambda_m$ for all $m\geq n$. We often also assume the sequence of finite volumes is 
{\em increasing}, i.e., $\Lambda_n\subset \Lambda_{n+1}$ for all $n$, and denote by $\Lambda_n\uparrow\Gamma$ a sequence of increasing and absorbing volumes.

We say that a frustration-free interaction is \emph{gapped}, if there exists a sequence of finite volumes $\Lambda_n \uparrow \Gamma$ 
such that
\begin{equation}\label{unif_gap}
\gamma_0 = \inf_{n\geq 1}\text{gap}(H_{\Lambda_n}) >0.
\end{equation}
In the situation that there is a sequence of uniformly gapped unperturbed Hamiltonians in the sense of \eqref{unif_gap}, we say that the spectral gap is \emph{stable} if for any $0 < \gamma < \gamma_0$
\begin{equation}\label{stable_gap}
s_\gamma := \inf_{n\geq 1} s^{\Lambda_n}_\gamma > 0,
\end{equation}
that is, if there is $s_\gamma>0$ such that $\inf_{n}\gap(H_{\Lambda_n}(s)) \geq \gamma$ for all $0\leq s \leq s_\gamma$.

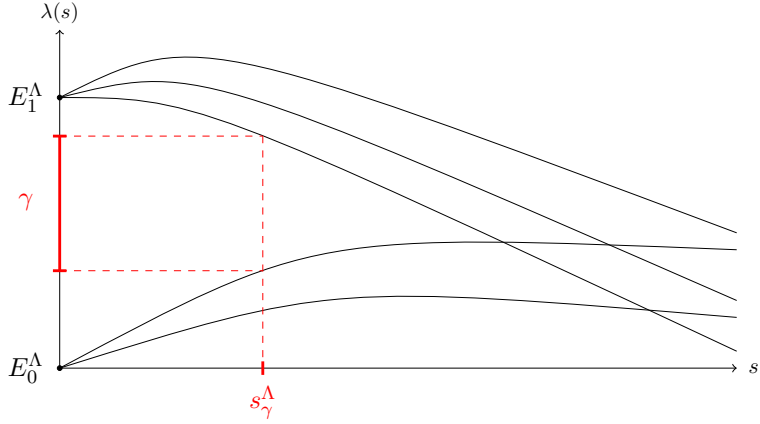
\begin{figure} \label{fig:eigsplitgap}
	\scalebox{.9}{
		\begin{tikzpicture}
		\draw[->] (-5,0) -- (5,0) node at (5.25,0) {\scalebox{.85}{$s$} };
		\draw[->] (-5,0) -- (-5,5) node at (-5,5.25) {\scalebox{.8}{$\lambda(s)$}};
		\draw [fill=black] (-5,0) circle (1pt);
		\draw [domain=-5:5] (-5,0) .. controls (-1.25,2) ..(5,1.75);
		\draw [domain=-5:5] (-5,0) .. controls (-1,1.25) ..(5,.75);
		\draw [fill=black](-5,4) circle (1pt);
		\draw  (-5,4) .. controls (-3,5) .. (5,2);
		\draw  (-5,4) .. controls (-3,4.5) .. (5,1);
		\draw  (-5,4) .. controls (-3,4) .. (5,.25);
		\node at (-5.5,0) {$ E_0^{\Lambda}$};
		\node at (-5.5,4) {$ E_1^{\Lambda}$};
		\draw[-, very thick, red] (-5,1.44) -- (-5,3.43);
		\node[red] at (-5.5,2.45) {$ \gamma$};
		\draw[red, dashed] (-5,1.44) -- (-2,1.44);
		\draw[red, dashed] (-5,3.43) -- (-2,3.43);
		\draw[red, very thick](-5.1,1.44) --(-4.9,1.44);
		\draw[red, very thick](-5.1,3.43) --(-4.9,3.43);
		\draw[red, dashed] (-2,3.43) -- (-2,0);
		\draw[red, very thick](-2,.1) --(-2,-.1);
		\node[red] at (-2,-.5) {$ s^{\Lambda}_\gamma$};
		\end{tikzpicture}}
	\caption{The spectral gap of $H_\Lambda(s)$ above the ground state energy (allowing for eigenvalue splitting) is at least $\gamma$ for all $0\leq s\leq s^{\Lambda}_\gamma$.}
\end{figure}

We analyze the stability of the ground state gap in the presence of small, local perturbations. Given two time-independent interactions $\eta, \Phi : \cP_0(\Gamma) \to \cA_{\Gamma}^{\rm loc}$, we consider local Hamiltonians of the form
\be\label{pert_hams}
H_\Lambda(s) = H_\Lambda + s V_{\Lambda^p},
\ee
where
$$
H_\Lambda = \sum_{X\subseteq\Lambda} \eta(X), \quad V_{\Lambda^p} 
= \sum_{\substack{X\subseteq \Lambda \\ X\cap \Lambda^p \neq \emptyset}} \Phi(X). 
$$
Here, $\eta$ is the background (or initial) interaction, $\Phi$ is the perturbation, and $\Lambda^p \subseteq \Lambda$ is the \emph{perturbation region}. Any subset of $\Lambda$ may be chosen as the perturbation region. In our application, the perturbation region will consist of all points $x\in\Lambda$ with a sufficient indistinguishability radius. This will be described in more detail in Section~\ref{sec:LTQO2}. The most traditional choice of perturbation region is $\Lambda^p = \Lambda$, for which $V_\Lambda\in\cA_\Lambda$ is a local Hamiltonian of the form defined in \eqref{ti_ham}.
In Section~\ref{sec:uniform_sequences}, we provide sufficient conditions on the unperturbed interaction $\eta$ so that for any $\Phi \in \cB_F$ with $F$ as in \eqref{stretch_exp_Ffun}, there are two sequences of finite volumes $\Lambda_n^p\subseteq\Lambda_n$ with $\Lambda_n^p\uparrow \Gamma$ so that \eqref{stable_gap} is satisfied.

\subsection{The Spectral Flow}\label{sec:spec_flow}
Much like the Heisenberg dynamics, see (\ref{fv_Heis_dyn}), the spectral flow is a family of automorphisms of the observable algebra $\mathcal{A}_{\Lambda}$. Consider the family of
Hamiltonians $\{ H_{\Lambda}(s) \}_{s \in [0,1]}$  given by (\ref{pert_hams}). For each $0 \leq s \leq 1$ and all $t\in\bR$, denote by $\tau_t^{(s)}:\cA_\Lambda \to \cA_\Lambda$ the Heisenberg dynamics associated with $H_\Lambda(s)$, i.e.
\begin{equation} \label{stab_heis_dyn}
\tau_t^{(s)}(A) := e^{itH_{\Lambda}(s)} A e^{-itH_{\Lambda}(s)}.
\end{equation}
To define the spectral flow, we first introduce its generator. For any $\xi >0$, define $D^\xi: [0, 1] \to \mathcal{A}_{\Lambda}$ by
\begin{equation} \label{gen_spec_flow}
D^\xi(s) = \int_{\mathbb{R}} \tau_t^{(s)}( V_{\Lambda^p}) \, W_{\xi}(t) \, dt \quad \mbox{for all } 0 \leq s \leq 1,
\end{equation}
where $V_{\Lambda^p}$ is the perturbation, see (\ref{pert_hams}), and $W_{\xi} \in L^1( \mathbb{R})$ is 
the real-valued weight function defined e.g. in \cite{nachtergaele:2019}[Section VI.B]. It is straight-forward to check that for any $\xi>0$, $D^\xi$ is pointwise self-adjoint (i.e. $D^\xi(s)^* = D^\xi(s)$ for all $s \in [0,1]$) and continuous in $s$. As such, there is a unique family of unitaries given by the solutions of
\begin{equation} \label{spec-flow-unis}
\frac{d}{ds} U^\xi(s) = - i D^\xi(s) U^\xi(s) \quad \mbox{with } U^\xi(0) = \idty.
\end{equation}
In terms of these unitaries, a family of automorphisms of $\mathcal{A}_{\Lambda}$ is defined by setting 
\begin{equation} \label{spec-flow-auto}
\alpha_s^\xi(A) = U^\xi(s)^* A U^\xi(s) \quad \mbox{for all } A \in \mathcal{A}_{\Lambda} \mbox{ and } 0 \leq s \leq 1 \, .
\end{equation} 
Given a choice of $\xi$, we refer to the family of automorphisms $\{ \alpha_s^\xi \}_{s \in [0,1]}$ as the 
\emph{spectral flow automorphisms}. This is due to the following property that these automorphisms satisfy:
Let $\spec(H_\Lambda(s))= \Sigma_1^\Lambda(s)\cup \Sigma_2^\Lambda(s)$ be as in \eqref{spec_sets} 
and for all $0\leq s \leq 1$ denote by $P(s)$ the spectral projection associated to 
$H_{\Lambda}(s)$ onto $\Sigma_1^{\Lambda}(s)$. Fix $0< \gamma < \gap(H_\Lambda(0))$. It is proven, e.g. in \cite[Theorem~6.3]{nachtergaele:2019}, that for any $0< \xi \leq \gamma$ the spectral flow automorphisms satisfy
\begin{equation} \label{spec_flow_proj}
\alpha_s^{\xi}(P(s)) = P(0) \quad \mbox{for all } 0 \leq s \leq s^{\Lambda}_\gamma.
\end{equation} 
As such, properties of the ground state projections of $H_\Lambda(0)$ extend to the spectral projection of $\alpha_s(H_\Lambda(s))$ associated with $\Sigma_1(s)$.

To ease notation, we will work with the fixed family of spectral flows determined by the choice $\xi = \gamma$, and denote it simply by $\{ \alpha_s \}_{s \in [0,1]}$. An important point, to which we will return in Section~\ref{sec:Step1}, is that if the family of Hamiltonians $\{ H_{\Lambda}(s) \}_{s \in [0,1]}$ given by (\ref{pert_hams}) has sufficiently decay, (e.g. $\Phi \in \mathcal{B}_F$ with $F$ as in \eqref{stretch_exp_Ffun}), then these spectral flow automorphisms satisfy an explicit quasi-locality estimate of the form (\ref{ql_map}). At the heart of the spectral stability argument are three crucial properties of the spectral flow.

\begin{claim}[Properties of the Spectral Flow]
	The spectral flow has the following properties:
	\begin{enumerate}
		\item The spectral flow is a family of automorphisms implemented by unitaries, see (\ref{spec-flow-unis}). As such, $\spec(\alpha_s(H))=\spec(H)$ for any Hamiltonian $H\in\cA_\Lambda$ and so one may analyze $\alpha_s(H)$ to establish spectral gap estimates for $H$.
		\item The spectral flow maps the spectral projection corresponding to the perturbed system
		back to the spectral projection corresponding to the unperturbed systems, see (\ref{spec_flow_proj}).
		\item Given sufficient decay of the perturbation $\Phi$, the spectral flow $\alpha_s$ from \eqref{spec-flow-auto} is quasi-local.
	\end{enumerate}
\end{claim}

\noindent The properties described in this claim, as well as others, were proven in detail in Section 6 of \cite{nachtergaele:2019}. 

\subsection{Anchored Interactions}  \label{sec:balled-int}

\subsubsection{Anchored Interactions} \label{sec:def-ball-int} 

For combinatorial reasons, it becomes cumbersome to work with interactions and perturbations as they are defined in Section~\ref{sec:QSS}. For this reason, we find it convenient to work with \emph{anchored interactions}. We define this notion here, and provide one method for transforming an interaction into an anchored interaction and vice-versa. The vital property of the anchoring procedure we introduce is that it preserves the decay properties of the original interaction, assuming it decays sufficiently fast. The procedure we discuss also has the convenient property that it preserves the local ground state spaces for balls. Anchored interactions can be time-dependent in complete analogy with the definitions and results for time-dependent interactions defined on arbitrary finite subsets.

\begin{defn} \label{def:ball-int}  Given a countable metric space $(\Gamma,d)$, an algebra of local observables $\cA_\Gamma^{\rm loc}$, and a subset $\Lambda \subseteq \Gamma$, we say that mapping $\Phi:\Lambda \times \bZ_{\geq 0}\to \cA_{\Lambda}^{\rm loc}$ is a \emph{anchored interaction on $\Lambda$} if $\Phi(x,n) = \Phi(x,n)^*\in \cA_{b_x^\Lambda(n)}$ for all $(x,n)\in \Lambda \times \bZ_{\geq 0}$.
\end{defn}

For fixed $x\in\Lambda$, we often use $\Vert \Phi(x,n)\Vert$, as a function of $n$, to express the decay of the interaction strength with distance. In that situation it is natural to require the following property:
\be\label{support_property}
\Phi(x,n)\neq 0 \implies  \mbox{ there is a pair of sites } y,z\in b_x^{\Lambda}(n) \mbox{ with } d(y,z) > n-1.
\ee
We will show that for a given Hamiltonian
we can always find an anchored interaction with this property.
Note that we only require $\diam(b_x^{\Lambda}(n)) > n-1$ and not $\supp(\Phi(x,n)) = b_x^\Lambda(n)$.
We say the term $\Phi(x,n)$ is \emph{anchored} at the site $x$, hence the terminology. Depending on $\Lambda$, it is possible that $b_x^\Lambda(n)=b_y^\Lambda(m)$ for two distinct sites $x, y\in\Lambda$. For an anchored interaction the anchoring site can hold significance, and so we allow for the possibility that $\Phi(x,n)\neq\Phi(y,m)$. This is not possible for an \emph{interaction} as defined in Section~\ref{sec:QSS} as this is a function of $X\in\cP_0(\Gamma)$. We impose the second condition on an anchored interaction for two reasons. First, it justifies associating the ball $b_x^\Lambda(n)$ to the interaction term $\Phi(x,n)$. Second, if $\Lambda$ is finite, then $\Phi(x,n)=0$ for all $n> \diam(\Lambda)$ and so $\Phi$ is nonzero for only a finite number of pairs $(x,n)\in \Lambda\times \bZ$.

Given any finite volume $\Lambda\subseteq\Gamma$, and an interaction $\Phi: \cP_0(\Gamma) \to \cA_{\Gamma}^{\rm loc}$ it is always possible rewrite the local Hamiltonian $H_\Lambda$, see \eqref{ti_ham}, as
\begin{equation}\label{fv_ham}
H_\Lambda = \sum_{x\in \Lambda}\sum_{n\geq 0} \Phi_\Lambda(x,n)
\end{equation}
where $\Phi_\Lambda : \Lambda \times \bZ_{\geq 0} \to \cA_\Lambda$ is an anchored interaction on $\Lambda$. In the next section we introduce one procedure for transforming a general interaction $\Phi$ into an equivalent anchored interaction. This is not the only procedure
one could use. For the results on stability, one only needs that the resulting anchored interaction satisfies an \emph{anchored $F$-norm} similar to that in Proposition~\ref{prop:balled_Fnorm} below.

One may also consider time-dependent anchored interactions, which are defined as follows:

\begin{defn} \label{def:ball-int-td}
	Given a countable metric space $(\Gamma,d)$, an algebra of local observables $\cA_\Gamma^{\rm loc}$, an interval $I\subseteq \bR$ (possibly infinite), and a subset $\Lambda \subseteq \Gamma$, we say that mapping $\Phi:\Lambda \times \bZ_{\geq 0}\times I\to \cA_{\Lambda}^{\rm loc}$ is a \emph{anchored interaction on $\Lambda$} if the following three conditions hold:
	\begin{enumerate}
		\item $\Phi(x,n,t)^*=\Phi(x,n,t) \in \cA_{b_x^\Lambda(n)}$ for all triplets $(x,n,t)$.
		\item For all $(x,n)\in \Lambda \times \bZ_{\geq 0}$, the mapping $t\mapsto \Phi(x,n,t)$ is continuous.
	\end{enumerate}
\end{defn} 
For time-dependent anchored interactions, it may again be convenient to require \eqref{support_property}, i.e. if 
\be\label{dep_support_property}
\Phi(x,n,t)\neq 0 \;\; \implies \;\; d(y,z)>n-1 \;\; \text{ for some } \;\;  y,z\in b_x^\Lambda(n).
\ee

\subsubsection{An Anchoring Procedure} \label{sec:ball-proc} As our main interest is perturbed Hamiltonians of the form \eqref{pert_hams}, we define a anchoring process that will respect Hamiltonians defined in terms of a perturbation region. Let $\Phi:\cP_0(\Gamma) \to \cA_\Gamma^{\rm loc}$ be an interaction and fix (possibly infinite) volumes $\Lambda^p \subseteq \Lambda \subseteq \Gamma$. With respect to these volumes, we define an anchored interaction on $\Lambda$, denoted $\Phi_{\Lambda^p}:\Lambda \times \bZ_{\geq 0} \to \cA_{\Lambda}^{\rm loc}$ so that
\[
\Phi_{\Lambda^p}(x,n)=0 \;\; \text{  if  }\;\; x\in\Lambda \setminus \Lambda^p.
\]
We further show that if $\Lambda$ is finite, the local Hamiltonian $H_{\Lambda,\Lambda^p}$ defined by
\be \label{pert_region_hams}
H_{\Lambda,\Lambda^p} := \sum_{\substack{X\subseteq \Lambda \\ X\cap \Lambda^p \neq \emptyset}} \Phi(X)
\ee
can be rewritten in terms of this anchored interaction as
\be \label{balled_fv_ham}
H_{\Lambda,\Lambda^p} = \sum_{\substack{x\in\Lambda^p \\n\geq 0}} \Phi_{\Lambda^p}(x,n).
\ee
Here, note that \eqref{pert_region_hams} is of the same form as the perturbation in \eqref{pert_hams}.  And, of course, \eqref{pert_region_hams} agrees with \eqref{ti_ham} if $\Lambda^p = \Lambda$, an essential requirement.

We will also use anchored versions of the a priori (unperturbed) Hamiltonian which is given by a finite-range, uniformly bounded, frustration-free interaction $\eta$.
The anchoring procedure we introduce benefits from preserving the finite-range and uniform boundedness conditions as well as leaving the local ground state spaces invariant. 

To define our anchoring procedure, let us first denote by $\cS(\Lambda^p)$ the set of all finite volumes that intersect the perturbation region, i.e.
\[\cS(\Lambda^p)=\{X\in \cP_0(\Lambda) \, : \, X\cap\Lambda^p \neq \emptyset\} .\]
We further partition this set by $\cS(\Lambda^p) = \bigcup_{n\geq 0} \cS_n(\Lambda^p)$ where, for all $n\geq 1$,
\be
\cS_n(\Lambda^p) = \{ X\in \cS(\Lambda^p) \, : \, \exists \, x\in\Lambda^p \mbox{ s.t. } \, X\subset b_x^\Lambda(n) \mbox{ and } \, \forall \, x\in \Lambda^p, X\not\subset b^\Lambda_x(n-1) \} .
\ee
Setting $\cS_0 = \{ \{x\} \mid x\in \Lambda^p\}$ for $n=0$, it is clear that $\{ \cS_n(\Lambda^p)\mid n\geq 0\}$ is a partition of $\cS(\Lambda^p)$. Therefore we define the {\em radius of $X$} by $r(X) = n$ if $X\in \cS_n(\Lambda^p)$, and the {\em multiplicity of $X$} as:
\be\label{eq:m(X)}
m(X) = \# \left\{ x\in\Lambda^p | X\subset b^\Lambda_x(r(X)) \right\}.
\ee
Note that $m(X) \geq 1$ for all $X\in\cS(\Lambda^p)$, that $m(X)$ is always finite even if $\Lambda^p$ is infinite,
and that $r(X) -1 <\diam(X)\leq 2r(X)$. The radius and multiplicity, in general, depend on $\Lambda$ and $\Lambda^p$. This, however, will not play an important role in our analysis.

Then, for $x\in\Lambda^p$, we define $\Phi_{\Lambda^p}(x,n)\in\cA_{b_x^\Lambda(n)}$ by
\begin{equation} \label{BalledInt}
\Phi_{\Lambda^p}(x,n) = \sum_{\substack{X \in \cS_n(\Lambda^p): \\X\subset b^\Lambda_x(n)}} \frac{1}{m(X)} \Phi(X),
\end{equation}
with the convention that $\Phi_{\Lambda^p}(x,n)=0$ for empty sums. We set $\Phi_{\Lambda^p}(x,n)=0$ for all $x\in \Lambda\setminus\Lambda^p$.
It is straightforward to see that $\Phi_{\Lambda^p}$ is an anchored interaction in the sense of Definition \ref{def:ball-int}. 
Moreover, (\ref{support_property}) holds. In fact, when $\Phi_{\Lambda^p}(x,n)\neq 0$, there is $X \in\cS_n(\Lambda^p)$ 
with $X \subset b_x^\Lambda(n)$ and $\Phi(X)\neq 0$. For this $X\in\cS_n(\Lambda^p)$, one sees that there exists $y,z\in X\subset b_x^\Lambda(n)$ with $d(y,z)>n-1$.

We now show that using this anchoring procedure, any finite-volume Hamiltonian $H_{\Lambda,\Lambda^p}$ of the form given in \eqref{pert_region_hams} for $\Lambda^p\subseteq\Lambda\in\cP_0(\Gamma)$ can be rewritten as described in \eqref{balled_fv_ham}. Given the definitions of $\cS(\Lambda^p)$ and $\cS_n(\Lambda^p)$, and \eq{BalledInt},  the Hamiltonian $H_{\Lambda,\Lambda^p}$ may be rewritten as
\be \label{int_to_ball}
H_{\Lambda,\Lambda^p} = \sum_{X\in \cS(\Lambda^p)} \Phi(X) = \sum_{n \geq 0} \sum_{X\in \cS_n(\Lambda^p)} \Phi(X).
\ee
Using the definition of $m(X)$, we have
\be\label{balled_equiv}
\sum_{X\in \cS_n(\Lambda^p)} \Phi(X)
= \sum_{X\in \cS_n(\Lambda^p)}\frac{1}{m(X)} \sum_{\substack{x \in \Lambda^p\\ X\subset b^\Lambda_x(n)}}\Phi(X) = \sum_{x \in \Lambda^p} \Phi_{\Lambda^p}(x,n),
\ee
Combining \eq{int_to_ball} and \eq{balled_equiv} we obtain the desired property \eq{balled_fv_ham}.

Next, we analyze the effects of our anchoring procedure on the initial interaction and its associated local Hamiltonians, see \eqref{pert_hams} and the subsequent discussion. Consider a finite-range, uniformly bounded, frustration-free interaction $\eta$ and define the anchored interaction $\eta_\Lambda$ as in \eq{BalledInt} for any (possibly infinite) $\Lambda \subset \Gamma$. Here, we note that in \eqref{BalledInt} one uses $\Lambda^p=\Lambda$ as the background interaction is defined extensively.

The range of $\eta$ is defined as the smallest $R\geq 0$ such that $\eta(X)=0$ for all $X$ with $\diam(X)>R$. For any anchored interaction 
$\Phi_{\Lambda^p}$, we define the \emph{maximal radius} as the smallest integer $R_{\Lambda^p}\geq 0$ such that $\Phi_{\Lambda^p}(x,n)=0$, for all 
$x\in\Lambda^p$ and $n>R_{\Lambda^p}$. Since $\eta_{\Lambda}$ is obtained from $\eta$ by the anchoring procedure defined above, we have the following relationship:
$R_{\Lambda} -1\leq R \leq 2 R_{\Lambda}$. Moreover, the anchored interaction preserves the uniform norm. Namely, letting $\|\eta\|:=\sup_{X\in\cP_0(\Gamma)}\|\eta(X)\|$ denote the uniform bound on the original interaction, for any $x\in\Lambda$ and $0\leq n \leq R_\Lambda$, \[\|\eta_\Lambda(x,n)\| \leq \|\eta\|2^{|b_x(n)|} \leq  \|\eta\|2^{\kappa R_\Lambda^\nu}\]
where we apply $\nu$-regularity to bound the number of terms in the summation from \eqref{BalledInt}. As a consequence of these properties, we may group the anchored terms into a single site-anchored operator
\begin{equation}\label{eq:anchored-background}
h_x : = \sum_{n=0}^{R_\Lambda}\eta_\Lambda(x,n) \in \cA_{b_x(R_\Lambda)} \quad \text{with} \quad \|h\|:=\sup_{x\in\Lambda}\|h_x\|<\infty.
\end{equation}
In the case that $\Lambda$ is finite, the associated local Hamiltonian is equal to $H_\Lambda = \sum_{x\in\Lambda}h_x$, which is nontrivially used in our analysis in Sections~\ref{sec:Step1}-\ref{sec:uniform_sequences}. 

Given the importance of ground state projections in this work, one may wonder if our anchoring procedure effects the ground state spaces of the local Hamiltonians. The natural definitions of the local Hamiltonians $H_{\Lambda_0}$, $\Lambda_0\in\cP_0(\Lambda)$, are as follows:
\be
H^\eta_{\Lambda_0} = \sum_{X\subset \Lambda_0} \eta(X), \quad H^{\eta_{\Lambda}}_{\Lambda_0} = \sum_{\substack{x\in \Lambda,n\geq 0\\ 
		b^\Lambda_x(n) \subset \Lambda_0}}\eta_{\Lambda}(x,n).
\ee
For the anchoring procedure introduced, we claim that 
\be\label{equal_gss}\ker H^\eta_{b^\Lambda_y(m)} = \ker H^{\eta_{\Lambda}}_{b^\Lambda_y(m)}
\ee for all $y\in\Lambda, m\geq 0$. To establish this, it suffices to notice the following two properties, which are easy to verify from the above definitions:
\begin{enumerate}
	\item[(i)] if $X\subset b^\Lambda_y(m)$ and $\eta(X) \neq 0$, then 
	\[ H^{\eta_{\Lambda}}_{b^\Lambda_y(m)}\geq \sum_{n\leq m}\eta_{\Lambda}(y,n) \geq m(X)^{-1}\eta(X),\]
	\item[(ii)] $H^{\eta_{\Lambda}}_{b^\Lambda_y(m)}\leq  H^\eta_{b^\Lambda_y(m)}$ for all $y\in \Lambda$ and $m\geq 0$.
\end{enumerate} 

While the ground state equality in \eqref{equal_gss} is convenient, it is not strictly necessary for the arguments in this work as we will always work with the ground state projections for the local Hamiltonians defined by the original (i.e. unanchored) interaction. In particular, the results of this work still hold for any anchoring procedure as long as for any $b^\Lambda_y(m) \in \cP_0(\Lambda)$,
\[
\ker(H^\eta_{b^\Lambda_y(m)}) \subseteq \ker(H^{\eta_\Lambda}_{b^\Lambda_y(m)}).
\]

\subsubsection{Decay Properties of Anchored Interactions}
As we mentioned before, the purpose of working with an anchored interaction is to simplify certain combinatorial arguments. If an anchored interaction is defined with respect to an interaction $\Phi:\cP_0(\Gamma) \to \cA_{\Gamma}^{\rm loc}$, e.g. as defined in \eqref{BalledInt}, we will want the anchored interaction to inherent properties possessed by $\Phi$, including decay properties. Therefore, we introduce an anchored version of the $F$-norm from \eqref{ti_Fnorm}, from which we will be able to verify similar decay estimates.

\begin{defn}\label{def:Balled_FNorm}
	Let $(\Gamma,d)$ be a countable metric space, and $\Phi_\Lambda:\Lambda \times \bZ_{\geq 0} \to \cA_{\Gamma}^{\rm loc}$ be an anchored interaction. For any integer $m\geq 0$, we say that $\Phi_\Lambda$ has a finite \emph{$m$-th moment anchored $F$-norm} with respect to an $F$-function $F$ and write $\Phi_\Lambda\in\cB_F^m$ if
	\be\label{balled_F_norm}
	\|\Phi_\Lambda\|_{m,F}:=\sup_{x,y\in\Lambda}\frac{1}{F(d(x,y))}\sum_{\substack{n\geq 0, \, z\in\Lambda: \\ x,y\in b_z^\Lambda(n)}}|b_z^\Lambda(n)|^m\|\Phi_\Lambda(z,n)\| < \infty.
	\ee 
	For $m=0$, we simply say that $\Phi_\Lambda$ has a finite anchored $F$-norm and denote this by $\|\Phi_\Lambda\|_F$.
\end{defn}

We note that while we use the same notation, $\cB_F^m$, for the decay classes of interactions and anchored interactions, the correct interpretation should be clear from context and so there should be no confusion. In Proposition~\ref{prop:balled_Fnorm} below, we show that given an interaction $\Phi\in\cB_F^m$ and the anchored interaction $\Phi_{\Lambda^p}$ as in \eqref{BalledInt}, then $\Phi_{\Lambda^p}$ will have a finite $m$-th anchored $F$-norm as long as $F$ decays sufficiently fast.

\begin{prop}\label{prop:balled_Fnorm}
	Let $\Gamma$ be a $\nu$-regular metric space, $\Phi\in \cB_F^m$ for an $F$-function $F$, and $\Lambda^p \subseteq \Lambda\subseteq\Gamma$. Fix any integer $m\geq 0$. Then, for any $x, \, y\in \Lambda$ the anchored interaction defined in \eqref{BalledInt} satisfies
	\begin{equation}\label{balled_Fnorm}
\sum_{\substack{n\geq 0, \, z\in\Lambda^p: \\ x,y\in b_z^\Lambda(n)}} |b_z^\Lambda(n)|^m \|\Phi_{\Lambda^p}(z,n)\|
	\leq
	2^\nu\kappa^{m+2}\|\Phi\|_{F} \sum_{n \geq \lceil d(x,y)/2\rceil} n^{(m+2)\nu}F(n-1).
	\end{equation}
\end{prop}
It is important to note that while the anchored interaction depends on $\Lambda$ and $\Lambda^p$, the bound on the RHS of \eqref{balled_Fnorm} is independent of both volumes. In the situation that there is an $F$-function $\tilde{F}$ for which
\[
\sum_{n \geq \lceil r/2\rceil} n^{(m+2)\nu}F(n-1) \leq \tilde{F}(r) \;\; \text{for}\;\;r\in[0,\infty),
\]
the above result shows that $\|\Phi_{\Lambda^p}\|_{m,\tilde{F}}<\infty$.

\begin{proof} Fix $x, \, y\in \Lambda$. Then a simple bound using \eqref{BalledInt} and $\nu$-regularity gives
	\begin{equation} \label{balled_first_bd}
	\sum_{n\geq 0} 
	\sum_{ \substack{ z\in \Lambda^p \\ x,y\in b_z^\Lambda(n)}} 
	|b_z^\Lambda(n)|^m\|\Phi_{\Lambda^p}(z,n)\|
	\leq
	\kappa^m\sum_{n \geq 0} n^{m\nu}
	\sum_{\substack{z\in\Lambda_p \\ x,y\in b_z^\Lambda(n)}}
	\sum_{\substack{X\in \cS_n(\Lambda^p) \\ X\subseteq b_z^\Lambda(n)}}\frac{1}{m(X)}
	\|\Phi(X)\|.
	\end{equation}
	Note that since $x,y\in b_z(n)$,
	\[
	d(x,y) \leq d(x,z) + d(y,z) \leq 2n
	\]
	and so the sum over $n$ can be optimized to $n\geq \lceil d(x,y)/2 \rceil$. 
	
	Consider the (inner) double summation from the RHS of \eqref{balled_first_bd}, for which the following change of variables is valid:
	\[
	\sum_{\substack{z\in\Lambda_p \\ x,y\in b_z^\Lambda(n)}}
	\sum_{\substack{X\in \cS_n(\Lambda^p) \\ X\subseteq b_z^\Lambda(n)}} 
	=
	\sum_{\substack{z\in\Lambda_p \\ x,y\in b_z^\Lambda(n)}}
	\sum_{X\in \cS_n(\Lambda^p)} {\rm Ind}_{X\subseteq b_z^\Lambda(n)} 
	=
	\sum_{X\in \cS_n(\Lambda^p)} \sum_{\substack{z\in\Lambda_p: \\ x,y\in b_z^\Lambda(n) \wedge X\subseteq b_z^\Lambda(n)}}
	\]
	where ${\rm Ind_{(\cdot)}}$ is the indicator function.
	Notice that the constraints $x,y\in b_z^\Lambda(n)$ and $X\subseteq b_{z}^\Lambda(n)$ immediately imply that
	\[
	X\subset b_x^\Lambda(2n)\cap b_y^\Lambda(2n).
	\]
	Combining these observations, we find	
	\begin{equation}
	\sum_{\substack{z\in\Lambda_p \\ x,y\in b_z^\Lambda(n)}} 
	\sum_{\substack{X \in \cS_n(\Lambda_p) \\ X\subseteq b_z^\Lambda(n)}}
	\frac{1}{m(X)}\|\Phi(X)\|
	=
	\sum_{\substack{X\in \cS_n(\Lambda^p):\\X\subset b_x^\Lambda(2n)\cap b_y^\Lambda(2n)}} \;\sum_{\substack{z\in\Lambda_p: \\ x,y\in b_z^\Lambda(n) \wedge X\subseteq b_z^\Lambda(n)}}
	\frac{1}{m(X)}	\|\Phi(X)\|.
	\end{equation}
	From the definition of $m(X)$, see \eqref{eq:m(X)}, we see that the inner sum above is bounded from about by $\|\Phi(X)\|$, i.e.
	\[
	\sum_{\substack{z\in\Lambda_p: \\ x,y\in b_z^\Lambda(n) \wedge X\subseteq b_z^\Lambda(n)}}
	\frac{1}{m(X)}	\|\Phi(X)\| \leq \|\Phi(X)\|.
	\]
	Since $\diam(X)>n-1$ for any $X\in \cS_n(\Lambda_p)$ and $\Phi\in \cB_F$ one can further bound as follows:
	\begin{align}
	\sum_{\substack{X\in \cS_n(\Lambda_p) \\ X\subseteq b_x(2n)\cap b_y(2n)}} 
	\|\Phi(X)\| 
	& \leq
	\sum_{\substack{w_1, \, w_2 \in b_x(2n) \\ n-1<d(w_1,w_2) \leq n}}
	\sum_{\substack{X\subseteq \Lambda \\ w_1, \, w_2 \in X}} 
	\|\Phi(X)\| \nonumber\\
	& \leq
	\|\Phi\|_F\sum_{\substack{w_1, \, w_2 \in b_x(2n) \\ n-1<d(w_1,w_2)\leq n}}
	F(d(w_1,w_2)) \nonumber\\
	& \leq
	\kappa^22^\nu\|\Phi\|_F n^{2\nu}F(n-1). \label{balled_uniform_bd}
	\end{align}
	In the final estimate above, we have use $\nu$-regularity and the fact that $w_2\in b_{w_1}(n)$. The result follows from inserting \eqref{balled_uniform_bd} and $n \geq \lceil d(x,y)/2 \rceil$ into \eqref{balled_first_bd}.
\end{proof}

We conclude by discussing how to define an interaction $\Phi:\cP_0(\Lambda)\to \cA_{\Lambda}^{\rm loc}$ from an anchored interaction $\Phi_{\Lambda}:\Lambda \times \bZ_{\geq 0}\to \cA_{\Lambda}^{\rm loc}$, and show that a finite anchored $F$-norm $\|\Phi_{\Lambda}\|_F$ is sufficient for proving that $\Phi\in\cB_F$. This implies that the finite volume Hamiltonians associated with an anchored interaction satisfy a Lieb-Robinson bound if it has a finite anchored $F$-norm. 

To define $\Phi$, for any $X\in \cP_0(\Lambda)$ we set
\be \label{int_from_balledint}
\Phi(X) = \sum_{\substack{(x,n) \in \Lambda \times \bZ_{\geq 0} \\ b_x^\Lambda(n)  = X }} \Phi_{\Lambda}(x,n),
\ee  
with convention that empty sums are taken to be zero. For this interaction we have the following result:
\begin{prop}\label{prop:TriIneq_Fnorm}
	Let $\Phi_{\Lambda}:\Lambda \times \bZ_{\geq 0}\to \cA_{\Lambda}^{\rm loc}$ be an anchored interaction for which $\|\Phi_\Lambda\|_{m,F}<\infty$ for an $F$-function $F$. Then $\Phi\in\cB_F^m$ where $\Phi:\cP_0(\Lambda)\to \cA_{\Lambda}^{\rm loc}$ is the interaction defined in \eqref{int_from_balledint}. In particular, $\|\Phi\|_{m,F} \leq \|\Phi_{\Lambda}\|_{m,F}$.
\end{prop}
\begin{proof}
	Fix any $x,y\in \Lambda$. Recall that $\Phi_\Lambda(x,n)\in\cA_{b_x^\Lambda(n)}$ Since $\Phi_\Lambda$ is an anchored interaction. Using \eqref{int_from_balledint} and applying the triangle inequality, one finds
	\bea
	\sum_{\substack{X\subseteq \Lambda : \\ x,y\in X}}|X|^m\|\Phi(X)\|
	& \leq &
	\sum_{\substack{X\subseteq \Lambda : \\ x,y\in X}} \sum_{\substack{(z,n)\in\Lambda\times\bZ_{\geq 0} : \\ b_z^\Lambda(n)=X}}|b_z^\Lambda(n)|^m \| \Phi_\Lambda(z,n) \| \nonumber\\
	& = & \sum_{\substack{(z,n)\in\Lambda\times\bZ_{\geq 0} :\\ x,y\in b_z^\Lambda(n)}}|b_z^\Lambda(n)|^m\| \Phi_\Lambda(z,n) \|
	\nonumber\\
	& \leq & \|\Phi_\Lambda\|_{m,F} F(d(x,y)), \label{int_bldint_bd}
	\eea
	which implies $\|\Phi\|_{m,F}\leq \|\Phi_{\Lambda}\|_{m,F}$ as claimed.
\end{proof}	

There are two important insights one can draw from this result.

First, notice that if there is at most one nonzero term in the summation of \eqref{int_from_balledint}, then the anchored interaction is actually an interaction. In this case the first inequality of \eqref{int_bldint_bd} is actually an equality and so the anchored $F$-norm and interaction $F$-norm are the same. This justifies using the same notation for both types of $F$-norms. We almost exclusively work with anchored interactions and anchored $F$-norms. However, it will always be clear from context which type of $F$-norm we are using. 

Second, in the situation that $\Lambda$ is finite, the local Hamiltonian defined by $\Phi_\Lambda$ satisfies
\[
H_\Lambda := \sum_{(x,n)\in\Lambda\times\bZ_{\geq 0}}\Phi_{\Lambda}(x,n) = \sum_{X\subseteq \Lambda}\Phi(X)
\]
where $\Phi$ is as defined in \eqref{int_from_balledint}. If $\|\Phi_{\Lambda}\|_F<\infty$, then Proposition~\ref{prop:TriIneq_Fnorm} implies that the Heisenberg dynamics $\tau_{t,s}^\Lambda$ associated to $H_\Lambda$ satisfies the Lieb-Robinson bound from Theorem~\ref{thm:LRBounds}. Moreover, we can use the anchored $F$-norm, $\|\Phi_{\Lambda}\|_F$, to bound the Lieb-Robinson velocity, specifically
\be \label{balled_LR_vel}
v_\Phi \leq 2C_F\|\Phi_{\Lambda}\|_F.
\ee

\section{On Perturbation Theory for Frustration Free Hamiltonians} \label{sec:Pert_Theory}

\subsection{Introduction}

The main goal of this work is to present an approach for proving persistence of the ground state gap for frustration-free models under a broad class of extensive perturbations. This approach, originally due to Bravyi, Hastings, and Michalakis (BHM) \cite{bravyi:2010}, has some elements in common with other methods that have appeared in the literature, including the recent work of Fr\" ohlich and Pizzo \cite{frohlich:2020,del-vecchio:2020}, but also older work by Albanese \cite{albanese:1989}, Kennedy and Tasaki \cite{kennedy:1992}, and Yarotsky \cite{yarotsky:2006}. Specifically, the aim in all these works is to prove a relative form bound in one way or another. In this section we discuss the general implications of relative form bounds on spectral gaps and also identify a situation for which a form bound can be proved in a straightforward manner. 

The strength of the BHM approach stems from applying a relative form bound after a unitary transformation that brings the problem into a form where the results of this section can be applied. The next several sections are devoted to the analysis necessary to establish this.
At first sight, the Lie-Schwinger block diagonalization approach of Fr\" ohlich and Pizzo does something similar in that it also rests on the construction of a suitable similarity transformation \cite{frohlich:2020}. An important difference, however, is that in the latter work the transformation itself is constructed 
perturbatively by a power series for which one needs to prove a positive radius of convergence.  This is not the case with the approach here. The transformation we use is well-defined for the full parameter range for which a non-vanishing ground state gap exists. 

The method of Fr\" ohlich and Pizzo has so far only been applied to initial Hamiltonians that are on-site and have a unique product ground state.
These are strong limitations but, due to the availability of a convergent power series, it also has the advantage of yielding analyticity of the ground state energy density as a function of the perturbation parameters within the radius of convergence \cite{del-vecchio:2019a}.

We now first prove some general results about relatively bounded perturbations and then look specifically at quantum lattice systems.

\subsection{General Perturbation Theory with Form Bounds} \label{sec:gen_pert}

In this section, we determine spectral gap estimates for perturbations of a gapped, self-adjoint operator, $H$ on a (possibly infinite-dimensional) Hilbert space in the situation that the perturbation is form bounded by $H$. The results we present in this section hold for both bounded and unbounded operators, and so we present the results in a general context. At the end of the section, we discuss how to apply the results to the gapped quantum spin systems of interest.

The first lemma can be regarded as a variational principle for spectral gaps. In the statement we use the convention 
\begin{equation}\label{eq:convention}
\inf \emptyset = +\infty, \qquad\sup \emptyset = -\infty.
\end{equation}

\begin{lemma}[Level Repulsion Principle]\label{lem:variational_principle_for_gaps}
	Let $\cH$ be a complex Hilbert space  and $H$ a densely defined self-adjoint operator on $\cH$ with domain $\cD$. Let $\cK$ be a closed linear subspace of $\cH$ such that $\cK \cap \cD$ is dense in $\cK$ and $\cK^\perp\cap\cD$ is dense in $\cK^\perp$. Define $a,b\in\Rl\cup\{\pm\infty\}$ as follows
	\be\label{VPbounds}
	a = \sup \{\langle \psi , H\psi\rangle \mid \psi \in \cK \cap\cD, \Vert \psi\Vert =1\}, \quad b = \inf \{\langle \psi , H\psi\rangle \mid \psi \in \cK^\perp\cap\cD, \Vert \psi\Vert =1\}.
	\ee
	Then, if $a<b$,
	\be\label{opengap}
	(a,b) \cap \spec (H) = \emptyset.
	\ee
\end{lemma}

\begin{proof}
	For $\cK=\{0\}$ or $\cK=\cH$, \eq{opengap} is trivially satisfied given \eqref{eq:convention}. The cases where either $a=\infty$ or $b=-\infty$ are also trivial. Therefore, we may assume that both $a$ and $b$ are finite and $a<b$.
	
	We show that $H-\lambda\idty$ has a bounded inverse for all $\lambda\in (a,b)$. Replacing $H$ by $H'=H-\lambda\idty$ in \eqref{VPbounds} changes the constants $a$ and $b$ to $a':= a-\lambda <0$ and $b':=b-\lambda>0$, and hence showing that $0\not\in\spec H'$. In other words, without loss of generality, we may assume that $a<0$ and $b>0$, and then show that $0\not\in\spec(H)$ or, equivalently, that $H$ has a bounded inverse.
	
	We first consider the case that $H$ is bounded, and hence $\mathcal{D} = \mathcal{H}$.  
	Let $P$ be the orthogonal projection onto $\cK$ and define $Q=\idty -P$. 
	Denote by $PHP: \mathcal{K} \to \mathcal{K}$ and $QHQ: \mathcal{K}^\perp \to \mathcal{K}^\perp$ 
	the bounded, self-adjoint restrictions of $H$ to $\mathcal{K}$ and $\mathcal{K}^{\perp}$ respectively.
	The definitions of $a$ and $b$ imply that $PHP \leq aP$ and $QHQ\geq bQ$. 
	Since $a<0$ and $b>0$, we find that $PHP$ is negative definite with a bounded inverse
	$(PHP)^{-1}\in\cB(\cK)$, and $QHQ$ is positive definite with a bounded inverse $(QHQ)^{-1}\in\cB(\cK^\perp)$. Therefore, there are 
	positive $A\in\cB(\cK)$ and $B\in\cB(\cK^\perp)$ such that $(PHP)^{-1} = - A^2$ and  $(QHQ)^{-1}=B^2$. 
	
	Consider the representation of $H$ as a block-operator acting on $\mathcal{H} = \mathcal{K} \oplus \mathcal{K}^\perp$.
	One finds that 
	\be\label{decomp}
	H=\begin{bmatrix} PHP & PHQ \\ QHP & QHQ \end{bmatrix}
	= \begin{bmatrix} A^{-1}&  0\\ 0 & B^{-1} \end{bmatrix} \begin{bmatrix} -\idty & X  \\ X^* & \idty \end{bmatrix} \begin{bmatrix} A^{-1}&  0\\ 0 & B^{-1} \end{bmatrix}
	\ee
	where we have used that $A$ and $B$ both have (bounded) inverses and denoted by 
	$X: \mathcal{K}^\perp \to \mathcal{K}$ be the operator $X = A(PHQ)B$.
	Let $Y$ denote the middle matrix operator above, which is clearly bounded and self-adjoint. One checks that
	$$
	Y^2 = \begin{bmatrix} \idty+ XX^* & 0\\ 0 & \idty +X^*X \end{bmatrix},
	$$
	and thus $Y^2 \geq \idty$. This shows that the interval $(-1,1)$ is contained in the resolvent set of $Y$. As such, 
	$Y^{-1}$ is bounded with norm at most $1$. In this case, we can invert \eq{decomp} and obtain
	\be\label{Hinverse}
	H^{-1} =  \begin{bmatrix} A&  0\\ 0 & B \end{bmatrix} \begin{bmatrix} -\idty & X  \\ X^* & \idty \end{bmatrix}^{-1} \begin{bmatrix} A&  0\\ 0 & B \end{bmatrix},
	\ee
	which is bounded as it is the product of bounded operators; in fact, $\Vert H^{-1}\Vert \leq \max ( |a|^{-1}, b^{-1})$. This concludes the proof
	for the case of bounded $H$.
	
	The general case of unbounded $H$ can be handled by considering a sequence $(H_n)_{n\geq 1}$ of bounded self-adjoint
	operators converging to $H$ in the strong resolvent sense. Let $P_n$, $n\geq 1$, denote the spectral projections of $H$ corresponding to $[-n,n]$,
	and define $H_n = P_n H P_n$. In this case, $H_n$ is bounded and $H_n\to H$ in the strong resolvent sense (see  \cite{bogli:2017} and also \cite[Satz 9.21]{weidmann:2000}). The constants $a_n$ and $b_n$ defined by
	\be\label{VPbounds2}
	a_n = \sup \{\langle \psi , H_n\psi\rangle \mid \psi \in \cK \cap\cD, \Vert \psi\Vert =1\}, \quad b_n = \inf \{\langle \psi , H_n\psi\rangle \mid \psi \in \cK^\perp\cap\cD, \Vert \psi\Vert =1\}
	\ee
	satisfy $a=\lim_n a_n$ and $b=\lim_n b_n$. Applying the result for bounded operators, we conclude that
	$(a_n,b_n) \cap \spec (H_n) = \emptyset$ for all $n \in \mathbb{N}$. 
	By strong resolvent convergence, see e.g. \cite[Theorem VIII.24 (a)]{reed:1980},
	for each $\lambda \in \spec(H)$, there is $\lambda_n \in \spec(H_n)$ with $\lambda_n \to \lambda$.
	This implies that $(a,b) \cap \spec (H) = \emptyset$ as desired.
\end{proof} 

Lemma \ref{lem:variational_principle_for_gaps} is optimal in the sense that it identifies a spectral gap exactly if $\cK$ is the spectral 
subspace of $H$ associated with the spectrum below the gap. If $\cK$ is not an invariant subspace of $H$, then the quantity $b-a$, if 
positive, is a lower bound for the gap. In that case, the lemma shows that the `off-diagonal' terms $PHQ + QHP$ can only push the two 
parts of spectrum further part. This can be regarded as a generalization of the level repulsion observed for a pair of eigenvalues of 
a diagonal Hermitian matrix when one considers the effect of non-vanishing off-diagonal matrix elements. 

\begin{lemma}[Relatively Bounded Perturbations]\label{lem:persistence_of_gaps}
	Let $\cH$ be a complex Hilbert space  and $H$ a densely defined self-adjoint operator on $\cH$ with domain $\cD$. Suppose $V$ is a self-adjoint operator on $\cH$ with $\cD\subset \dom V$, and suppose there exist constants $\alpha\geq 0$ and $\beta \in [0,1)$, such that
	\be\label{relative_bound}
	\vert\langle\psi,V\psi\rangle\vert \leq \alpha \Vert \psi\Vert^2 + \beta \langle \psi, H\psi\rangle, \text{ for all } \psi \in \cD.
	\ee
	Then,
	\be\label{global_bounds}
	(1-\beta)\inf\spec (H) -\alpha \leq \inf\spec(H+V) \text{ and }  \sup\spec(H+V) \leq (1+\beta)\sup\spec (H) +\alpha.
	\ee
	If, in addition, $a<b\in\Rl$ are such that $(a,b) \cap \spec (H) = \emptyset$, then
	\be\label{gap_estimate}
	((1+\beta)a+\alpha,(1-\beta)b-\alpha)\cap \spec (H+V) =\emptyset.
	\ee
\end{lemma}

\begin{proof}
	If $\beta=0$, then $\Vert V\Vert\leq \alpha$ and the statements in the lemma follow from standard perturbation theory for bounded perturbations \cite{kato:1995}. Thus, we assume that $\beta>0$. In this case, \eq{relative_bound} implies that $H$ is bounded below by $-\alpha\beta^{-1}$, and
	in particular, $H+V$ is self-adjoint on $\cD$ \cite[Theorem X.17]{reed:1975}. The estimates in \eq{global_bounds} follow directly from the relative boundedness expressed by \eq{relative_bound} as
	\be\label{first_pert_bd}
	(1-\beta)\braket{\psi}{H\psi}-\alpha\|\psi\|^2 \leq \braket{\psi}{(H+V)\psi} \leq (1+\beta)\braket{\psi}{H\psi}+\alpha\|\psi\|^2
	\ee
	for all $\psi \in \cD$.
	
	To prove \eq{gap_estimate}, first consider the case of bounded $H$ and $V$, and let $P$ be the spectral projection of $H$ corresponding to the interval $[-\alpha\beta^{-1}, a]$. Applying \eqref{first_pert_bd}
	\bea
	\langle\psi, (H+V) \psi\rangle &\leq& (1+\beta) \langle\psi,H\psi\rangle + \alpha \Vert\psi\Vert^2\leq [(1+\beta) a +\alpha]\Vert\psi\Vert^2, \psi \in P\cH \label{gap_lb}\\
	\langle\psi, (H+V) \psi\rangle &\geq& (1-\beta) \langle\psi,H\psi\rangle - \alpha \label{gap_up} \Vert\psi\Vert^2\geq  [(1-\beta) b -\alpha]\Vert\psi\Vert^2,\psi \in (\idty- P)\cH.
	\eea
	As the conditions on the domain of $H$ and the perturbation $V$ in Lemma \ref{lem:variational_principle_for_gaps} are satisfied with $\cK=\ran P$, the spectral subspace of $H$ corresponding to $[-\alpha \beta^{-1},a]$, the gap from \eq{gap_estimate} now follows by taking the $\sup$ and $\inf$, respectively, of \eqref{gap_lb}-\eqref{gap_up} and applying Lemma \ref{lem:variational_principle_for_gaps}.
	
	To treat the general case in which $H$ and $V$ may be unbounded, let $P_n$ be the spectral projection of $H$ corresponding to the interval $[-\alpha\beta^{-1}, a+n]$,
	for all $n\geq 1$.
	Then, $H_n=P_n H P_n$ is bounded on $\cH_n=P_n\cH$, and \eq{relative_bound} implies that $V_n=P_n V P_n$ is bounded too. $H_n$ and $V_n$ satisfy \eq{relative_bound}
	with the same constants $\alpha$ and $\beta$. We also have $(a,b) \cap \spec (H_n) = \emptyset$ if this condition holds for $H$. Therefore, the argument in the previous paragraph 
	shows that
	\be\label{gap_n}
	((1+\beta)a+\alpha,(1-\beta)b-\alpha)\cap \spec (H_n+V_n) =\emptyset, \text{ for all } n\geq 1.
	\ee
	Since the sequence $(H_n+V_n)_{n\geq 1}$ converges to $H+V$ in the strong resolvent sense,
	an application of \cite[Theorem VIII.24 (a)]{reed:1980} shows that 
	$$
	((1+\beta)a+\alpha,(1-\beta)b-\alpha)\cap \spec (H+V) =\emptyset.
	$$
\end{proof} 
As can be seen from the previous proof, the estimate of the gap may be optimized by proving form bounds separately on each of the spectral subspaces. This is the content of the following corollary.
\begin{cor}\label{cor:fb_gap}
	Let $\cH$ be a complex Hilbert space  and $H$ a densely defined self-adjoint operator on $\cH$ with domain $\cD$ such that $(a,b)  \cap \spec (H) = \emptyset$. Suppose that $V$ is a self-adjoint operator on 
	$\cH$ with $\cD\subset \dom V$ for which there are constants $\alpha',\alpha''\geq 0$ and $\beta \in [0,1)$ such that
	\beann
	|\braket{\psi}{V \psi}| & \leq & \alpha'\|\psi\|^2 + \beta\braket{\psi}{H\psi}, \quad \psi \in P\cH \cap \cD \\
	|\braket{\psi}{V \psi}| & \leq & \alpha''\|\psi\|^2 + \beta\braket{\psi}{H\psi}, \quad \psi \in (\idtyty-P)\cH\cap \cD
	\eeann
	where $P$ is the spectral projection of $H$ associated to $(-\infty,a]$. Then,
	\be\label{optimized_stability}
	((1+\beta)a + \alpha', (1-\beta)b-\alpha'') \cap \spec(H+V) = \emptyset.
	\ee
\end{cor}

\begin{proof}
	The proof follows just as the proof of Lemma~\ref{lem:persistence_of_gaps} from replacing $\alpha$ with $\alpha'$, resp. $\alpha''$, in \eqref{gap_lb}, resp. \eqref{gap_up}.
\end{proof}

We now turn to the situation of interest: the stability of the spectral gap of a Hamiltonian, $H$, that is gapped above the ground state energy. In this case, we will consider perturbed Hamiltonians of the form $H(s) = H + V(s)$ where $s\in \bR$. We will always assume that $V(s):\cH\to\cH$ is self-adjoint with $\dom H \subseteq \dom V(s)$ for all $s$. Generally, we will also assume that the perturbation is strongly differentiable in $s$. However, this is not necessary for the next result and so we will not assume this here.

\begin{thm}\label{thm:Pert_Est}
	Let $H\geq 0$ be a self adjoint operator on a dense domain $\cD$ for which $0\in\spec(H)$ and $(0,\gamma)\cap \spec(H) = \emptyset$ for some $\gamma>0$, and denote by $P$ the ground state projection of $H$. Suppose that $H(s), \, s\in\bR$, is a family of perturbed Hamiltonians of the form
	\be \label{pert_ham_decomp}
	H(s) = H + V(s) + A(s) + C(s)\idtyty,
	\ee
	for which there are constants $\alpha, \alpha', \alpha'', \beta >0$ so that for all $s\in\bR$ the following hold:
	\begin{enumerate}
		\item[(i)] $C(s)\in\bR$
		\item[(ii)] $A(s)^*=A(s)\in\cB(\cH)$ with $\|A(s)\|\leq s\alpha$, $\|PA(s)P\|\leq s\alpha'$, and $\|(\idtyty-P)A(s)(\idtyty-P)\|\leq s\alpha''$. 
		\item[(iii)] $V(s)^*=V(s)$ with $\cD\subseteq \dom V(s)$ and $|\braket{\psi}{V(s)\psi}|\leq s\beta\braket{\psi}{H\psi}$ for all $\psi \in \cD$.
	\end{enumerate}
	Then, for all $0\leq s< \beta^{-1}$, $\spec(H(s))=\Sigma_1(s)\cup \Sigma_2(s)$ where
	\[
	\Sigma_1(s) \subseteq [C(s)-s\alpha, C(s)+s\alpha'], \quad \Sigma_2(s) \subseteq [C(s)+(1-s\beta)\gamma-s\alpha'', \, \infty).
	\]
\end{thm}

Note that the conditions imply $A(0)=V(0)=0$. As such, the result is trivial for $s=0$ as $H(0)$ is just a constant shift of $H$ by $C(0)$. In many applications, we will have that all quantities $A(s), \, V(s),$ and $C(s)$ are continuous in $s$ and, moreover, $C(s) \to 0$ as $s\to 0$.
\begin{proof}
	Suppose that $H(s)$ has the form described above and fix $0\leq s<\beta^{-1}$. Without loss of generality we may assume $C(s) = 0$. As $(0,\gamma)\cap \spec(H)=\emptyset$ and
	\bea
	\braket{\psi}{V(s)+A(s) \psi} & \leq & s\beta\braket{\psi}{H\psi} + s\alpha'\|\psi\|^2, \;\;  \psi\in P\cH\cap \cD \label{eq:ub1}\\
	\braket{\psi}{V(s)+A(s) \psi} & \leq & s\beta\braket{\psi}{H\psi} + s\alpha''\|\psi\|^2, \;\;  \psi\in (\idty-P)\cH\cap \cD \label{eq:ub2}
	\eea
	Corollary~\ref{cor:fb_gap} implies
	\begin{equation}\label{pert_residue1}
	(s\alpha', (1-s\beta)\gamma - s\alpha'')\cap \spec(H(s)) = \emptyset.
	\end{equation}
	Trivially, $(-n, 0)\cap \spec(H) = \emptyset$ for all $n\in \bN$ and for all $\psi\in \cD$
	\[
	|\braket{\psi}{V(s)+A(s) \psi} |\leq s\beta\braket{\psi}{H\psi} + s\alpha\|\psi\|^2.
	\]
	Applying Lemma~\ref{lem:persistence_of_gaps} shows that
	\be \label{pert_residue2}
	(-\infty, -s\alpha)\cap \spec(H(s)) = \emptyset.
	\ee
	Combining \eqref{pert_residue1} and \eqref{pert_residue2}, it immediately follows that $\spec(H(s))= \Sigma_1(s)\cup\Sigma_2(s)$ where
	\beann
	\Sigma_1(s) & = & \spec(H(s)) \cap [C(s)-s\alpha, C(s)+s\alpha'] \\ 
	\Sigma_2(s) & = & \spec(H(s))\cap [C(s)+(1-s\beta)\gamma-s\alpha'', \, \infty).
	\eeann
\end{proof}

In the above result, we derive stability of the spectral gap above the ground state energy assuming that the Hamiltonian $H(s)$ has a decomposition of the form \eqref{pert_ham_decomp} that satisfies conditions (i)-(iii) of Theorem~\ref{thm:Pert_Est}. In our application, the decomposition of this kind that we find depends extensively on properties of the ground state projections. However, whenever such a decomposition exists, Lemma~\ref{lem:persistence_of_gaps} implies stability of all higher gaps $(a,b)\cap \spec(H)$, not just the ground state gap. This is summarized in the following corollary.

\begin{cor}[Higher Order Gaps]\label{cor:higher_gaps}
	Suppose that $H(s)$, $s\in\bR$, has a decomposition of the form \eqref{pert_ham_decomp} satisfying the assumptions of Theorem~\ref{thm:Pert_Est}, and that $(a,b)\cap \spec(H) = \emptyset$ where $0<a<b$. Then, for all $0\leq s<\beta^{-1}$,
	\be\label{higher_gap}
	((1+s\beta)a+s\alpha, (1-s\beta)b-s\alpha) \cap \spec(H(s)) = \emptyset.
	\ee
	where $\alpha$ is as in Theorem~\ref{thm:Pert_Est}.
\end{cor}

\begin{proof}
	The proof of this result runs just as that of Theorem~\ref{thm:Pert_Est} with the alteration that $P$ is the spectral projection associated with $(-\infty,a]$. By replacing both $\alpha'$ and $\alpha''$ with $\alpha$ in \eqref{eq:ub1}-\eqref{eq:ub2}, one finds that \eqref{higher_gap} again follows from applying Lemma~\ref{lem:persistence_of_gaps}.
\end{proof}

In general, it is far from obvious when a Hamiltonian $H(s)$ has a decomposition of the form \eqref{pert_ham_decomp}. 
In the time-honored quantum mechanics tradition, considering $\tilde{H}(s)=U(s)^*H(s)U(s)$ with a cleverly chosen unitary transformation 
$U(s)$, can be very helpful. We will use the  unitary transformations given by the spectral flow, see \eqref{gen_spec_flow}-\eqref{spec-flow-auto}.
While the main focus of this paper is spectral gap stability of quantum spin systems, the spectral flow automorphism is well defined and its key property, i.e. \eqref{spec_flow_proj}, holds in a more general context, see \cite{nachtergaele:2019}. As a consequence, the approach presented in this paper can be applied to more general systems. 

%
%
%

\subsubsection{Using the Spectral Flow for Spectral Gap Stability} \label{sec:sf-gap-stab}

Recall that given a quantum spin system on a $\nu$-regular metric space $(\Gamma, d)$ with
quasi-local algebra $\cA_\Gamma$, we consider a family of finite volume Hamiltonians of the form
\[
H_\Lambda(s) = H_\Lambda + s V_{\Lambda^p}, \;\; s\in\bR
\]
with $\Lambda^p\subseteq\Lambda\in\cP_0(\Gamma)$ as defined in \eqref{pert_hams}. We denote by $\gamma_\Lambda$ the spectral gap of $H_\Lambda$ above the ground state energy (which we normalize to be zero), and define $\Sigma_1^\Lambda(s)$ and 
$\Sigma_2^\Lambda(s)$ as before, see \eq{spec_sets}. For any $0<\gamma<\gamma_\Lambda$, the goal is to find a lower bound for 
\begin{equation}
s^\Lambda_\gamma = \sup\{s'\in [0,1] \, : \, \gap(H_\Lambda(s))=\text{dist}(\Sigma_1^\Lambda(s),\Sigma_2^\Lambda(s))\geq \gamma \; \text{  for all } \; 0\leq s \leq s' \}.
\end{equation}
This will be achieved by applying Theorem~\ref{thm:Pert_Est} to the Hamiltonian $\tilde{H}_\Lambda(s)= \alpha_s(H_\Lambda(s))$ 
where $\alpha_s:\cA_\Lambda \to \cA_\Lambda$ is the spectral flow, defined as in \eqref{gen_spec_flow}-\eqref{spec-flow-auto} 
with $\xi =\gamma$. The advantage of considering $\tilde{H}_\Lambda(s)$ is that its spectral projection $\tilde{P}(s)$ associated to $\Sigma^\Lambda_1(s)$ is 
constant for all $0\leq s \leq s^\Lambda_\gamma$. Specifically, if $P(s)$ is the spectral projection of $H_\Lambda(s)$ associated with $\Sigma^\Lambda_1(s)$, then by \eqref{spec_flow_proj}
\[
\tilde{P}(s) = \alpha_s(P(s)) = P(0), \;\;\text{for all}\;\; 0 \leq s \leq s_\Lambda(\gamma).
\] 

If the model defined by $H_\Lambda$ is frustration-free and the ground states satisfy an LTQO condition (i.e. are sufficiently indistinguishable) as
described in
Section \ref{sec:LTQO2}, then we can construct a decomposition of the form \eqref{pert_ham_decomp} from Theorem~\ref{thm:Pert_Est} for any interaction $\Phi$ defining the perturbation $V_{\Lambda^p}$ which has a finite $F$-norm for an $F$-function of sufficient decay. Specifically, we will construct the following type of decomposition:

\begin{claim}[Decomposition of the Equivalent Hamltonian]\label{clm:decompositionH} Under appropriate assumptions on the ground states of $H_\Lambda$, and with sufficient decay on $\Phi$, one can write
	\be\label{qss_pert_decomp}
	\alpha_s(H_\Lambda(s)) = H_\Lambda + V_\Lambda(s)+\Delta_\Lambda(s)+E_\Lambda(s) + C_\Lambda(s)\idtyty
	\ee
	with the following properties:
	\begin{enumerate}
		\item[(i)] $P(0)\Delta_\Lambda(s)P(0) = \Delta_\Lambda(s)$, and there is a constant $\delta_\Lambda>0$ so that $\|\Delta_\Lambda(s)\|\leq s \delta_\Lambda$.
		\item[(ii)] There is a constant $\epsilon_\Lambda >0$ so that $\|E_\Lambda(s)\|\leq s\epsilon_\Lambda$.
		\item[(iii)] There is a constant $\beta_\Lambda>0$ so that $|\braket{\psi}{V_\Lambda(s)\psi}|\leq s\beta_\Lambda \braket{\psi}{H_\Lambda \psi}$ for all $\psi \in \cH_\Lambda$.
	\end{enumerate}
\end{claim}

The conditions above imply that $V_\Lambda(0)=\Delta_\Lambda(0)=E_\Lambda(0)=0$. For our definition of $C_{\Lambda}(s)$ in Section~\ref{sec:LTQO2}, see specifically \eqref{const-s}, one can verify that $C_{\Lambda}(s) \to 0$ as $s\to 0$.
In any case, Theorem~\ref{thm:Pert_Est} applies with $A(s) =\Delta_\Lambda(s)+E_\Lambda(s)$, $\alpha=\alpha'=(\delta_\Lambda+\epsilon_\Lambda)$ and $\alpha'' = \epsilon_\Lambda$, to give
\[
\Sigma^\Lambda_1(s) \subseteq [C(s)-s(\delta_\Lambda+\epsilon_\Lambda), \, C(s)+s(\delta_\Lambda+\epsilon_\Lambda)],
\quad
\Sigma^\Lambda_2(s) \subseteq [C(s)+(1-s\beta_\Lambda)\gamma_\Lambda-s\epsilon_\Lambda, \infty).
\]
These inclusions imply
\be\label{stab_gap_est}
\gap(H_\Lambda(s)):=\text{dist}(\Sigma^\Lambda_1(s),\Sigma^\Lambda_2(s)) \geq \gamma_{\Lambda}-s(\beta_\Lambda \gamma_\Lambda+\delta_\Lambda+2\epsilon_\Lambda),
\ee
from which the following lower bound holds:
\be\label{s_gamma_est}
s^{\Lambda}_\gamma\geq \frac{\gamma_{\Lambda} - \gamma}{\beta_\Lambda \gamma_\Lambda+\delta_\Lambda+2\epsilon_\Lambda}.
\ee
Using a similar estimate combined with Corollary~\ref{cor:higher_gaps}, we can also obtain a lower bound on the range of $s$ for which higher order gaps remain open.

It is important to note that, in general, the constants $\beta_\Lambda,\delta_\Lambda$, and $\epsilon_\Lambda$ appearing in the lower bound for $s^{\Lambda}_\gamma$ depend on the finite volume $\Lambda$.  Our proof of 
stability of the spectral gap will require an increasing and absorbing sequence of finite volumes $\Lambda_n\uparrow\Gamma$ such that 
$$
\gamma_0 =\inf_{n \geq 1} {\rm gap}(H_{\Lambda_n}) = \inf_{n \geq 1} \gamma_{\Lambda_n}>0 .
$$
In this case, we have defined {\em stability of the spectral gap} by the property that 
$$
\inf_{n \geq 1} s^{\Lambda_n}_\gamma>0 \quad \mbox{for all} \quad 0< \gamma<\gamma_0 \, .
$$ 
This form of stability implies that the quantum spin model has a non-vanishing gap in the thermodynamic limit whenever the perturbation parameter is sufficiently small (see Sections \ref{sec:uniform_sequences} and \ref{sec:automorphic-equivalence} for details). For our proof of stability, we will also require that our
increasing and absorbing sequence of finite volumes can be associated with a suitable choice of perturbation regions $\Lambda_n^p\uparrow \Gamma$, and
that given any $\epsilon, \delta>0$, there is $N \geq 1$ sufficienly large so that $\delta_{\Lambda_n} < \delta$ and $\epsilon_{\Lambda_n} <\epsilon$ for all $n\geq N$. Moreover, we must also show that $\sup_n\beta_{\Lambda_n}\leq \beta<\infty$. This will be the task in the next couple of sections. But first, we 
identify a condition on the perturbation for which a relative form bound is straightforward to derive.

%
%

\subsection{A class of form bounded interactions} \label{sec:form_bd}

In this section we consider a class of perturbations of frustration-free models for which a relative form bound can be derived quite simply.
The unperturbed model is defined on $\cH_\Lambda=\bigotimes_{x\in\Lambda} \cH_x$, where $\cH_x$ are arbitrary, not necessarily finite-dimensional   complex Hilbert spaces. The unperturbed Hamiltonian $H_\Lambda\geq 0$ is assumed to be frustration-free but not necessarily bounded.
As detailed below, the perturbation $V_\Lambda$ will be assumed to be given in terms of a bounded anchored interaction. We will show that if the interaction
terms of the perturbation annihilate the ground states of $H_\Lambda$, one can calculate a constant $\beta >0$ such that
\be \label{foam_bd}
|\braket{\psi}{V_\Lambda \psi}| \leq \beta \braket{\psi}{H_\Lambda \psi} \; \text{ for all } \; \psi\in\text{dom}(H_\Lambda).
\ee

In this section we study a system defined on a finite set $\Lambda$ equipped with a metric $d$. Often, and in later sections, $\Lambda$ will be a 
finite subset of a $\nu$-regular metric space $(\Gamma, d)$, and $\Gamma$ is thought of as infinite, but this does not play a role here.

To state the result we will use two families of subsets of $\Lambda$. The first are the balls, which are labeled by $x\in\Lambda$ and $n\geq 0$:
$$
b_x^\Lambda(n) = \{ y\in \Lambda \mid d(x,y) \leq n\}.
$$
The second family is also labeled by $x\in\Lambda$ and $n\geq 0$, and we denote those sets by $\Lambda(x,n)$. We require 
$b^\Lambda_x(n)\subset
\Lambda(x,n)$, for all $x\in\Lambda$ and $n\geq 0$. For example, the $\Lambda(x,n)$ could be balls for another metric. In some 
situations we may take $\Lambda(x,n)=b_x^{\Lambda}(n)$. The balls are used to describe the decay of the perturbations and to define the indistinguishability 
radius needed for the LTQO condition, see \eqref{LTQO_length}, while the gap of the local Hamiltonians $H_{\Lambda(x,n)}$ will feature prominently in the relative form bound we derive. The indistinguishability radius and LTQO on the one hand and the local ground state gaps on the other, in principle, are two unrelated aspects of the models we study. Therefore, good choices for these two families of finite sets need not be the same. The balls with respect to a natural choice of metric
for expressing the indistinguishability radius, may not be the most convenient shape of finite volumes for estimating the local gaps. Therefore, we maintain the freedom to chose $\Lambda(x,n)$ distinct from the balls $b_x^\Lambda(n)$. 

There will be a further property we require of the sets $\Lambda(x,n)$. We formulate it here for the case where $\ell:=\diam(\Lambda)$ is finite, but the definition and discussion carries over to the infinite situation without change. Let $\cS = \{ \Lambda(x,n)\subset\Lambda \mid x\in\Lambda, 1 \leq n\leq \ell\}$. 
In the following definition $c$ and $\zeta$ are positive constants.


\begin{defn}[Family of Partitions of $(c,\zeta)$-Polynomial Growth Separating $\cS$] \label{ass:sep_part}
	$\cT=\{\cT_n \mid 1 \leq n\leq \ell\}$ is a family of partitions of $(c,\zeta)$-polynomial growth separating $\cS$ if for each $n$,
	$\caT_n = \{T_n^i : i\in \caI_n \}$ is a partition of $\Lambda$ with $|\caI_n|\leq c n^\zeta$ and such that
	\be \label{sep_property}
	\Lambda(x,n)\cap\Lambda(y,n) = \emptyset \;\text{  for all } \; x,y\in T_n^i \; \text{with}\; x \neq y.
	\ee
\end{defn}

The canonical choice of $ \Lambda(x,0) := \{ x \}$ and $\mathcal{T}_0 = \{ \Lambda(x,0) : x \in \Lambda \}$ allows for
these notions to be extended to $n=0$, but this will not play an important role in our arguments. Moreover, in application, such partitions $\cT_n$ may only be required for values $n$ larger than some threshold $R$, see, e.g. Theorem~\ref{thm:Step3General} below.

The collection of volumes $\caS$ that can be separated by a family of partitions of polynomial growth enters the proof of Theorem~\ref{thm:Step3General} through a combinatorial argument. It is because of this argument that working with anchored interactions is convenient. 

\begin{ex}[Separating Partition on $\bZ^\nu$] Consider $\Gamma = \bZ^\nu$ with, e.g., the $\ell^\infty$-metric, and let $\Lambda = [-L, L]^{\nu}$, $L \geq 1$. 
	Take $\cS$ to be the collection of balls, i.e. $\Lambda(x,n) = b_x^\Lambda(n)$ for all relevant $x$ and $n$. 
	We construct a family of partitions $\cT=\{\cT_n \, : \, 1\leq n\leq 2L+1\}$ which is of $(3^{\nu}, \nu)$-polynomial growth and separates $\cS$.
	To define the $n$-th partition, first set $\caI_n = [0,2n+1)^\nu$. Clearly, this set satisfies the polynomial condition as
	\[| \caI_n| = (2n+1)^\nu \leq (3n)^\nu.\]
	We then define the $n$-th parition $\caT_n = \{T_n^x \, : \, x\in I_n \}$ of $\Lambda$ by
	\[T_n^x = \{z\in \Lambda \, : \, z_i \equiv x_i \mod(2n+1), \, i = 1, \ldots, \nu  \}. \]
	Fix $x\in\cI_n$. By construction, $d(y,z)\geq 2n+1$ for any two distinct sites $y,z\in\caT_n^x$ and so $b_y^\Lambda(n)\cap b_z^\Lambda(n) = \emptyset$. Thus, $\cT_n$ separates $\caS$ as desired.
\end{ex}

%

The local Hamiltonians $H_{\Lambda_0}$, for any $\Lambda_0\subset \Lambda$, are defined in terms of 
a finite-range, frustration-free interaction $\eta: \cP_0(\Lambda) \to \cA_{\Lambda}^{\rm loc}$.
As mentioned before, $\Lambda$ is a fixed finite set here. Let $P_{\Lambda_0}$ be the orthogonal projection 
onto the ground state space, $\ker(H_{\Lambda_0})$, where
\begin{equation} \label{gen_ff_ham}
H_{\Lambda_0} = \sum_{X \subset \Lambda_0} \eta(X).
\end{equation}
Recall that the frustration-free property guarantees that if $\Lambda_0 \subset \Lambda_1$, then 
\begin{equation}\label{FFProp}
P_{\Lambda_0}P_{\Lambda_1} = P_{\Lambda_1}P_{\Lambda_0} = P_{\Lambda_1}.
\end{equation}
%

Let $R\geq 0$ denote the range of the interaction $\eta$ and assume $R\leq \ell:=\diam(\Lambda)$. In general, $R$ should be thought of 
`small' relative to $\ell=\diam(\Lambda)$. 
%
Given a collection of finite volumes $\cS=\{\Lambda(x,n)\}$, we define the \emph{local gaps}, $\gamma(n)$, $1\leq n\leq \ell$, by
\be\label{local_gap}
\gamma(n) = \inf_{x\in\Lambda} {\rm gap}(H_{\Lambda(x,n)}).
\ee
In concrete situations of interest, we usually have $H_{\Lambda(x,n)}\neq 0$,  for all $x\in\Lambda$ and $n\geq R$. To deal with the situations where
some $H_{\Lambda(x,n)}= 0$, we define ${\rm gap}(0)=\infty$. 


Given $\cS$ with a family of separating partitions of $(c,\zeta)$-growth as defined in Definition ~\ref{ass:sep_part}, Theorem~\ref{thm:Step3General} provides conditions on which a self-adjoint operator  $\Phi \in \mathcal{A}_{\Lambda}$ written as
\begin{equation} \label{balled_int}
V_\Lambda = \sum_{x \in \Lambda} \sum_{n=R}^{\ell}\Phi(x,n), \quad \Phi(x,n)^*= \Phi(x,n) \in \mathcal{A}_{b_x^\Lambda(n)} 
\end{equation}
can be form-bounded by $H_{\Lambda}$ with an explicit constant $\beta$, as in (\ref{foam_bd}). Note that in \eq{balled_int} we have grouped terms together 
so that the summation starts with $n=R$; compare, e.g., with more general anchored interactions as in (\ref{fv_ham}). We do this as there are many cases of interest for which $H_{b_x^\Lambda(n)}=0$ (i.e. $P_{b_x^\Lambda(n)}=\idty$) for $n< R$. Thus, given \eqref{annih_terms} below, we choose to use the range $R$ as the lower bound in \eqref{balled_int}. However, this is not strictly necessary as long as \eqref{annih_terms} holds.

\begin{thm}\label{thm:Step3General}
	Let $H_{\Lambda}$ be a frustration-free Hamiltonian (not necessarily bounded), and $\cS$ be a 
	collection of sub-volumes of $\Lambda$ with positive local gaps $\gamma(n)>0$ for $n\geq R$. Assume there exists a family of partitions, $\cT = \{ \cT_n : R\leq n\leq \ell \}$, of $(c, \zeta)$-polynomial growth
	that separates $\cS$. Suppose $V_\Lambda \in \mathcal{A}_{\Lambda}$ is as in (\ref{balled_int}) and satisfies
	\begin{equation} \label{annih_terms}
	\Phi(x,n)P_{b_x^\Lambda(n)} = P_{b_x^\Lambda(n)} \Phi(x,n) = 0,\quad x\in\Lambda,\, R\leq n\leq \ell.
	\end{equation}
	Then, for all $\psi \in \dom{H_\Lambda}$,
	\begin{equation}
	\label{RelBdGeneral}
	\left| \braket{\psi}{V_\Lambda \, \psi} \right|
	\; \leq \;
	\beta
	\braket{\psi}{H_\Lambda\, \psi}
	\end{equation}
	where, given $G(n)=\max_{x\in\Lambda} \|\Phi(x,n)\|$:
	\be
	\beta= c\sum_{n=R}^{\ell} \frac{n^\zeta G(n)}{\gamma(n)}. 
	\ee
\end{thm}

Note that in the case that $\gamma>0$ is a local uniform gap for $\caS$, that
\begin{equation}\label{eq:rel_bd_loc_gap}
\beta \leq  \frac{c}{\gamma} \sum_{n=R}^{\ell} n^{\zeta} G(n).
\end{equation}

The proof of Theorem~\ref{thm:Step3General} below, which follows the argument of \cite[Proposition 2]{michalakis:2013}, uses the collection of finite volumes $\caS$ as well as the associated family of partitions of polynomial growth, $\cT$, separating $\cS$ to define self-adjoint operators $Q_n^i$ and $\Phi_n^i$ as follows. For each $n$ with $R\leq n\leq \ell$, 
let $P_{\Lambda(x,n)}$ be the orthogonal projection onto the ground state space of $H_{\Lambda(x,n)}$, and define $Q_{\Lambda(x,n)}=
\idty - P_{\Lambda(x,n)}$. Denoting by $\caT_n = \{ T_n^i :i \in \caI_n \}$ the $n$-th separating partition, for each $i \in \caI_n$ we
define self-adjoint operators
\begin{equation} \label{two_ops}
H_n^i = \sum_{x\in T_n^i} H_{\Lambda(x,n)}, \quad Q_n^i = \sum_{x\in T_n^i} Q_{\Lambda(x,n)},
\quad
V_{n}^i = \sum_{x\in T_n^i} \Phi(x,n).
\end{equation}
Since the partition part $T_n^i$ has the separating property, i.e. $\Lambda(x,n)\cap \Lambda(y,n) = \emptyset$ for all distinct pairs $x, \, y\in T_n^i$, we see that each of these operators is a sum of commuting terms, i.e. for all $x, y \in T_n^i$
\begin{equation} \label{commuting_hams}
[H_{\Lambda(x,n)}, H_{\Lambda(y,n)}] = 0, \quad [Q_{\Lambda(x,n)}, Q_{\Lambda(y,n)}] = 0, \quad [\Phi(x,n), \Phi(y,n)] = 0 .
\end{equation} 
Moreover, since each $H_{\Lambda(x,n)}$ is non-negative, the separating property of $T_n^i$ also implies that $H_n^i \leq H_\Lambda$. Therefore, since the local gaps necessarily satisfy $\gamma(n)Q_{\Lambda(x,n)} \leq H_{\Lambda(x,n)}$, the following operator inequalities hold:
\begin{equation} \label{fb_base_ham_est}
\gamma(n)  \sum_{i \in \caI_n} Q_n^i\leq \sum_{i \in \caI_n} H_n^i \leq |\caI_n| H_{\Lambda} \leq cn^\zeta H_\Lambda.
\end{equation}

The proof below uses these facts to determine the claimed form bound.

\begin{proof}[Proof of Theorem~\ref{thm:Step3General}]
	
	As above, fix $R \leq n \leq \ell$, and for each $i \in \caI_n$ denote by $C_i$ the collection of
	configurations associated to the part $T_n^i \in \caT_n$, i.e.
	\begin{equation}
	C_i  = \{0,1\}^{T_n^i}  = \left\{ (\sigma_x )_{x \in T_n^i} : \sigma_x \in \{0, 1 \} \right\} \,.
	\end{equation}
	For any $\sigma = (\sigma_x)\in C_i$, define the quantity
	\begin{equation}
	S( \sigma) = \prod_{x\in T_n^i}\left[ \sigma_x(\idty-P_{\Lambda(x,n)}) + (1-\sigma_x)P_{\Lambda(x,n)}\right].
	\end{equation}
	By the separating property, i.e. \eqref{sep_property}, it follows that $[P_{\Lambda(x,n)}, P_{\Lambda(y,n)}]=0$ for all $x,y \in T_n^i$. As a consequence, the set $\{ S( \sigma) : \sigma \in C_i \}$ forms a mutually orthogonal family of
	orthogonal projections that sum to the identity, i.e.
	\begin{equation} \label{SigmaProps}
	S(\sigma) = S(\sigma)^*, 
	\quad
	S(\sigma)S(\sigma') = \delta_{\sigma, \, \sigma'}S(\sigma).
	\;\, \text{and} \;\,
	\sum_{\sigma \in C_i} S(\sigma) = \idtyty,
	\end{equation}
	
	Let $V_n^i$ be as in \eqref{two_ops}. We first show that for any $i \in \caI_n$ and each $\psi \in \mathcal{H}_{\Lambda}$,
	\begin{equation} \label{fb_est_1}
	| \langle \psi, V_n^i \psi \rangle | \leq G(n) \langle \psi, S_n^i \psi \rangle 
	\end{equation}
	where $S_n^i\in \cA_\Lambda$ is defined by
	\begin{equation}
	S_n^i = \sum_{ \sigma \in C_i} | \sigma| S( \sigma) \quad \mbox{with} \quad | \sigma| = \sum_{x \in T_n^i} \sigma_x.
	\end{equation}
	Since $b_x^\Lambda(n) \subseteq \Lambda(x,n)$, the frustration-free property \eqref{FFProp} implies
	\[P_{\Lambda(x,n)} = P_{b_x^\Lambda(n)} P_{\Lambda(x,n)} = P_{\Lambda(x,n)} P_{b_x^\Lambda(n)}.\] 
	Considering \eqref{sep_property} and (\ref{annih_terms}), the above implies
	\begin{equation}
	[\Phi(x,n), P_{\Lambda(y,n)}]=0 \quad \mbox{for all } x,y \in T_n^i. 
	\end{equation} 
	As a consequence,
	$[ V_n^i, S(\sigma)]=0$ for all $\sigma \in C_i,$ and so
	\begin{equation}
	S(\sigma)V_n^iS(\sigma') = \delta_{\sigma, \, \sigma'} \sum_{\substack{x \in T_n^i :\\ \sigma_x=1}}
	S(\sigma) \Phi(x,n) S(\sigma)
	\end{equation}
	for all $\sigma, \sigma' \in C_i$ where we have used both (\ref{SigmaProps}) and (\ref{annih_terms}). The bound
	\begin{equation} \label{Ssig_bd}
	\| S(\sigma)V_n^iS(\sigma) \| \leq | \sigma| G(n) 
	\end{equation}
	readily follows. Given this and (\ref{SigmaProps}), the estimate
	\begin{eqnarray}
	| \langle \psi, V_n^i \psi \rangle| & \leq & \sum_{\sigma, \sigma' \in C_i} | \langle \psi, S(\sigma) V_n^i S(\sigma') \psi \rangle|  \nonumber\\ & \leq & \sum_{\sigma \in C_i} \| S(\sigma) V_n^i S(\sigma) \|  \langle \psi, S(\sigma) \psi \rangle \nonumber \\
	& \leq & G(n) \langle \psi, S_n^i \psi \rangle
	\end{eqnarray}
	holds for any $\psi \in \mathcal{H}_{\Lambda}$ as claimed in \eq{fb_est_1}.
	
	Now, let $Q_n^i$ be as in \eqref{two_ops}. We claim that for any $i \in \caI_n$ and each $\psi \in \mathcal{H}_{\Lambda}$,
	\begin{equation} \label{fb_est_2}
	\langle \psi, S_n^i \psi \rangle  =  \langle \psi, Q_n^i \psi \rangle, 
	\end{equation}
	Since $V_\Lambda= \sum_{n=R}^{\ell} \sum_{i \in \caI_n} V_n^i$, (\ref{fb_est_1}) and (\ref{fb_est_2}) would imply
	\begin{equation}
	| \langle \psi, V_\Lambda \psi \rangle| \leq \sum_{n=R}^{\ell} G(n) \sum_{i \in \caI_n}  \langle \psi, S_n^i \psi \rangle \leq  \sum_{n=R}^{\ell} G(n) \sum_{i \in \caI_n}  \langle \psi, Q_n^i \psi \rangle,
	\end{equation}
	and (\ref{RelBdGeneral}) would follow from (\ref{fb_base_ham_est}). Thus, to complete the proof we need only verify (\ref{fb_est_2}).
	
	To prove (\ref{fb_est_2}), first note that from the separating property of $T_n^i$, and the fact that $P_{\Lambda(x,n)}$ is the orthogonal projection onto the kernel of $H_{\Lambda(x,n)}$, one has
	\begin{equation}
	[Q_{\Lambda(x,n)}, P_{\Lambda(y,n)}]=0 \quad \forall \; x,y \in T_n^i \quad \Rightarrow \quad [Q_{\Lambda(x,n)}, S(\sigma)]=0 \quad \forall \;\sigma \in C_i.
	\end{equation} 
	From this, we conclude that for all $\sigma, \sigma' \in C_i$,
	\be
	S(\sigma) Q_{\Lambda(x,n)} S(\sigma') = \delta_{\sigma, \sigma'} Q_{\Lambda(x,n)} S(\sigma) =  \delta_{\sigma, \sigma'} ( \idtyty - P_{\Lambda(x,n)}) S(\sigma)  =  \delta_{\sigma, \sigma'} \sigma_x S(\sigma).  
	\ee
	The following identities are then straightforward:
	\begin{eqnarray}
	\langle \psi, Q_n^i \psi \rangle & = & \sum_{x \in T_n^i} \sum_{\sigma, \sigma' \in C_i} \langle \psi, S(\sigma) Q_{\Lambda(x,n)} S(\sigma') \psi \rangle \nonumber \\ & =& \sum_{\sigma \in C_i} \sum_{x \in T_n^i} \sigma_x \langle \psi, S(\sigma) \psi \rangle \nonumber \\ & =&  \langle \psi, S_n^i \psi \rangle
	\end{eqnarray}
	as claimed in (\ref{fb_est_2}). 
\end{proof}

Note that the proof provides a form bound of $V_\Lambda$ by $\sum_{x\in\Lambda,n\geq R} Q_{\Lambda(x,n)}$, which in general is stronger than
\eq{RelBdGeneral}.

%
%

\newcommand{\Vee}{V}

\section{Initial Steps and Quasi-Locality} \label{sec:Step1}

\subsection{Introduction}\label{sec:intro_Step1}

In this section, we start the analysis of the transformed quantum spin Hamiltonians as a first step toward the establishing the properties
outlined in Claim~\ref{clm:decompositionH}. In general, we will use the set-up and notations introduced in Section~\ref{sec:set-up}.
Concretely, as in Section~\ref{sec:stab+sf}, we consider local Hamiltonians of the form
\begin{equation} \label{ham_label}
H_{\Lambda}(s) = H_{\Lambda} + s \Vee_{\Lambda^p} \quad \mbox{with} \quad 0 \leq s \leq 1.
\end{equation}
Here and throughout this section $\Lambda$ is a fixed finite subset of  a $\nu$-regular metric space $(\Gamma, d)$. Balls are defined with respect to $\Lambda$:  
$b_x^{\Lambda}(n) = \{ y \in \Lambda : d(x, y) \leq n \}$ for $x \in \Lambda$ 
and $n \geq 0$. We will refer to $H_{\Lambda}(0) = H_{\Lambda}$ as the {\it initial Hamiltonian}. The perturbation, $\Vee_{\Lambda^p}$, is defined with reference to 
a {\em perturbation region} $\Lambda^p\subseteq\Lambda$.
As discussed in Section~\ref{sec:balled-int}, we will further assume that both $H_{\Lambda}$ and $\Vee_{\Lambda^p}$ have been written in 
{\it anchored form}, and in particular, we take
\begin{equation} \label{ham+int}
H_{\Lambda} = \sum_{x \in \Lambda} h_x \quad \mbox{and} \quad \Vee_{\Lambda^p} = \sum_{x \in \Lambda^p} \sum_{n \geq R} \Phi(x,n). \, 
\end{equation}
We will assume that the initial Hamiltonian has an interaction radius bounded by some $R \geq 0$, meaning
$h_x^* = h_x \in \mathcal{A}_{b_x^{\Lambda}(R)}$ for all $x \in \Lambda$. 
Later we will also assume that the initial Hamiltonian is
generated by a frustration-free interaction, however, it is not needed for this section. 
For convenience, we will always assume that $\inf\spec H_\Lambda =0$.
The perturbation $\Vee_{\Lambda^p}$ is an anchored interaction on $\Lambda$ as in Definition~\ref{def:ball-int}:
$\Phi(x,n)^*=\Phi(x,n) \in \mathcal{A}_{b_x^{\Lambda}(n)}$. 
Note that, in general, by redefining $\Phi(x,R)$ we can assume without loss of generality that
$\Phi(x,n) =0$ if $n< R$, as we have above.
The anchored forms of $H_\Lambda$ and $\Vee_{\Lambda^p}$ may have been derived by the procedure 
described in Section~\ref{sec:ball-proc}, but this is not necessary for the analysis which follows. 
For notational convenience we will write
\begin{equation} \label{pert-fam-1}
H(s) = H + s \Vee
\end{equation}
for this family of Hamiltonians satisfying the assumptions detailed above.

Our analysis investigates the ground state gap, ${\rm gap}(H(s))$, as defined in (\ref{fv_pert_gap}). 
It is convenient to set
\be \label{def:gam_0}
\gamma_\Lambda = \gap(H(0)) = \gap(H).
\ee
Then, for any  $0 < \gamma < \gamma_\Lambda$, recall the quantity $s_\gamma^\Lambda$ is as defined in \eq{def:s_lam_gam}:
\begin{equation} \label{def:s_gam}
s_\gamma^\Lambda = \sup\{s' \in [0,1] : \gap(H(s)) \geq \gamma \mbox{ for all } 0 \leq s \leq s' \}.
\end{equation}
If $s_\gamma^\Lambda<1$, then $ \gap(H(s_\gamma^{\Lambda})) = \gamma$. In other words, adding
$s_\gamma^{\Lambda} \Vee$ to $H$ reduces the gap from $\gamma_\Lambda$ to $\gamma$.

The first step towards the results in Claim~\ref{clm:decompositionH} is to define an anchored interaction $\Phi^{(1)}$ 
such that the Hamiltonian transformed by the spectral flow satisfies
\be\label{expandPhi1}
\alpha_s(H(s)) = H + \Vee^{(1)}(s) \quad \mbox{with} \quad \Vee^{(1)}(s) = \sum_{x\in\Lambda}\Phi_x^{(1)}(s)
\quad \mbox{and} \quad \Phi_x^{(1)}(s) = \sum_{n\geq R}\Phi^{(1)}(x,n,s).
\ee
Moreover, the following two properties hold.
First, with $P(0)$ the spectral projection onto the ground state space of $H(0) = H$, we have
\be\label{trans_pert_commute}
[\Phi_x^{(1)}(s), \, P(0)] = 0 \quad \mbox{for all } x \in \Lambda  \mbox{ and } 0 \leq s\leq s_\gamma^\Lambda.
\ee
Second, the anchored terms described in (\ref{expandPhi1}) satisfy the estimate
\be \label{trans_pert_decay}
\|\Phi^{(1)}(x,n,s)\| \leq sG^{(1)}(n) \quad  \mbox{for all } x \in \Lambda, n\geq R, \mbox{ and }  0 \leq s\leq 1.
\ee
Here $G^{(1)}(n)$ vanishes as $n\to\infty$ at a certain rate, specified by a {\em decay class}, a notion we define in Definition \ref{def:dec-class}.
These two properties are proved in Proposition~\ref{prop:step1-commute} and Theorem \ref{thm:Step1}, respectively.

We will express the decay assumptions on the perturbation using a particular form of $F$-function on $(\Gamma, d)$.
Let $F_0 : [0, \infty) \to (0, \infty)$ be an $F$-function on $(\Gamma, d)$. By $\nu$-regularity, we can take $F_0$ as in (\ref{poly-dec-F}), 
but it is not necessary to assume this explicit form here. In any case, we will call this $F_0$ the base $F$-function. As in Section~\ref{sec:Fnorm}, 
given any $g : [0, \infty) \to [0, \infty)$ which is non-decreasing and sub-additive, the function $F: [0, \infty) \to (0, \infty)$
defined by
\begin{equation} \label{weighted-F-fun}
F(r) = e^{- g(r)} F_0(r) \quad \mbox{for any } r \geq 0
\end{equation} 
is also an $F$-function on $(\Gamma, d)$. 
We will further assume that the weight $e^{-g}$ decays at least as fast as a {\it stretched exponential}, i.e. there is some $a>0$ and $0< \theta \leq 1$ for which
\begin{equation} \label{stab_g_grows}
g(r) \geq a r^{\theta} \,  \quad \mbox{for all } r \geq 0 \, .
\end{equation}

For ease of later reference, we summarize the basic assumptions on the initial (unperturbed) Hamiltonian $H$ and the perturbation $\Vee$, see (\ref{pert-fam-1}), 
in the following.

\begin{assumption}[Assumptions on $H$ and $\Vee$]\label{ass:basic}
	For a quantum spin system defined on a finite $(\Lambda,d)$ we impose the following conditions.
	\begin{enumerate}
		\item[(i)]$H$ is described by a finite-range, frustration-free interaction in anchored form of maximal radius $R\geq 0$ as follows: 
		\be
		H=\sum_{x\in\Lambda} h_x \quad \mbox{with} \quad h^*_x=h_x\in \cA_{b^\Lambda_x(R)}.
		\ee
		\item[(ii)]The perturbation is given in terms of $\Lambda^p\subset \Lambda$ and $\Phi(x,n)^*=\Phi(x,n) \in \cA_{b^\Lambda_x(n)}$ for $R\leq n \leq \diam(\Lambda)$, i.e.
		\be
		\Vee = \sum_{x\in\Lambda^p}\sum_{n\geq R} \Phi(x,n),
		\ee
		and there is a weighted $F$-function $F$ as in \eq{weighted-F-fun} for $g$ of the form \eq{stab_g_grows}, 
		and a constant $\Vert \Phi \Vert_{1,F}$ such that
		\begin{equation} \label{ini-int-dec}
		\sum_{x \in \Lambda^p} \sum_{\stackrel{n \geq R:}{y,z \in b_x^{\Lambda}(n)}} | b_x^{\Lambda}(n)| \| \Phi(x,n) \| \leq \| \Phi \|_{1,F} F(d(y,z)) \, ,
		\mbox{ for all } y,z\in\Lambda,
		\end{equation}
		and
		\begin{equation} \label{int-bd-con}
		\| \Phi(x,m) \| \leq 
		\| \Phi \|_{1,F} F(\max(0,m-1)), \quad m\geq R.
		\end{equation}
	\end{enumerate}
\end{assumption}

We note that for an anchored interaction satisfying \eqref{support_property}, for example the ones derived for a general interaction as in Section~\ref{sec:ball-proc},
\eq{int-bd-con} follows from \eq{ini-int-dec}.

\subsection{Application of the Spectral Flow} \label{sec:HGSF}\label{sec:conv-lab}

An essential tool for our analysis here is the spectral flow discussed in Section \ref{sec:stab+sf}.
Consider a fixed value $0 < \gamma < \gamma_\Lambda$. For each $0 \leq s \leq 1$, we denote by $\alpha_s$ the 
spectral flow automorphism of $\mathcal{A}_{\Lambda}$ as defined in (\ref{spec-flow-auto}). Here we
have taken $\xi = \gamma$ and we suppress this in our notation. A crucial property is that 
for $s\in [0,s_\gamma^\Lambda]$, that is when $\gap(H(s))\geq \gamma$, we have $\alpha(P(s))=P(0)$, where
$P(s)$ is the spectral projection of $H(s)$ corresponding to $\Sigma_1(s)$ as defined in \eq{spec_sets}
with $\Sigma_1(0)=\{0\}$.

For each fixed $s\in [0,1]$, we have the Heisenberg dynamics associated to (\ref{pert-fam-1}):
\begin{equation} \label{stab_dynamics}
\tau_t^{(s)}(A) = e^{it H(s)} A e^{-itH(s)}, \quad \mbox{for } A \in \mathcal{A}_{\Lambda} \mbox{ and } t \in \mathbb{R}.
\end{equation}
We also consider the family of linear maps $\{ \mathcal{F}_s \}_{s \in [0,1]}$ with $\mathcal{F}_s : \mathcal{A}_{\Lambda} \to \mathcal{A}_{\Lambda}$ given by 
\begin{equation} \label{def_wio_F}
\mathcal{F}_s(A) = \int_{\mathbb{R}} \tau_t^{(s)}(A) w_{\gamma}(t) \, dt \quad \mbox{for all } A \in \mathcal{A}_{\Lambda} \mbox{ and } 0 \leq s \leq 1\, .
\end{equation}
Here $w_{\gamma}$ is the real-valued function in $L^1( \mathbb{R})$ defined in (6.32) of Section VI.B in \cite{nachtergaele:2019}, and
we will refer to $\{ \mathcal{F}_s \}_{s \in [0,1]}$ as the family of integral operators with 
weight function $w_{\gamma}$. As in the case of the spectral flow, we have suppressed the dependence of this family on
the value of $\gamma >0$ to ease notation. One readily checks that:
\begin{enumerate}
	\item[(i)] Since $H(s)$ generates the Heisenberg dynamics $\tau_t^{(s)}$,
	\begin{equation} \label{F-gen-ham-inv}
	\mathcal{F}_s(H(s)) = H(s) \quad \mbox{for all } 0 \leq s \leq 1.
	\end{equation}
	\item[(ii)] With this particular choice of weight function $w_{\gamma}$,
	\begin{equation} \label{FA_commutes}
	[ \mathcal{F}_s(A), P(s)] = 0 \quad \mbox{for all } A \in \mathcal{A}_{\Lambda} \mbox{ and each } 0 \leq s \leq s_{\gamma}^{\Lambda} \, .
	\end{equation}
\end{enumerate}
Equation (\ref{FA_commutes}) follows from the fact that the Fourier transform of $w_{\gamma}$ has support in $[-\gamma,\gamma]$,
which immediately implies $(\idty- P(s)) \mathcal{F}_s(A) P(s) =0$. See, e.g., \cite{hastings:2004b} or \cite[Lemma 6.8 ]{nachtergaele:2019}. 

Consider now the difference
\begin{equation}
\alpha_s(H(s)) - H =   \alpha_s( \mathcal{F}_s(H(s))) - \mathcal{F}_0(H) \quad \mbox{for all } 0 \leq s \leq 1 \, ,
\end{equation}
where we have used (\ref{F-gen-ham-inv})  to insert the corresponding integral operators which leave their generating Hamiltonians
invariant. Using $H(s) = H + s\Vee$, this difference can be rewritten as
\begin{eqnarray} \label{rewrite-diff}
\alpha_s(H(s)) - H 
& = & ( \alpha_s - {\rm id})( \mathcal{F}_s(H)) + ( \mathcal{F}_s - \mathcal{F}_0)(H) + s \alpha_s( \mathcal{F}_s(\Vee)) \nonumber \\
& = & \mathcal{K}_s^1(H) + \mathcal{K}_s^2(H) + \mathcal{K}_s^3 (\Vee), 
\end{eqnarray}
where we introduced three families of linear maps $\{ \mathcal{K}_s^i\}_{s \in [0,1]}$,
with $\mathcal{K}_s^i : \mathcal{A}_{\Lambda} \to \mathcal{A}_{\Lambda}$ for each $0 \leq s \leq 1$ and
$i=1,2,3$, given by
\begin{equation} \label{3_fams_maps}
\mathcal{K}_s^1 = (\alpha_s - {\rm id}) \circ \mathcal{F}_s, \quad \mathcal{K}_s^2 = \mathcal{F}_s - \mathcal{F}_0, \quad \mbox{and} \quad \mathcal{K}_s^3 = s \alpha_s \circ \mathcal{F}_s \, .
\end{equation} 

With an eye towards our goal of applying Theorem~\ref{thm:Pert_Est}, we summarize (\ref{rewrite-diff}) differently,
\begin{equation} \label{step1_transform}
\alpha_s(H(s)) = H + \Vee^{(1)}(s) \quad \mbox{for all }  0 \leq s \leq 1, 
\end{equation}
where we have set $\Vee^{(1)}(s) =  \mathcal{K}_s^1(H) + \mathcal{K}_s^2(H) + \mathcal{K}_s^3 (\Vee)$. By expanding 
as in \eq{expandPhi1}
we further write
\begin{equation} \label{phi1x-def}
\Vee^{(1)}(s) = \sum_{x \in \Lambda} \Phi_x^{(1)}(s) \quad \mbox{with} \quad 
\Phi_x^{(1)}(s) = \mathcal{K}_s^1(h_x) +  \mathcal{K}_s^2(h_x) + \chi_{\Lambda^p}(x) \cdot \sum_{n \geq R} \mathcal{K}_s^3( \Phi(x,n)).
\end{equation}
Here $\chi_{\Lambda^p}$ is the characteristic function of the perturbation region $\Lambda^p \subseteq \Lambda$, see (\ref{ham+int}).

\begin{prop} \label{prop:step1-commute}
	With assumptions and notation as above, for all $x \in \Lambda$,
	\begin{equation} \label{global-terms-commute}
	[ \Phi_x^{(1)}(s), P(0)] = 0 \quad \mbox{for all } 0 \leq s \leq s_{\gamma}^{\Lambda} \, ,
	\end{equation}
	where $P(0)$ is the spectral projection onto the ground state space of $H(0) = H$ and
	$s_{\gamma}^{\Lambda}$ is as in (\ref{def:s_gam}).
\end{prop}

\begin{proof}
	Note that for each $x \in \Lambda$ and any $0 \leq s \leq 1$, one has
	\begin{equation}
	\mathcal{K}_s^1(h_x) +  \mathcal{K}_s^2(h_x)  = (\alpha_s - {\rm id})( \mathcal{F}_s(h_x)) + (\mathcal{F}_s - \mathcal{F}_0)(h_x) = (\alpha_s \circ \mathcal{F}_s)(h_x) - \mathcal{F}_0(h_x) \, .
	\end{equation}
	Thus for $x$ and $s$ as above,
	\begin{equation} \label{phi_1_x_comm}
	\Phi^{(1)}_x(s) = ( \alpha_s \circ \mathcal{F}_s)(h_x) - \mathcal{F}_0(h_x) + s \chi_{\Lambda^p}(x) \sum_{n \geq R} (\alpha_s \circ \mathcal{F}_s)( \Phi(x,n)) \, .
	\end{equation}
	We now use (\ref{FA_commutes}). In fact, the case of $s=0$ implies that $[ \mathcal{F}_0(h_x), P(0)]=0$. Moreover, 
	\begin{equation}
	[ (\alpha_s \circ \mathcal{F}_s)(A), P(0)] = \alpha_s \left( [ \mathcal{F}_s(A), P(s)] \right) = 0  \mbox{ for all } A \in \mathcal{A}_{\Lambda}
	\mbox{ whenever } 0 \leq s \leq s_\gamma^\Lambda \, ,
	\end{equation}
	where we have also used (\ref{spec_flow_proj}). Given (\ref{phi_1_x_comm}), we now conclude that $[ \Phi^{(1)}_x(s), P(0)]=0$ for all $x \in \Lambda$ and $0 \leq s \leq s_\gamma^\Lambda$.  This completes the proof.
\end{proof}

%
%

\subsection{Application of Quasi-Locality and Local Decompositions} \label{sec:QL+LD}
In this section, we review the notions of quasi-locality and local decompositions.
For the interested reader, more details on this can be found in \cite[Section IV]{nachtergaele:2019}.

A linear map $\mathcal{K} : \mathcal{A}_{\Lambda} \to \mathcal{A}_{\Lambda}$ is said to satisfy a quasi-locality bound of order 
$q \geq 0$ if there is a non-increasing function $G:[0, \infty) \to [0, \infty)$ with $\lim_{r \to \infty}G(r) =0$
for which given any sets $X,Y \subset \Lambda$ and observables $A \in \mathcal{A}_X$ and $B \in \mathcal{A}_Y$, the bound
\begin{equation} \label{gen_ql_map}
\| [ \mathcal{K}(A), B ] \| \leq |X|^q \| A \| \| B \| G(d(X,Y)) 
\end{equation}
holds. Any linear map $\mathcal{K}$ satisfying (\ref{gen_ql_map}) will be referred to as {\it quasi-local}. 
As is well-known, Lieb-Robinson bounds are useful in demonstrating quasi-locality of the Heisenberg dynamics
associated to sufficiently short-range interactions. For the analysis at hand, we will need 
explicit quasi-locality bounds for other maps as well; for example, the 
spectral flow automorphism, as introduced in Section~\ref{sec:stab+sf}, and various weighted integral operators, see e.g. (\ref{def_wio_F}).  
Before describing these particular estimates, let us continue with some generalities. 

One important fact about quasi-local maps is that they can be approximated by strictly local maps with
errors quantified in terms of their decay function, i.e. the function $G$ in (\ref{gen_ql_map}) above.
To make this precise, we recall the notion of localizing maps and local decompositions.
For any $X \subseteq \Lambda$, denote by $\tilde{\Pi}_X^{\Lambda}$ the normalized 
partial trace over $\mathcal{H}_{\Lambda \setminus X}$, i.e. the unique linear map 
$\tilde{\Pi}_{X}^{\Lambda} : \mathcal{A}_{\Lambda} \to \mathcal{A}_{X}$ for which 
\begin{equation} \label{pi_partial_trace}
\tilde{\Pi}_{X}^{\Lambda}(A \otimes B) = {\rm Tr}[B] A \quad \mbox{for all } A \in \mathcal{A}_{X} \mbox{ and } B \in \mathcal{A}_{\Lambda \setminus X} \, ,
\end{equation}
where ${\rm Tr}[B]$ denotes the normalized trace of $B$ as an operator on the finite dimensional $\mathcal{H}_{\Lambda \setminus X}$.
If $X = \Lambda$, then the above map is understood to act as the identity, i.e.
$\tilde{\Pi}_{\Lambda}^{\Lambda}(A) = A$ for all $A \in \mathcal{A}_{\Lambda}$.
We will denote by $\Pi_{X}^{\Lambda}$ the map from $\mathcal{A}_{\Lambda}$ to $\mathcal{A}_{\Lambda}$
defined by $A \mapsto \tilde{\Pi}_{X}^{\Lambda}(A) \otimes \idty_{\Lambda \setminus X}$. 
We refer to these maps $\{ \Pi_{X}^{\Lambda} \}_{X \subset \Lambda}$ 
as localizing, or strictly local, maps on $\mathcal{A}_{\Lambda}$. The choice of notation stems from a more
general discussion of localizing maps for quantum lattice systems, see \cite[Section IV]{nachtergaele:2019}.
As we are only considering quantum spin systems here, the normalized 
partial trace is a simple realization of these more general maps. 

Given these localizing maps, it is also convenient to introduce the corresponding local 
decompositions. Let $x \in \Lambda$, $n \geq 0$, and for any $m \geq n$ define 
a map $\Delta_{x,n; m}^{\Lambda} : \mathcal{A}_{\Lambda} \to \mathcal{A}_{\Lambda}$ by setting
\begin{equation} \label{def:Delta}
\Delta_{x,n;m}^{\Lambda} = \left\{ \begin{array}{cl} \Pi_{b_x^{\Lambda}(n)}^{\Lambda} & \mbox{if } m =n, \\
\Pi_{b_x^{\Lambda}(m)}^{\Lambda} - \Pi_{b_x^{\Lambda}(m-1)}^{\Lambda} & \mbox{if } m >n. \end{array} \right.
\end{equation}
Note that each $\Delta_{x,n;m}^{\Lambda}$ has range contained in $\mathcal{A}_{b_x^{\Lambda}(m)}$, regarded as a sub-algebra of $\mathcal{A}_{\Lambda}$.
Moreover, one has that
\begin{equation} \label{sum-to-m} 
\sum_{m=n}^M \Delta_{x,n;m}^{\Lambda}(A) = \Pi_{b_x^{\Lambda}(M)}^{\Lambda}(A) \quad \mbox{for each } M \geq n \mbox{ and all } A \in \mathcal{A}_{\Lambda} \, .
\end{equation}
Since $\Lambda$ is finite, for each $x \in \Lambda$, there is $M$ sufficiently large so that $b_x^{\Lambda}(M) = \Lambda$. In this case,
for any integer $n$, we have that 
\begin{equation} \label{telescope}
A = \sum_{m \geq n} \Delta_{x,n;m}^{\Lambda}(A) \quad \mbox{for all } A \in \mathcal{A}_{\Lambda} \, 
\end{equation}
and the sum on the RHS above is finite.

A more general version of the following lemma is given in \cite[Lemma 5.1]{nachtergaele:2019}. 
It provides a simple estimate for 
strictly local approximations of quasi-local maps.

\begin{lemma} \label{lem:qlm_loc_est}
	Let $\mathcal{K} : \mathcal{A}_{\Lambda} \to \mathcal{A}_{\Lambda}$ be a quasi-local map satisfying (\ref{gen_ql_map}).
	For any $x \in \Lambda$, $n \geq 0$, $A \in \mathcal{A}_{b_x^{\Lambda}(n)}$, and each $m \geq n$, one has that 
	\begin{equation}
	\| \mathcal{K}(A) - \Pi_{b_x^{\Lambda}(m)}^{\Lambda}(\mathcal{K}(A)) \| \leq 2 |b_x^{\Lambda}(n)|^q \| A \| G(m-n) 
	\end{equation}
	and as a result
	\begin{equation} \label{Delta_bd}
	\| \Delta_{x,n;m}^{\Lambda} (\mathcal{K}(A)) \| \leq 4 |b_x^{\Lambda}(n)|^q \| A \| G(m-n-1) \quad \mbox{for all } m > n \, .
	\end{equation}
\end{lemma}

%
%

\subsubsection{First Estimates} \label{sec:1ests}

The goal of this section is to prove Theorem~\ref{thm:Step1}.
%
We begin by stating a crucial technical estimate that we will prove in Section~\ref{sec:just-tech}. 

\begin{lemma} \label{lem:qlm_ests}
	Under Assumption \ref{ass:basic}, the three families of maps $\{ \mathcal{K}_s^i \}_{s \in [0,1]}$, $i=1,2,3$,
	as defined in (\ref{3_fams_maps}), satisfy the following locally bounded and quasi-local estimates:
	\begin{enumerate}
		\item[(i)] Locally bounded: There are numbers $B_i \geq 0$ and $p_i \geq 0$ for which: given any $X \subset \Lambda$,
		\begin{equation} \label{ki_lb_est}
		\| \mathcal{K}_s^i(A) \| \leq s \cdot B_i |X|^{p_i} \| A \| \quad \mbox{for any } A \in \mathcal{A}_X \, ,
		\end{equation}
		In fact, $p_1=2$, $p_2 = 1$, and $p_3=0$.
		\item[(ii)]Quasi-local: There are numbers $q_i \geq 0$ and non-increasing functions $G_i:[0, \infty) \to [0, \infty)$ for which: 
		given $X,Y \subset \Lambda$ and observables $A \in \mathcal{A}_X$ and $B \in \mathcal{A}_Y$, 
		\begin{equation}  \label{ki_ql_est}
		\| [ \mathcal{K}_s^i(A) , B ] \| \leq s \cdot |X|^{q_i} \| A \| \| B \| G_i(d(X,Y)) \, .
		\end{equation}
		In fact, $\lim_{r \to \infty} G_i(r) = 0$ and $q_1 =2$, $q_2 =2$, and $q_3 =1$.
	\end{enumerate}
\end{lemma}

Lemma~\ref{lem:qlm_ests} demonstrates that the three families of maps introduced  in \eq{3_fams_maps}
are locally bounded and quasi-local with estimates that are uniform in $s\in [0,1]$.
As will be established in the next section, our quasi-locality bounds
yield explicit decay functions $G_i$, for $i=1,2,3$. Rather than compiling all the various estimates we obtain, we prefer to
describe a class of decay functions which captures, in principle, the worst case scenario
in all these bounds. To this end, for any $\xi >0$, introduce a function $f_\xi:[0, \infty) \to (0, \infty)$  by setting
\begin{equation} \label{fxi-dec}
f_\xi(r) = \left\{ \begin{array}{cl} \frac{e^2}{4} & \mbox{if } 0 \leq r \leq \xi^{-1} e^2, \\
\frac{ \xi r}{ (\ln( \xi r))^2} & \mbox{if } r > \xi^{-1} e^2. \end{array} \right.
\end{equation} 
In what follows, we will frequently make reference to the following decay class:
\begin{defn} \label{def:dec-class} Let $\eta$, $\xi$, and $\theta$ be positive numbers. 
	We will say that a function $G: [0, \infty) \to (0, \infty)$ is of \emph{decay class $(\eta, \xi, \theta)$} if for every $0 < \eta' < \eta$, there are positive numbers 
	$C_1$, $C_2$, $a$, and $d$, with $C_1 \geq C_2e^{-\eta' f_{\xi}(a d^{\theta})}$, for which the estimate
	\begin{equation}
	G(r) \leq \left\{ \begin{array}{cl} C_1 & \mbox{if } 0 \leq r \leq d \\
	C_2 e^{- \eta' f_{\xi}(ar^{\theta})} & \mbox{if } r > d \end{array} \right.
	\end{equation}
	holds for all $r \geq 0$. Here, $f_{\xi}$ is as in (\ref{fxi-dec}) above.
\end{defn} 

\begin{remk} \label{rem:dec-class} 
	Our estimates will frequently use basic properties of these decay classes. 
	First, note that each of these decay classes is closed under addition and 
	multiplication by non-negative scalars. Next, if $G$ is in a particular decay class, then
	for any $p > 0$ and $c>0$, the functions
	\begin{equation}
	G_1(r) = (1+r)^p G(r) \quad \mbox{and moreover} \quad G_2(r) = \sum_{n \geq \lfloor cr \rfloor} n^p G(n)
	\end{equation}
	are both in the same decay class, as one easily checks.
\end{remk}

\begin{remk} \label{rem:dec-triple} The proof of Lemma~\ref{lem:qlm_ests} actually establishes that each of the functions $G_i$
	is of a particular decay class. More precisely, let $\eta>0$ be the number in (\ref{def:eta}) below. 
	Take $\xi = \gamma / 2v$ to be the ratio of $\gamma$, the fixed number 
	(strictly less than $\gamma_\Lambda$) which
	is used in the definition of the spectral flow, and $2v$ where $v$ is an estimate on the Lieb-Robinson velocity for
	the dynamics corresponding to the family of Hamiltonians $H(s)$ under investigation, see (\ref{stab-dyn-lrb}) and (\ref{lrb-vest}). 
	Finally, let $\theta$ be as in (\ref{stab_g_grows}).  The proof of Lemma~\ref{lem:qlm_ests} demonstrates that
	each $G_i$ is of decay class $( \eta, \frac{\gamma}{2v}, \theta)$. 
	In particular, if $\gamma$, $v$, and $\| \Phi \|_{1, F}$ satisfy volume-independent estimates,
	then these decay functions may be chosen independent of the volume.   
\end{remk}

We can now state the main result of this section.

\begin{thm} 
	\label{thm:Step1}
	Under Assumption \ref{ass:basic}, the transformed Hamiltonian, see (\ref{step1_transform}), can be written as
	\begin{equation} \label{uni-equi-rep-step-1}
	\alpha_s(H(s)) = H + \Vee^{(1)}(s)  \quad \mbox{with} \quad \Vee^{(1)}(s) = \sum_{x \in \Lambda} \sum_{m \geq R} \Phi^{(1)}(x,m,s)
	\end{equation}
	for all $0 \leq s \leq 1$. Here the terms above satisfy 
	\begin{equation}
	\Phi^{(1)}(x,m,s)^* = \Phi^{(1)}(x,m,s) \in \mathcal{A}_{b_x^{\Lambda}(m)} \quad \mbox{for all } x \in \Lambda, \, m \geq R, \mbox{ and } 0 \leq s \leq 1,
	\end{equation}
	and for $x$, $m$, and $s$ as above, the bound 
	\begin{equation} \label{phi1_est}
	\| \Phi^{(1)}(x, m, s) \| \leq s \cdot G^{(1)}(m)  \, 
	\end{equation}
	holds for some function $G^{(1)} : [0, \infty) \to (0, \infty)$ in the decay class $(\eta, \frac{\gamma}{2v}, \theta)$. Here, the parameters in this
	triple are as in Remark~\ref{rem:dec-triple}.
\end{thm}

\begin{remk}
	\indent
	\begin{enumerate}
		\item[(i)] As will be clear from the proof below, a choice for the decay function $G^{(1)}$ 
		can be made explicit in terms of various, previously introduced decay functions.  More importantly, when 
		$\gamma$, $v$, and $\| \Phi \|_{1, F}$ satisfy volume-independent estimates,
		$G^{(1)}$ may be chosen in a volume-independent manner as well.  
		\item[(ii)] We note that, technically, no gap assumption is used in the proof of Theorem~\ref{thm:Step1}. 
		In fact, if the spectral flow used to transform the Hamiltonian, see (\ref{uni-equi-rep-step-1}), is defined
		with respect to any $\xi>0$, as in (\ref{spec-flow-auto}), then Theorem~\ref{thm:Step1} still holds
		and the resulting decay function $G^{(1)}$ is in the decay class $(\eta, \frac{\xi}{2v}, \theta)$.  
	\end{enumerate}
\end{remk}

%

\begin{proof}
	Recall from Section~\ref{sec:conv-lab}, see (\ref{rewrite-diff})-(\ref{phi1x-def}), that for all $0 \leq s \leq 1$
	\begin{equation} \label{uni-equi-rep-step-1_2}
	\alpha_s(H(s)) = H + \Vee^{(1)}(s)  \quad \mbox{with} \quad \Vee^{(1)}(s) =  \sum_{x \in \Lambda} \Phi^{(1)}_x(s),
	\end{equation}
	where 
	\begin{equation} \label{phi1x-2}
	\Phi_x^{(1)}(s) = \mathcal{K}_s^1(h_x) +  \mathcal{K}_s^2(h_x) + \chi_{\Lambda^p}(x) \cdot \sum_{n \geq R} \mathcal{K}_s^3( \Phi(x,n))  \, .
	\end{equation}
	For $i=1,2,3$, the families of maps $\{ \mathcal{K}_s^i \}_{s \in [0,1]}$ are as in (\ref{3_fams_maps}),
	and $\chi_{\Lambda^p}$ is the characteristic function of $\Lambda^p \subseteq \Lambda$. {F}rom Lemma~\ref{lem:qlm_ests}, each of
	the maps $\mathcal{K}_s^i$ is locally bounded and quasi-local. In this case, the terms on the right-hand-side
	of (\ref{phi1x-2}) can be approximated by strictly local terms with error estimates controlled using Lemma~\ref{lem:qlm_loc_est} as follows:
	\begin{equation} \label{phi_1_x_local}
	\Phi^{(1)}_x(s) = \sum_{m \geq R} \Phi^{(1)}(x,m,s)
	\end{equation}
	where we used the local decompositions from \eqref{def:Delta} to define
	\begin{equation} \label{phi_1_x_bits}
	\Phi^{(1)}(x,m,s) = \Delta_{x, R; m}^{\Lambda}(\mathcal{K}_s^1(h_x)) +  \Delta_{x,R;m}^{\Lambda}(\mathcal{K}_s^2(h_x))+  \chi_{\Lambda^p}(x) \cdot \sum_{n=R}^m \Delta_{x,n;m}^{\Lambda}( \mathcal{K}_s^3(\Phi(x,n) )) \, .
	\end{equation}
	We need only estimate these terms as in (\ref{phi1_est}). 
	
	First, consider the case of $m=R$. One has that
	\begin{eqnarray} \label{phi_1_x_bits_1_bd}
	\| \Phi^{(1)}(x, R, s)  \| & \leq & \| \Pi_{b_x^{\Lambda}(R)}^{\Lambda}(\mathcal{K}_s^1(h_x) )\| +  \| \Pi_{b_x^{\Lambda}(R)}^{\Lambda}( \mathcal{K}_s^2(h_x) )\|  + \| \Pi_{b_x^{\Lambda}(R)}^{\Lambda}( \mathcal{K}_s^3( \Phi(x,R))) \|   \nonumber \\
	& \leq & s \| h_x \| \sum_{i=1}^2 B_i |b_x^{\Lambda}(R)|^{p_i}  + s \| \Phi(x,R) \| B_3 |b_x^{\Lambda}(R)|^{p_3} 
	\end{eqnarray}
	where we have used the form of the local decompositions, see (\ref{def:Delta}), and Lemma~\ref{lem:qlm_ests}(i).
	Since $\Lambda$ is finite, each of $\max_{x \in \Lambda} \| h_x \|$, $\max_{x \in \Lambda} \|\Phi(x,R) \|$, and 
	$\max_{x \in \Lambda}|b_x^{\Lambda}(R)|$ are as well, and it is clear that an estimate of the form (\ref{phi1_est}) holds.
	
	For $m>R$, we estimate as follows:
	\begin{eqnarray} \label{l_n_bits_est}
	\| \Phi^{(1)}(x, m, s)  \| & \leq & \| \Delta_{x,R;m}^{\Lambda}(\mathcal{K}_s^1(h_x) )\| +  \| \Delta_{x,R;m}^{\Lambda}( \mathcal{K}_s^2(h_x) )\|  + 
	\sum_{n=R}^m \| \Delta_{x,n;m}^{\Lambda}( \mathcal{K}_s^3( \Phi(x,n))) \| \nonumber \\
	& \leq & 4 s \left( \| h_x \| \sum_{i=1}^2  |b_x^{\Lambda}(R)|^{q_i} G_i(m-R-1) + \sum_{n=R}^{m-1} |b_x^{\Lambda}(n)|^{q_3} \| \Phi(x,n) \| G_3(m-n-1) \right)  \nonumber \\
	& \mbox{ } & \quad + s B_3 |b_x^{\Lambda}(m)|^{p_3} \| \Phi(x,m) \| 
	\end{eqnarray}
	where we have used Lemma~\ref{lem:qlm_ests}(ii) as input for the bounds on the local decompositions proven
	in Lemma~\ref{lem:qlm_loc_est}.The first two terms on the right-hand-side of (\ref{l_n_bits_est}) clearly decay in $m$, and moreover, the
	final term, which corresponds to $n=m$ in (\ref{phi_1_x_bits}) and arises from 
	local bounds as in (\ref{phi_1_x_bits_1_bd}),  may be further estimated using (\ref{int-bd-con}). 
	We need only extract decay from the terms with $R \leq n \leq m-1$. Since $G_3$ is non-increasing and $q_3=1$, we find that 
	\begin{equation} \label{sum_est_s1_1}
	\sum_{n=R}^{\lfloor m/2 \rfloor -1}  |b_x^{\Lambda}(n)|^{q_3} \| \Phi(x,n) \| G_3(m-n-1) \leq \| \Phi \|_{1,F} F(R-1) G_3(m/2)
	\end{equation}
	whereas
	\begin{equation}  \label{sum_est_s1_2}
	\sum_{n= \lfloor m/2 \rfloor}^{m -1}  |b_x^{\Lambda}(n)|^{q_3} \| \Phi(x,n) \| G_3(m-n-1) \leq G_3(0) \| \Phi \|_{1,F} F(\lfloor m/2 \rfloor -1)
	\end{equation}
	both using (\ref{int-bd-con}) to control the interaction terms. 
	
	To summarize, we have shown that for $m>R$,
	\begin{eqnarray} \label{explicit_est_1}
	\| \Phi^{(1)}(x, m, s)  \| & \leq & s \left( 4 \max_{x \in \Lambda} \| h_x \| \right) \sum_{i=1}^2  |b_x(R)|^{q_i} G_i(m-R-1)  \nonumber \\
	& \mbox{ } & + s \| \Phi \|_{1,F} \left[ 4 F(R-1) G_3(m/2) + 4 G_3(0) F(\lfloor m/2 \rfloor -1) + B_3 \kappa^{p_3} m^{p_3 \nu} F(m-1) \right]  
	\end{eqnarray}
	As indicated in Remark~\ref{rem:dec-triple}, each function $G_i$, for $i =1,2,3$, is in the decay class $(\eta, \frac{\gamma}{2v}, \theta)$, and 
	the function $F$ decays even faster. Using Remark~\ref{rem:dec-class}, we conclude that the bound above
	is also in this decay class, and this completes the proof.
\end{proof}

\begin{remk} 
	If the initial perturbation $V$ is a anchored interaction, as in Definition~\ref{def:ball-int-td}, satisfying \eqref{support_property} then
	by arguing as in Appendix~\ref{sec:est-trans-balled-ints}, it is clear that $\Vee^{(1)}(s)$ is a
	$s$-dependent anchored interaction that satisfies \eqref{dep_support_property}. In fact, one easily checks that 
	$s \mapsto \Phi^{(1)}(x,m,s)$ is continuous for each choice of $(x,m)$, and this
	result also follows from the more general discussion found in \cite[Section IV.B.1]{nachtergaele:2019}. 
\end{remk}

We end this section with an estimate of the global terms $\Phi^{(1)}_x$ anchored at sites $x$ {\em outside} the original perturbation region
$\Lambda^p$.

\begin{lemma} \label{lem:FarDecay}
	Under Assumption~\ref{ass:basic}, consider the transformed Hamiltonian $\alpha_s(H(s))$, as in (\ref{step1_transform}).
	Let $N \geq R$ and take $(\eta, \frac{\gamma}{2v}, \theta)$ as in Remark~\ref{rem:dec-triple}.
	Then there is a function $G:[0, \infty) \to (0, \infty)$ of 
	decay class $(\eta, \frac{\gamma}{2v}, \theta)$ for which the global term  $\Phi_x^{(1)}(s)$ as in (\ref{phi1x-def}) satisfies
	\begin{equation}\label{FarDeltaDecay}
	\|\Phi_x^{(1)}(s)\| \leq 2 s \| h_x \| \cdot G(N)
	\end{equation}
	 for all $x \in \Lambda$ with $d(x,\Lambda^p) \geq N$ and $0 \leq s \leq 1.$
\end{lemma}
We prove this lemma at the end of the next section. 

\subsubsection{Technical details of the quasi-local estimates} \label{sec:just-tech}

In this section, we will prove the technical estimates claimed in Lemma~\ref{lem:qlm_ests}
as well as Lemma~\ref{lem:FarDecay} which will be useful for arguments in Section~\ref{sec:LTQO2}.

The estimates which form the core of Lemma~\ref{lem:qlm_ests} are proven using various bounds
on the composition of quasi-local maps as established in \cite[Section V.C]{nachtergaele:2019}. As input, we must first
quantify quasi-local bounds on the Heisenberg dynamics, various integral operators, and the
spectral flow. We turn to this topic first. 

Under Assumption~\ref{ass:basic}, it is well-known that the Heisenberg dynamics associated to the Hamiltonians $H(s)$, 
see (\ref{pert-fam-1}), satisfies a quasi-locality bound. In fact, 
an application of Theorem~\ref{thm:LRBounds}, here we first use Proposition~\ref{prop:TriIneq_Fnorm}, shows that
for $X,Y \subset \Lambda$ with $X \cap Y = \emptyset$ and any $A \in \mathcal{A}_X$ and $B \in \mathcal{A}_Y$, one has that
\begin{equation} \label{stab-dyn-lrb}
\| [ \tau_t^{(s)}(A), B] \| \leq \frac{2 \| F_0 \|}{C_F} |X| \| A \| \| B \| e^{v|t| - g(d(X, Y))} \, ,
\end{equation}
see e.g. (\ref{velocity_bd}). For the sake of uniform estimates, note that $v$ is no larger than
\begin{equation} \label{lrb-vest}
v \leq  2 C_F \left( \frac{\kappa R^\nu}{F(R)} \cdot \max_{x \in \Lambda} \| h_x \| + \| \Phi\|_{1, F} \right) \, .
\end{equation}
To be clear, if $H$ is obtained from a uniformly bounded interaction and (\ref{ini-int-dec}) holds with
$\Lambda^p = \Lambda = \Gamma$, then this bound on $v$ is uniform with respect to $s \in [0,1]$ and
all finite volumes $\Lambda \subset \Gamma$. 

Let us now turn to estimates for two families of integral operators. 
For each $\xi >0$, define two families of 
linear maps $\{ \mathcal{F}^\xi_s \}_{s \in [0,1]}$ and $\{ \mathcal{G}^\xi_s \}_{s \in [0,1]}$,
with $\mathcal{F}^\xi_s \, , \mathcal{G}^\xi_s : \mathcal{A}_{\Lambda} \to \mathcal{A}_{\Lambda}$, given by  
\begin{equation} \label{F+G-maps}
\mathcal{F}^\xi_s(A) = \int_{\mathbb{R}} \tau_t^{(s)}(A) w_{\xi}(t) \, dt \quad \mbox{and} \quad \mathcal{G}^\xi_s(A) = \int_{\mathbb{R}} \tau_t^{(s)}(A) W_{\xi}(t) \, dt
\end{equation}
for all $A \in \mathcal{A}_{\Lambda}$ and $0 \leq s \leq 1$. As above, $\tau_t^{(s)}$ is the Heisenberg dynamics associated to $H(s)$ 
and here $w_{\xi}, W_{\xi} \in L^1(\mathbb{R})$ are the real-valued weight functions defined in \cite[Section VI.B]{nachtergaele:2019}. 
Both of these families of maps are bounded uniformly in $s$. In fact, one has that
\begin{equation} \label{lb_wios}
\| \mathcal{F}^\xi_s(A) \| \leq \| A \| \quad \mbox{and} \quad \| \mathcal{G}^\xi_s(A) \| \leq \| W_{\xi} \|_1 \| A \| \quad \mbox{for all } A \in \mathcal{A}_{\Lambda} \mbox{ and } 0 \leq s \leq 1 \, ,
\end{equation}
where we have used that $w_{\xi}$ is $L^1$-normalized.

An important consequence of the results proven in \cite[Section VI]{nachtergaele:2019}, see specifically Lemma 6.5, Lemma 6.10, and Lemma 6.11, 
is that the integral operators defined above in (\ref{F+G-maps}) satisfy quasi-locality estimates that are uniform with respect to
$0 \leq s \leq 1$. The following lemma summarizes the above-mentioned results proven in \cite{nachtergaele:2019}.
Before we state it, recall that in (\ref{fxi-dec}) we introduced a sub-additive, non-decreasing function 
$f_\xi : [0, \infty) \to (0, \infty)$ for any $\xi>0$. Moreover, let $\eta>0$ be the number defined by setting
\begin{equation} \label{def:eta}
\eta \left(1 + \sum_{n=2}^ \infty \frac{1}{n \ln(n)^2} \right) = 1 \, .
\end{equation}
One readily checks that $\eta \in (2/7,1)$. 

\begin{lemma} \label{lem:wio-dec} For each $\xi>0$, let $\{ \mathcal{F}_s^{\xi} \}_{s \in [0,1]}$ and 
	$\{ \mathcal{G}_s^{\xi} \}_{s \in [0,1]}$ denote the families of integral operators introduced in (\ref{F+G-maps}) above.
	If the corresponding Heisenberg dynamics satisfies (\ref{stab-dyn-lrb}), then with $\eta >0$ as in (\ref{def:eta}): given any
	$0<\epsilon <1$ and all $X,Y \subset \Lambda$ the bound
	\begin{equation} \label{ql_wios}
	\sup_{0 \leq s \leq 1} \| [ \mathcal{K}_s^\xi(A), B] \| \leq 2 \| A \| \| B \| |X| G_{\mathcal{K}}^\epsilon(d(X,Y)) 
	\end{equation}
	holds for all $A \in \mathcal{A}_X$, $B \in \mathcal{A}_Y$, and $\mathcal{K} \in \{ \mathcal{F}, \mathcal{G} \}$. There is a
	number $d_\epsilon^*>0$ for which one may take 
	\begin{equation} \label{gen_F_ql_dec_est}
	G_{\mathcal{F}}^{\epsilon}(d) = \left\{ \begin{array}{cl} 1 \,  & \mbox{if } 0 \leq d \leq d_{\epsilon}^* \\ 
	\min\left\{ 1, c\left( \frac{C \xi}{v} + \frac{27}{7} e^4 f_{\xi_{\epsilon}}(g(d))^2 \right) e^{- \eta f_{\xi_{\epsilon}}(g(d))} \right\}\,  & \mbox{otherwise}, \end{array} \right.
	\end{equation}
	and
	\begin{equation} \label{gen_G_ql_dec_est}
	G_{\mathcal{G}}^{\epsilon}(d) = \left\{ \begin{array}{cl} \| W_{\xi} \|_1 \,  & \mbox{if } 0 \leq d \leq d_{\epsilon}^* \\ 
	\min\left\{ \| W_{\xi} \|_1, \left( \frac{C }{2v} + \frac{243}{49 \xi \eta} c e^4 f_{\xi_{\epsilon}}(g(d))^3 \right) e^{- \eta f_{\xi_{\epsilon}}(g(d))} \right\} \,  & \mbox{otherwise}. \end{array} \right.
	\end{equation}
	In fact, one may take $d_{\epsilon}^*$ to be the smallest value of $d$ for which
	\begin{equation} \label{ld_ge}
	\max\left[ 9, \sqrt{ \frac{ \eta \xi_{\epsilon}}{ \epsilon}} \right] \leq \ln( \xi_{\epsilon} g(d)) \quad \mbox{where } \xi_{\epsilon} = \frac{(1 - \epsilon) \xi}{v} \, . 
	\end{equation}
	Here $v$ and $g$ are as in (\ref{stab-dyn-lrb}) and we set $C = \frac{2 \| F_0 \|}{C_F}$. In addition, $c$ is related to the $L^1$-normalization of $w_\xi$, see \cite[Section VI.B]{nachtergaele:2019}. 
\end{lemma}

For this proof of stability, decay functions, such as those in (\ref{gen_F_ql_dec_est}) and (\ref{gen_G_ql_dec_est}) above, will frequently enter.
Rather than tracking the precise details of all of these bounds, we find it convenient use the notion of decay class 
introduced in Definition~\ref{def:dec-class} to characterize the worst case estimate these bounds produce. With this in mind,
we re-state a more pragmatic version of Lemma~\ref{lem:wio-dec}. 

\begin{lemma} \label{lem:wio-prag} Under Assumption~\ref{ass:basic},  take $0 < \gamma < \gamma_\Lambda$
	where $\gamma_\Lambda$ denotes the initial spectral gap as in (\ref{def:gam_0}). 
	Denote by $\{ \mathcal{F}_s \}_{s \in [0,1]}$ and $\{ \mathcal{G}_s \}_{s \in [0,1]}$
	the families of integral operators defined as in (\ref{F+G-maps}) with $\xi = \gamma$ in both cases.  
	For each $\mathcal{K} \in \{ \mathcal{F}, \mathcal{G} \}$, there is a function $G_{\mathcal{K}}$ of decay class 
	$(\eta, \frac{\gamma}{2v}, \theta)$ for which given any $X, Y \subset \Lambda$,  
	\begin{equation} \label{wios_prag_dec}
	\sup_{0 \leq s \leq 1} \| [ \mathcal{K}_s(A), B] \| \leq 2 \| A \| \| B \| |X| G_{\mathcal{K}}(d(X,Y)) 
	\end{equation}
	for all $A \in \mathcal{A}_X$ and $B \in \mathcal{A}_Y$. 
\end{lemma}
\begin{proof}
	This lemma is just a special case of Lemma~\ref{lem:wio-dec} where we have taken $\xi = \gamma$. Although the assumption that
	$\gamma$ is related to the size of the initial spectral gap is not used in this estimate, this is the specific value of
	$\xi$ which will be used in all our applications, e.g. it is for this value that Proposition~\ref{prop:step1-commute} holds. Moreover, we have taken $\epsilon = 1/2$ to be concrete, but this is not crucial. Finally, 
	the particular parameters of the decay class are: 
	$\eta> 2/7$ is as in (\ref{def:eta}),
	$v$ may be taken as in (\ref{lrb-vest}), and $\theta$ is from the stretched exponential decay of the weight, see (\ref{stab_g_grows}). 
\end{proof}

\begin{remk} As is clear from Lemma~\ref{lem:wio-dec}, the decay functions, 
	denoted by $G_{\mathcal{K}}$ in Lemma~\ref{lem:wio-prag} above, depend only on $\gamma$ and $v$. 
	In particular, if $\gamma$ and $v$ are assumed to be volume independent, then so too are these
	decay functions.  
\end{remk}
The final step before proving Lemma~\ref{lem:qlm_ests} is to recall the quasi-locality estimates for the spectral flow 
established in \cite[Section VI.E.2]{nachtergaele:2019}. As before, take 
$0< \gamma < \gamma_\Lambda$ and denote the generator of the
spectral flow by 
\begin{equation} \label{gen_spec_flow2}
D(s)  = \int_{\mathbb{R}} \tau_t^{(s)}( \Vee) \, W_{\gamma}(t) \, dt = \mathcal{G}_s( \Vee) \quad \mbox{for all } 0 \leq s \leq 1.
\end{equation}
Recalling (\ref{gen_spec_flow}) and (\ref{F+G-maps}), we have taken $\xi = \gamma$ and suppressed the
dependence of the maps $D$ and $\mathcal{G}$ on $\gamma$. Note that we have written the above 
generator $D$ as the composition of a 
strictly local interaction, i.e. $\Vee$, with a quasi-local map, i.e. $\mathcal{G}_s$.  
A proof of quasi-locality estimates for the dynamics generated by such a
quasi-locally {\it transformed} interaction is the content of \cite[Section V.D]{nachtergaele:2019}. For the convenience of the
reader, we briefly review these proofs, specifically in the context of anchored interactions, in Appendix~\ref{sec:est-trans-balled-ints} below.
The basic idea is that local decompositions, see (\ref{def:Delta}), can be used to re-write the generator $D(s)$ as 
\begin{equation} \label{gen_spec_flow_asint}
D(s) = \sum_{x \in \Lambda^p} \sum_{m \geq R} \Psi(x, m, s) 
\end{equation}
where for each $x \in \Lambda^p$, $m \geq R$, and $0 \leq s \leq 1$, we have set 
\begin{equation}
\Psi(x,m, s) = \sum_{n=R}^m \Delta_{x,n;m}^{\Lambda}( \mathcal{G}_s( \Phi(x,n))) \, .
\end{equation}

For the desired quasi-locality bounds, we will need estimates on the above interaction terms.
Here we use results proven in Appendix~\ref{sec:est-trans-balled-ints}. 
First, recall that the family of maps $\{ \mathcal{G}_s \}_{s \in [0,1]}$ 
satisfies a uniform (in $s$) local bounded, see (\ref{lb_wios}),
as well as a uniform quasi-local estimate of order one, see Lemma~\ref{lem:wio-prag}. 
Next, recall that $\Vee$ satisfies (\ref{ini-int-dec}). Together, these estimates imply that
Theorem~\ref{thm:trans-int-bd} holds pointwise in $s$, and moreover, since the corresponding decay functions $G_{\mathcal{G}}$ and $F$
are both in the decay class $(\eta, \frac{\gamma}{2v}, \theta)$, we have further satisfied the assumptions of 
Corollary~\ref{cor:trans-int-F-dec}. We conclude that for every $0 < \mu < \eta$, there is an $F$-function $F_{\Psi}^{\mu}$ on
$(\Gamma, d)$ for which 
\begin{equation} \label{Psi-int-bd}
\sum_{x \in \Lambda^p} \sum_{\stackrel{m \geq R:}{y,z \in b_x(m)}} \| \Psi(x,m,s) \| \leq F_{\Psi}^{\mu}(d(y,z))
\end{equation} 
and we stress this bound is uniform with respect to $0 \leq s \leq 1$. 
Moreover, for any $\zeta > \nu +1$ there are positive numbers $C_1$, $C_2$, $a$, and $d$, with $C_1 \geq C_2 e^{- \mu f_{\gamma / 2v}(a d^{\theta})}$, 
for which one may take $F_{\Psi}^{\mu}$ with the form $F_{\Psi}^{\mu} = F_{\Psi, 0} \cdot F_{\Psi, \mu}^{\rm dec}$ where:
\begin{equation} \label{F-Psi-factors}
F_{\Psi,0}(r) = \frac{1}{(1+r)^{\zeta}} \quad \mbox{and} \quad F_{\Psi, \mu}^{\rm dec}(r) = \left\{ \begin{array}{cl} C_1 & \mbox{if } 0 \leq r \leq d, \\
C_2 e^{- \mu f_{\gamma/ 2v}(a r^{\theta})} & \mbox{if } r >d. \end{array} \right.
\end{equation}

\begin{remk} With an eye towards future statements of uniformity, note that Theorem~\ref{thm:trans-int-bd} 
	and Corollary~\ref{cor:trans-int-F-dec} together demonstrate that the choice of decay functions $F_{\Psi}^{\mu}$ appearing
	in (\ref{Psi-int-bd}) above can be made explicit in terms of the decay function $F$ associated to $\Vee$ and the decay function
	$G_{\mathcal{G}}$ as introduced in Lemma~\ref{lem:wio-prag}. In cases where $\gamma$ and $v$ can be estimated
	uniformly in the finite volume, in particular a volume-independent analogue of (\ref{ini-int-dec}) is assumed, then the choices of
	decay functions in (\ref{Psi-int-bd}) and (\ref{F-Psi-factors}) may be taken volume-independent as well.
\end{remk}

As a consequence of (\ref{Psi-int-bd}), we obtain quasi-local estimates for the spectral flow from Lieb-Robinson bounds. 
More precisely, as is discussed in Section~\ref{sec:stab+sf}, see (\ref{spec-flow-unis}) and (\ref{spec-flow-auto}),
the map $D(s)$ is the generator of the spectral flow automorphism $\alpha_s$. In this case, 
an application of Theorem~\ref{thm:LRBounds} (combined again with Proposition~\ref{prop:TriIneq_Fnorm}) 
shows that for any $0< \mu < \eta$ and any $X,Y \subset \Lambda$ with $X \cap Y = \emptyset$, the quasi-local estimate
\begin{equation} \label{ql_spec_flow}
\| [ \alpha_s(A), B ] \|  \leq  \frac{2 \| A \| \|B \|}{C} \left( e^{2sC} - 1 \right) \sum_{x \in X} \sum_{y \in Y} F_{\Psi}^{\mu}(d(x,y)) 
\end{equation}
holds for all $A \in \mathcal{A}_X$, $B \in \mathcal{A}_Y$, and $0 \leq s \leq 1$. Here, to ease notation, we have denoted by
$C = C_{F_{\Psi}^{\mu}}$ the convolution constant associated to the $F$-function $F_{\Psi}^{\mu}$.

We can now present to proof of Lemma~\ref{lem:qlm_ests}. 

\begin{proof}[Proof of Lemma~\ref{lem:qlm_ests}:]
	We will treat each family of maps separately.
	
	{\it Estimates for $\{ \mathcal{K}_s^1 \}_{s \in [0,1]}$:} Recall that for each $0 \leq s \leq 1$, the map $\mathcal{K}_s^1 : \mathcal{A}_{\Lambda} \to \mathcal{A}_{\Lambda}$ is defined by 
	\begin{equation}
	\mathcal{K}_s^1(A) = [ (\alpha_s - {\rm id}) \circ \mathcal{F}_s](A) \quad \mbox{for all } A \in \mathcal{A}_{\Lambda} \, .
	\end{equation}
	For this family of maps, we will use \cite[Lemma 5.10]{nachtergaele:2019} to obtain both the local bound and the
	quasi-local estimate. To apply Lemma 5.10, we need 
	a priori estimates for the maps being composed. 
	
	Let us first consider $\mathcal{F}_s$. A local bound of order zero for $\mathcal{F}_s$ was established in (\ref{lb_wios}). 
	Moreover, the bound (\ref{wios_prag_dec}) in Lemma~\ref{lem:wio-prag}
	demonstrates a quasi-locality estimate for $\mathcal{F}_s$ of order 1. In the latter bound, Lemma~\ref{lem:wio-prag} also establishes that 
	the decay function $G_{\mathcal{F}}$ is in the decay class $(\eta, \frac{\gamma}{2v}, \theta)$, see
	Definition~\ref{def:dec-class} for more details. As the notation suggests, we stress that both of these 
	estimates hold uniformly with respect to $0 \leq s \leq 1$. 
	
	Let us now consider $\alpha_s - {\rm id}$. Here it will be crucial that the pre-factors in the estimates for $\alpha_s - {\rm id}$
	are linear in $s$. To see this, we proceed as follows. Note that for each $A \in \mathcal{A}_{\Lambda}$ and $0 \leq s \leq 1$, the equality
	\begin{equation} \label{spec_flow_-idty}
	(\alpha_s - {\rm id})(A)  = \alpha_s(A) - A = \int_0^s \frac{d}{dr} \alpha_r (A) \, dr = i \int_0^s \alpha_r([D(r), A] ) \, dr, 
	\end{equation}
	holds. For any $0< \mu < \eta$ and each $A \in \mathcal{A}_X$, the bound
	\begin{eqnarray} \label{lb_spec_flow_-1}
	\| ( \alpha_s - {\rm id})(A) \| & \leq & 2 \| A \| \sum_{z \in \Lambda^p} \sum_{\stackrel{m \geq R:}{b_z^{\Lambda}(m) \cap X \neq \emptyset}} \int_0^s \| \Psi(z,m,r) \| \,dr \nonumber \\
	& \leq & 2 \| A \| \sum_{x \in X} \sum_{w \in \Lambda^p} \int_0^ s \sum_{z \in \Lambda^p} \sum_{\stackrel{m \geq R:}{x,w  \in b_z^{\Lambda}(m)}} \| \Psi(z,m,r) \| \, dr  \nonumber \\
	& \leq & 2 s \| A \| \| F_{\Psi}^{\mu} \| |X| 
	\end{eqnarray}
	follows from (\ref{gen_spec_flow_asint}) and (\ref{Psi-int-bd}). In this case, the local bound claimed in (\ref{ki_lb_est})
	now follows from \cite[Lemma 5.10(i)]{nachtergaele:2019}. Here we have used that the local bound for $\alpha_s- {\rm id}$ is
	of order one for all $0 \leq s \leq 1$, and moreover, each moment of the decay function $G_{\mathcal{F}}$ is finite. 
	One may take $p_1=2$. 
	
	For the quasi-local estimate on $\alpha_s - {\rm id}$, with a linear pre-factor in $s$, we argue as follows. It is clear that 
	for any $X,Y \subset \Lambda$, the bound
	\begin{equation}
	\| [ (\alpha_s - {\rm id})(A), B ] \| \leq 2 \| (\alpha_s - {\rm id})(A) \| \|B \| 
	\end{equation}
	holds for all $A \in \mathcal{A}_X$ and $B \in \mathcal{A}_Y$. For any $0 < \mu < \eta$, we have 
	the local bound (\ref{lb_spec_flow_-1}) and moreover, if $X \cap Y = \emptyset$, then 
	\begin{equation}
	\| [ (\alpha_s - {\rm id})(A), B ] \|  = \| [ \alpha_s(A), B ] \|  \leq  4 se^{2 C} \| F_{\Psi,0} \| |X| \| A\| \| B \| F_{\Psi, \mu}^{\rm dec}(d(X,Y))
	\end{equation}
	where for the final inequality above, we have used (\ref{ql_spec_flow}), the mean value theorem, and the factorized form of the
	$F$-function $F_{\Psi}^{\mu}$, see (\ref{F-Psi-factors}).
	In this case, the estimate claimed in (\ref{ki_ql_est}) now follows 
	from \cite[Lemma 5.10(ii)]{nachtergaele:2019}. One may take $q_1 =2$. For sufficiently large $r$, the resulting decay function $G_1$ has the form
	\begin{equation}
	G_1(r) \sim (r/2)^{\nu} F_{\Psi, \mu}^{\rm dec}(r/2) + \sum_{n= \lfloor r/2 \rfloor}^{\infty} (1+ n)^{\nu}G_{\mathcal{F}}(n) \, .
	\end{equation}
	Here $\sim$ indicates that we have ignored certain $r$-independent pre-factors. These include
	factors from the local bounds on $\alpha_s - {\rm id}$ and $\mathcal{F}_s$, factors from the
	quasi-local bounds on $\alpha_s - {\rm id}$ and $\mathcal{F}_s$, and factors of $\kappa$ from the
	$\nu$-regularity assumption. One can, however, make an explicit choice for the resulting decay function using the statement of \cite[Lemma 5.10]{nachtergaele:2019}.  In any case, we conclude that $G_1$ is of decay class $(\eta, \frac{\gamma}{2v}, \theta)$, see e.g. 
	comments in Remark~\ref{rem:dec-class}.

	{\it Estimates for $\{ \mathcal{K}_s^2 \}_{s \in [0,1]}$:} Recall that for each $0 \leq s \leq 1$, the map $\mathcal{K}_s^2 : \mathcal{A}_{\Lambda} \to \mathcal{A}_{\Lambda}$ is defined by 
	\begin{equation}
	\mathcal{K}_s^2(A) = \mathcal{F}_s(A) - \mathcal{F}_0(A) \quad \mbox{for all } A \in \mathcal{A}_{\Lambda} \, .
	\end{equation}
	To estimate, we find it useful to observe that $\mathcal{K}_s^2$ can be re-written as a composition. 
	Recall that the mapping $\delta^{\Vee} : \mathcal{A}_{\Lambda} \to \mathcal{A}_{\Lambda}$ defined by 
	\begin{equation}
	\delta^{\Vee}(A) = i [ \Vee, A] \quad \mbox{for all } A \in \mathcal{A}_{\Lambda} \,
	\end{equation}
	is called the derivation associated to $\Vee$. In terms of this mapping, note that
	\begin{eqnarray} \label{K2_comp}
	\mathcal{K}_s^2(A) = \mathcal{F}_s(A) - \mathcal{F}_0(A) & = & \int_{\mathbb{R}} \left( \tau_t^{(s)}(A) - \tau_t^{(0)}(A) \right) w_{\gamma}(t) \, dt \nonumber \\
	& = & \int_{\mathbb{R}} \int_0^t \frac{d}{dr} \tau_r^{(s)} \circ \tau^{(0)}_{t-r}(A) \,  dr \, w_{\gamma}(t) \, dt \nonumber \\
	& = & \int_{\mathbb{R}} \int_0^t i \tau_r^{(s)} \left( [ H(s) - H(0), \tau^{(0)}_{t-r}(A)] \right) \,  dr \, w_{\gamma}(t) \, dt \nonumber \\
	& = & s ( \mathcal{G}_s \circ \delta^{\Vee})(A).
	\end{eqnarray}
	In the last line above, we used that the distributional derivative of $W_{\gamma}$ satisfies 
	\begin{equation}
	\frac{d}{dt} W_{\gamma}(t) = - w_{\gamma}(t) + \delta_0(t) 
	\end{equation}
	as is discussed in \cite[Section VI.B]{nachtergaele:2019}.
	
	We are now in a position to apply \cite[Lemma 5.8]{nachtergaele:2019}. 
	As before, we first collect the relevant a priori bounds on the maps being composed. 
	For $\mathcal{G}_s$, (\ref{lb_wios})
	establishes a local bound of order zero, while (\ref{wios_prag_dec}) in Lemma~\ref{lem:wio-prag}
	provides a quasi-locality estimate of order one. Again, the corresponding decay function $G_{\mathcal{G}}$ is 
	in the decay class $(\eta, \frac{\gamma}{2v}, \theta)$ and both of these 
	estimates hold uniformly with respect to $0 \leq s \leq 1$.
	
	The assumed norm bound on $\Vee$, see (\ref{ini-int-dec}), guarantees that the corresponding derivation $\delta^{\Vee}$ is locally bounded and
	quasi-local. More precisely, the local bound of order one
	\begin{equation}
	\| \delta^{\Vee}(A) \| \leq 2 \| \Phi \|_{1,F} \| F \| |X| \| A \| \quad \mbox{for all } A \in \mathcal{A}_X \mbox{ with } X \subset \Lambda \, ,
	\end{equation}
	holds, and whenever $X,Y \subset \Lambda$ with $X \cap Y = \emptyset$, one has that 
	\begin{equation}
	\| [ \delta^{\Vee}(A), B ] \| \leq 4 \| \Phi \|_{1,F} \| F_0\| |X| \|A \| \| B \| e^{-g(d(X,Y))} 
	\end{equation}
	for any $A \in \mathcal{A}_X$ and $B \in \mathcal{A}_Y$. More details on these calculations 
	can be found in \cite[Section V.B, Example 5.4]{nachtergaele:2019}. Note that the weight $e^{-g}$
	decays at the stretched exponential rate governed by (\ref{stab_g_grows}) which is faster than the previously indicated decay classes. 
	
	Applying \cite[Lemma 5.8 (i)]{nachtergaele:2019}, we find a local bound of the form (\ref{ki_lb_est}) with $p_2 =1$. 
	Using \cite[Lemma 5.8 (ii)]{nachtergaele:2019}, a quasi-locality bound of the form (\ref{ki_ql_est}) holds with $q_2=2$.
	The corresponding decay function $G_2$ is, for sufficiently large values of $r$, given by
	\begin{equation}
	G_2(r) \sim (r/2)^{\nu} G_{\mathcal{G}}(r/2) + e^{-g(r/2)} 
	\end{equation}
	where we have again ignored certain pre-factors. Since $G_{\mathcal{G}}$ is of decay class $(\eta, \frac{\gamma}{2v}, \theta)$, so too is $G_2$.
	
	{\it Estimates for  $\{ \mathcal{K}_s^3 \}_{s \in [0,1]}$:} Recall that for each $0 \leq s \leq 1$, the map $\mathcal{K}_s^3 : \mathcal{A}_{\Lambda} \to \mathcal{A}_{\Lambda}$ is defined by 
	\begin{equation}
	\mathcal{K}_s^3(A) = s (\alpha_s \circ \mathcal{F}_s)(A) \quad \mbox{for all } A \in \mathcal{A}_{\Lambda} \, .
	\end{equation}
	It is clear that, for each $0 \leq s \leq 1$, both $\alpha_s$ and $\mathcal{F}_s$ are of norm one. As such,
	a local bound of the form (\ref{ki_lb_est}) holds for $\mathcal{K}_s^3$, and one may take $p_3=0$. Quasi-locality bounds of order one for $\alpha_s$
	and $\mathcal{F}_s$ have already been discussed, see (\ref{ql_spec_flow}) and (\ref{wios_prag_dec}) respectively.  
	An application of \cite[Lemma 5.8 (ii)]{nachtergaele:2019} demonstrates a quasi-locality estimate for 
	$\mathcal{K}_s^3$. One may take $q_3=1$ and a corresponding decay function $G_3$ is, for sufficiently large values of
	$r$, given by 
	\begin{equation}
	G_3(r) \sim (r/2)^{\nu} F_{\Psi, \mu}^{\rm dec}(r/2) + G_{\mathcal{F}}(r/2)  
	\end{equation} 
	As before, we conclude that $G_3$ is of decay class $(\eta, \frac{\gamma}{2v}, \theta)$.
\end{proof}

Finally, we include the proof of Lemma~\ref{lem:FarDecay}.

\begin{proof}[ Proof of Lemma~\ref{lem:FarDecay}:] Let $N \geq R$ and take $x \in \Lambda$ with $N \leq d(x, \Lambda^p)$. By (\ref{phi1x-def}), we have that
	\begin{equation} \label{phi1x-far}
	\Phi_x^{(1)}(s) = \mathcal{K}_s^1(h_x) + \mathcal{K}_s^2(h_x) \quad \mbox{for all } 0 \leq s \leq 1 \, . 
	\end{equation}
	We will estimate the norm of the terms on the right-hand-side of (\ref{phi1x-far}) separately. Recall that
	\begin{equation}
	\mathcal{K}_s^1(h_x) = (\alpha_s - {\rm id})( \mathcal{F}_s(h_x)) \, .
	\end{equation}
	Let us denote by $A = \Pi_{b_x^{\Lambda}( \lfloor N/2 \rfloor)}^{\Lambda}( \mathcal{F}_s(h_x))$. In this case, it is clear that
	\begin{equation}
	\| \mathcal{K}_s^1(h_x) \| \leq \| (\alpha_s - {\rm id})(A) \| + \| (\alpha_s - {\rm id})( \mathcal{F}_s(h_x) - A) \| \, .
	\end{equation}
	Arguing as in (\ref{spec_flow_-idty}) - (\ref{lb_spec_flow_-1}), we find that
	\begin{equation}
	\| ( \alpha_s - {\rm id})(A) \| \leq 2 \| h_x \| \sum_{y \in b_x^{\Lambda}( \lfloor N/2 \rfloor)} \sum_{w \in \Lambda^p} \int_0^s \sum_{z \in \Lambda^p} 
	\sum_{\stackrel{m \geq R:}{y,w \in b_z^{\Lambda}(m)}} \| \Psi(z,m, r) \| \, dr 
	\end{equation}
	Here we have used that $\| A \| \leq \|h_x \|$. Given this, we conclude from (\ref{Psi-int-bd}) that 
	\begin{eqnarray}
	\| ( \alpha_s - {\rm id})(A) \| & \leq & 2 s \| h_x \|   \sum_{y \in b_x( \lfloor N/2 \rfloor)} \sum_{w \in \Lambda^p} F_{\Psi}^{\mu}(d(y,w)) \nonumber \\
	& \leq & 2 s \kappa \| h_x \| \| F_{\Psi,0} \| ( N/2)^{\nu} F_{\Psi, \mu}^{\rm dec}(N/2)  
	\end{eqnarray}
	where in the final bound we used that $d(y,w) \geq N/2$ for each choice of $y$ and $w$ as above. In fact,
	$$
	N \leq d(x, \Lambda^p) \leq d(x,w) \leq d(x,y) + d(y,w) \leq \lfloor N/2 \rfloor + d(y,w) 
	$$
	This term has decay as claimed in (\ref{FarDeltaDecay}).
	
	Using the telescoping property of the local decompositions, see (\ref{sum-to-m}) and (\ref{telescope}), one sees that
	\begin{equation}
	\mathcal{F}_s(h_x) - A = ({\rm id} - \Pi^{\Lambda}_{b_x^{\Lambda}( \lfloor N/2 \rfloor )})(\mathcal{F}_s(h_x)) = \sum_{m \geq \lfloor N/2 \rfloor + 1} \Delta^{\Lambda}_{x,R;m} ( \mathcal{F}_s(h_x)) \, .
	\end{equation} 
	The norm bound
	\begin{eqnarray}
	\| ( \alpha_s - {\rm id})( \mathcal{F}_s(h_x) - A) \|  & \leq &  \sum_{m \geq \lfloor N/2 \rfloor + 1} \| (\alpha_s - {\rm id})(\Delta^{\Lambda}_{x,R;m} ( \mathcal{F}_s(h_x))) \| 
	\nonumber \\
	& \leq & 2 s \| F_{\Psi}^{\mu} \| \sum_{m \geq \lfloor N/2 \rfloor + 1} | b_x^{\Lambda}(m)| \| \Delta_{x,R;m}^{\Lambda}( \mathcal{F}_s(h_x)) \| \nonumber \\ 
	& \leq & 16 s \| F_{\Psi}^{\mu} \| \|h_x\| |b_x^{\Lambda}(R)|  \sum_{m \geq \lfloor N/2 \rfloor + 1} | b_x^{\Lambda}(m)| G_{\mathcal{F}}(m-R-1) 
	\end{eqnarray}
	follows from the local bound proven in (\ref{lb_spec_flow_-1}) and an application of (\ref{Delta_bd}) from Lemma~\ref{lem:qlm_loc_est} 
	which applies given the result of Lemma~\ref{lem:wio-prag}. This term also has decay as claimed in (\ref{FarDeltaDecay}).
	
	Lastly, we note that using (\ref{K2_comp}) and (\ref{lb_wios}), the bound
	\begin{equation}
	\| \mathcal{K}_s^2(h_x) \| = s \| ( \mathcal{G}_s \circ \delta^{\Vee})(h_x) \| \leq s \|W_{\gamma} \|_1 \| \delta^{\Vee}(h_x) \|
	\end{equation}
	is clear. Moreover, the estimate
	\begin{eqnarray}
	\| \delta^{\Vee}(h_x) \| & \leq & \sum_{z \in \Lambda^p} \sum_{\stackrel{n \geq R:}{b_z^{\Lambda}(n) \cap b_x^{\Lambda}(R) \neq \emptyset}} \| [ \Phi(z,n), h_x] \| \nonumber \\
	& \leq & 2 \| h_x \| \sum_{y \in b_x^{\Lambda}(R)} \sum_{w \in \Lambda^p} \sum_{z \in \Lambda^p} \sum_{\stackrel{n \geq R:}{w,y \in b_z^{\Lambda}(n)}} \| \Phi(z,n) \|  \nonumber \\
	& \leq & 2 \| h_x \| \| \Phi \|_{1,F} \sum_{y \in b_x(R)} \sum_{w \in \Lambda^p} F(d(y,w)) \nonumber \\
	& \leq & 2 \| h_x \| \| \Phi \|_{1,F} \| F_0 \| |b_x^{\Lambda}(R)| e^{-g(N-R)} 
	\end{eqnarray}
	follows from (\ref{ini-int-dec}) and the form of the corresponding weighted $F$-function. Note also that 
	$$
	N \leq d(x, \Lambda^p) \leq d(x,w) \leq d(x,y) + d(y,w) \leq R + d(y,w) 
	$$
	
	We have shown that 
	\begin{eqnarray}
	\| \Phi_x^{(1)}(s) \| & \leq & \| \mathcal{K}_s^1(h_x) \| + \| \mathcal{K}_s^2(h_x) \| \nonumber \\
	& \leq & 2 s \| h_x \| \left( \kappa \| F_{\Psi,0} \| (N/2)^{\nu} F_{\Psi, \mu}^{\rm dec}(N/2) + 8 \| F_{\Psi}^{\mu} \| |b_x^{\Lambda}(R)| \sum_{m \geq \lfloor N/2 \rfloor} |b_x^{\Lambda}(m+1)| G_{\mathcal{F}}(m-R) \right) \nonumber \\
	& \mbox{ } & \quad +  2 s \| W_{\gamma} \|_1 \| h_x \| \| \Phi \|_{1,F} \|F_0 \| |b_x^{\Lambda}(R)| e^{-g(N-R)} 
	\end{eqnarray}
	Since all functions of $N$ can be appropriately estimated, this completes the proof.
\end{proof}

\section{Local topological quantum order and conditions for relative boundedness}\label{sec:LTQO2}

For the finite-volume family of Hamiltonians $H(s) = H + s \Vee \in \mathcal{A}_{\Lambda}$, as in (\ref{pert-fam-1}), 
we showed in Theorem \ref{thm:Step1} that the unitarily equivalent family $\alpha_s(H(s))$ can be re-written as
\begin{equation} \label{step1-dec}
\alpha_s(H(s)) = H + \Vee^{(1)}(s) \quad \mbox{with} \quad \Vee^{(1)}(s) = \sum_{x \in \Lambda,m\geq R} \Phi^{(1)}(x,m,s) \, ,
\end{equation} 
with $\Vert  \Phi^{(1)}(x,m,s) \Vert \leq s G^{(1)} (m)$ for a function $G^{(1)}$ of decay class $(\eta,\frac{\gamma}{2v},\theta)$.
The goal of this section is to
complete the decomposition described in Section~\ref{sec:sf-gap-stab} and show that $\alpha_s(H(s))$ can be further re-written as
\begin{equation} \label{step2-dec}
\alpha_s(H(s)) = H + \Vee^{(2)}(s) + \Delta(s) + E(s) +C(s) \idty 
\end{equation}
with terms satisfying the properties described in Claim~\ref{clm:decompositionH}, see Theorems~\ref{thm:step2-(i)+(ii)} and 
\ref{thm:rel_bound_cond} below. 
To do this we need to assume an additional property of the ground states of the initial Hamiltonian $H$. This property, which 
is referred to as local topological quantum order (LTQO), is expressed in terms of the indistinguishability radius we introduced in 
Section~\ref{sec:LTQO}.

The idea is that the ground states are indistinguishable by perturbations acting in a region where `LTQO holds,' which typically excludes the 
boundary of $\Lambda$. This is motivations the following definition of the perturbation region, $\Lambda^p\subseteq\Lambda$.
Fix a non-increasing function $\Omega: \mathbb{R} \to [0, \infty)$ and let $r_x^\Omega$ be the corresponding indistinguishability radius associated to
$x \in \Lambda$, see Definition \ref{LTQO_def}. Let $K, L\geq 0$  with $K \geq R$, the bound on the interaction radius of the initial Hamiltonian $H$, which we furthermore assume to be frustration free. Then,
define a perturbation region $\Lambda^p = \Lambda^p(K,L)$ by setting
\begin{equation} \label{def:L_p}
\Lambda^p  = \{x\in\Lambda \, : r_y^\Omega \geq L+K \text{ for all  } y\in b_x^{\Lambda}(K)\}.
\end{equation} 
The estimates proven in this section will depend on $K$, $L$, and various decay functions.  
When considering the thermodynamic limit, appropriate choices for $L$ and $K$ will, in particular, depend on 
the rate at which these functions decay; more on this in Section~\ref{sec:uniform_sequences}.
Let us further introduce 
\begin{equation} \label{eff_pert_reg}
\Lambda^p(K) = \{ x \in \Lambda \, : \, d(x, \Lambda^p) \leq K \}
\end{equation}
which we refer to as the {\it effective perturbation region}.
Given (\ref{def:L_p}) and (\ref{eff_pert_reg}), all sites in this
effective perturbation region are guaranteed to have an indistinguishability radius of at least $L+K$:
\be \label{LTQO_length_bd}
L+K \leq r_x^\Omega
\quad \text{for all} \quad x\in \Lambda^p(K) \, .
\ee

\begin{figure} \label{fig:PertRegion}
	\begin{center}
		\begin{tikzpicture}
		\fill[black!7!white, draw = black] (-3,0) rectangle (5,5);
		\draw[red, very thick](-2.25,.75) rectangle (4.25,4.25);
		\draw[blue, very thick] (-1.5,1.5) rectangle (3.5,3.5);
		\node[red] at (-1.5, 1)  {$\Lambda^p(K)$};
		\node at (-2.75,.25) {$\Lambda$};
		\node[blue] at (-1, 1.75) {$\Lambda^p$};
		\end{tikzpicture}
	\end{center}
	\caption{The perturbation region $\Lambda^p$ and effective perturbation region $\Lambda^p(K)$. Sites in $\Lambda \setminus \Lambda^p(K)$ have
		a fixed distance from the perturbation region. Sites in $\Lambda^p(K)$ have indistinguishability radii with a fixed lower bound.} 
\end{figure}
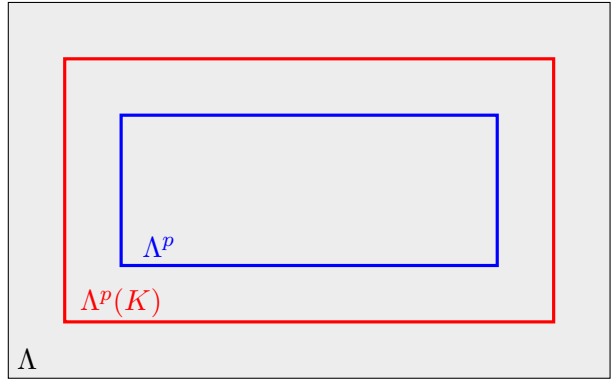

We use this effective perturbation region to partition the global terms of $\Vee^{(1)}(s)$, see (\ref{phi1x-def}), into a main term and a remainder term as follows:
\begin{equation} 
\Vee^{(1)}(s) =  \Vee^{(1)}_{\rm eff}(s) + E(s),
\end{equation}
with
\bea
\Vee^{(1)}_{\rm eff}(s) &=& \sum_{x \in \Lambda^p(K)} \Phi^{(1)}_x(s) \label{def:phi_eff} \\
E(s) &=& \sum_{x \in \Lambda \setminus \Lambda^p(K)} \Phi^{(1)}_x(s) \, .  \label{def:E}
\eea
Let $\omega$ denote the ground state functional associated to the initial Hamiltonian $H$, see  (\ref{ground_state_fun}), and set 
\begin{equation} \label{const-s}
C(s) = \omega \left( \Vee^{(1)}_{\rm eff}(s) \right) \, 
\end{equation}
to be the ground state expectation of $\Vee^{(1)}_{\rm eff}(s)$. Given (\ref{def:phi_eff}) and (\ref{const-s}), an 
application of Theorem~\ref{thm:Step1} shows that, in finite volume, $C(s) \to 0$ as $s \to 0$. In any case,  
we can now write
\be \label{final_decomp}
\Vee^{(1)}(s) 
=  \Vee^{(2)}(s) + \Delta(s) + E(s) +C(s) \idty \, ,
\ee
where we have set
\begin{equation} \label{phi2-s}
\Vee^{(2)}(s) = ( \idty - P_{\Lambda} ) \left( \Vee^{(1)}_{{\rm eff}}(s)  - C(s) \idty \right)  ( \idty - P_{\Lambda} )
\end{equation}
and 
\begin{equation} \label{delta-s}
\Delta(s) = P_{\Lambda} \left( \Vee^{(1)}_{{\rm eff}}(s)  - C(s) \idty \right) P_{\Lambda} \, .
\end{equation}
Note that (\ref{final_decomp}) with $\Vee^{(2)}(s)$ and $\Delta(s)$ as defined in (\ref{phi2-s}) and (\ref{delta-s}), 
holds for all $0 \leq s \leq s_{\gamma}^{\Lambda}$ by Proposition~\ref{prop:step1-commute}. The off-diagonal terms vanish
as all the global terms $\Phi^{(1)}_x(s)$ commute with the ground state projection $P_{\Lambda}$ associated to $H$;
note that in the notation of Section~\ref{sec:Step1}, $H = H(0)$ and $P_{\Lambda} = P(0)$. Thus, we have established the desired
form of (\ref{step2-dec}).

We now show that the terms in (\ref{step2-dec}) satisfy the properties described in Claim~\ref{clm:decompositionH}. 
It is easiest to estimate the remainder terms $\Delta(s)$ and $E(s)$ and so we do this first. Before we do so, we
introduce the following quantity as it appears in a number of our estimates. Set 
\begin{equation} \label{def:ckl}
C(K,L) = 2 \sum_{m \geq K+1} G^{(1)}(m) + \kappa \left( \sum_{m \geq 0} m^{\nu} G^{(1)}(m) \right) \Omega(L)
\end{equation}
where $G^{(1)}$ is the decay function obtained in the proof of Theorem~\ref{thm:Step1}, $\kappa$ and $\nu$ are from (\ref{nu-reg}), 
and $\Omega$ is the non-increasing function used to define the indistinguishability radius.
In applications, this quantity will be small for large values of $K$ and $L$. The following result makes explicit that properties (i) and (ii) of Claim~\ref{clm:decompositionH}
hold for the decomposition (\ref{step2-dec}) just obtained. 

\begin{thm}\label{thm:step2-(i)+(ii)}
	Under Assumption~\ref{ass:basic}, fix a non-increasing function $\Omega: \mathbb{R} \to [0, \infty)$ and
	define $E(s)$ and $\Delta(s)$, as in (\ref{def:E}) and (\ref{delta-s}) respectively, for all $0 \leq s \leq 1$.
	\begin{enumerate}
		\item[(i)] For all $s$, one has that $P_\Lambda\Delta(s)P_\Lambda = \Delta(s)$ and $\| \Delta (s) \| \leq s \delta$ where one may take
		\be \label{LTQO_delta}
		\delta =\vert \Lambda^{p}(K)\vert C(K,L) \,
		\ee
		and $C(K,L)$ is as in (\ref{def:ckl}) above. 
		\item[(ii)] For all $s$, $\|E(s)\| \leq s\varepsilon$ where one may take 
		\be \label{LTQO_alpha}
		\varepsilon = 2 \left( \sum_{x\in\Lambda \setminus \Lambda^p(K)}  \Vert h_x\Vert \right) G(K) 
		\ee
		and $G$ is the decay function obtained in the proof of Lemma~\ref{lem:FarDecay}.
	\end{enumerate}
\end{thm}

\begin{proof}
	We first prove (ii). Using (\ref{def:E}), it is clear that
	\begin{equation}
	\| E(s) \|  \leq \sum_{x \in \Lambda \setminus \Lambda^p(K)} \| \Phi_x^{(1)}(s) \| \leq 2 s \left( \sum_{x \in \Lambda \setminus \Lambda^p(K)} \| h_x \| \right) G(K)
	\end{equation}
	where the final inequality above follows from (\ref{eff_pert_reg}) and an application of Lemma~\ref{lem:FarDecay}.
	
	Now, consider (i). The claim that $P_{\Lambda} \Delta(s) P_{\Lambda} = \Delta(s)$ for all $0 \leq s \leq 1$ is immediate
	given (\ref{delta-s}). To prove (\ref{LTQO_delta}), it is convenient to introduce some additional notation. 
	Recall (\ref{def:phi_eff}) and (\ref{const-s}). Linearity guarantees that we may write
	\begin{equation} \label{phi1-x-c}
	\Vee^{(1)}_{{\rm eff}}(s) - C(s) \idty = \sum_{x \in \Lambda^p(K)} \Phi^{(1)}_{x, \omega}(s) \quad \mbox{where} \quad  \Phi^{(1)}_{x, \omega}(s) =  \Phi^{(1)}_x(s) -
	\omega \left( \Phi^{(1)}_x(s) \right) \idty \, .
	\end{equation}
	Moreover, in terms of the decomposition established in the proof of Theorem~\ref{thm:Step1}, see specifically (\ref{phi_1_x_local}), we may further write
	\begin{equation} \label{phi1-x-m-c}
	\Phi^{(1)}_{x, \omega}(s) = \sum_{m  \geq R} \Phi^{(1)}_{x, \omega, m}(s) \quad \mbox{where} \quad  \Phi^{(1)}_{x, \omega, m}(s) =  \Phi^{(1)}(x,m,s) 
	- \omega \left( \Phi^{(1)}(x,m,s) \right) \idty \, .
	\end{equation}
	
	Now, using (\ref{delta-s}) and (\ref{phi1-x-c}), the triangle inequality yields
	\begin{equation} \label{silly-del-est}
	\| \Delta(s) \| \leq \sum_{x \in \Lambda^p(K)} \| P_{\Lambda} \Phi_{x, \omega}^{(1)}(s) P_{\Lambda} \| \, .
	\end{equation}
	For each fixed $x \in \Lambda^p(K)$, using (\ref{phi1-x-m-c}), we may further estimate
	\begin{eqnarray} \label{Del-est-1}
	\| P_{\Lambda} \Phi_{x, \omega}^{(1)}(s) P_{\Lambda} \| & \leq & \sum_{m \geq R} \| P_{\Lambda} \Phi_{x, \omega, m}^{(1)}(s) P_{\Lambda} \| \nonumber \\
	& = & \sum_{m = R}^K \| P_{\Lambda} \Phi_{x, \omega, m}^{(1)}(s) P_{\Lambda} \| + \sum_{m > K} \| P_{\Lambda} \Phi_{x, \omega, m}^{(1)}(s) P_{\Lambda} \|  \, .
	\end{eqnarray}
	With the final term above, we use the bound (\ref{phi1_est}) proven in Theorem~\ref{thm:Step1}, i.e. 
	\begin{equation} \label{Del-est-2}
	\sum_{m > K} \| P_{\Lambda} \Phi_{x, \omega, m}^{(1)}(s) P_{\Lambda} \|  \leq 2 \sum_{m > K} \| \Phi^{(1)}(x,m,s) \| \leq 2 s \sum_{m>K}G^{(1)}(m)
	\end{equation}
	For the remaining term, we use the frustration-free property of the ground state, which implies that for each $x \in \Lambda^p(K)$,
	we have
	\begin{equation}
	P_{\Lambda} = P_{\Lambda}P_{b_x(L+K)} = P_{b_x(L+K)}P_{\Lambda} \, .
	\end{equation}
	Using the bound (\ref{LTQO_length}) from Definition~\ref{LTQO_def}, as well as (\ref{LTQO_length_bd}), we find that for  $x \in \Lambda^p(K)$ and $R \leq m \leq K$, 
	\begin{eqnarray} \label{use-LTQO-1}
	\| P_{\Lambda} \Phi^{(1)}_{x, \omega, m}(s) P_{\Lambda} \| & \leq & \| P_{b_x(L+K)} \Phi^{(1)}(x, m, s)P_{b_x(L+K)} - \omega( \Phi^{(1)}(x, m, s)) P_{b_x(L+K)}\| \nonumber \\
	& \leq & |b_x^{\Lambda}(m)| \| \Phi^{(1)}(x,m,s) \| \Omega(L+K - m). 
	\end{eqnarray}
	
	Combining (\ref{Del-est-1}), (\ref{Del-est-2}), and (\ref{use-LTQO-1}), we obtain that for each $x \in \Lambda^p(K)$
	\begin{equation} \label{proj-phi1-x-c-est}
	\| P_{\Lambda} \Phi_{x, \omega}^{(1)}(s) P_{\Lambda} \|  \leq s \left( 2 \sum_{m \geq K+1} G^{(1)}(m) + \Omega(L) \sum_{m \geq R} |b_x^{\Lambda}(m)| G^{(1)}(m) \right) ,
	\end{equation}
	where we have again applied (\ref{phi1_est}) from Theorem~\ref{thm:Step1}; now to the right-hand-side of (\ref{use-LTQO-1}).
	Recalling (\ref{def:ckl}), the bound claimed in (\ref{LTQO_delta}) follows from (\ref{silly-del-est}), (\ref{proj-phi1-x-c-est}), and $\nu$-regularity, i.e. (\ref{nu-reg}).
\end{proof} 

The remainder of the section is devoted to proving property (iii) in Claim~\ref{clm:decompositionH}. 
To establish this form bound, we will show that the term $\Vee^{(2)}(s)$, see (\ref{phi2-s}) above, can be written as an
$s$-dependent, anchored interaction which satisfies the assumptions of Theorem~\ref{thm:Step3General}. 
That is the content of Theorem~\ref{thm:rel_bound_cond} below. First, we prove the following lemma. 
For its statement, recall the notation established in the proof of Theorem~\ref{thm:step2-(i)+(ii)}; namely (\ref{phi1-x-c}) and (\ref{phi1-x-m-c}) 

\begin{lemma}\label{lem:phi_2_est}
	Under the assumptions of Lemma~\ref{lem:qlm_ests}, fix a non-increasing function $\Omega: \mathbb{R} \to [0, \infty)$.
	Let $x\in \Lambda^p(K)$. For all $0 \leq s \leq s_{\gamma}^{\Lambda}$ and any $m \leq n \leq r_x^\Omega$, we have
	\begin{eqnarray} \label{SumLTQO}
	\left\| \sum_{k=R}^m \Phi_{x, \omega, k}^{(1)}(s)P_{b_x^{\Lambda}(n)} \right\|
	& \leq & s \left[ 2 \sum_{k \geq m+1} G^{(1)}(k) + \left( \sqrt{ 8\kappa} \sum_{k \geq 0} G^{(1)}(k) \right) \sqrt{m^\nu\Omega(n-m)} \right] \nonumber \\
	& \mbox{ } & \quad + s \left[  2 \sum_{k \geq K+1} G^{(1)}(k) + \kappa \left( \sum_{k \geq 0} k^{\nu} G^{(1)}(k) \right) \Omega(L) \right]  \, .
	\end{eqnarray}
\end{lemma}
As we did with the statement of Theorem~\ref{thm:step2-(i)+(ii)}, see specifically (\ref{def:ckl}), 
it is convenient to label the some of the terms above as they frequently appear below.
For any $0 \leq m \leq n$, set 
\begin{equation} \label{def:Dmn}
D(m,n)  =  2 \sum_{k \geq m+1}G^{(1)}(k) + \left( \sqrt{ 8\kappa} \sum_{k \geq 0} G^{(1)}(k) \right) \sqrt{m^\nu\Omega(n-m)} \, . 
\end{equation}
The quantity above behaves similarly to $C(K,L)$ in the sense that
$D(m,n)$ should be {\it small} if $m$ and $n-m$ are all sufficiently {\it large}. 
\begin{proof}
	Fix $x \in \Lambda^p(K)$ and consider the observable $A_m = \sum_{k=R}^m \Phi^{(1)}_{x, \omega, k}(s) \in \cA_{b_x^\Lambda(m)}$. 
	Since $m \leq n \leq r_x^{\Omega}$, an application of Proposition~\ref{cor:LTQO} shows that
	\begin{equation} \label{app_cor}
	\| A_m P_{b_x^{\Lambda}(n)} \| \leq \| A_m P_{\Lambda} \| + \| A_m \| \sqrt{2|b_x^{\Lambda}(m)|\Omega(n-m)} \, .
	\end{equation}
	Going back to (\ref{phi1-x-m-c}), we have that
	\begin{equation}
	\Phi^{(1)}_{x,\omega}(s) = \sum_{k \geq R} \Phi^{(1)}_{x, \omega, k}(s) = A_m + \sum_{k \geq m+1} \Phi^{(1)}_{x, \omega, k}(s)
	\end{equation}
	and therefore, 
	\begin{equation} \label{ap0}
	\| A_m P_{\Lambda} \| \leq \left\| \Phi^{(1)}_{x,\omega}(s) P_{\Lambda} \right\| + 2 \sum_{k \geq m+1} \| \Phi^{(1)}(x,k, s) \| .
	\end{equation}
	Using Proposition~\ref{prop:step1-commute}, the first term above may be re-written as in the LHS of (\ref{proj-phi1-x-c-est}) and
	estimated by $sC(K,L)$. The bound claimed in (\ref{SumLTQO}) now follows from the naive bound
	\begin{equation} \label{tri-ineq-A}
	\| A_m \|  \leq  \sum_{k=R}^m \| \Phi^{(1)}_{x,\omega, k}(s) \|  \leq  2 \sum_{k=R}^m \| \Phi^{(1)}(x,k,s) \| 
	\end{equation}
	and two applications of (\ref{phi1_est}) from Theorem~\ref{thm:Step1}; once for the final term on the RHS of (\ref{ap0}) and 
	once for final estimate in (\ref{tri-ineq-A}).
\end{proof}

Finally, to allow for some additional flexibility in the application of the estimates below, we divide the terms in certain 
sums according to a function $f$ with specified properties. 
This function should be regarded as an additional free parameter in this set-up. To keep track of terms, 
we find is convenient to introduce the quantity $\ell_x = \ell_x(\Lambda)$ defined for each $x \in \Lambda$ by setting
\begin{equation} \label{max-ball-rad-x}
\ell_x = \min \{ n \in\bZ_{\geq 0} : b_x^{\Lambda}(n) = \Lambda \} \, .
\end{equation}

\begin{thm} \label{thm:rel_bound_cond}
	Let $f:[0,\infty)\to[0,\infty)$ be any differentiable function with $f(0) =0$ and $0<f'(t)< 1$ for all $t \geq 0$. 
	Under Assumption~\ref{ass:basic}, one can write 
	\begin{equation*}
	\Vee^{(2)}(s) = \sum_{x \in \Lambda^p(K)} \sum_{n=R}^{\ell_x} \Phi^{(2)}(x,n,s)
	\end{equation*}
	with terms satisfying: for all $x \in \Lambda^p(K)$, $n \geq R$, and $0\leq s \leq s_{\gamma}^{\Lambda}$, the following holds 
	\begin{enumerate}
		\item[(i)] $\Phi^{(2)}(x,n,s)^* = \Phi^{(2)}(x,n,s) \in \mathcal{A}_{b_x^{\Lambda}(n)}$
		\item[(ii)] $ P_{b_x^{\Lambda}(n)} \Phi^{(2)}(x,n,s) = \Phi^{(2)}(x,n,s) P_{b_x^{\Lambda}(n)} = 0$
		\item[(iii)] $\| \Phi^{(2)}(x,n,s) \| \leq  2s G^{(2)}(n)$, where
		\begin{equation} \label{Phi2_norm_bds}
		G^{(2)}(n) = 
		\begin{cases}
		G^{(1)}(f(n))+D(\lceil f(n) \rceil -1, n-1) + C(K,L) & R \leq n < L+K \\
		\sum_{k \geq f(L+K)} G^{(1)}(k) + D(\lceil f(L+K)\rceil -1, L+K-1) +C(K,L)    & \phantom{ R \leq }n \geq L+K
		\end{cases}
		\end{equation}
		and the quantities $C(K,L)$ and $D(m,n)$ are as in (\ref{def:ckl}) and (\ref{def:Dmn}) respectively.
	\end{enumerate}
	
\end{thm}

\begin{proof}[Proof of Theorem~\ref{thm:rel_bound_cond}]
	Recall that $\Vee^{(2)}(s)$ is as defined in (\ref{phi2-s}).
	To obtain the terms $\Phi^{(2)}(x,n,s)$, satisfying the conditions of Theorem~\ref{thm:rel_bound_cond},
	we use the notation introduced in the proof of Theorem~\ref{thm:step2-(i)+(ii)} to expand.
	By inserting (\ref{phi1-x-c}), we have that
	\[
	\Vee^{(2)}(s) = \sum_{x \in \Lambda^p(K)} (\idtyty-P_\Lambda)\Phi_{x,\omega}^{(1)}(s)(\idtyty-P_\Lambda)
	\]	
	and moreover, with (\ref{phi1-x-m-c}) we have that for each $x \in \Lambda^p(K)$, 
	\begin{equation} \label{phi2_sum_3}
	(\idtyty - P_\Lambda) \Phi^{(1)}_{x,\omega}(s) ( \idtyty - P_\Lambda) = \sum_{m = R}^{\ell_x} (\idtyty - P_\Lambda) \Phi^{(1)}_{x, \omega, m}(s) ( \idtyty - P_\Lambda).
	\end{equation}
	In the argument below, it will be convenient to write $\Phi^{(2)}(x,n,s)$ as the sum of two terms
	\begin{equation*}
	\Phi^{(2)}(x,n,s) = \Theta_1(x,n,s) + \Theta_2(x,n,s)
	\end{equation*}
	each of which will separately satisfy the conditions of Theorem~\ref{thm:rel_bound_cond}.
	
	First, to any $x \in \Lambda^p(K)$, set 
	\begin{equation} \label{def:theta1-ellx}
	\Theta_1(x, \ell_x, s) = \sum_{m \geq f(L + K) } (\idtyty - P_\Lambda) \Phi^{(1)}_{x, \omega, m}(s) ( \idtyty - P_\Lambda) \, .
	\end{equation} 
	Since $b_x^{\Lambda}(\ell_x) = \Lambda$, such terms clearly satisfy conditions (i) and (ii) above.
	Moreover, an application of Theorem~\ref{thm:Step1} shows that
	\begin{equation} \label{Theta1Decay_large}
	\|  \Theta_1(x, \ell_x, s) \| \leq 2 \sum_{m \geq  f(L+K) } \| \Phi^{(1)}(x, m, s) \|  \leq 2s  \sum_{m \geq  f(L+K) } G^{(1)}(m) \, .
	\end{equation}
	
	Only those terms in (\ref{phi2_sum_3}) with $R \leq m <  f(L+K) $ remain to be analyzed. 
	For these, let us introduce the following sequence of operators:
	\begin{equation} \label{Ejprojs}
	E_j 
	\; = \; 
	\left\{ \begin{array}{ll} 
	\idtyty - P_{b_x^{\Lambda}(R)} & j=R \\ 
	P_{b_x^{\Lambda}(j-1)} - P_{b_x^{\Lambda}(j)} & R+1 \leq j \leq \ell_x \\ 
	P_\Lambda & j = \ell_x+1 \,.
	\end{array} \right.
	\end{equation}
	One readily checks that each $E_j$ is an orthogonal projection, and moreover,
	\begin{equation}
	E_jE_k = \delta_{j,k} E_k \quad \mbox{with} \quad \idtyty = \sum_{j=R}^{\ell_x+1}E_j \, .
	\end{equation}
	In words, this family of projections is mutually orthogonal and sums to the identity. It is also useful to observe that partial sums telescope, i.e. 
	\begin{equation} \label{telescopic_prop}
	\idtyty - P_{b_x^{\Lambda}(n)} = \sum_{j=R}^n E_j \quad \mbox{for all } R \leq n \leq \ell_x \, .
	\end{equation} 
	
	Now, for $R \leq m < f(L+K) $, set $m_f = \lfloor f^{-1}(m) \rfloor$. Using \eqref{telescopic_prop}, a short calculation shows
	\begin{align}
	(\idtyty - P_\Lambda) \Phi_{x, \omega, m}^{(1)}(s) (\idtyty - P_\Lambda)
	& = 
	\sum_{j,k=R}^{\ell_x} E_j \Phi_{x, \omega, m}^{(1)}(s) E_k \label{Phi2Decompose} \\
	& = 
	(\idtyty-P_{b_x^{\Lambda}(m_f)}) \Phi_{x, \omega, m}^{(1)}(s) (\idtyty-P_{b_x^{\Lambda}(m_f)}) \nonumber
	\quad+ 
	\sum_{f^{-1}(m) < j \leq \ell_x} A_{j,m},
	\end{align}
	where, for all $j > f^{-1}(m)$,  we define
	\begin{equation}\label{def_ajm}
	A_{j,m} = E_j \Phi_{x, \omega, m}^{(1)}(s) (\idtyty - P_{b_x^{\Lambda}(j-1)})  + (\idtyty - P_{b_x^{\Lambda}(j)}) \Phi_{x, \omega, m}^{(1)}(s) E_j. 
	\end{equation}
	
	{F}rom the expression (\ref{Phi2Decompose}), we will extract two types of terms: $\Theta_1$-terms and
	$\Theta_2$-terms. Let us first continue defining the $\Theta_1$-terms. Note that for each integer $m$ with
	$R< m < f(L+K)$, 
	\[
	(m-1)_f < m_f < L+K \, .
	\]
	This follows from our assumptions on $f$ and, e.g., an application of the mean-value theorem.
	As such, each choice of $m_f$ corresponds to a unique integer $m$ and so the term
	\begin{equation} \label{Theta1}
	\Theta_1(x, m_f ,s) 
	\; = \;
	(\idtyty-P_{b_x^{\Lambda}(m_f)}) \Phi_{x, \omega, m}^{(1)}(s) (\idtyty-P_{b_x^{\Lambda}(m_f)})
	\end{equation}
	is well-defined. Note further that we are only considering values of 
	$m_f < L+K \leq r_x^\Omega \leq \ell_x$, and so there is no overlap with (\ref{def:theta1-ellx}).
	For these terms in (\ref{Theta1}), conditions (i) and (ii) are clear. 
	In fact, since $f(m)<m$ and hence $m<f^{-1}(m)$, we have that
	$m \leq \lfloor f^{-1}(m) \rfloor = m_f$, and thus $\Theta_1(x,m_f,s) \in \mathcal{A}_{b_x^{\Lambda}(m_f)}$. 
	To obtain a norm bound, note that since $m_f \leq f^{-1}(m)$ and
	$f$ is increasing, we have $f(m_f) \leq m$. In this case, again
	Theorem~\ref{thm:Step1} shows that 
	\begin{equation} \label{Theta1Decay_small}
	\|\Theta_1(x, m_f,s)\|
	\; \leq \;
	2 \|\Phi^{(1)}(x, m,s)\|
	\; \leq \;
	2sG^{(1)}(f(m_f)).
	\end{equation} 
	Setting $\Theta_1(x, m, s)=0$ for any integer values not considered above, 
	we have shown that the $\Theta_1$-terms satisfy the conditions of Theorem~\ref{thm:rel_bound_cond} with the norm estimate
	\be \label{Theta1_Decay}
	\|\Theta_1(x,n,s)\| \leq 2s \cdot 
	\begin{cases}
		G^{(1)}(f(n)) & R \leq n < L+K \\
		\sum_{m \geq f(L+K)} G^{(1)}(m) & \phantom{R \leq} n \geq L+K
	\end{cases}
	\ee
	Going back to (\ref{phi2_sum_3}), let us summarize the progress: for each $x \in \Lambda^p(K)$, we have written
	\begin{equation} \label{Theta1Decay2}
	(\idtyty - P_\Lambda) \Phi^{(1)}_{x,\omega}(s) ( \idtyty - P_\Lambda)
	= \sum_{n=R}^{\ell_x} \Theta_1(x,n,s) +  \sum_{R\leq m <f(L+K)} \sum_{f^{-1}(m) < j \leq \ell_x}  A_{j,m}.
	\end{equation}
	We will re-organize the final sum above and label these as $\Theta_2$-terms. 
	
	For defining the $\Theta_2$-terms, it will be convenient to note that the 
	operators $A_{j,m}$, see (\ref{def_ajm}), satisfy the following for each $R+1 \leq j \leq \ell_x$:
	\begin{enumerate}
		\item[(i)] $A_{j,m}^* = A_{j,m} \in \mathcal{A}_{b_x^{\Lambda}(j)}$ for all $m$ with $f^{-1}(m)<j$.
		\item[(ii)] $A_{j,m}P_{b_x^{\Lambda}(j)}= 0$ for all $m$ with $f^{-1}(m)<j$.
		\item[(iii)] For any $n$ with $f^{-1}(n)<j$, Lemma~\ref{lem:phi_2_est} becomes relevant as
		\begin{equation} \label{ApplyLTQOLem}
		\left\|\sum_{m=R}^n A_{j,m}\right\| 
		\; \leq \; 
		2 \left\|\sum_{m=R}^n\Phi_{x, \omega, m}^{(1)} E_j\right\|
		\; \leq \;
		2 \left\|\sum_{m=R}^n\Phi_{x, \omega, m}^{(1)} P_{b_x^{\Lambda}(j-1)}\right\|.
		\end{equation} 
		Here, the final inequality follows from using \eqref{ff_prop} to write $ E_j =  P_{b_x^{\Lambda}(j-1)}E_j$. 
	\end{enumerate}
	
	We now define the $\Theta_2$-terms. To do so, we interchange the double sum from (\ref{Theta1Decay2}) and isolate those terms for which Lemma~\ref{lem:phi_2_est} applies. Namely, note that 
	\begin{equation*}
	\sum_{R\leq m <f(L+K)} \; \sum_{f^{-1}(m) < j \leq \ell_x} 
	A_{j,m}
	\; = \;
	\sum_{R \leq m < f(L+K)} \;
	\sum_{f^{-1}(m)<j < L+K}
	A_{j,m}
	+
	\sum_{R \leq m < f(L+K)} \;
	\sum_{L+K \leq j  \leq \ell_x}
	A_{j,m}.
	\end{equation*}
	For the second collection of terms above, one has that
	\begin{equation}
	\sum_{R \leq m < f(L+K)} \;
	\sum_{L+K \leq j  \leq \ell_x}
	A_{j,m}
	=
	\sum_{L+K \leq j  \leq \ell_x} \;
	\sum_{R \leq m < f(L+K)}
	A_{j,m},
	\end{equation}
	and so we set
	\begin{equation}
	\label{Theta2_LargeSupp}
	\Theta_2(x,j,s) = \sum_{R \leq m < f(L+K)} A_{j,m} \quad \mbox{for all } L+K \leq j \leq \ell_x \, .
	\end{equation}
	Based on the observations above, it is clear that these terms satisfy conditions (i) and (ii). 
	Moreover,
	\begin{align}
	\|\Theta_2(x,j,s) \| 
	\; & \leq \;
	2 \left\|\sum_{m=R}^{\lceil f(L+K)\rceil -1}\Phi_{x, \omega, m}^{(1)} P_{b_x^{\Lambda}(j-1)}\right\|
	\; \leq \;
	2 \left\|\sum_{k=R}^{\lceil f(L+K)\rceil -1}\Phi_{x, \omega, m}^{(1)} P_{b_x^{\Lambda}(L+K-1)}\right\|.
	\end{align}
	Here, for the final inequality we have used $P_{b_x^{\Lambda}(j-1)} = P_{b_x^{\Lambda}(L+K-1)}P_{b_x^{\Lambda}(j-1)}$. Applying Lemma~\ref{lem:phi_2_est}, 
	for $L+K \leq j \leq \ell_x$ we obtain the following uniform estimate:
	\begin{equation} \label{Theta2Decay_large}
	\|\Theta_2(x,j,s) \| \leq 2s \left[ D(\lceil f(L+K)\rceil -1, L+K-1) + C(K,L) \right]
	\end{equation}
	where $C(K,L)$ and $D(m,n)$ are as defined in (\ref{def:ckl}) and (\ref{def:Dmn}) respectively.	
	
	For the remaining terms, observe that 
	\begin{equation*}
	\sum_{R \leq m < f(L+K)} \;
	\sum_{f^{-1}(m)<j < L+K}
	A_{j,m}
	\; = \;
	\sum_{f^{-1}(R)<j<L+K} \;
	\sum_{R \leq m < f(j)}
	A_{j,m}
	\end{equation*}
	and thus if we set 
	\begin{equation}
	\Theta_2(x,j,s)
	\; = \; 
	\sum_{R \leq m < f(j)}
	A_{j,m} \quad \; \mbox{for } \; f^{-1}(R)<j<L+K \, ,
	\end{equation}
	then, as before, these terms satisfy conditions (i) and (ii), and
	\begin{align} 
	\|\Theta_2(x,j,s)\|
	\; \leq & \;
	2 \left\| \sum_{m=R}^{\lceil f(j) \rceil -1}
	\Phi_{x, \omega, m}^{(1)}P_{b_x^{\Lambda}(j-1)}  \right\| \, 
	\; \leq \; 2s \left[ D( \lceil f(j) \rceil -1, j-1) + C(K,L) \right] \, .	\label{Theta2Decay_small} 
	\end{align}
	Once again, we define $\Theta_2(x, j,s)=0$ for any values of $j$ not considered above.
	
	Finally, set $\Phi^{(2)}(x,n,s) = \Theta_1(x,n,s) + \Theta_2(x,n,s)$. By construction, it is clear that these $\Phi^{(2)}$-terms satisfy the conditions of
	Theorem~\ref{thm:rel_bound_cond}, and moreover, combining the decay bounds from \eqref{Theta1_Decay}, \eqref{Theta2Decay_large}, and \eqref{Theta2Decay_small} we find that
	\begin{equation*}
	\|\Phi^{(2)}(x,n,s)\| 
	\; \leq \; 
	2s G^{(2)}(n)
	\end{equation*}
	as desired.		
\end{proof}

\section{Uniform Sequences and the Thermodynamic Limit}\label{sec:uniform_sequences}

\subsection{Introduction}

So far, we have determined conditions under which we can estimate the spectral gaps of a continuous family of quantum spin Hamiltonians on a fixed finite volume. Often we are interested in a uniform lower bound for the spectral gap for a collection of finite volume families of arbitrarily large size so that we can derive stability properties of the infinite systems. In this section, we focus on formulating conditions on families of models labeled by finite sets $\Lambda$ that imply such uniform estimates. We leave the discussion of the thermodynamic limit itself to Section~\ref{sec:automorphic-equivalence}.  We focus on conditions that let us establish a uniform lower bound for the spectral gap above the ground state along a sequence of finite systems of increasing size as this is also the foundation for studying higher gaps, see e.g. Corollary~\ref{cor:higher_gaps}. We will refer to such sequences as {\em uniform sequences}. 

Each finite volume Hamiltonian is defined in terms of one or more interactions, such as the maps $\eta$ and $\Phi$  in Section~\ref{sec:LTQO2}. These interactions will for the most part {\em not} depend on $n$ except for necessary modifications generally designated as {\em boundary conditions}. Boundary conditions can be expressed in a number of different ways and we discuss two common cases: 1) open boundary conditions, and 2) boundary conditions arising from modifying the metric on $\Lambda_n$. The typical situation we have in mind for the second case is {\em periodic} or {\em twisted periodic boundary conditions}; for example, when a finite rectangle in $\Ir^2$ is embedded on a torus. More generally, it is often of interest to define the model on a sequence of triangulations (or other discretizations) of a compact manifold.

In the latter case we assume one can 
extend the interaction to include `boundary' terms in the natural way. More precisely, these are the situations where there is {\em no} boundary and
we will refer to this case as {\em geometric boundary conditions}. There are other ways to define boundary conditions, which may involve $n$-dependence of both the interaction and the metric. These can be handled by small modifications of the discussion below.

\subsection{Uniform sequences of finite systems}\label{sec:uniformity}

The goal of this section is to describe conditions on a model which allow for the results obtained in 
Section~\ref{sec:Step1} and Section~\ref{sec:LTQO2} to hold uniformly along a sequence of finite-volumes. 
Our discussion of {\it uniformity} covers both common cases of boundary conditions discussed above. 

{\em Uniform Sequences of Finite Volumes:} 
Most interesting quantum spin models are defined over a metric space $(\Gamma, d)$ for which
$\Gamma$ has infinite cardinality. To apply the results of Section~\ref{sec:Step1} and Section~\ref{sec:LTQO2}, we
must first restrict the model to an appropriate choice of finite subsets. 
Let $\{ \Lambda_n \}_{n\geq 1}$ be an increasing and absorbing sequence of finite subsets of $\Gamma$.
To allow for possible boundary conditions, we will further regard each finite subset as a metric space on its own, i.e.
for each $n \geq 1$, associate a metric space $(\Lambda_n, d_n)$ to $\Lambda_n$ with a metric $d_n$ satisfying: for each $x,y \in \Gamma$, there is
$n(x,y) \geq 1$ sufficiently large so that $d_n(x,y) = d(x,y)$ for all $n \geq n(x,y)$. By allowing $n$-dependence of the metric, we are really including cases where $\Lambda_n$ does not have a natural embedding in $\Gamma$ but, strictly speaking, it can always be considered as a subset.

Finally, we also assume that these metrics are uniformly $\nu$-regular, i.e. there are positive numbers $\kappa$ and $\nu$ for which given $n \geq 1$, $m \geq 1$, and $x \in \Lambda_n$, 
\[
|b_x^{\Lambda_n}(m)| \leq \kappa m^\nu \quad \text{where} \quad b_x^{\Lambda_n}(m) = \{y\in \Lambda_n \mid d_n(x,y)\leq m\}.
\]
We will use $\diam_n(X)$ to represent the diameter of a set $X \subseteq \Lambda_n$ with respect to the metric $d_n$.
We will refer to any sequence $\{ \Lambda_n \}_{n \geq 1}$ of finite subsets of $\Gamma$ satisfying the conditions described above
as a {\it uniform sequence of finite volumes}. 

{\em Uniformity in the Initial Hamiltonian:}
We will assume that the initial, unperturbed Hamiltonian can be associated with a finite-range, uniformly bounded, 
frustration free interaction $\eta$. To make this association precise, let $\{ \Lambda_n \}_{n \geq 1}$ be a 
uniform sequence of finite volumes. For each $n \geq 1$, we 
assume that there is a (non-negative) frustration free interaction $\eta_n$ on $\Lambda_n$ and write
\begin{equation} \label{unif_unpert_ham}
H_{\Lambda_n} = \sum_{X \subseteq \Lambda_n} \eta_n(X) \, . 
\end{equation}
We assume that the sequence $\{ \eta_n \}_{n \geq 1}$ approximates $\eta$ in the sense that
for each $X \in \mathcal{P}_0( \Gamma)$, there is $n(X) \geq 1$ for which 
$\eta_n(X) = \eta(X)$ for all $n \geq n(X)$. We further assume that the sequence $\{ \eta_n \}_{n \geq 1}$ 
has a uniform finite range, in that there is a number $R' \geq 0$ for which 
$\eta_n(X) =0$ whenever $\diam_n(X)>R'$, and is uniformly bounded in the sense that
\begin{equation}
\sup_{n \geq 1} \| \eta_n \|_{\infty} < \infty \quad \mbox{where} \quad \| \eta_n \|_{\infty} = \sup_{X \subseteq \Lambda_n} \| \eta_n(X) \| \,.
\end{equation}
Lastly, we assume that there is a uniform gap above the ground state energy meaning that
\begin{equation} \label{g_infty} 
\gamma_0 = \inf_{n \geq 1} {\rm gap}(H_{\Lambda_n}) > 0 \, .
\end{equation}
We will refer to any sequence $\{ H_{\Lambda_n} \}_{n \geq 1}$ of Hamiltonians generated by a uniform sequence of finite volumes and a corresponding
sequence of interactions $\{ \eta_n \}_{n \geq 1}$ satisfying the constraints above as a {\it uniformly gapped sequence of 
	initial Hamiltonians}. 

\begin{ex}
	The situation $d_n = d\restriction_{\Lambda_n}$ and $\eta_n = \eta$ for each $n\geq 1$ corresponds to traditional open boundary conditions. 
	By modifying the metric on each finite volume $\Lambda_n$, models with periodic boundary conditions
	can also be accommodated as above. In fact, this construction allows for models with various 
	boundary conditions such as when $d_n = d\restriction_{\Lambda_n}$ but $\eta$ is modified along the boundary, i.e. $\eta_n = \eta + \eta_{\partial \Lambda_n}$.
\end{ex}

\begin{remk} \label{rem:ball_ini_ham}
	Given a uniformly gapped sequence of initial Hamiltonians $\{ H_{\Lambda_n} \}_{n \geq 1}$, 
	the anchoring procedure described in Section~\ref{sec:ball-proc}
	applies. In this case, for each $n \geq 1$, one may write 
	\begin{equation}
	H_{\Lambda_n} = \sum_{x \in \Lambda_n} h_x^{(n)}
	\end{equation}
	with $h_x^{(n)} \geq 0$ and $h_x^{(n)} \in \mathcal{A}_{b_x^{\Lambda_n}(R)}$ for all $x \in \Lambda_n$. 
	Here $R \geq 0$ is the maximal radius associated with this anchoring. This $R$ is independent of $n$ and 
	satisfies $R \leq R'+1$. Moreover, as discussed in Section~\ref{sec:ball-proc}, 
	one has that
	\begin{equation} \label{global_bd_ini_ham}
	\| h \| = \sup_{n \geq 1} \| h^{(n)} \|_{\infty} < \infty \quad \mbox{where} \quad \| h^{(n)} \|_{\infty} = \sup_{x \in \Lambda_n} \| h_x^{(n)} \| \, .
	\end{equation}
\end{remk}

\begin{remk} \label{rem:LTQO_reg} Let $\Omega: \mathbb{R} \to [0, \infty)$ be a non-increasing function.
	Given a uniformly gapped sequence of initial Hamiltonians $\{ H_{\Lambda_n} \}_{n \geq1}$ and two
	sequences of non-negative numbers $\{ K_n \}_{n \geq 1}$, with $K_n \geq R$, and $\{ L_n \}_{n \geq 1}$, one can define 
	LTQO regions:
	\begin{equation} \label{def:L_n^p}
	\Lambda_n^p = \{ x \in \Lambda_n : r_y^{\Omega}( \Lambda_n) \geq K_n + L_n \mbox{ for all } y \in b_x^{\Lambda_n}(K_n) \} \, ,
	\end{equation}
	compare with (\ref{def:L_p}), on which our stability argument allows for perturbations.  Note that here the quantity $r_y^{\Omega}(\Lambda_n)$ represents
	the indistinguishability radius of $H_{\Lambda_n}$ at $y \in \Lambda_n$, see
	Definition~\ref{LTQO_def}. One analogously defines effective perturbation regions
	\begin{equation} \label{def:L_n^pK_n}
	\Lambda_n^p(K_n) = \{ x \in \Lambda_n : d_n(x, \Lambda_n^p) \leq K_n \} \, ,
	\end{equation}
	compare with (\ref{eff_pert_reg}). We note that, for all $n \geq 1$, both of these subsets of $\Lambda_n$ 
	are defined with respect to the same fixed non-increasing function $\Omega$. 
	As we will see, the notion that a model (one for which a choice of $\Omega$ has already been made) satisfies 
	our stability bounds {\it uniformly} requires an appropriate choice of the sequences $\{ K_n \}_{n \geq 1}$ and $\{ L_n \}_{n \geq 1}$. Such a 
	choice, however, is not independent of the perturbation; more on that soon. 
\end{remk} 

{\em Uniformity in the Separating Partitions:}
In Section~\ref{sec:form_bd}, see specifically Theorem~\ref{thm:Step3General}, we made use
of separating partitions, and we here briefly remark on how this notion can also be made uniform. 
Let $\{ \Lambda_n \}_{n \geq 1}$ be a uniform sequence of finite volumes. 
For each $n \geq 1$, denote by $\ell_n = \lceil\diam_n\Lambda_n\rceil$. 
Choose $\mathcal{S}^{(n)}$ to be
a collection of subsets of $\Lambda_n$ satisfying
\begin{equation}\label{unif_gap_volumes}
\caS^{(n)} = \left\{\Lambda_{n}(x,m) \mid x\in \Lambda_n,\, R\leq m \leq \ell_n\right\}, \mbox{ with } b_x^{\Lambda_n}(m)\subset \Lambda_n(x,m).
\end{equation}
For brevity, we say that any such sequence $\{ \mathcal{S}^{(n)} \}_{n \geq 1}$ is a sequence of subsets associated to $\{ \Lambda_n \}_{n \geq 1}$
which contains balls. Corresponding to such a sequence $\{ \mathcal{S}^{(n)} \}_{n \geq 1}$, 
we will require that there exists a sequence $\{ \mathcal{T}^{(n)} \}_{n \geq 1}$ of 
families of separating partitions which satisfies a polynomial growth bound 
independent of $n$. More precisely, we will assume that there are positive numbers $c$ and $\zeta$ for which given any $n \geq 1$, 
there is a family $\mathcal{T}^{(n)}$ of partitions of $\Lambda_n$ which separates $\mathcal{S}^{(n)}$ and is of 
$(c, \zeta)$-polynomial growth in the sense of Definition~\ref{ass:sep_part}. This means, if we 
write $\mathcal{T}^{(n)} = \{ \mathcal{T}_m^{(n)} \mid 1 \leq m \leq \ell_n \}$ and denote 
each partition of $\Lambda_n$ by $\mathcal{T}_m^{(n)} = \{ \mathcal{T}_{m,i}^{(n)} \mid i \in \mathcal{I}_m^{(n)} \}$,
then
\begin{enumerate}
	\item[(i)] \emph{Separation:} $\Lambda_n(x,m)\cap\Lambda_n(y, m)=\emptyset$ for any distinct pair $x, y\in T^{(n)}_{m,i}$.
	\item[(ii)] \emph{Uniform polynomial growth:} $| \cI_m^{(n)}| \leq c m^\zeta$ for all $n,m$.
\end{enumerate}
In this case, we say that $\{ \mathcal{T}^{(n)} \}_{n \geq 1}$ is a sequence of families of partitions which 
separates $\{ \mathcal{S}^{(n)} \}_{n \geq 1}$ and satisfies a uniform polynomial growth bound. 

\begin{remk}
	Generally, the existence of such families of partitions with uniform polynomial growth is not hard to establish.  
	Typically, one knows the existence of such sets on $\Gamma$. 
	In fact, if there is a collection of finite volumes 
	\be
	\caS = \{\Gamma(x,m) \in \cP_0(\Gamma) \mid x\in \Gamma,\, m\geq R\},
	\mbox{ with } b_x^{\Gamma_n}(m)\subset \Gamma_n(x,m).
	\ee
	and a corresponding family $\mathcal{T} = \{ \mathcal{T}_m \mid m \geq R \}$ of partitions $\mathcal{T}_m$ of $\Gamma$
	with $\mathcal{T}_m = \{\mathcal{T}_{m,i} \mid i \in \mathcal{I}_m \}$ satisfying:
	\begin{enumerate}
		\item[(i)] $| \cI_m| \leq c m^\zeta$;
		\item[(ii)] for $x\neq y\in T^i_m$, $\Gamma(x,m)\cap\Gamma(y,m)=\emptyset$,
	\end{enumerate}
	then, along any uniform sequence $\{ \Lambda_n \}_{n \geq 1}$, an obvious choice for $\mathcal{S}^{(n)}$ and $\mathcal{T}^{(n)}$ is
	obtained through intersection, namely
	\begin{equation}
	\Lambda_n(x,m)= \Gamma(x,m)\cap \Lambda_n \quad \mbox{and} \quad \cT^{(n)}_m=\{ T_{m,i}\cap\Lambda_n \mid i\in \cI_m\}.
	\end{equation} 
	Moreover, using $\nu$-regularity, one can show that the set of balls $\Gamma(x,m) = b_x(m)$ always corresponds to a 
	family of partitions $\mathcal{T}$ which separates these balls and is of $(\kappa,\nu)$-polynomial growth. 
\end{remk}

\begin{remk} \label{rem:uni_local_gaps}
	Let $\{ \Lambda_n \}_{n \geq 1}$ be a uniform sequence of finite volumes,
	$\{ \mathcal{S}^{(n)} \}_{n \geq 1}$ be any sequence of subsets associated to $\{ \Lambda_n \}_{n \geq 1}$
	which contains balls, and $\{ H_{\Lambda_n} \}_{n \geq 1}$ be a uniformly gapped sequence of initial Hamiltonians.
	For each $n \geq 1$ and any $R \leq m \leq \ell_n$, a {\it local gap} is defined by setting
	\be \label{loc_gap}
	\gamma_n(m) =\inf\left\{ \gap(H_{\Lambda_n(x,m)}) \mid \Lambda_n(x,m) \in \mathcal{S}^{(n)} \right\}
	\ee
	where the corresponding local Hamiltonians are given by 
	\be
	H_{\Lambda_{n}(x,m)} = \sum_{X\subseteq \Lambda_{n}(x,m)} \eta_n(X).
	\ee
	Since $\Lambda_n =\Lambda_n(x,\ell_n)\in S^{(n)}$ for any $n\geq 1$ and $x\in\Lambda_n$, the infimum of these local gaps 
	produces a lower bound on $\gamma_0$, as in (\ref{g_infty}).
\end{remk}

\begin{ex}
	The freedom of choosing appropriate sub-volumes $\caS^{(n)}$ can be useful for optimizing the lower bound on $\gamma_0$. 
	Consider the one-species PVBS model on $\Gamma = \bZ^\nu$ 
	as analyzed, e.g., in \cite{bishop:2016a}. This is an example of a model where local gaps 
	are sensitive to the boundary geometry. In fact, for a particular choice of parameters, 
	the spectral gap for the Hamiltonians associated with balls $b_x^{\Lambda_n}(m)$ closes as $n\to\infty$, but 
	the corresponding gaps remain non-vanishing on volumes $\Lambda_n(x,m)\supset b_x^{\Lambda_n}(m)$ with slightly slanted boundaries, 
	see \cite{bishop:2016a}. 
\end{ex}

{\em Uniformity in the Perturbations:}
Here we discuss a class of perturbations to which our stability results 
will apply. Let $F_0$ be an $F$-function on $(\Gamma, d)$. 
This base $F$-function depends only on $\Gamma$, and one may take it 
as in (\ref{poly-dec-F}). Let $\theta \in (0, 1]$. The class of perturbations
we consider are determined by a weighted $F$-function $F(r) = e^{-g(r)}F_0(r)$
with a weight $e^{-g}$ for which there is some $a>0$ such that
\begin{equation} \label{g_grows_theta}
g(r) = a r^{\theta} \quad \mbox{for all } r \geq 0 \, .
\end{equation}
Let $F$ denote such a weighted $F$-function and $\{ \Lambda_n \}_{n \geq 1}$ a uniform
sequence of finite volumes. We say that a sequence of interactions $\{ \Phi_n \}_{n \geq 1}$, 
with each $\Phi_n$ an anchored interaction on $\Lambda_n$,
decays like $F$ uniformly along $\{ \Lambda_n \}_{n \geq 1}$
if there is a non-negative number $\| \Phi \|_{1,F}$ for which
\begin{equation} \label{global_term_bd}
\| \Phi_n(x,m) \| \leq \| \Phi \|_{1,F} F( \max(0, m-1)) \quad \mbox{ for all } n \geq 1, \, x \in \Lambda_n, \mbox{ and } R \leq m \leq \ell_n \, ,
\end{equation}
and moreover, $\sup_{n \geq 1} \| \Phi_n \|_{1,F} \leq \| \Phi \|_{1, F}$ where
\begin{equation} \label{phin_1F_norm}
\| \Phi_n \|_{1,F} = \sup_{x,y \in \Lambda_n} \frac{1}{F(d_n(x,y))} \sum_{z \in \Lambda_n} \sum_{\stackrel{R\leq m \leq \ell_n:}{x,y \in b_z^{\Lambda_n}(m)}} |b_z^{\Lambda_n}(m)| 
\| \Phi_n(z,m) \| .
\end{equation}

Here, to be consistent with the notation of Section~\ref{sec:Step1} and Section~\ref{sec:LTQO2},
we organize the terms in the anchored interactions to start with $m=R$, the maximal radius associated with
the initial Hamiltonians as discussed in Remark~\ref{rem:ball_ini_ham}. 

In the case of open boundary conditions, for any anchored interaction $\Phi \in \mathcal{B}^1_{F}( \Gamma)$
one may take 
\begin{equation}
\| \Phi \|_{1,F} = \sup_{x,y \in \Gamma} \frac{1}{F(d(x,y))} \sum_{z \in \Gamma} \sum_{\stackrel{m \geq R:}{x,y \in b_z(m)}} 
|b_z(m)| \| \Phi(z,m) \| \, .
\end{equation}

The previous discussions motivate the following definition of a class of {\it perturbation models}.
For these we can then formulate the uniformity assumptions in Assumption~\ref{ass:uni_pert_model} below that are 
sufficient to prove stability in Theorem~\ref{thm:uni_stab}. 


\begin{defn}[Perturbation Models] \label{def:pert_model} Consider a quantum spin system defined on a
	$\nu$-regular metric space $(\Gamma, d)$. A perturbation model on $(\Gamma, d)$ 
	consists of the following:
	\begin{enumerate}
		\item[(i)] A uniform sequence of finite volumes $\{ \Lambda_n \}_{n \geq 1}$.
		\item[(ii)] A sequence $\{ \mathcal{S}^{(n)} \}_{n \geq 1}$ of subsets of $\{ \Lambda_n \}_{n \geq 1}$ 
		containing balls and a corresponding sequence $\{ \mathcal{T}^{(n)} \}_{n \geq 1}$ of families of 
		partitions which separates $\{ \mathcal{S}^{(n)} \}_{n \geq 1}$ and satisfies a uniform polynomial growth bound. 
		\item[(iii)] A uniformly gapped sequence of frustration-free 
		Hamiltonians $\{ H_{\Lambda_n} \}_{n \geq 1}$.
		\item[(iv)] A non-increasing function $\Omega : \mathbb{R} \to [0, \infty)$ with $\lim_{r\to\infty}\Omega(r)=0$ that defines the indistinguishability 
		radii, $r_x^\Omega(\Lambda_n)$, for each initial Hamiltonian $H_{\Lambda_n}$, and two
		sequences of non-negative numbers $\{ K_n \}_{n \geq 1}$ and $\{ L_n \}_{n \geq 1}$, which define the perturbation regions $\Lambda^p_n$ as in \eqref{def:L_n^p}.
		\item[(v)] A weighted $F$-function $F$ on $(\Gamma, d)$ with weight 
		satisfying (\ref{g_grows_theta}) for some $\theta \in (0,1]$,  and a sequence of anchored interactions $\{ \Phi_n \}_{n \geq 1}$ which decays like $F$ uniformly along $\{ \Lambda_n \}_{n \geq 1}$.
	\end{enumerate}
	For each perturbation model and any $0 \leq s \leq 1$, a sequence of perturbed Hamiltonians is given by
	\begin{equation} \label{pert_model_hams}
	H_{\Lambda_n}(s) = H_{\Lambda_n} + s \Vee_{\Lambda_n^p} \quad \mbox{where} \quad 
	\Vee_{\Lambda_n^p} = \sum_{x \in \Lambda_n^p} \sum_{\stackrel{m \geq R:}{b_x(m) \subseteq \Lambda_n}} \Phi_n(x,m) \, .
	\end{equation}
\end{defn}

\begin{remk} \label{rem:uni_vel_bd} Consider a perturbation model
	on $(\Gamma, d)$ in the sense of Definition~\ref{def:pert_model}. 
	For each $n \geq 1$, the results of Section~\ref{sec:Step1} and Section~\ref{sec:LTQO2}
	apply to the perturbed Hamiltonians $H_{\Lambda_n}(s)$ as in (\ref{pert_model_hams}) above. In fact, if $\| \Phi \|_{1,F}$ 
	is the value estimating the corresponding sequence of perturbations $\{ \Phi_n \}_{n \geq 1}$,
	as in (\ref{global_term_bd}) and (\ref{phin_1F_norm}), then a bound on the Lieb-Robinson velocity
	associated to the Heisenberg dynamics generated by $H_{\Lambda_n}(s)$ is 
	\begin{equation}
	v = 2C_F \left( \frac{ \kappa R^{\nu}}{F(R)} \| h \| + \| \Phi \|_{1, F} \right) 
	\end{equation}
	compare with (\ref{lrb-vest}). Here $\| h \|$ is the uniform estimate on the initial Hamiltonians, see
	comments in Remark~\ref{rem:ball_ini_ham}, and we stress that this value of $v$ is uniform in
	$n \geq 1$ and $0 \leq s \leq 1$. In fact, regarding the sequence of initial Hamiltonians $\{ H_{\Lambda_n} \}_{n \geq 1}$
	and the weighted $F$-function $F$ fixed, this value of $v$ is further uniform with respect to any
	sequence of perturbations $\{ \Phi'_n \}_{n \geq 1}$ satisfying $\| \Phi' \|_{1,F} \leq \| \Phi \|_{1,F}$. 
\end{remk}

In order for a perturbation model as defined above to have a stable spectral gap
we will further need to assume uniform boundedness of certain estimates.
To turn this property into an assumption for models, let us first recall some notation from Section~\ref{sec:LTQO2}. 

Consider a perturbation model as in Definition~\ref{def:pert_model}. 
Take $\gamma >0$ with  $\gamma < \gamma_0$, the uniform ground state gap, see
(\ref{g_infty}), of our sequence of initial Hamiltonians. One checks that, for each $n \geq 1$, the
Hamiltonian $H_{\Lambda_n}$ and perturbation $V_{\Lambda_n^p}$, see e.g. 
Remark~\ref{rem:uni_vel_bd}, satisfy the conditions of Assumption~\ref{ass:basic}.
In this case, Theorem~\ref{thm:step2-(i)+(ii)} applies, and the numbers
\begin{equation} \label{def:delta_n}
\delta_n = | \Lambda_n^p(K_n)| C(K_n, L_n) 
\end{equation}
compare with (\ref{LTQO_delta}), and
\begin{equation} \label{def:epsilon_n}
\epsilon_n = 2 \left( \sum_{x \in \Lambda_n \setminus \Lambda_n^p(K_n)} \| h_x \| \right) G(K_n)
\end{equation}
compare with (\ref{LTQO_alpha}), are relevant for our stability analysis. The quantity $C(K_n,L_n)$, which 
appears in (\ref{def:delta_n}) above, is as defined in (\ref{def:ckl}). This quantity 
is defined with respect to a decay function which is obtained by applying Theorem~\ref{thm:Step1}
in the finite volume $\Lambda_n$. Given the uniformity imposed by Definition~\ref{def:pert_model},
there is a single choice which works in all finite volumes, and we 
continue to denote this particular decay function by $G^{(1)}$. To be clear, this function $G^{(1)}$
depends on the choice of $\gamma$ in that it depends on estimates for the spectral flow $\alpha_s^{\Lambda_n}( \cdot)$
which is defined for $\xi = \gamma$ as in Section~\ref{sec:spec_flow}. Moreover, $G^{(1)}$ also depends on the weighted $F$-function
$F$, and the sequence of perturbations $\{ \Phi_n \}_{n \geq 1}$ through $\| \Phi \|_{1,F}$. As previously observed, 
given $F$ and a value of $\| \Phi \|_{1,F}$, this decay function $G^{(1)}$ holds uniformly for
all sequences of anchored interactions $\{ \Phi'_n\}_{n \geq 1}$ satisfying $\| \Phi' \|_{1,F} \leq \| \Phi \|_{1,F}$.  
Arguing similarly, there is a single choice of decay function corresponding to the proof of Lemma~\ref{lem:FarDecay}
which holds uniformly in the sense described above. We call this function $G$, and it is
the function which we use in (\ref{def:epsilon_n}).

For the same perturbation model, let us also fix, independent of $n$, a differentiable function $f: [0, \infty) \to [0, \infty)$ 
with $f(0) =0$ and $0< f'(t) <1$ for all $t \geq 0$.  Given this, one further checks that Theorem~\ref{thm:rel_bound_cond} applies 
in each finite volume $\Lambda_n$  and determines an anchored interaction $\Vee_n^{(2)}(s)$  for all
$0 \leq s \leq s_{\gamma}^{\Lambda_n}$. Recall that $s_{\gamma}^{\Lambda_n}>0$ is as in (\ref{def:s_lam_gam}).
Given the conclusions of Theorem~\ref{thm:rel_bound_cond}, it is clear that $H_{\Lambda_n}$, $\mathcal{S}^{(n)}$, and $\Vee_n^{(2)}(s)$
satisfy the conditions of Theorem~\ref{thm:Step3General} for each $n \geq 1$ and all $0 \leq s \leq s_{\gamma}^{\Lambda_n}$.
As a result, each $\Vee_n^{(2)}(s)$ is form bounded by $H_{\Lambda_n}$ with a pre-factor given by
\begin{equation} \label{def:beta_n}
\beta_n = 2 c \sum_{m = R}^{\ell_n} \frac{m^{\zeta} G_n^{(2)}(m)}{ \gamma_n(m)} \, .
\end{equation}
Recall that the family of partitions of $\Lambda_n$ which separates $\mathcal{S}^{(n)}$ is uniformly of $(c, \zeta)$-polynomial growth, and
$\gamma_n(m)$ is the local gap defined in (\ref{loc_gap}). Moreover, $G_n^{(2)}$ is the decay function obtained in the
application of Theorem~\ref{thm:rel_bound_cond}; namely, setting $M_n = K_n+L_n$
\begin{equation} \label{dec_fun_g2}
G^{(2)}_n(m) = \begin{cases}
G^{(1)}(f(m))+ D(\lceil f(m) \rceil - 1, m -1) + C(K_n,L_n) & R\leq m < M_n \\
\sum_{k \geq f(M_n)} G^{(1)}(k) + D( \lceil f(M_n) \rceil - 1, M_n -1) + C(K_n,L_n) & \phantom{R\leq \;} m \geq M_n . \\
\end{cases}
\end{equation} 
\begin{assumption}[Uniform Perturbation Model] \label{ass:uni_pert_model} We say that a perturbation model, as in Definition~\ref{def:pert_model},
	is a uniform perturbation model if:
	\begin{enumerate}
		\item[(i)] The quantities $\delta_n$ and $\epsilon_n$ from \eqref{def:delta_n}-\eqref{def:epsilon_n} are uniformly bounded from above, i.e.
		\begin{equation}
		\delta = \sup_{n \geq 1} \delta_n < \infty \quad \mbox{and} \quad \epsilon = \sup_{n \geq 1} \epsilon_n < \infty \, .
		\end{equation}
		\item[(ii)] There exists a function $f$ as above for which $\beta_n$ from \eqref{def:beta_n} satisfies
		\begin{equation}
		\beta  = \sup_{n \geq 1} \beta_n < \infty \, .
		\end{equation}
	\end{enumerate}
\end{assumption}

We now state our main result on stability for perturbation models as in Definition~\ref{def:pert_model}. Recall that {\it stability of the spectral gap}, as described 
at the end of Section~\ref{sec:sf-gap-stab}, is the property that 
\begin{equation}
s_{\gamma} = \inf_{n \geq 1} s_{\gamma}^{\Lambda_n} >0 \quad \mbox{for all} \quad 0 < \gamma < \gamma_0 \, .
\end{equation}

\begin{thm} \label{thm:uni_stab} 
	Every uniform perturbation model has a stable spectral gap.
\end{thm}

\begin{proof}
	Consider any uniform perturbation model, i.e. a model as in Definition~\ref{def:pert_model} that satisfies Assumption~\ref{ass:uni_pert_model}, and
	let $0< \gamma < \gamma_0$, where $\gamma_0$ is uniform ground state gap from (\ref{g_infty}). 
	Gathering our results, we have shown that for any $n \geq 1$ and
	all $0 \leq s \leq s_{\gamma}^{\Lambda_n}$ the decomposition:
	\begin{equation}
	\alpha_s^{\Lambda_n}(H_{\Lambda_n}(s)) = H_{\Lambda_n} + \Vee_n^{(2)}(s) + \Delta_n(s) + E_n(s) + C_n(s) \idty
	\end{equation}
	holds; we refer to the beginning of Section~\ref{sec:LTQO2} for a review of the relevant notation. Moreover,
	we have checked that the properties listed in Claim~\ref{clm:decompositionH} are satisfied with the parameters $( \delta_{\Lambda}, \epsilon_{\Lambda}, \beta_{\Lambda})$
	replaced by $( \delta_n, \epsilon_n, \beta_n)$. In fact, by estimating with the corresponding supremums, for these uniform perturbation
	models, the same parameters may be replaced by $( \delta, \epsilon, \beta)$. Similar to the discussion after Claim~\ref{clm:decompositionH}, a finite volume
	application of Theorem~\ref{thm:Pert_Est} yields that for all $0 \leq s < \min\{\beta^{-1}, s_{\gamma}^{\Lambda_n} \}$, 
	\begin{equation}
	\Sigma_1^{\Lambda_n}(s) \subseteq [C_n(s) - s( \delta + \epsilon), C_n(s) + s( \delta + \epsilon)] \quad \mbox{and} \quad
	\Sigma_2^{\Lambda_n}(s) \subseteq [C_n(s) + ( 1 - s \beta) \gamma_0 - s \epsilon , \infty).
	\end{equation}
	Note that in each application of Theorem~\ref{thm:Pert_Est} above we use $\gamma_0$ for $\gamma$ in the statement of the theorem. This demonstrates that 
	\be\label{split_diamter}
	{\rm diam}( \Sigma_1^{\Lambda_n}(s)) \leq 2 s ( \delta + \epsilon)
	\ee and moreover, 
	\begin{equation} \label{f_uni_gap_est}
	{\rm gap}(H_{\Lambda_n}(s)) = {\rm dist} \left( \Sigma_1^{\Lambda_n}(s), \Sigma_2^{\Lambda_n}(s) \right) \geq \gamma_0 -
	s( \gamma_0 \beta + \delta + 2 \epsilon) \, .
	\end{equation}
	{F}rom (\ref{f_uni_gap_est}), it is clear that ${\rm gap}(H_{\Lambda_n}(s)) \geq \gamma$ holds whenever $s$ is small enough so that
	\begin{equation} 
	\gamma_0 - s( \gamma_0 \beta + \delta + 2 \epsilon) \geq \gamma \quad \Longleftrightarrow \quad s\leq \frac{\gamma_0 - \gamma}{ \gamma_0 \beta + \delta + 2 \epsilon} \, .
	\end{equation}
	In this case, 
	\begin{equation}
	s_\gamma=\inf_{n \geq 1} s_{\gamma}^{\Lambda_n} \geq \frac{\gamma_0 - \gamma}{ \gamma_0 \beta + \delta + 2 \epsilon} > 0, 
	\end{equation}
	and thus the model is stable.
\end{proof}

For many applications, one is primarily interested in establishing stability for large finite volumes.
In such cases, it suffices to replace the supremums considered in Assumption~\ref{ass:uni_pert_model} with 
limit superiors, and the corresponding gap estimates likely improve. 

Another common situation concerns uniform perturbation models for which the quantities $\delta_n$ and $\epsilon_n$,
see (\ref{def:delta_n}) and (\ref{def:epsilon_n}), become vanishingly small, i.e. 
\begin{equation} \label{lim=0}
\lim_{n \to \infty} ( \delta_n + \epsilon_n) = 0 \, .
\end{equation} 
In this situation, for any $0\leq s \leq s_\gamma$, the diameter of the ground state splitting $\diam(\Sigma_1^{\Lambda_n}(s))$, see \eqref{split_diamter}, tends to zero as $n\to\infty$. It is for this case we will be able to show spectral gap stability in the thermodynamic limit. This is the topic of Section~\ref{sec:automorphic-equivalence}, see specifically Corollary~\ref{cor:GNSgap}. For convenience of later reference, we state this as a corollary.

\begin{cor} \label{cor:zerod+e} 
	Consider a uniform perturbation model for which (\ref{lim=0}) holds. In this case, for each $0< \gamma < \gamma_0$, 
	there is $s_{\gamma} >0$ for which we have $\gamma_n'(s)$ and $ \epsilon_n'(s)$ such that
	\begin{equation}
	{\rm spec}\left(H_{\Lambda_n}(s) - \gsE_{\Lambda_n}(s) \idty \right) \subseteq [0, \epsilon_n'(s)] \cup [\gamma_n'(s) + \epsilon_n'(s), \infty) \quad 
	\mbox{for all } 0 \leq s \leq s_{\gamma} \, ,
	\end{equation}
	where $\gsE_{\Lambda_n}(s) = {\rm min \, spec} \left(H_{\Lambda_n}(s) \right)$ and
	\begin{equation}
	\liminf_{n \to \infty} \gamma_n'(s) \geq \gamma \quad \mbox{while} \quad \lim_{n \to \infty} \epsilon_n'(s) = 0 \, .
	\end{equation}
\end{cor}

The previous results also provide estimates on the convergence rates for $\gamma_n'(s)$ and $ \epsilon_n'(s)$. 

\subsection{Applications}\label{sec:uniform_sequences_geometric-bc}

In Definition~\ref{def:pert_model} and Assumption~\ref{ass:uni_pert_model} we formulated general conditions
under which we proved stability of the spectral gap uniform in the volume. Most items in these conditions are
part of the standard setting and straightforward to verify: a suitable sequence of finite volumes, a uniformly gapped 
sequence of initial Hamiltonians, and a class of perturbations that are sufficiently short range.
Following the philosophy of Bravyi-Hastings-Michalakis, we also introduced a decaying function
$\Omega$ to give a quantitive expression of the LTQO property through our definition of the 
indistinguishability radius \eq{LTQO_length}. Physically, this property expresses the local indistinguishability of the different
ground states of the finite-volume Hamiltonians. Verifying a sufficient quantitative version of the LTQO condition
is, however, less straightforward. In particular, it is clear from the discussion preceding Assumption~\ref{ass:uni_pert_model}
that doing this involves a combination of a number of characteristics of the models all at once.
In this section, we discuss what this usually comes down to in practice in both settings of boundary conditions expressed by 
boundary terms, and geometric boundary conditions.
%

\subsubsection{Verifying Assumptions} \label{sec:veri_ass}
In this section, we briefly discuss Assumption~\ref{ass:uni_pert_model} and identify certain situations under which it clearly holds. 
Throughout, we assume $\Gamma$ is infinite. 
First, observe that if Assumption~\ref{ass:uni_pert_model} (ii) holds, then the sequences $\{ K_n \}_{n \geq 1}$ and $\{ L_n \}_{n \geq 1}$
must satisfy $K_n \to \infty$ and $L_n \to \infty$. In fact, their sum, $K_n+L_n$, cannot be bounded because 
the function $G_n^{(2)}$, see (\ref{dec_fun_g2}), is eventually equal to a non-zero constant whereas the truncated moment in (\ref{def:beta_n}) 
grows as $\ell_n = \lceil {\rm diam}_n( \Lambda_n) \rceil \to \infty$. Therefore, from the structure of $G_n^{(2)}$, it is clear that both these
sequences must be unbounded. For this reason, Assumption~\ref{ass:uni_pert_model} implies that any point $x$ which is
eventually in all LTQO regions, i.e. $x \in \bigcap_{n=m}^{\infty} \Lambda_n^p$, for some $m$, must satisfy
$\lim_{n\to\infty} r_x^{\Omega}( \Lambda_n) = \infty$. 

One concludes that in the general formulation we have given, the approach is not particularly well-suited when $\Gamma$ has a boundary,
as this is where LTQO can often fail. Note, however, that in one dimension the method can be adapted to yield useful results \cite{Moon:2018}. A similar 
modification could also handle other cases of systems with a finite boundary.


One can check that Assumption~\ref{ass:uni_pert_model} (i) holds whenever
\begin{equation} \label{check_ass_1}
\sup_{n \geq 1} \ell_n^{\nu} \left[ C(K_n, L_n) + G(K_n) \right] < \infty \, .
\end{equation}
In fact, since $\nu$-regularity guarantees that $| \Lambda_n| \leq \kappa \ell_n^{\nu}$, one has that
\begin{equation}
\delta_n \leq \kappa \ell_n^{\nu} C(K_n, L_n) \quad \mbox{and} \quad \epsilon_n \leq 2 \kappa \| h \| \ell_n^{\nu}G(K_n) \, .
\end{equation}
Similarly, if
\begin{equation}
\limsup_{n \to \infty } \ell_n^{\nu} \left[ C(K_n, L_n) + G(K_n) \right] =0 \, ,
\end{equation}
then (\ref{lim=0}) holds. 

Under the additional assumption that the local gaps associated to the unperturbed Hamiltonians decay no faster than a power law,
we can formulate a similar statement about Assumption~\ref{ass:uni_pert_model} (ii). More precisely, recall that for each $n \geq 1$
and all $R \leq m \leq \ell_n$, we defined finite-volume, local gaps $\gamma_n(m)$ in (\ref{loc_gap}). We will say that these
local gaps decay no faster than a power law if there are numbers $C>0$ and $k \geq 0$ for which, given any $n \geq 1$,  
\begin{equation} \label{power_law_local_gap}
\gamma_n(m) \geq \frac{C}{m^k} \quad \mbox{for all } R \leq m \leq \ell_n \, . 
\end{equation}
The situation $k=0$ corresponds to the case of uniformly bounded local gaps, and can be checked in some applications.  Given (\ref{power_law_local_gap}),
one derives from (\ref{def:beta_n}) the bound
\begin{equation}
\beta_n \leq \frac{2c}{C} \sum_{m=R}^{\ell_n} m^{ \zeta + k} G_n^{(2)}(m) 
\end{equation}
From this we conclude that when the local gaps decay no faster than a power law, as in (\ref{power_law_local_gap}), 
Assumption~\ref{ass:uni_pert_model} (ii) is satisfied when
\begin{equation} \label{check_ass_2}
\sum_{m=R}^{\infty} m^{\zeta +k} \left( G^{(1)}(f(m)) +D( \lceil f(m) \rceil -1, m-1) \right) < \infty \quad \mbox{and} \quad 
\sup_{n \geq 1} \ell_n^{ \zeta + k +1}G_n^{(2)}(K_n+L_n) < \infty \, .
\end{equation}

Since we only considered perturbation models satisfying Definition~\ref{def:pert_model}(v), the
functions $G$ and $G^{(1)}$ which frequently enter above are both of decay class $(\eta, \frac{\gamma}{2v}, \theta)$, see 
Definition~\ref{def:dec-class}. In this case, all moments of these functions are necessarily finite. As a result, 
verifying conditions (\ref{check_ass_1}) and (\ref{check_ass_2}) above primarily entails checking that
one has adequate decay of $\Omega$, coupled with appropriate growth of the sequences 
$\{ K_n \}_{n \geq 1}$ and $\{L_n \}_{n \geq 1}$.

Finally, let us remark that if one assumes at least power law bounds for the decay of $\Omega$ and the growth of
$\{ K_n \}_{n \geq 1}$ and $\{ L_n \}_{n \geq 1}$, a sufficient condition for stability can be given in terms of an 
inequality for the exponents. Specifically, suppose there are $\alpha_1, \alpha_2 \in (0,1)$
such that the sequences $\{ K_n \}_{n \geq 1}$ and $\{ L_n \}_{n \geq 1}$ satisfy
\begin{equation} \label{def:L_nK_n}
K_n \geq \left\lceil\ell_n^{\alpha_1}\right\rceil \quad \mbox{and} \quad L_n \geq  \left\lceil\ell_n^{\alpha_2}\right\rceil \quad \mbox{for each } n \geq 1, 
\end{equation}
where $\ell_n = \lceil {\rm diam}_n( \Lambda_n) \rceil$. 
Then, Assumption~\ref{ass:uni_pert_model} holds if $\Omega(r)$ is $O(r^{-q})$ with $q$ sufficiently large. Assumption~\ref{ass:uni_pert_model} (i) holds if $q \geq \alpha_2^{-1} \nu$.
Recalling (\ref{def:ckl}), we find that (\ref{check_ass_1}) holds whenever
\begin{equation} \label{make_go_0}
\ell_n^{\nu} \Omega(\ell_n^{\alpha_2}) < \infty \, .
\end{equation}
Here we have used that both $G$ and $G^{(1)}$ are in decay class $(\eta, \frac{\gamma}{2v}, \theta)$. Also one sees that (\ref{lim=0}) holds if $q > \alpha_2^{-1} \nu$. 
Furthermore, it suffices to take $f(x) = ax$ for some $a\in (0,1)$ and Assumption~\ref{ass:uni_pert_model} (ii) holds if
\begin{equation} \label{large_q}
q \geq \underline{\alpha}^{-1} \left( 2 (1+ \zeta +k) + \overline{\alpha} \nu \right)
\end{equation}
where we have set  $\underline{\alpha} = \min(\alpha_1, \alpha_2 )$ and $\overline{\alpha} = \max(\alpha_1, \alpha_2 )$. One checks
that (\ref{large_q}) implies Assumption~\ref{ass:uni_pert_model} (i), and so for any such value of $q$, Assumption~\ref{ass:uni_pert_model}  
holds with  (\ref{lim=0}). To see that (\ref{large_q}) is sufficient, note that (\ref{check_ass_2}) holds
whenever both
\begin{equation} \label{omeg_int_bd}
\sum_{m=R}^{\infty} m^{\zeta +k + \nu/2} \sqrt{ \Omega((1-a)m -1)} < \infty 
\end{equation}
and
\begin{equation} \label{omeg_sup_bd}
\sup_{n \geq 1} \ell_n^{\zeta + k +1} \left( \Omega(\ell_n^{\alpha_2}) + (\ell_n')^{\nu/2} \sqrt{\Omega( (1-a)\ell_n' -1) } \right)< \infty
\end{equation}
hold. Here we have set $\ell_n' = \ell_n^{\alpha_1} + \ell_n^{\alpha_2}$ and used (\ref{def:Dmn}). These claims are readily checked.

\subsubsection{The case of geometric boundary conditions} \label{subsec:PBC}

In \cite{bravyi:2010,bravyi:2011} and \cite{michalakis:2013} only sequences of finite systems defined on boxes in $\Ir^\nu$ with 
periodic boundary conditions are considered. Periodic boundary conditions are a special case of what we have called geometric 
boundary conditions induced by embedding finite subsets of lattices in a compact manifold (without a boundary).
Another example are twisted embeddings on a torus, which is a natural setting to address Lieb-Schultz-Mattis type questions\cite{yao:2019}. 
All these situations are described by an increasing and absorbing sequence of finite $\Lambda_n\subset\Gamma$ that are equipped with a metric 
$d_n$ that pointwise converges to the metric $d$ on $\Gamma$.	

The detailed estimates of Section \ref{sec:LTQO2} are more than is required to handle this
situation and allow us to identify specific classes of interactions for which we can prove simpler and more useful bounds. 

We consider here a perturbation model as in Definition~\ref{def:pert_model}. For typical examples such as, e.g. the Toric Code Model, there is a natural choice of decay function $\Omega$ for which the 
corresponding indistinguishability radius $r_x^\Omega(\Lambda_n)$, see  \eq{LTQO_length},
is proportional to the size of the smallest topologically non-trivial closed path in $\Lambda_n$ that contains $x$.
In such case, the following assumption is typically known, or readily checked.

\begin{assumption} \label{ind_rad_grow} For a given perturbation model, there is an increasing sequence $\{ r_n \}_{n \geq 1}$ of
	positive numbers with 
	\be\label{LTQO_geometric-bc}
	r^\Omega_x(\Lambda_n) \geq r_n \quad  \mbox{ for all } n\geq 1 \mbox{ and } x\in\Lambda_n.
	\ee
\end{assumption}
Due to assumptions we will make later we will also require $r_n\to\infty$. Given Assumption~\ref{ind_rad_grow}, our arguments simplify, and we here briefly describe these changes.

Consider a perturbation model such that Assumption~\ref{ind_rad_grow} holds. In such a case, it is convenient to take $K_n = L_n = r_n/2$ in Definition~\ref{def:pert_model}. With this choice, both the perturbation region, see (\ref{def:L_n^p}), and the effective perturbation region, see (\ref{def:L_n^pK_n}),
are extensive; in fact, $\Lambda_n^p = \Lambda_n^p(r_n/2) = \Lambda_n$ for all $n \geq 1$. As a result, the quantity $\epsilon_n$, see (\ref{def:epsilon_n}), satisfies $\epsilon_n=0$. 
%
%
%
%
%
The analogue of a \emph{uniform perturbation model}, i.e. Assumption~\ref{ass:uni_pert_model}, in this case is as follows.
\begin{assumption} \label{ass:new_uni} We say that a perturbation model satisfying Assumption~\ref{ind_rad_grow} is uniform if:
	\begin{enumerate}
		\item One has that $\delta = \sup_{n \geq 1} \delta_n < \infty$.
		\item There is a non-negative function $f$ with $f(0) = 0$ and $0< f'(t) < 1$ for all $t \geq 0$ such that
		\begin{equation}
		\beta = \sup_{n \geq 1} \beta_n < \infty \quad \mbox{where} \quad \beta_n = \sum_{m=R}^{\ell_n}\frac{m^\zeta}{\gamma_n(m)}G_n^{(2)}(m) \, .
		\end{equation}
	\end{enumerate}
\end{assumption}

Arguing as before, the following is clear.

\begin{thm}\label{thm:unif_stab_geometric-bc}
	Every perturbation model satisfying Assumption~\ref{ass:new_uni} has a stable spectral gap. 
\end{thm}

In fact, if such a model additionally satisfies $\lim_{n \to \infty} \delta_n =0$, then the analogue of 
Corollary~\ref{cor:zerod+e} holds as well.

%
%

\section{Infinite Systems and Automorphic Equivalence of Gapped Phases}\label{sec:automorphic-equivalence}

\subsection{Introduction}

In the previous sections we studied perturbations of systems defined on a finite set $\Lambda$ with one or more frustration-free 
ground states and a spectral gap. Theorem \ref{thm:uni_stab}, Corollary \ref{cor:zerod+e}, and Theorem \ref{thm:unif_stab_geometric-bc}
specify conditions under which we have a uniform positive lower bound for the spectral gap of a family of perturbed Hamiltonians defined on a sequence 
of finite volumes $\Lambda_n\uparrow\Gamma$ . We are now interested in applying these results to analyze the gap of the corresponding infinite model. 

The main goal of this section is to prove a lower bound for the spectral gap of the GNS Hamiltonians of the perturbed models in the thermodynamic limit. Such a bound would follow directly from strong resolvent convergence of the finite-volume perturbed Hamiltonians represented as operators acting on the GNS Hilbert space. Since the perturbations can spoil the frustration-free property of the Hamiltonians, though, it is not clear one can expect strong resolvent convergence in general. Therefore, in this section we develop a more direct approach to obtain bounds on the spectral gap in the thermodynamic limit, see Theorem \ref{thm:gap_ineq_TL}. In particular, we show that under the assumption of uniform LTQO, see Assumption~\ref{ind_rad_grow}, there is a unique gapped ground state in the thermodynamic limit.  We will also show that the perturbed models for which the stability result applies also have indistinguishable ground states. That is, LTQO is a stable property itself. The Toric Code model and, more generally, Kitaev's quantum double models, satisfy the conditions of this section. They and their perturbations have a translation invariant gapped pure ground state on $\Ir^2$. The one-dimensional AKLT model also satisfies the conditions of this section.

%
%

\subsection{Description of the infinite system}\label{sec:infinite_system}

We consider uniform perturbation models on an infinite set $\Gamma$, see Definition~\ref{def:pert_model} and Assumption~\ref{ass:uni_pert_model}. As a consequence, the spectral gap of the associated sequence of finite-volume perturbed Hamiltonians
\begin{equation}\label{pert_hams_TL}
H_{\Lambda_n}(s) = H_{\Lambda_n} + s \Vee_{\Lambda_n^p} \;\;\text{with}\;\; \Lambda_n\uparrow \Gamma
\end{equation}
is stable in the sense of Theorem~\ref{thm:uni_stab}, meaning that
\[
s_\gamma = \inf_{n\geq 1}s_\gamma^{\Lambda_n} >0 \;\; \text{for all} \;\; 0 < \gamma < \gamma_0
\]
where $s_{\gamma}^{\Lambda_n}$ is as in (\ref{def:s_lam_gam}) and $\gamma_0$ is the uniform lower-bound on the non-vanishing spectral gap above the ground state energy of the initial Hamiltonians. For each $n\geq 1$, $s_{\gamma}^{\Lambda_n}$ is a bounded non-increasing function of $\gamma$ and, hence, so is $s_\gamma$. Therefore, the following limit exists:
\be\label{s_gamma_inf} 
s_0 := \lim_{\gamma\to 0} s_{\gamma}.
\ee
By our definitions and assumptions, we can assume that $s_0\in(0,1]$, and for all $s\in [0,s_0)$ there exists $\gamma\in (0,\gamma_0)$ such that 
$s\in [0,s_\gamma)$, meaning that $\gap(H_{\Lambda_n}(s))\geq \gamma$ for all $n$. Said differently, for all $s\in [0,s_0)$, $\inf_n\gap(H_{\Lambda_n}(s)) >0$.

So far, we have not required that a {\em uniform perturbation model} converges in any sense as $n\to\infty$. We only imposed that for each $n$ the perturbation satisfied conditions that allowed us to prove that the gap above the ground state remains open uniformly in $n$. For a limiting perturbed infinite system to exist, we now add the assumption that the perturbing interactions $\Phi_n$ as described in Definition~\ref{def:pert_model}(v) converge locally in $F$-norm for the given $F$-function to an interaction $\Phi\in\cB_{F}$ on $\Gamma$. For static interactions which is the case we consider here, this notion
of convergence simply means that for all $\Lambda\in\cP_0(\Gamma)$,
\be\label{local_convergence}
\lim_{n\to\infty}\|(\Phi-\Phi_n)\restriction_{\Lambda} \|_F \to 0,
\ee
where $\Phi\restriction_\Lambda$ denotes the restriction of the interaction to $\Lambda$. This holds, for example, if the perturbations are eventually constant: there exists $\Phi\in\cB_{F}$ such that for all finite $X$ there is an $N$ so that $\Phi_n(X) = \Phi(X)$ for all $n\geq N$.

It was shown in \cite[Theorem 3.8]{nachtergaele:2019} that \eq{local_convergence} implies that the thermodynamic limit of the dynamics corresponding to $\Phi_n$ exists and equals the dynamics generated by $\Phi$. This implies
that there exists strongly continuous dynamics $\{\tau^s_t\}_{t\in\Rl}$ and $\{\alpha_s\}_{s\in\Rl}$ on $\cA_\Gamma$,  defined by
\bea
 \tau_t^s(A) = \lim_{n\to\infty} \tau^{s,\Lambda_n}_t (A), \quad \alpha_s(A) =  \lim_{n\to\infty} \alpha^{\Lambda_n}_s (A), \mbox{ for all }
 A\in\cA_\Gamma^{\rm loc}.
\eea
In the case of  $\alpha_s$ we left implicit the choice of the parameter $\xi>0$, which is kept fixed in this limit.
The convergence is uniform on any compact range of $t$ and $s$ and, as a consequence, the limit is strongly continuous in these parameters (see \cite{nachtergaele:2019} for proofs of these statements). It follows that $\{\tau_t^s\}_{t\in\Rl}$ is generated by a closed derivation 
$\delta_s$ for which $\cA_\Gamma^{\rm loc}$ is a core \cite{bratteli:1997}. Moreover, it is the limit of the finite-volume generators:
\be\label{delta_s}
\delta_s(A) = \lim_{n\to\infty}[H_{\Lambda_n}(s),A],\quad A\in\cA_{\Gamma}^{\rm loc}.
\ee

In order to express and study stability of the spectral gap in the thermodynamic limit, we will consider the GNS representation of an infinite-volume
ground state obtained as the thermodynamic limit of finite-volume ground states. As we will show, the set up considered in this section implies that such a
limiting state is pure and unique. In the next section we will discuss some important situations in which it is not unique.

Let $\tau:=\{\tau_t=e^{it\delta}\}_{t\in\Rl}$ be a strongly continuous dynamics on the $C^*$-algebra $\cA_\Gamma$, with a generator $\delta$ as in \eq{delta_s}, and let $(\pi_\omega,\cH_\omega, \Omega_\omega)$ denote the GNS representation of a $\tau$-invariant state $\omega$.
Then, by standard arguments (see, e.g.,  \cite{bratteli:1997} or \cite{naaijkens:2017}) the derivation is implemented by a self-adjoint operator, $H_\omega$, with dense domain $\dom H_\omega\subset \cH_\omega$ for which
\be
\pi_\omega(\delta(A)) = [H_\omega,\pi_\omega(A)] , \mbox{ for all } A\in\cA_{\Gamma}^{\rm loc}.
\ee
One has that $\cA_{\Gamma}^{\rm loc}$ is a core for $\delta$ (as a densely defined closed operator) and $\pi(\cA_{\Gamma}^{\rm loc})\Omega_\omega$ is a core for $H_\omega$. The spectrum of $H_\omega$ is then what we refer to as the spectrum of the infinite system. This is sometimes referred to as the {\em bulk spectrum}. Our main goal is to establish a spectral gap above the ground state of such a GNS Hamiltonian, $H_\omega$.

It is easy to see that $\omega$ is $\tau$-invariant if and only if  $\omega(\delta(A))=0$ for all $A\in\cA_{\Gamma}^{\rm loc}$. We recall that a state $\omega$ is called a \emph{ground state} for $\delta$ if 
\be\label{TL_gs}
\omega(A^*\delta(A)) \geq 0, \mbox{ for all } A\in\cA_{\Gamma}^{\rm loc}.
\ee
A simple argument shows that  $\omega(\delta(A))=0$ for all $A\in\cA_{\Gamma}^{\rm loc}$ if and only if $\omega(A^*\delta(A)) \in\bR$ for all $A\in\cA_{\Gamma}^{\rm loc}$. Hence, any ground state for $\delta$ is necessarily a $\tau$-invariant state. 

The GNS Hamiltonian for any ground state $\omega$ is non-negative and the cyclic vector $\Omega_\omega$ 
satisfies $H_\omega \Omega_\omega=0$, i.e. $\min\spec(H_\omega)=0$. As a consequence, if there is a $\gamma>0$ for which
\be\label{TL_gapped_gs}
\braket{\pi_\omega(A)\Omega_\omega}{H_\omega \, \pi_\omega(A)\Omega_\omega} = \omega(A^*\delta(A)) \geq \gamma \omega(A^*A) = \gamma\|\pi_\omega(A)\Omega_\omega\|^2
\ee
for all $A\in\cA_\Gamma^{\rm loc}$ such that $\omega(A)=0$, then the ground state of $H_\omega$ is unique and, moreover, 
\[\gap(H_\omega) :=\sup\{\delta\geq 0 : \spec(H_\omega)\cap (0,\delta) = \emptyset\} \geq \gamma.\] 
Thus, we say that $\omega$ is a \emph{unique gapped ground state} if \eqref{TL_gapped_gs} is satisfied 

We now return to the situation of interest: uniform perturbation models with perturbations that converge locally in $F$-norm. In Section~\ref{sec:TL_gap} we study the limiting infinite volume state for each $0\leq s <s_0$, and analyze the spectral gap of these states in Section~\ref{sec:infinite_gap}

\subsection{Stability of LTQO and the existence of a pure infinite volume state}\label{sec:TL_gap}

Recall that the regions $\Lambda_n^p$ for a perturbation model are defined using the indistinguishability radius, see Defintion~\ref{def:pert_model}(iv). Since the initial interactions are frustration-free, the indistinguishability radius implies that the ground state space of each unperturbed Hamiltonian satisfies the following estimate: for each $n \geq 1$,
$x \in \Lambda_n$, $0 \leq k \leq r_x^{\Omega}( \Lambda_n)$, and $A \in \mathcal{A}_{b_x^{\Lambda_n}(k)}$, 
\begin{equation} \label{LTQO_est_TL} 
\| P_{\Lambda_n}(0) A P_{\Lambda_n}(0) - \omega_0^{(n)}(A) P_{\Lambda_n}(0) \| \leq |b_x^{\Lambda_n}(k)| \| A \| \Omega(r_x^{\Omega}(\Lambda_n) -k)
\end{equation} 
where $P_{\Lambda_n}(0)$ is the ground state projection associated to $H_{\Lambda_n}$ and $\omega_0^{(n)}$ is the corresponding
ground state functional, see \eqref{eq:gs} below. As discussed in Section~\ref{sec:LTQO} and demonstrated in Sections~\ref{sec:LTQO2}-\ref{sec:uniform_sequences}, this 
LTQO property is crucial for stability of the gap. When studying the thermodynamic limit, one is often interested in the perturbation regions becoming extensive, i.e. $\Gamma^p = \Gamma$ where
\[
\Gamma^p = \left\{x\in \Gamma \mid \; \exists \; m\geq 1 \; \text{ s.t. } \; x\in \textstyle{\bigcap_{n\geq m}}\Lambda_n^p \, \right\}.
\]
As discussed in Section~\ref{sec:veri_ass}, when $\Gamma$ is infinite the conditions of a uniform perturbation model guarantee that $r_x^{\Omega}(\Lambda_n)\to \infty$ for any $x\in \Gamma^p$. This motivates us to consider uniform perturbation models that are indistinguishable everywhere in the following sense:

\begin{defn}\label{ass:bf} We say a perturbation model with decay function $\Omega$ as in Definition~\ref{def:pert_model}(iv) is \emph{everywhere indistinguishable} if for all $x \in \Gamma$,
	\begin{equation}\label{rinfinity}
	r_x^{\Omega}( \Lambda_n) \to \infty \quad \mbox{as} \quad n \to \infty.
	\end{equation}
\end{defn}
From the definition of indistinguishability radius it follows immediately that for all
$x,y\in\Lambda_n$, $r_x^{\Omega}( \Lambda_n) \geq r_y^{\Omega}( \Lambda_n) -d(x,y)$ and, hence, $|r_x^{\Omega}( \Lambda_n)- r_y^{\Omega}( \Lambda_n)|\leq d(x,y)$. Therefore, \eq{rinfinity} holds for all $x\in\Gamma$ if and only if it holds for some $x\in\Gamma$.

Clearly, any everywhere indistinguishable uniform perturbation model has an LTQO estimate as in (\ref{LTQO_est_TL}) which 
becomes vanishingly small in the thermodynamic limit for any $x\in\Gamma$. One of the main goals of this section is to show that for all $0\leq s< s_0$ the perturbed model has a similarly vanishing LTQO estimate and, moreover, the finite volume states
\be\label{eq:gs}
\omega^{(n)}_s(A) = \frac{\Tr P_{\Lambda_n}(s) A}{\Tr P_{\Lambda_n}(s)} \quad\text{for}\quad  A\in\cA_{\Lambda_n}
\ee
converge to a pure infinite volume state $\omega_s$ on $\cA_\Gamma$. Here, we recall that $P_{\Lambda_n}(s)$ is the spectral projection of $H_{\Lambda_n}(s)$ onto $\Sigma_1^{\Lambda_n}(s)$ as defined in \eqref{spec_sets}. We prove the stability of the LTQO estimate and existence of the limiting infinite volume state in Theorem~\ref{thm:convergence_LTQO}, and show that the state is pure in Corollary~\ref{cor:uniqueTL}.

The finite and infinite volume spectral flow automorphisms play a key role in the proof of Theorem~\ref{thm:convergence_LTQO}. As discussed, e.g., in Section~\ref{sec:uniformity}, for any $0< \gamma < \gamma_0$ there exists a function of decay class $( \eta, \frac{ \gamma}{2 v}, \theta)$, which we denote here by $G_{\alpha}^\gamma$, that can be used in the quasi-locality estimates for the finite volume spectral flows $\alpha_s^{(n)}$ uniformly in $n \geq 1$ and $0 \leq s \leq s_\gamma$. This decay function may also be used in the quasi-locality estimates for the limiting spectral flow automorphisms $\alpha_s$ for the same range of $s$. We use such a function in the statement of Theorem~\ref{thm:convergence_LTQO}.

\begin{thm}\label{thm:convergence_LTQO}
	For an everywhere indistinguishable uniform perturbation model with a sequence $\Phi_n$ that converges locally in $F$-norm in $\cB_{F}$ (see \eq{local_convergence}), the pointwise limit
	\begin{equation}\label{convergence_s}
	\omega_s(A) = \lim_{n \to \infty} \omega^{(n)}_s(A), \quad  A\in\mathcal{A}_{\Gamma}^{\rm loc}
	\end{equation}
	exists and defines a state on $\mathcal{A}_{\Gamma}$ for every $0 \leq s < s_0$. 
	Moreover, for any $0<\gamma<\gamma_0$ such that $s\leq s_\gamma$ and any local observable $A \in \mathcal{A}_{b_x(k)}$ with $x\in\Gamma$ and $k\geq0$, one has that for any $m \geq 0$ and all	$n \geq 1$ sufficiently large
	\be
	\Vert P_{\Lambda_n}(s) A P_{\Lambda_n}(s) - \omega^{(n)}_s(A) P_{\Lambda_n}(s) \Vert \leq
	|b_x(k)| \|A\| \left(\Omega(r_x^{\Omega}(\Lambda_n)-k-m) + 4 G_\alpha^\gamma(m)\right).
	\label{LTQO_s}\ee

\end{thm}

We remark that \eq{LTQO_s} is an LTQO property for $\omega^{(n)}_s$. In fact, our estimates will show
\be
\vert \omega_s^{(n)}(A) - \omega_s(A)\vert \leq  |b_x(k)| \Vert A \Vert\left(2\Omega(r_x^{\Omega}(\Lambda_n)-k-m)+6 G_\alpha^\gamma(m)\right)
\ee
where the quantities are as in \eqref{LTQO_s}. Therefore, one can replace $\omega_s^{(n)}$ with 
$\omega_s$ in \eq{LTQO_s} and a similar bound holds with an appropriate change to the estimates on the right-hand-side. 

\begin{proof}
Fix $s\in[0,s_0)$ and let $\gamma>0$ be such that $s\leq s_\gamma$, which is guaranteed to exist by \eqref{s_gamma_inf}. We begin by considering the finite volume state $\omega^{(n)}_s$. Denote by $\alpha_s^{(n)}$ the spectral flow automorphism associated with $H_{\Lambda_n}(s)$ and $\xi = \gamma$ as in \eqref{spec-flow-auto}.
Since $\alpha^{(n)}_s(P_{\Lambda_n}(s)) = P_{\Lambda_n}(0)$, see (\ref{spec_flow_proj}), we can rewrite the perturbed finite volume states in terms of the initial state via 
\be
\omega^{(n)}_s = \omega^{(n)}_0\circ \alpha_s^{(n)}.
\label{spectralflow_states}\ee
The results of this theorem for values $s>0$ follow from establishing the analogous properties for $\omega_0^{(n)}$ and the uniform quasi-locality of $\alpha_s^{(n)}$. In particular, for sufficiently large $n$, quasi-locality is used to approximate observables of the form $\alpha_s^{(n)}(A)$ by an $n$-independent local operator. We first discuss this in more detail.

Consider an observable $A\in\cA_{b_x^{\Lambda_n}(k)}$ for $n \geq 1$, $x \in \Lambda_n$ and $k \geq 0$. For each $m \geq 0$, denote by $A_m^{(n)}(s)\in \cA_{b_x^{\Lambda_n}(k+m)}$ the strictly local approximation of $\alpha_s^{(n)}(A)$ given by
\begin{equation}
A_m^{(n)}(s) = \Pi_{b_x^{\Lambda_n}(k+m)}^{\Lambda_n}(\alpha_s^{(n)}(A))
\end{equation}
where we use the localizing maps introduced in Section~\ref{sec:QL+LD}, see \eqref{pi_partial_trace} and the subsequent discussion. Recall that the spectral flow automorphisms $\alpha_s^{(n)}$ converge strongly to $\alpha_s$ on $\mathcal{A}_{\Gamma}$.
Using the consistency relation $\Pi_X^\Lambda(A)\otimes \idty_{\Lambda'\setminus\Lambda} = \Pi_{X}^{\Lambda'}(A)$ for any $A\in \cA_{\Lambda}$ and $X\subset \Lambda \subset \Lambda'$,
we find that for $n' \geq n$ sufficiently large
\begin{align}\label{eq:approx_convergence}
\left\|A_m^{(n)}(s)- A_m^{(n')}(s)\right\|
\, = \, \left\|\Pi_{b_x(k+m)}^{\Lambda_{n'}}\left(\alpha_s^{(n)}(A)-\alpha_s^{(n')}(A)\right)\right\|  \, \leq \, \left\|\alpha_s^{(n)}(A)-\alpha_s^{(n')}(A)\right\| \,.
\end{align}
Thus, strong continuity implies that $\{ A_m^{(n)}(s) \}_{n \geq 1}$ is uniformly Cauchy and therefore converges, i.e.
\begin{equation}\label{local_stong_limit}
\lim_{n \to \infty} A_m^{(n)}(s) = A_m(s) 
\end{equation}
for some $A_m(s) \in \mathcal{A}_{b_x(k+m)}$. As a consequence, for each $A\in\cA_{b_x(k)}$ there is an $N$ so that for all $n \geq N$, 
\be
\Vert \alpha_s^{(n)}(A) - A_m(s)\Vert \leq \|\alpha_s^{(n)}(A) - A_m^{(n)}(s)\|+ \|A_m^{(n)}(s)- A_m(s)\| \leq 3|b_x^{\Lambda_n}(k)| \Vert A\Vert G_\alpha^\gamma(m).
\label{alphas_local}
\ee
Here, we have use \eqref{local_stong_limit}, that $s\leq s_\gamma$, and applied Lemma~\ref{lem:qlm_loc_est}. Said differently, given $m\geq 0$, the same local operator $A_{m}(s)\in \cA_{b_x(k+m)}$ can be used to approximate the transformed operator $\alpha_s^{(n)}(A)$ for all $n$ sufficiently large. 

	We now prove \eq{convergence_s} for $s=0$. Fix $x \in \Gamma$, $k \geq 0$, and $A \in \mathcal{A}_{b_x(k)}$. Note $b_x(k) \subset \Lambda_n$ for all
	$n$ sufficiently large. Moreover, Definition~\ref{ass:bf} implies
	$k\leq \min\{r_x^{\Omega}(\Lambda_n), \, r_x^{\Omega}(\Lambda_{n'})\}$ for $n' \geq n$ sufficiently large.
	In this case, using (\ref{LTQO_est_TL}) and the frustration free property of the initial ground state projectors, i.e.
	$P_{\Lambda_n}(0)P_{\Lambda_{n'}}(0) = P_{\Lambda_{n'}}(0)$, we find
	\bea
	\vert \omega_0^{(n)} (A) -  \omega_0^{(n^\prime)} (A) \vert &=& \Vert \left( \omega_0^{(n)} (A)  -  
	\omega_0^{(n')} (A) \right) P_{\Lambda_{n^\prime}}(0) \Vert\nonumber\\
	&\leq & \Vert  \omega^{(n)}_0(A) P_{\Lambda_{n^\prime}}(0) - P_{\Lambda_{n^\prime}}(0) A P_{\Lambda_{n^\prime}}(0) \Vert \nonumber\\
	&&+  \Vert P_{\Lambda_{n^\prime}}(0) A P_{\Lambda_{n^\prime}}(0) - \omega^{(n^\prime)}_0(A) P_{\Lambda_{n^\prime}}(0)\Vert\nonumber \\
	&\leq &2 |b_x(k)|  \Vert A\Vert \Omega(\min\{r_x^{\Omega}(\Lambda_n), \, r_x^{\Omega}(\Lambda_{n'})\}-k).
	\label{estimate0}\eea
	Here we have used that $\Omega$ is non-increasing. Since we assumed $r_x^{\Omega}(\Lambda_n)\to\infty$,
	it follows that $\omega_0^{(n)}(A)$ converges for all  $A\in \cA_\Gamma^{\rm loc}$.
	
	Now, consider $0<s\leq s_\gamma$ for some $0<\gamma<\gamma_0$. For each $n\geq 1$, let $\alpha_s^{(n)}$ be the spectral flow automorphism with $\xi=\gamma$. We use \eq{spectralflow_states} and \eq{alphas_local} to obtain similar estimates for the perturbed models. Given the parameters above, for each choice of $m \geq 0$ the quantity $l =\min\{r_x^{\Omega}(\Lambda_n), \, r_x^{\Omega}(\Lambda_{n'})\}-k-m\geq 0$
	for sufficiently large $n \leq n'$. In this case, for $A\in \cA_{b_x(k)}$ as above, an application of (\ref{alphas_local}) shows
	\bea
	\vert \omega_s^{(n)} (A) -  \omega_s^{(n^\prime)} (A) \vert 
	&=& \vert \omega_0^{(n)} (\alpha^{(n)}_s(A)) - \omega_0^{(n^\prime)}(\alpha^{(n')}_s(A)) \vert\nonumber\\
	&\leq &  \vert \omega_0^{(n)} (A_m(s)) - \omega_0^{(n^\prime)}(A_m(s)) \vert
	+6 |b_x(k)|  \Vert A\Vert G_\alpha^\gamma(m) \, .
	\eea
	Combining this with \eq{estimate0}, we have
	\be
	\vert \omega_s^{(n)} (A) -  \omega_s^{(n^\prime)} (A) \vert \leq |b_x(k)|\Vert A\Vert(2\Omega(l) +  6G_\alpha^\gamma(m)), \,
	\ee 
	from which it is clear that the limit in \eq{convergence_s} exists.
	
	To prove \eq{LTQO_s}, we argue similarly. Recall that $P_{\Lambda_n}(0) = \alpha^{(n)}_s(P_{\Lambda_n}(s))$.
	Using that $ \alpha^{(n)}_s$ is an automorphism and \eq{spectralflow_states} , we find that with $l=r_x^{\Omega}(\Lambda_n)-k-m$
	\bea
	\Vert P_{\Lambda_n}(s) A P_{\Lambda_n}(s) - \omega^{(n)}_s(A) P_{\Lambda_n}(s)\Vert
	&=& \Vert P_{\Lambda_n}(0)  \alpha^{(n)}_s(A) P_{\Lambda_n}(0) - \omega^{(n)}_0(\alpha^{(n)}_s(A)) P_{\Lambda_n}(0)\Vert \nonumber\\
	&\leq& \Vert P_{\Lambda_n}(0)A_m^{(n)}(s) P_{\Lambda_n}(0) - \omega^{(n)}_0(A_m^{(n)}(s)) P_{\Lambda_n}(0)\Vert \nonumber \\
	&&+ \; 2\Vert \alpha^{(n)}_s(A) -A_m^{(n)}(s)\Vert\nonumber\\
	&\leq& |b_x(k)| \Vert A\Vert (\Omega(l) + 4 G_\alpha^\gamma(m) ).
	\label{forLTQO_s}\eea
	For the last inequality, we have again used \eqref{LTQO_est_TL} and applied Lemma~\ref{lem:qlm_loc_est}.
\end{proof}

We now turn to showing that the states $\omega_s$ are pure for each $0\leq s < s_0$. In fact, we use LTQO to show that these states are unique in the sense that any sequence of finite-volume states defined by density matrices $\rho_n$ contained in the range of the spectral projections $P_{\Lambda_n}(s)$ necessarily converge to $\omega_s$.

\begin{cor}\label{cor:uniqueTL}  Consider an everywhere indistinguishable uniform perturbation model with
a sequence $\Phi_n$ that converges locally in $F$-norm in $\cB_{F}$ (see \eq{local_convergence}), and fix $0\leq s < s_0$.
	For any sequence of density matrices $\rho_n = P_{\Lambda_n}(s)\rho_n\in\cA_{\Lambda_n}$ the limit
	\be\label{eq:state_convergence}
	\lim_{n\to\infty} \Tr \rho_n A = \omega_s(A)
	\ee
	holds for all $A\in\cA^{\rm loc}_\Gamma$, and $\omega_s$ is a pure state on $\mathcal{A}_{\Gamma}$ for each $0\leq s < s_0$.
\end{cor}
\begin{proof}
	Note that if $P_{\Lambda_n}(s)\rho_n = \rho_n$ for all $n \geq 1$, then
	$$
	\Tr \rho_n A = \Tr \rho_nP_{\Lambda_n}(s) A P_{\Lambda_n}(s) = \omega_s^{(n)}(A) + 
	\Tr \rho_n [P_{\Lambda_n}(s)A P_{\Lambda_n}(s) - \omega^{(n)}_s(A) P_{\Lambda_n}(s)],
	$$
	and so the first claim follows from \eq{convergence_s} and \eq{LTQO_s}. 
	
	To see that $\omega_s$ is pure, we use the thermodynamic limit of the spectral flow to relate it to $\omega_0$ via
	\be
	\omega_s(A) = \omega_0(\alpha_s(A)),
	\label{alpha_s_TL}\ee
	and prove that $\omega_0$ is pure.
	
	Assume $\eta$ is a state that is majorized by $\omega_0$, i.e. $\eta(A^*A) \leq c \omega_0(A^* A)$ for some $c\geq 1$. Since $\omega_0$ is the ground state of a frustration-free system, it follows that $\omega_0(H_{\Lambda_n}) = 0$ for all $n$. Restricting $\eta$ to $\cA_{\Lambda_n}$ produces a state implemented by a density matrix $\eta_n$. By the majorizing assumption and the frustration-free property, this matrix satisfies $\eta_n = P_{\Lambda_n}(0)\eta_n$. 
	Therefore, applying \eq{eq:state_convergence} with $s=0$, one finds the states defined by $\eta_n$ necessarily converge to $\omega_0$. Hence $\eta=\omega_0$, and $\omega_0$ is a pure state.
	Since $\alpha_s$ is an automorphism, \eq{alpha_s_TL} implies that $\omega_s$ is also pure.
\end{proof}

\subsection{Spectral gap stability of the GNS Hamiltonian}\label{sec:infinite_gap}

We will now provide conditions under which the state $\omega_s$, whose existence is guaranteed by Theorem~\ref{thm:convergence_LTQO},
is a gapped ground state with respect to the dynamics $\delta_s$ from \eqref{delta_s}. Since we will apply similar arguments to systems with 
discrete symmetries in the next section, we first prove a more general result. 

\begin{thm}\label{thm:gap_ineq_TL}
	Let $\Lambda_n\uparrow\Gamma$ and assume that $H_n=H_n^*\in\cA_{\Lambda_n}$ is a sequence of Hamiltonians for which there is a derivation $\delta$ on $\cA_\Gamma$ with
	\begin{equation}
	\delta(A) = \lim_{n \to \infty} [H_n, A]
	\end{equation}
	for all $A \in \mathcal{A}_{\Gamma}^{\rm loc}$.	Set $E_n = {\rm min \, spec}(H_n)$ and suppose there are sequences of non-negative numbers $\{ \gamma_n \}_{n \geq 1}$ and $\{ \epsilon_n \}_{n \geq 1}$ so that:
	\begin{enumerate}
		\item[(i)] $\epsilon_n \to 0$ as $n \to \infty$,
		\item[(ii)]  $\limsup_{n} \gamma_n >0$,
		\item[(iii)]  The spectral projection $P_n$ of $H_n - E_n\idty$ onto $[0,\epsilon_n]$ satisfies
		\begin{equation} \label{general_gap_n}
		(\idty - P_n)(H_n-E_n \idty) \geq \gamma_n(\idty - P_n) \, .
		\end{equation}
	\end{enumerate}
	Then, for any state $\omega$ on $\cA_\Gamma$, if there exists a sequence $Q_n\in\cA_{\Lambda_n}$ of nonzero orthogonal projections $Q_n \leq P_n$ such that
	\be
	\lim_{n\to\infty} \Vert P_n A Q_n - \omega(A) Q_n\Vert = 0 \quad \mbox{for all } A \in  \cA_\Gamma^{\rm loc},
	\label{LTQO_PQ}\ee
	then $\omega$ is a unique gapped ground state for $\delta$. In particular, for any $A\in \cA^{\rm loc}_\Gamma$ with $\omega(A)=0$,
	\be
	\omega(A^* \delta(A))\geq \left(\limsup_n \gamma_n \right)\omega(A^*A).
	\label{gap_ineq_TL}\ee
\end{thm}

\begin{proof}
	Without loss of generality, we will assume that $E_n = 0$ for all $n \geq 1$. 
	We begin with two observations. First, we claim that for all $A \in \mathcal{A}_{\Gamma}^{\rm loc}$, 
	\be\label{eq:state_convergence2}
	\lim_{n\to\infty}\Tr\rho_nA = \omega(A)
	\ee
	for any sequence $\{ \rho_n \}_{n \geq 1}$ of density matrices with $Q_n \rho_n = \rho_n \in \mathcal{A}_{\Lambda_n}$ for all $n \geq 1$. 
	Here we argue as in Corollary~\ref{cor:uniqueTL}. In fact, since $Q_n \leq P_n$, each of these density matrices satisfies  $\rho_n = \rho_nQ_n = Q_n\rho_nP_n$.
	In this case, for any $A \in \mathcal{A}_{\Gamma}^{\rm loc}$  and $n \geq 1$ sufficiently large, 
	\[
	\Tr \rho_n A - \omega (A) = \Tr \rho_n [P_n A Q_n  - \omega(A) Q_n]  \, .
	\]
	By \eq{LTQO_PQ}, the above tends to zero, and thus we have (\ref{eq:state_convergence2}). 
	
	Next, we prove that $\omega$ is invariant under the dynamics $e^{it \delta}$ by showing that $\omega(\delta(A)) = 0$ for all $A \in \mathcal{A}_{\Gamma}^{\rm loc}$.
	To see this, note that 
	\be\label{eq:wd_limit}
	\omega(\delta(A)) = \lim_{n\to \infty}\Tr\rho_n[H_n,A] = 0\,  \quad \mbox{for all } A \in \mathcal{A}_{\Gamma}^{\rm loc} \, .
	\ee
	Here the second equality above follows as $\rho_n=\rho_nP_n$, $\epsilon_n \to 0$, and 
	\be\label{eq:trace_bound}
	|\Tr\rho_n[H_n,A]| = |\Tr\rho_n[H_nP_n,A]| \leq 2\epsilon_n\|A\|
	\ee
	where for the inequality we used that $P_n$ is the spectral projection of $H_n$ onto $[0,\epsilon_n]$.
	We approximate to see that the first equality in (\ref{eq:wd_limit}) is true. For any $m \leq n$ and each $A \in \mathcal{A}_{\Gamma}^{\rm loc}$, 
	\begin{align}
	|\omega(\delta(A))-\Tr\rho_n[H_n,A]| \; \leq \;&  |\omega(\delta(A))-\omega([H_m,A])| + |\omega([H_m,A])-\Tr\rho_n[H_m,A]|\label{eq:wd_convergence} \\
	&+ |\Tr\rho_n([H_m,A]-[H_n,A])|.  \nonumber \, 
	\end{align}
	The existence of $\delta(A)$ guarantees that for $n\geq m \geq 1$ sufficiently large enough, 
	both the first and last term above can be made arbitrarily small. For any such $m$, the second term above can be made small, 
	using (\ref{eq:state_convergence2}), and a possibly larger choice of $n \geq 1$. This completes the
	proof of (\ref{eq:wd_limit}). 
	
	We now show that $\omega$ is a ground state for $\delta$. Arguing as above, we find that for any 
	$A\in\cA_{\Gamma}^{\rm loc}$,
	\begin{equation}
	\omega(A^*\delta(A))  = \lim_{n\to\infty} \Tr\rho_nA^*[H_n,A] \label{eq:gs_convergence}
	\end{equation}
	and in addition, the estimate
	\begin{equation}
	|\Tr\rho_nA^*AH_n| \leq \epsilon_n\|A\|^2  \label{eq:remainder_bound}
	\end{equation}
	holds. Since $H_n\geq 0$, we also have that 
	\begin{equation} \label{non-neg-term}
	\Tr\rho_nA^*[H_n,A]+\Tr\rho_nA^*AH_n = \Tr\rho_nA^*H_nA\geq 0.
	\end{equation}
	The fact that $\omega$ is invariant under the dynamics implies $\omega(A^*\delta(A))\in\bR$ for 
	all local $A$, and so for any $n \geq 1$,
	\beann
	\omega(A^*\delta(A)) &\geq& \omega(A^*\delta(A))- \Tr\rho_nA^*H_nA \\
	&\geq& -|\omega(A^*\delta(A))-\Tr\rho_nA^*[H_n,A]|-|\Tr\rho_nA^*AH_n|.
	\eeann
	where we used (\ref{non-neg-term}) for the final estimate above.
	{F}rom \eqref{eq:gs_convergence} and \eqref{eq:remainder_bound}, we conclude that $\omega$ satisfies (\ref{TL_gs}) and hence
	is a ground state for $\delta$.
	
	To argue that $\omega$ is a gapped ground state for $\delta$, we establish (\ref{gap_ineq_TL}).
	Note that since $\omega$ is a ground state of $\delta$,
	we conclude from (\ref{eq:gs_convergence}) that 
	\begin{equation} 
	\omega(A^*\delta(A))  = \lim_{n\to\infty} | \Tr\rho_nA^*[H_n,A] |
	\end{equation}
	for all $A\in \cA^{\rm loc}_\Gamma$. Now, since $(\idty-P_n)\rho_n=0$, we may re-write 
	\be \label{re-write-in-proof}
	\Tr\rho_nA^*[H_n,A] = \Tr\rho_nA^*H_n(\idty-P_n)A + \Tr\rho_nA^*[H_nP_n,A].
	\ee
	The first term above is non-negative. In particular, an application of \eqref{general_gap_n} shows that
	\[
	\Tr\rho_nA^*H_n(\idty-P_n)A \geq \gamma_n\Tr\rho_nA^*(\idty-P_n)A\geq 0.
	\]
	For the second term in (\ref{re-write-in-proof}), we find that
	\[
	|\Tr\rho_nA^*[H_nP_n,A]| \leq \|A^*[H_nP_n,A]\| \leq 2\epsilon_n \|A\|^2.
	\]
	As a result, we have the following the lower bound
	\be\label{eq:abs_LB}
	|\Tr\rho_nA^* [H_n, A]| \geq \gamma_n\Tr  \rho_n A^*(\idty - P_n)A- 2\epsilon_n\Vert A\Vert^2.
	\ee
	
	Now, let $A \in \mathcal{A}_{\Gamma}^{\rm loc}$ and suppose that $\omega(A)=0$. Observe that 
	for such an observable, \eq{LTQO_PQ} implies that
	\be
	\Tr  \rho_n A^* P_n A = \Tr  \rho_n A^* [P_n A Q_n - \omega(A)  Q_ n ] \to 0.
	\ee
	Since $\epsilon_n\to0$ and $\Tr(\rho_nA^* A) \to \omega(A^*A)$, we conclude from \eqref{eq:abs_LB} that
	\be
	\limsup_n |  \Tr(\rho_nA^* [H_n , A])| \geq  \limsup_n \gamma_n \Tr(\rho_nA^* A)
	=  \left(\limsup_n \gamma_n \right) \omega(A^*A),
	\ee
	and this completes the proof.
\end{proof}

The previous theorem implies that the uniform lower bound obtained for uniform sequences of finite systems in Section \ref{sec:uniform_sequences}
carries over to the gap for the GNS Hamiltonian $H_{\omega_s}$ of the corresponding thermodynamic limit. Since we are interested in infinite volume ground states, we require that the splitting of the lower part of the spectrum $\Sigma_1^{\Lambda_n}(s)$ tends to a single point in the sense that $\diam(\Sigma_1^{\Lambda_n}(s))\to 0$  as $n\to \infty$. This is the case if \eqref{lim=0} holds since $\diam(\Sigma_1^{\Lambda_n}(s)) \leq 2s(\delta_n+\epsilon_n)$. We finish this section with a precise statement of this fact for the perturbation models we have been considering.

\begin{cor}\label{cor:GNSgap}
	Assume that (\ref{lim=0}) holds for a everywhere indistinguishable uniform perturbation model for which the perturbations converge locally in $F$-norm. Then, for any $0< \gamma < \gamma_0$ and each $0 \leq s \leq s_{\gamma}$, the GNS Hamiltonian associated with the pure state $\omega_s$ from Theorem~\ref{thm:convergence_LTQO}	has a simple ground state eigenvalue 0 with a spectral gap above it bounded below 
	by $\gamma$.
\end{cor}
\begin{proof}
	Consider such a perturbation model, and fix $0\leq s \leq s_\gamma$ for some positive $\gamma<\gamma_0$. Since this model satisfies Assumption~\ref{ass:uni_pert_model} and	Definition~\ref{ass:bf}, the results of Theorem~\ref{thm:convergence_LTQO} and Corollary~\ref{cor:uniqueTL} hold. In particular, the state $\omega_s$ on $\mathcal{A}_{\Gamma}$ from Theorem~\ref{thm:convergence_LTQO} is pure. Since this model satisfies (\ref{lim=0}), Corollary~\ref{cor:zerod+e} also holds, and so there are non-negative sequences $\{ \epsilon_n'(s) \}_{n \geq 1}$ and $\{ \gamma_n'(s) \}_{n \geq 1}$
	for which 
	\begin{enumerate}
		\item[(i)] $\epsilon_n'(s) \to 0$ as $n \to \infty$,
		\item[(ii)] $\limsup_{n} \gamma_n'(s)\geq\gamma$,
		\item[(iii)] For all $n \geq 1$, 
		\[
		\spec (H_{\Lambda_n}(s)-E_n(s)\idty)\subset [0,\epsilon_n'(s)]\cup [\epsilon_n'(s)+\gamma_n'(s),\infty)  \, .
		\]
	\end{enumerate} 
	This shows that the conditions of Theorem~\ref{thm:gap_ineq_TL} hold where we take $Q_n =P_n = P_{\Lambda_n}(s)$ and observe that (\ref{LTQO_PQ}) holds by \eqref{LTQO_s}, (\ref{convergence_s}), and Assumption~\ref{ass:bf}.	Our claims about the gap for the corresponding GNS Hamiltonian now follow from (\ref{gap_ineq_TL}) and the comments following (\ref{TL_gapped_gs}).
\end{proof}


%
%

\section{Symmetry Restricted Stability and the Thermodynamic limit of the ground states with discrete symmetry breaking}\label{sec:TL_G}

\subsection{Discrete symmetries} \label{sec:discrete_sym}

In many interesting systems, the interactions have symmetries. When considering the thermodynamic limit we need to 
allow for the possibility that symmetries of the model are spontaneously broken. In the case of a continuous symmetry, such as the spin 
rotations about an axis, the Goldstone theorem \cite{landau:1981} implies that, under quite general conditions, there is no gap in the spectrum 
above the ground state in the thermodynamic limit. Therefore, in our context of gapped ground states, only discrete symmetries need to be considered. 

An important consequence of the results in this section is the stability of the gapped portion of the ground state phase diagram 
of a variety of quantum lattice models, which includes many special cases studied previously in the literature
\cite{ginibre:1969, albanese:1989, kennedy:1992, borgs:1996, datta:1996, yarotsky:2006, szehr:2015}.

We now proceed to setting up the class of models with discrete symmetry breaking for which we prove stability of the symmetry breaking and the ground state gap. We find a compromise between generality and an effort to state the assumptions succinctly and transparently. Instead of attempting to describe the most general situation, we will focus on three types of discrete symmetry breaking (described below) that cover a large number of models considered in the literature. We start with unperturbed models defined on an increasing and absorbing sequence of finite volumes $\Lambda_n$, $n\geq 1$, that have a symmetry described by a finite set of automorphisms $\sigma_g$ labeled by $g\in G$. These automorphism act on $\cA_{\Lambda_n}$ and they are $n$-dependent in that sense. 

Local topological quantum order expresses the indistinguishability of the ground states by local observables, which is made precise by our notion of the indistinguishability radius. If a spontaneous symmetry breaking occurs that can be detected by a local order parameter, then clearly one cannot expect LTQO to hold for all local observables. However, if the perturbation respects the symmetry, then stability can be again verified using a modified notion of LTQO. We introduce two indistinguishability radii that take into account the model symmetry: the $G$-symmetric indistinguishability radius (see Definition~\ref{LTQO_def_G}) and the $G$-broken indistinguishability radius (see Assumption \ref{assumption:LTQO_N}). We show that the uniform finite-volume stability results from Section~\ref{sec:uniform_sequences} hold when perturbing at sites with a sufficiently large $G$-symmetric indistinguishability radius. In this case, though, it is not clear if the uniform gap stability extends to the infinite system. To this end, we show that a sufficiently large lower bound on the $G$-broken indistinguishability radius guarantees a non-vanishing spectral gap for the GNS Hamiltonian.

Let us now describe three types of symmetry breaking to which our arguments apply. In short, they are (i) a finite group of local gauge symmetries, (ii) partial breaking of translation invariance to an infinite subgroup (periodic states), and (iii) finite lattice symmetries in translation invariant systems
(reflections and rotations). In each case stability for the uniform sequence of finite systems follows from a ground state indistinguishability condition for a subalgebra of the local observables generated by the symmetry, which we denote by $\cAG$. The superscript $G$ refers to the symmetry as it is represented in the system and not just the abstract symmetry group. In each case $G$ labels a finite set of automorphisms that commute with the infinite system's initial dynamics as well as the perturbed dynamics.
When $\cA_\Gamma$ carries a representation of $\Ir^d$ by translations, we denote these automorphisms by $\rho_a, a\in \Ir^d$. 

For each type of symmetry breaking, the automorphisms and algebra $\cAG$ are as follows:
\begin{enumerate}
	\item[(S1)]\emph{Local Gauge Symmetry:} $G$ is a finite group and for each $x\in \Gamma$ there is a representation of $G$ by automorphisms $\sigma_g^x, g\in G$, on
	$\cA_{\{x\}}$ for which $\sigma_g=\bigotimes_{x\in\Gamma} \sigma_g^x$ denotes the corresponding automorphism on $\cA_\Gamma$. In this case, the gauge symmetry is broken in the ground states and $\cAG$ is the $G$-invariant elements of $\cA_\Gamma^{\rm loc}$:
	\be\label{AGi}
	\cAG = \{ A\in \cA_\Gamma^{\rm loc} \mid \sigma_g(A) =A, \; \forall \; g\in G\}.
	\ee
	\item[(S2)]\emph{Translation-Invariant:} The infinite system has a $d$-dimensional translation invariance represented by automorphisms $\rho_a, a\in \Ir^d$, and this symmetry is broken in 
	the set of ground states to the subgroup $(N_1\Ir)\times \cdots \times (N_d\Ir)$ for integers $N_1,\ldots,N_d >1$. In this situation we take
	$G=\Ir_{N_1} \times \cdots \times\Ir_{N_d}$ where we identify $\Ir_{N_i}$ as a set with $\{0,\ldots,N_i-1\}\subset \Ir$. We consider the subset of $\cA_\Gamma^{\rm loc}$ consisting of observables that reflect this symmetry (but are not invariant):
	\be\label{AGii}
	\cAG = \left\{\sum_{a\in G} \rho_a(A)\mid A\in \cA^{\rm loc}_\Gamma \right\}.
	\ee 
	\item[(S3)]\emph{Finite Group of Lattice Symmetries:} $G$ is a finite group of symmetries of $\Gamma$ acting on $\cA_\Gamma$ as automorphisms 
	$\sigma_r$, $r\in G$, and we assume that the system also has a translation symmetry, acting by automorphisms $\{\rho_a \mid a\in \Ir^d\}$,
	that remains unbroken in the ground states of the initial dynamics. In this situation, define
	\be\label{AGiii}
	\cAG = \{ A\in \cA_\Gamma^{\rm loc} \mid \; \forall \; r\in G, \; \exists \; a_r \in \Ir^d \mbox{ s.t. } \sigma_r (A) = \rho_{a_r} (A)\} .
	\ee
\end{enumerate}

Key for the analysis below will be how the automorphisms in each of these cases behave under composition with the localizing operators $\Pi_X^\Lambda$ from Section~\ref{sec:Step1}, see \eqref{pi_partial_trace}. For (S1), the action of $\sigma_g$ on $\cA_{\Lambda}$ is given by conjugating with $U_\Lambda(g) = \bigotimes_{x\in\Lambda} U_x(g)$ where $U_x(g)$ is a unitary that implements $\sigma_g^x$. As a consequence, each $\sigma_g$ commutes with the partial trace $\tr_{\Lambda\setminus X}:\cA_\Lambda \to \cA_X$, and hence also with the localizing operators:
\be\label{g_loc_commute}
\sigma_g\circ \Pi_X^\Lambda = \Pi_X^\Lambda\circ \sigma_g.
\ee
For both (S2) and (S3), we will assume periodic boundary conditions on $\Lambda$ and therefore each of the automorphisms $\rho_a$ and $\sigma_r$ has a well-defined restriction onto $\cA_\Lambda$. In this case, one does not have commutativity but rather a covariant relation:
\be\label{covariance}
\rho_a\circ \Pi_X^\Lambda = \Pi_{X+a}^\Lambda\circ \rho_a, \qquad \sigma_r\circ \Pi_X^\Lambda = \Pi_{r(X)}^\Lambda\circ \sigma_r
\ee
for all $a\in\bZ^d$ and $r\in G$.

The differences between \eqref{g_loc_commute}-\eqref{covariance} as well as the various choices for $\cAG$ cause a small change in the arguments for stability below. We provide the full argument for the case (S1) in Sections~\ref{subsec:Gsymmetry_radius}-\ref{subsec:Gbroken_radius} and discuss the necessary alterations for cases (S2) and (S3) in Section~\ref{sec:casesS2S3}.

\subsection{Symmetry restricted indistinguishability and stability of the spectral gap}\label{subsec:Gsymmetry_radius}
For the case (S1), we consider the same set-up as in Section \ref{sec:TL_gap} with a few modifications due to the gauge symmetry $G$. Once again there is a sequence $(\Lambda_n,d_n)$ of increasing and absorbing finite subsets of $(\Gamma, d)$ for which the unperturbed Hamiltonians $H_{\Lambda_n}$ are frustration-free, uniformly finite-range (with range $R$), and uniformly bounded. Moreover, we assume that the interaction is gauge symmetric, and so
$$
H_{\Lambda_n} = \sum_{X\subseteq\Lambda_n}\eta_n(X)
$$
where $\eta_n(X) \in \cAG$. As before, we assume a non-vanishing spectral gap:
\[
\gamma_0 = \inf_{n\geq 1}\gap(H_{\Lambda_n})>0.
\]

The perturbations are given by interactions $\Phi_n$ which take values in $\cAG$ and have a finite $F$-norm as in Definition~\ref{def:pert_model}(v). To ensure that the conditions of Section~\ref{sec:infinite_system} are satisfied, and in particular \eqref{local_convergence}, we assume both $\eta_n$ and $\Phi_n$ eventually become constant for any finite $X\subset \Gamma$, i.e. there are interactions $\eta$ and $\Phi$ so that $\eta_n(X) = \eta(X)$ and $\Phi_n(X)=\Phi(X)$ for $n$ sufficiently large. In particular, this implies that the perturbations converge locally in $F$-norm, and so Section~\ref{sec:automorphic-equivalence} is relevant. The perturbed Hamiltonians are then given by
\be\label{Gpert_ham_TL}
H_{\Lambda_n}(s) = H_{\Lambda_n} + s\sum_{\substack{X\subseteq\Lambda_n \\ X\cap \Lambda_n^{p}\neq \emptyset}} \Phi_n(X)
\ee
where the perturbation regions $\Lambda_n^p$ are chosen similarly to \eqref{def:L_n^p} for an indistinguishability radius that reflects the symmetry of the model. We will consider two possible candidates for this radius. First, we consider
\be\label{eq:G-pert_region}
\Lambda_n^p = \{ x \in \Lambda_n : r_y^{\Omega,G}( \Lambda_n) \geq K_n + L_n \mbox{ for all } y \in b_x^{\Lambda_n}(K_n) \}
\ee
where $K_n$, $L_n$ are chosen appropriately and $r_y^{\Omega,G}(\Lambda_n)$ is the $G$-symmetric indistinguishability radius:

\begin{defn}[$G$-symmetric indistinguishability radius]\label{LTQO_def_G}
	Let $\Omega:\bR \to [0,\infty)$ be a non-increasing function. The \emph{G-symmetric indistinguishability radius} of $H_\Lambda$ at $x\in\Lambda$, is the largest integer $r_x^{\Omega,G}(\Lambda)\leq \diam(\Lambda)$ such that for all integers $0 \leq k \leq n \leq r_x^{\Omega,G}(\Lambda)$ and all observables $A\in \cA_{b_x^\Lambda(k)}\cap\cAG$,
	\begin{equation}\label{G_LTQO_length}
	\|P_{b_x^\Lambda(n)} A P_{b_x^\Lambda(n)} - \omega_\Lambda(A) P_{b_x^\Lambda(n)}\| \leq |b_x^\Lambda(k)| \|A\| \Omega(n-k)
	\end{equation}
	where $\omega_\Lambda(A) = \Tr(AP_\Lambda)/\Tr(P_\Lambda).$
\end{defn}	

With perturbation regions defined using the $G$-symmetric indistinguishability radius, the main difficulty in adapting the framework from the previous sections is showing the results from Section~\ref{sec:LTQO2} still hold as we no longer assume indistinguishability for all observables. The key observation is that in the proofs of Theorem~\ref{thm:step2-(i)+(ii)} and Theorem~\ref{thm:rel_bound_cond}, the indistinguishability condition is only applied to the anchored observables $\Phi^{(1)}(x,m,s)$ constructed in Section~\ref{sec:Step1}, see \eqref{phi_1_x_bits}. Hence, these results are also valid for a Hamiltonian $H_{\Lambda_n}(s)$ as in \eqref{Gpert_ham_TL}-\eqref{eq:G-pert_region} as long as the corresponding operators $\Phi_n^{(1)}(x,m,s)$ belong to $\cAG$.

We first note that since $\cAG$ is an algebra, the anchoring procedure provided in Section~\ref{sec:def-ball-int} also produces terms that again belong to the algebra, and so one can assume an anchored form for $H_{\Lambda_n}(s)$ comprised of terms belonging to $\cAG$. With the usual decay assumptions on the perturbation, it is clear that the results of Section~\ref{sec:Step1} still apply to $H_{\Lambda_n}(s)$. Considering the definitions of the spectral flow, see \eqref{gen_spec_flow}-\eqref{spec-flow-auto}, and the integral operator $\caF_s^{(n)}$, see \eqref{def_wio_F}, the symmetry assumptions on $H_{\Lambda_n}(s)$ guarantee that both of these quasi-local maps commute with the automorphisms $\sigma_g$, $g\in G$. Since $\Phi_n^{(1)}(x,m,s)$ is defined in terms of compositions of these quasi-local maps and the localizing maps acting on the interaction terms (see \eqref{phi_1_x_bits}) for the gauge symmetry case (S1) it follows from \eqref{g_loc_commute} that $\Phi_n^{(1)}(x,m,s)\in\cAG$ as desired. For the cases (S2) and (S3), \eqref{g_loc_commute} does not hold and the argument needs to be modified. This is main difference between the different symmetry cases, and is the topic of Section~\ref{sec:casesS2S3}. 

If the sequence $H_{\Lambda_n}(s)\in \cAG$ constructed as in \eqref{Gpert_ham_TL}-\eqref{eq:G-pert_region} is a uniform perturbation model, i.e. satisfies Assumption~\ref{ass:uni_pert_model}, then it is clear that the spectral gaps are stable in the sense of Theorem~\ref{thm:uni_stab}. In particular, for each $0<\gamma<\gamma_0$ there is an $s_\gamma>0$ so that the finite volume gaps are uniformly bounded as follows:
\[
\inf_n\gap(H_{\Lambda_n}(s))\geq \gamma, \quad 0\leq s \leq s_\gamma.
\]
If additionally $\delta_n,\epsilon_n\to 0$, see \eqref{lim=0}, then Corollary~\ref{cor:zerod+e} holds and one can consider the stability of the ground state gap in the thermodynamic limit as in Section~\ref{sec:automorphic-equivalence}.

For the thermodynamic limit, we additionally assume that the model is everywhere $G$-indistinguishable in the sense that
\be\label{eq:Gindisting_cond} 
r_x^{\Omega,G}(\Lambda_n)\to \infty  \quad \text{as} \quad n\to\infty,
\ee
for all $x\in\Gamma$. 
Under the assumption of \eqref{eq:Gindisting_cond} many of the results from Section \ref{sec:automorphic-equivalence} are obtained for observables $A\in\cAG$ with virtually no modification. Using the $G$-symmetric indistinguishability radius, the two statements of Theorem~\ref{thm:convergence_LTQO} still hold for $A\in\cAG$ as well as the convergence from Corollary~\ref{cor:uniqueTL}, i.e.
\be
\omega_s(A) = \lim_{n\to\infty} \omega_s^{(n)}(A)=\lim_{n\to\infty}\frac{\Tr P_{\Lambda_n}(s) A}{\Tr P_{\Lambda_n}(s)}
=\lim_{n\to\infty}\Tr\rho_n A, \quad
\ee
for $A\in\cAG$ and all density matrices $\rho_n= \rho_nP_{\Lambda_n}(s)\in\cA_{\Lambda_n}$. The latter can be extended to a unique $G$-symmetric state on $\cA_\Gamma$ by
\be\label{sym_op}
\omega_s(A) = \omega_s(A_G), \qquad A_G := \frac{1}{|G|}\sum_{g\in G}\sigma_g(A).
\ee

One can consider the spectral gap for the GNS Hamiltonian associated with the state $\omega_s$. However, not all results from Section~\ref{sec:automorphic-equivalence} hold when using $G$-symmetric indistinguishability radius. Spontaneous breaking of the $G$-symmetry means that $\omega_s$ is {\em not} a pure state. While the main inequality \eq{gap_ineq_TL}
from Theorem~\ref{thm:gap_ineq_TL} holds for $A\in \cAG$, since this algebra is not dense in $\cA_\Gamma$ it is not immediately clear what this implies for the spectral gap of the GNS Hamiltonian associated to $\omega_s$. To address this, we impose more detailed assumptions suitable to cover the symmetry broken situation. 


\subsection{Symmetry breaking and its stability}\label{subsec:Gbroken_radius}

The goal of this section is to prove that in the case of a $G$-broken LTQO condition, see Assumption~\ref{assumption:LTQO_N} below, the simplex of infinite volume ground states is preserved for sufficiently small $s$ and, moreover, the GNS Hamiltonian associated with each pure ground state has a nonzero spectral gap. We will assume a sequence of finite volume Hamiltonians of the form \eqref{Gpert_ham_TL} with respect to perturbation regions defined using the \emph{G-broken indistinguishability radius} rather than the $G$-symmetric indistinguishability radius (see \eqref{approx_ortho} below). As we will show, the latter radius is necessarily bounded from below by the former, and so the discussion from the previous section still applies when using the $G$-broken indistinguishability radius to define the perturbation regions.

Let $\cS_s$ denote the set of all states on $\cA_\Gamma$ that can be obtained as weak limits of states on $\cA_{\Lambda_n}$ given by density matrices $\rho_n$ satisfying $P_{\Lambda_n}(s)\rho_n = \rho_n$. Recall that the perturbations converge locally in $F$-norm. In the situation that Theorem~\ref{thm:uni_stab} holds, the family of infinite volume spectral flows $\alpha_s:\cA_\Gamma \to \cA_\Gamma$ are a strongly continuous co-cycle of automorphisms, and moreover for $0\leq s\leq s_\gamma$:
\be
\cS_s = \{ \omega \circ \alpha_s \mid \omega \in \cS_0\},
\label{auto_equiv}\ee
see, e.g., \cite[Theorem 7.4]{nachtergaele:2019}. We assume that $\cS_0$ is a simplex with the pure, (gauge) symmetry-broken ground states as its extreme points.
If we denote the set of pure states of $\cS_s$ by $\cE_s$, then the relation \eq{auto_equiv} implies that $\cS_s$ is also a simplex of the same dimension 
as
\be
\cE_s = \{ \omega \circ \alpha_s \mid \omega \in \cE_0\}.
\label{auto_equiv_extreme}\ee
Hence, the structure of the symmetry-broken ground states is preserved. 

It is left to consider the spectral gap of the GNS Hamiltonians associated with $\omega\in \cE_0$. Assume that $\cE_0=\{\omega^1,\ldots,\omega^N\}$ for mutually disjoint $\omega^i$, meaning that their GNS representations are inequivalent. To prove a lower bound on the spectral gap of the GNS Hamiltonian associated with each $\omega_s^i := \omega^i \circ \alpha_s$, we assume the following $G$-broken local topological order condition on the unperturbed Hamiltonians. 

\begin{assumption} \label{assumption:LTQO_N}
	We say that a model with local Hamiltonians $H_{\Lambda_n}$ satisfies
	\emph{local topological quantum order with $N$ $G$-broken phases} (with decay function $\Omega$) if $G$ through composition with $\sigma_g$ acts transitively on $\cE_0$, and there are $N$ non-zero orthogonal projections $P^1_{b_x^{\Lambda_n}(m)},\ldots, P^N_{b_x^{\Lambda_n}(m)}$ onto subspaces of 
	$\ker H_{b_x^{\Lambda_n}(m)}$ such that the following properties hold:
	\begin{enumerate}
		\item[(i)] There is a constant $C$ such that for all $m\geq R$
		\be
		\left\Vert P_{b_x^{\Lambda_n}(m)} - \sum_{i=1}^N P^i_{b_x^{\Lambda_n}(m)}\right\Vert\leq C \Omega(m);
		\label{Pi1}\ee
		\item[(ii)] 
		There is a one-to-one correspondence between the projections $P^i_{\Lambda_n}$ and the pure states $\omega^i$ via:
		\be
		\omega^i(A)=\lim_{n\to\infty} \frac{\Tr P^i_{\Lambda_n} A}{\Tr P^i_{\Lambda_n}},
		\label{Pi2}\ee
		\item[(iii)]
		The \emph{$G$-broken indistinguishability radius} diverges for each $x\in \Gamma$. That is, $r_x^{\Omega,\cE_0}(\Lambda_n)\to \infty$ where $r_x^{\Omega,\cE_0}(\Lambda_n)\leq \diam_n(\Lambda_n)$ is the largest integer so that for all $0\leq k \leq m \leq r_x^{\Omega,\cE_0}(\Lambda_n)$ and for all local observables $A\in\caA_{b_x^{\Lambda_n}(k)}$
		\begin{equation}
		\left\Vert P^i_{b_x^{\Lambda_n}(m)} A P^j_{b_x^{\Lambda_n}(m)} -\delta_{ij} \omega^i(A) P^i_{b_x^{\Lambda_n}(m)} \right\Vert
		\leq |b_x^{\Lambda_n}(k)| \Vert A\Vert  \Omega(m-k).
		\label{approx_ortho}\end{equation}
	\end{enumerate}
\end{assumption}

Before stating the main result, we show that transitivity of the group action implies that the $G$-broken indistinguishability radius is a lower bound on the $G$-symmetric indistinguishability radius. From the action of $G$ on $\cE_0$, it is clear that the state
\be\label{sym_def}
\overline{\omega} _0(A) = \frac{1}{N} \sum_{i=1}^N \omega^i(A), \quad A\in\cA_{\Gamma}
\ee
is $G$-symmetric. In general, $\omega(A) = \omega(A_G)$ for any symmetric state $\omega\in\cS_0$ where $A_G$ is as in \eqref{sym_op}. Since such a state is a convex combinations of $\omega^i\in\cE_0$, transitivity then guarantees that $\overline{\omega}_0$ is the unique $G$-symmetric state as for all $i=1,\ldots,N$,
\be\label{sym_equal}
\omega^{i}(A_G) = \frac{1}{|G|}\sum_{g\in G}\omega^1(\sigma_{g_ig}(A)) = \omega^{1}(A_G) 
\ee
where $g_i\in G$ is such that $\omega^i = \omega^1\circ \sigma_{g_i}$.  More generally, transitivity implies that for all $i$
\begin{equation}\label{eq:sym_equality}
\overline{\omega}_0(A) = \omega^i(A), \quad A\in\cAG.
\end{equation}
Here, we use \eqref{sym_def}, \eqref{sym_equal} and that $\overline{\omega}_0$ is $G$-symmetric. This argument could be simplified using the invariance of $A\in\cAG$ under $\sigma_g$, $g\in G$. However, the more general justification above also holds for the cases (S2) and (S3) considered in Section~\ref{sec:casesS2S3}.

Now, given \eqref{eq:sym_equality}, Assumption~\ref{assumption:LTQO_N} implies \eqref{G_LTQO_length}. Specifically, there is $C'>0$ so that for any $0 \leq k \leq m \leq r_x^{\Omega,\cE_0}(\Lambda_n)$ and $A\in \caA_{b_x^{\Lambda_n}(k)}\cap \cAG$
\begin{equation}\label{G-indistinguish}
\|P_{b_{x}^{\Lambda_n}(m)}AP_{b_{x}^{\Lambda_n}(m)}-\overline{\omega}_0(A)P_{b_{x}^{\Lambda_n}(m)}\| \leq C'|b_x^{\Lambda_n}(k)|\|A\|\Omega(m-k) .
\end{equation}
Since $\omega_{\Lambda_n}(A) = \Tr(P_{\Lambda_n}A)/\Tr(P_{\Lambda_n})$ is a $G$-symmetric state on $\cA_{\Lambda_n}$, the pointwise limit $\omega_{\Lambda_n} \to \overline{\omega}_0$ holds as $\overline{\omega}_0$ is the unique $G$-symmetric state in $\cS_0$. Therefore, \eqref{G-indistinguish} implies \[ r_x^{\Omega,\cE_0}(\Lambda_n) \leq r_x^{C'\Omega,G}(\Lambda_n)\] for $n$ sufficiently large.

The main result of this section proves stability for a sequence of Hamiltonians $H_{\Lambda_n}(s)$ as in \eqref{Gpert_ham_TL} with perturbation regions defined using the $G$-broken indistinguishability radius:
\be\label{eq:broken_LTQO_regions}
\Lambda_n^p = \{ x \in \Lambda_n : r_y^{\Omega,\cE_0}( \Lambda_n) \geq K_n + L_n \mbox{ for all } y \in b_x^{\Lambda_n}(K_n) \}.
\ee
In particular, we require this sequence forms a \emph{uniform gauge symmetry-breaking perturbation model} by which we mean:
\begin{enumerate}
	\item[(i)] Both the initial interaction and perturbation take values in the algebra $\cAG$.
	\item[(ii)]  It is a uniform perturbation model as in Definition~\ref{def:pert_model} and Assumption~\ref{ass:uni_pert_model} with perturbation regions as in \eqref{eq:broken_LTQO_regions}.
\end{enumerate}

\begin{thm}\label{thm:sym_stability} Assume that $\cS_0$ is a simplex with pure, gauge symmetry broken extreme points, and that $H_{\Lambda_n}(s)$, $n\geq1$ is a uniform gauge symmetry-breaking perturbation model for which \eqref{lim=0} and Assumption~\ref{assumption:LTQO_N} hold. Fix $0<\gamma<\gamma_0$. Then for any $0\leq s\leq s_\gamma$ the following properties hold:
	
	(i) The orthogonal projections $P^i_{\Lambda_n}(s)$ for which $\alpha_s^{(n)}(P^i_{\Lambda_n}(s))=P^i_{\Lambda_n}$ satisfy:
	\bea
	&&\left\Vert P_{\Lambda_n} (s)- \sum_{i=1}^N P^i_{\Lambda_n}(s)\right\Vert\leq C\Omega(\diam_n(\Lambda_n));
	\label{Pis1}\eea
	where $C$ is as in \eqref{Pi1}. Moreover, defining $\omega_s^i = \omega^i \circ \alpha_s$ one has
	\bea
	&&
	\omega_s^i(A)=\lim_{n\to\infty} \frac{\Tr P^i_{\Lambda_n} (s)A}{\Tr P^i_{\Lambda_n}(s)}
	\label{Pis2}
	\eea
	and for all $x,k,m$, $A\in \cA_{b_x(k)}$ and sufficiently large $n$:
	\begin{equation}
	\left\Vert P^i_{\Lambda_n}(s) A P^j_{\Lambda_n}(s) -\delta_{ij} \omega^i_s(A) P^i_{\Lambda_n}(s) \right\Vert
	\leq |b_x(k)|\|A\| \left(\Omega(r_x^{\Omega,\cE_0}(\Lambda_n)-k-m) + 5 G_\alpha(m)\right).
	\label{LTQO_is}\end{equation}
	
	(ii) The set of limiting ground states $\cS_s$ is an $N$-dimensional simplex
	satisfying \eq{auto_equiv} and the GNS-Hamiltonian for each of its extreme points $\omega_s^i\in\cE_s$ has a positive spectral gap $\gamma_s$ above a unique ground state that satisfies
	\be
	\gamma_s \geq \limsup_n \gamma(H_{\Lambda_n}(s))\geq \gamma.
	\ee
\end{thm}

\begin{proof}
	For (i), \eq{Pis1} is immediate from \eqref{Pi1} since $\alpha_s^{(n)}$ is norm preserving, and \eqref{Pis2} follows from \eqref{Pi2} and the strong convergence $\alpha_s^{(n)}\to \alpha_s$. The LTQO property in \eq{LTQO_is} follows from \eqref{approx_ortho} and is proven using a similar argument as in Section~\ref{sec:automorphic-equivalence}, see specifically \eqref{forLTQO_s}, and that for sufficiently large $n$,
	\[
	\|\omega^i(\alpha_s(A)-\alpha_s^{(n)}(A))\| \leq |b_x(k)|\|A\|G_\alpha^\gamma(m).
	\]

	For (ii), note that the assumptions of Corollary~\ref{cor:zerod+e} guarantee that Theorem~\ref{thm:uni_stab} holds, and so the phase structure is preserved as discussed above, see \eqref{auto_equiv} and \eqref{auto_equiv_extreme}. Since the set of ground states $\cS_s$ is a simplex, its extreme points are disjoint states, i.e., their GNS representations are inequivalent. This implies that the the GNS Hamiltonian $H_{\omega^i_s}$ in each of these representations has a non-degenerate ground state.
	
	It remains to show that the finite-volume lower bounds for the spectral gap carry through in the thermodynamic limit. This follows from an 
	application of Theorem \ref{thm:gap_ineq_TL} with $P_n = P_{\Lambda_n}(s)$,  and $Q_n = P^i_{\Lambda_n}(s), i=1,\ldots, N$. For this application we use Corollary~\ref{cor:zerod+e} to verify \eqref{general_gap_n}, and \eqref{LTQO_is} to prove \eqref{LTQO_PQ}. For $0\leq s \leq s_\gamma$, \eq{gap_ineq_TL} implies that $H_{\omega^i_s}$ has spectral gap above it bounded by 
	$\limsup_n \gamma(H_{\Lambda_n}(s))\geq \gamma$.
\end{proof}

\subsection{Modifications to handle the cases (S2) and (S3)}\label{sec:casesS2S3}

In this section, we consider how to modify the arguments in Sections~\ref{subsec:Gsymmetry_radius}-\ref{subsec:Gbroken_radius} so that they apply to models with symmetry breaking of the type (S2) or (S3). The main challenge here is to prove that the anchored interaction terms constructed in Section~\ref{sec:Step1} belong to the algebra $\cAG$ so that the $G$-symmetric indistinguishability condition can be used to establish the results from Section~\ref{sec:LTQO2}. In each case, we first outline the assumptions on the lattice and interactions, and then move to discussing the necessary modifications. We begin by considering the case (S3) as it is most similar to the gauge invariant case (S1).

\subsubsection{The case (S3)} As outlined above, we consider that $\Gamma$ is a lattice with two actions: one by $\bZ^d$ via translations, and the other by a finite group of lattice symmetries $G$ (such as reflections and rotations). We require that the metric $d$ on $\Gamma$ respects these symmetries in the sense that for any $x,y\in\Gamma$
\begin{equation}\label{lattice_sym_d}
d(x,y)=d(r(x),r(y))=d(x+a,y+a)
\end{equation}
for all $r\in G$ and $a\in\bZ^d$. This is typically satisfied, e.g., by the lattice graph distance. These symmetries are realized on the local algebra $\cA_\Gamma^{\rm loc}$ through automorphisms $\sigma_r$ and $\rho_a$ whose action on $\cAG$ is compatible with the lattice symmetries in the following sense: if $A\in\cAG$ and $\supp(A)=X$, then 
\be\label{alg_lattice_action}
\sigma_r(A) = \rho_a(A) \iff r(X) = X+a.
\ee
The analogous result holds in the case that we consider these automorphisms acting on a finite volume with periodic boundary conditions given that we identify points $a,a'\in\bZ^d$ that are equivalent under the periodicity.

For the finite volume systems, as always we assume an increasing and absorbing sequence of finite volumes $(\Lambda_n,d_n)$ and associated finite volume Hamiltonians $H_{\Lambda_n}(s) \in\cA_{\Lambda_n}$ as in \eqref{Gpert_ham_TL} where $\eta_n$ is frustration-free, uniformly finite range, and uniformly bounded. In addition, we require that both the initial interaction and perturbation are translation invariant and take values in $\cAG$. Note that the assumption of translation invariance implies $\Lambda_n^p = \Lambda_n$. To model the group actions on $\Lambda_n$ and $\cA_{\Lambda_n}$, we also impose periodic boundary conditions and assume that the finite volume metrics $d_n$ also satisfy \eqref{lattice_sym_d}.

To prove that the observables constructed in Section~\ref{sec:Step1} belong to $\cAG$, one only needs to show that (for a finite volume $\Lambda$ with periodic boundary conditions) 
\be\label{S3_goal}
\Pi_{b_x^\Lambda(m)}^\Lambda(\cK(A)) \in \cA_{b_x^\Lambda(m)}\cap \cAG
\ee
for all $m\geq k$ where $A\in \cA_{b_x^\Lambda(k)}\cap \cAG$, and $\cK:\cA_\Lambda \to \cA_\Lambda$ is a quasi-local map that commutes with the symmetry automorphism. Once this is established one can proceed in the same way as in the gauge invariant case (S1).

The assumption on the metric \eqref{lattice_sym_d} guarantees that for any $X\subseteq \Lambda$, $(r(X))(m)=r(X(m))$ and $(X+a)(m)=X(m)+a$ for all $m>0$, $r\in G$ and $a\in \bZ^d$. As a result
\be\label{consistent_growth}
r(X)=X+a_r \implies r(X(m)) = X(m)+a_r.
\ee

If $A\in\cA_{X}\cap\cAG$, then by definition of $\cAG$, for any $r\in G$ there is an $a_r\in \bZ^d$ for which $A = \rho_{-a_r}(\sigma_r(A)).$ Since $\cK$ commutes with the symmetry actions, it follows that $\cK(A) = \rho_{-a_r}(\sigma_r(\cK(A)))$. Combining the covariance relation \eqref{covariance} with  \eqref{alg_lattice_action} and \eqref{consistent_growth} then shows that for all $r\in G$
\[
\Pi_{X(m)}^\Lambda(\cK(A))=\rho_{-a_r}\circ\sigma_r\circ\Pi_{X(m)}^\Lambda(\cK(A)) \in \cAG.
\]
The claim in \eqref{S3_goal} then follows from verifying \eqref{consistent_growth} applies to $X=b_x^\Lambda(k)$. This can be seen by using \eqref{lattice_sym_d} to show that for each $r\in G$ there is $a_r\in \bZ^d$ for which $r(b_x^\Lambda(k)) = b_x^\Lambda(k) + a$. 

The above shows that models with (S3) symmetry breaking will satisfy spectral stability as described in Theorem~\ref{thm:sym_stability} if they satisfy Assumption~\ref{assumption:LTQO_N} and the uniform perturbation model criterion from Section~\ref{sec:uniform_sequences}. Since the translation invariance implies that $\Lambda_n^p = \Lambda_n$, Section~\ref{subsec:PBC} becomes relevant. In particular, the uniform perturbation model criterion will be satisfied if $r_x^{\Omega,\cE_0}(\Lambda_n)\geq r_n$ for all $x\in\Lambda_n$ for a sequence $r_n \to \infty$ sufficiently fast so that Assumption~\ref{ass:new_uni} holds with $\delta_n\to 0$. 

\vskip12pt
\subsubsection{The case (S2)} We again consider a lattice $\Gamma$ endowed with an action of $\bZ^d$ by translations and assume that the metric respects this action. For the local Hamiltonians  $H_{\Lambda_n}(s)$ we assume the same construction as in the case (S3) with one alteration: it is \emph{not} required that the initial interaction or the perturbation take values in $\cAG$. Note that the periodic boundary conditions and translation invariance again imply $\Lambda_n^p=\Lambda_n$ for all $n$, and so the results from Section~\ref{subsec:PBC} are relevant here as they were in the case (S3). Also, one can assume the anchored representation of the local Hamiltonians are again translation invariant as the anchoring procedure from Section~\ref{sec:balled-int} preserves this property.

Contrasted to the other two cases, the terms $\Phi^{(1)}(x,m,s)\in\cA_{b_x^\Lambda(m)}$ constructed in Section~\ref{sec:Step1} for a fixed finite volume $\Lambda$ usually do \emph{not} belong to the algebra $\cAG$, see \eqref{phi_1_x_bits}. Thus, the approach for the other two cases does not work here and an alternate argument is needed. 

Due to the covariance property \eqref{covariance} along with the assumptions of periodic boundary conditions and translation invariance, the anchored terms are again translation invariant, that is
\be\label{trans_inv}
\rho_a(\Phi^{(1)}(x,m,s)) = \Phi^{(1)}(x+a,m,s)
\ee
for all $a\in\bZ^d$, $x\in\Lambda$, $m\geq R$, and $0\leq s \leq s_\gamma^\Lambda$. Using this, we will recombine the anchored terms to produce elements of $\cAG$ for which the conclusions of Proposition~\ref{prop:step1-commute} and Theorem~\ref{thm:Step1} still hold. This is sufficient for applying the results of Section~\ref{sec:LTQO2} and Section~\ref{sec:uniform_sequences} as well as the theory developed in this section.

Recall that $G=\Ir_{N_1} \times \cdots \times\Ir_{N_d}$ is the set of translations generating $\cAG$, and denote by $R_G := \max_{a\in G} |a|$. Then, for $m\geq R$ define
\begin{equation}\label{phiG_def}
\Phi_G^{(1)}(x,m+R_G,s) = \frac{1}{|G|}\sum_{a\in G} \Phi^{(1)}(x+a,m,s). 
\end{equation}
Applying \eqref{trans_inv} one trivially finds that $\Phi_G^{(1)}(x,m+R_G,s)\in \cA_{b_x^\Lambda(m+R_G)}\cap \cAG$ as desired. It is clear that these terms are self-adjoint, and due to periodic boundary conditions:
\be\label{TI1}
V^{(1)}(s) := \sum_{x\in\Lambda}\sum_{m\geq R}\Phi^{(1)}(x,m,s) = \sum_{x\in \Lambda}\sum_{m\geq R+R_G}\Phi_G^{(1)}(x,m,s).
\ee
Furthermore, applying Theorem~\ref{thm:Step1} to $\Phi^{(1)}(x,m,s)$ shows that for $m\geq R+R_G$
\be\label{TI2}
\|\Phi_G^{(1)}(x,m,s)\| \leq sG^{(1)}_{\rm sym}(m)
\ee
where $G^{(1)}_{\rm sym}(m) = G^{(1)}(m-R_G)$ is also of decay class $(\eta, \frac{\gamma}{2v},\theta)$. Hence, the conclusions of Theorem~\ref{thm:Step1} also hold for the averaged terms $\Phi_G^{(1)}(x,m,s)$.

To verify the conclusions of Proposition~\ref{prop:step1-commute}, for each $x\in\Lambda$ we introduce the $G$-averaged global operators
\[
\Phi_{x,G}^{(1)}(s) := \sum_{m\geq R+R_G}\Phi_G^{(1)}(x,m,s) = \frac{1}{|G|}\sum_{a\in G} \Phi_{x+a}^{(1)}(s),
\]
where we use \eqref{phiG_def} for the second equality. If we denote by $P_\Lambda$ the ground state projection onto the unperturbed Hamiltonian, then applying Proposition~\ref{prop:step1-commute} to each $\Phi_{x+a}^{(1)}(s)$, shows that
\be\label{TI3}
[P_\Lambda, \Phi_{x,G}^{(1)}(s)] = \frac{1}{|G|}\sum_{a\in G}[P_\Lambda,\Phi_{x+a}^{(1)}(s)]=0
\ee
for all $0 \leq s\leq s_\gamma^\Lambda$ as desired.

From \eqref{TI1}-\eqref{TI3}, it is clear that one can continue through the arguments of Sections~\ref{sec:LTQO2}-\ref{sec:uniform_sequences} using the $G$-averaged local and global terms along with the decay function $G_{\rm sym}^{(1)}$. Since these terms belong to the algebra $\cAG$, the $G$-symmetric indistinguishability condition holds and the results from this section can again be applied with no modifications.

\subsection{A class of one-dimensional examples with discrete symmetry breaking}\label{sec:MPSexamples}

In this section we discuss a class of frustration-free quantum spin chains ($\Gamma=\Ir$) with discrete symmetry breaking for which 
the conditions for stability with symmetry breaking as discussed in the previous section can be explicitly verified. 

Suppose $\omega$ is a pure, translation invariant matrix product state of a quantum spin chain with a $n$-dimensional single-site Hilbert space. 
It is well-known, see \cite{fannes:1992}, that there exists a positive integer $R$, and a frustration-free interaction, $0\leq h^\omega\in\cA_{[0,R]}$ such 
that $\omega$ is the unique zero-energy ground state on the whole chain for the model with local Hamiltonians
\be\label{MPS_ham}
H^\omega_\Lambda = \sum_{x\in \Ir\atop
	[x,x+R]\subset\Lambda} h^\omega_x,
\ee
where $h_x^\omega\in\cA_{[x,x+r]}$ is a translated copy of $h^\omega$.
Let us denote by $P^\omega_\Lambda$ the orthogonal projection onto 
$\ker H^\omega_\Lambda$. Then, there exists $\lambda\in [0,1)$,
and a constant $C$ such that for all $A\in\cA_{[a,b]}$ and $[a,b]\subset[l,r]$ we have ground state indistinguishability with an exponential rate:
\be\label{MPS_LTQO}
\Vert P_{[l,r]} A P_{[l,r]}  - \omega(A) P_{[l,r]} \Vert \leq C \Vert A\Vert (\lambda^{a-l}+\lambda^{r-b}).
\ee
The explicit details of this argument can be found in Theorem~\ref{thm:MPS_LTQO}. Thus, it is clear that $r_x^{\Omega}(\Lambda)\geq d(x,\partial\Lambda)$ for $\Omega(n)=2C\lambda^n$. In particular, given a perturbation model for initial Hamiltonians and $\Omega(r)$ as above, it is necessarily a \emph{uniform perturbation model}. Moreover, the conditions of Corollary~\ref{cor:zerod+e} hold with perturbation regions $\Lambda_n^p \uparrow \bZ$, and hence the results from Section~\ref{sec:automorphic-equivalence} apply.

Next, consider $N$ distinct pure, translation invariant matrix product states $\omega_1,\cdots,\omega_N$ of a quantum spin chain. 
It was proved in \cite[page 570, Theorem 1]{nachtergaele:1996} that there exists a finite-range frustration-free interaction $0\leq h\in\cA_{[0,R]}$ for which the set of zero-energy ground states of the infinite chains is exactly given by the convex hull of $\cS_0:=\{\omega_1,\cdots,\omega_N\}$. In the same paper, it was also proved that $h$ can be taken such that the orthogonal projection onto the ground state space is given by $P_{[a,b]}=\bigvee_{i=1}^N P^{\omega_i}_{[a,b]}$ where $P_{[a,b]}^{\omega_i}$ is the orthogonal projection onto $\ker(H_{[a,b]}^{\omega_i})$, and that there is a uniform positive lower bound on the spectral gaps: \[\gamma_0:=\inf_{[a,b] \in\cP_0(\bZ)}\gap(H_{[a,b]})>0. \]

For models with this construction, one necessarily has that \eqref{Pi1}-\eqref{approx_ortho} from Assumption~\ref{assumption:LTQO_N} are satisfied. First, \eqref{Pi2} holds trivially with $P_{[a,b]}^i = P_{[a,b]}^{\omega_i}$. Second, as the matrix product states are distinct, there is $\lambda\in[0,1)$ and $C>0$ so that
\be
\left\Vert P^{\omega_i}_{[a-n,b+n]} A P^{\omega_j}_{[a-n,b+n]}  -\delta_{ij} \omega_i(A) P^{\omega_i}_{[a-n,b+n]} \right\Vert \leq C \Vert A\Vert \lambda^n.
\label{LTQO_ij}
\ee
For $i=j$ this is clear from \eqref{MPS_LTQO}, and the case of $i\neq j$ is discussed in Appendix~\ref{app:N_MPS}. Finally, applying the above with $A=\idty$ shows that the projections $P^{\omega_i}_{[a-n,b+n]}$ are nearly pairwise orthogonal for large $n$:
\be
\Vert P^{\omega_i}_{[a-n,b+n]} P^{\omega_j}_{[a-n,b+n]}\Vert \leq C\lambda^n , \mbox{ for } i\neq j,
\ee
from which \eqref{Pi1} holds.

We now specialize to the situation in which the distinct pure states $\omega_i$ are related by a finite symmetry.
Suppose we have a unitary representation $G\ni g\mapsto U_g$ of a finite group $G$
on the single site Hilbert space $\cH_x$, and let $\sigma_g$ 
denote the corresponding automorphisms acting on the algebra of quasi-local observables $\cA_\Gamma$. Given a pure, translation invariant matrix product state $\omega$, consider the set of pure states, $\cE_0$, defined by 
\be\label{G_MPS}
\cE_0 = \{ \omega\circ\sigma_g \mid g\in G\}.
\ee
Then, $\cE_0$ is a finite set of mutually disjoint translation-invariant pure states to which the previous discussion can be applied, and moreover,
the corresponding frustration-free finite-range interaction can be chosen such that $\sigma_g(h)=h$. For such models we have the following theorem for stability.

\begin{thm}\label{thm:stability1D}
	Let $h\geq 0$ be a $G$-symmetric, frustration-free interaction for the set of distinct MPS $\cE_0$ as in \eqref{G_MPS}. For any $\Lambda_n\uparrow \bZ$ and any interaction $\Phi\in\cB_F$ with $F$ for a weighted $F$-function as in \eqref{g_grows_theta}, there is a sequence of perturbation regions $\Lambda_n^p\uparrow\bZ$ so that
	\[
	H_{\Lambda_n}(s) = \sum_{x, [x,x+R]\subset \Lambda_n}h_x + s\sum_{\substack{X\subseteq \Lambda_n \\ X\cap \Lambda_n^p \neq \emptyset}}\Phi(X)
	\]
	is a uniform gauge symmetry-breaking perturbation model. Moreover, there is an $s_\gamma>0$ for each $0<\gamma<\gamma_0$ so that
	\begin{enumerate}
		\item[(i)] $\cS_s$ is a simplex of ground state of the infinite spin chain ($\Gamma=\Ir$).
		\item[(ii)] The GNS Hamiltonian in each of the extreme points of $\cS_s$ has a simple ground state and a spectral gap bounded below by $\gamma$.
	\end{enumerate}
\end{thm}

\begin{proof}
	From \eqref{LTQO_ij} and the surrounding discussion, we see that Assumption~\ref{assumption:LTQO_N} holds with an exponentially decaying function $\Omega$ for which
	\[
	r_x^{\Omega,\cE_0}(\Lambda_n)\geq d(x,\partial\Lambda_n).
	\] 
	This decay implies that for any sequence of increasing and absorbing intervals $\Lambda_n\uparrow \bZ$, there exists a sequence perturbation regions
	\[
	\Lambda_n^p = \{ x \in \Lambda_n : r_y^{\Omega,\cE_0}( \Lambda_n) \geq K_n + L_n \mbox{ for all } y \in b_x^{\Lambda_n}(K_n) \}
	\]
	for which Corollary~\ref{cor:zerod+e} applies. For example, one can choose the sequences $L_n$ and $K_n$ as described in Section~\ref{sec:uniform_sequences_geometric-bc}, see \eqref{def:L_nK_n}. The claims regarding the infinite volume then follow from Theorem~\ref{thm:sym_stability}.
\end{proof}

An identical construction using site-blocking of the local Hilbert spaces can be applied to obtain models with spontaneous breaking of the lattice translation invariance that have $p$ distinct $p$-periodic
ground states, for any $p\geq 2$, see e.g. \cite{nachtergaele:1996}. Other symmetries, such as lattice reflection and charge conjugation can be treated in the same way. In each case 
Assumption~\ref{assumption:LTQO_N} holds with exponential decay as well as an analog of Theorem~\ref{thm:stability1D}.

\appendix

\section{Estimating Transformations of Anchored Interactions} \label{sec:est-trans-balled-ints}

In this section, we review some basic estimates concerning quasi-local transformations
of interactions. Many of the results proven here run parallel to estimates that can
be found in \cite[Section V.D.]{nachtergaele:2019}. However, we restrict our attention to anchored interactions here,
and this causes slight differences in some arguments. 

We begin with a simple lemma. For comparison, this lemma will play the role as \cite[Lemma A.9]{nachtergaele:2019}. Let us introduce the following notation. For any $x \in \mathbb{R}$, set 
\begin{equation} \label{non-neg-part}
| x |_+ = \max(x,0) \, .
\end{equation}
Furthermore, we will say that a function $f:[0, \infty) \to [0, \infty)$ is summable if
\begin{equation}
\| f \| = \sum_{n=0}^{\infty} f(n) < \infty \, .
\end{equation}
\begin{lemma} \label{lem:sum-dec-est} Let $F,G :[0, \infty) \to [0, \infty)$ be summable functions.
	If $G$ is also non-increasing, then for any $R \geq 0$, one has that
	\begin{equation} \label{sum-dec-est}
	\sum_{k \geq 0} G(k) \sum_{n \geq |R-k-1|_+} F(n) \leq \min\left( \| G \| \| F \|, H(R) \right)
	\end{equation} 
	where $H:[0, \infty) \to [0, \infty)$ is given by
	\begin{equation}
	H(R) = G(0) \cdot \lfloor R/2 \rfloor \sum_{n \geq \lfloor R/2 \rfloor} F(n) + \| F \| \sum_{n \geq \lfloor R/2 \rfloor } G(n). 
	\end{equation}
\end{lemma}
\begin{proof}
	First note that for any $R \geq 0$, one has the naive bound
	\begin{equation}
	\sum_{k \geq 0} G(k) \sum_{n \geq |R-k-1|_+} F(n) \leq  \| G \|  \| F \| \, .
	\end{equation}
	It is also clear that for any $0 \leq R < 2$, $H(R) = \| G \| \| F \|$. 
	
	For $R \geq 2$, a different estimate holds. Note that, in this case, 
	for $0 \leq k \leq \lfloor R/2 \rfloor -1$, one has
	\begin{equation}
	R-k-1 \geq R - \lfloor R/2 \rfloor \geq R/2 
	\end{equation} 
	and so since $G$ is non-increasing,
	\begin{equation}
	\sum_{k=0}^{\lfloor R/2 \rfloor -1} G(k) \sum_{n \geq |R-k-1|_+} F(n) \leq G(0) \cdot \lfloor R/2 \rfloor \sum_{n \geq \lfloor R/2 \rfloor} F(n) .
	\end{equation}
	For the remaining terms, it is clear that
	\begin{equation}
	\sum_{k \geq \lfloor R/2 \rfloor } G(k) \sum_{n \geq |R-k-1|_+} F(n) \leq \| F \| \sum_{k \geq \lfloor R/2 \rfloor } G(k)
	\end{equation} 
	This proves the estimate in (\ref{sum-dec-est}). 
\end{proof}

Let us now review the notion of a transformed interaction.
Consider a $\nu$-regular metric space $(\Gamma, d)$, and assume that $\Lambda \subset \Gamma$ is finite and there is an associated quantum lattice system $\mathcal{A}_{\Lambda} = \mathcal{B}( \mathcal{H}_{\Lambda})$,
see Section~\ref{sec:QSS} for more details. For any $x \in \Lambda$ and
$n \geq 0$, let us also denote by $b_x^{\Lambda}(n) = \{ y \in \Lambda : d(x,y) \leq n \}$.   Let $V \in \mathcal{A}_{\Lambda}$ denote the Hamiltonian associated with an anchored interaction $\Phi:\Lambda\times\bZ_{\geq 0}\to \cA_\Lambda$, see Definition~\ref{def:ball-int}, i.e.
\begin{equation} \label{Phi-balled}
V = \sum_{x \in \Lambda} \sum_{n \geq 0} \Phi(x,n) \, .
\end{equation}
For any linear map $\mathcal{K} : \mathcal{A}_{\Lambda} \to \mathcal{A}_{\Lambda}$, we will refer to the composition
\begin{equation} \label{apply-K}
\mathcal{K}(V) = \sum_{x \in \Lambda} \sum_{n \geq 0} \mathcal{K}( \Phi(x,n))
\end{equation}
as a {\it transformed interaction}. If the map $\mathcal{K}$ commutes with the involution,  i.e.
\begin{equation} \label{K-com-inv}
\mathcal{K}(A)^* = \mathcal{K}(A^*) \quad \mbox{for all } A \in \mathcal{A}_{\Lambda} \, ,
\end{equation}
then one may re-write this composition $\mathcal{K}( V)$ as an anchored interaction. In fact, 
using the local decompositions described in Section~\ref{sec:QL+LD}, see (\ref{def:Delta}), 
we see that for each $x \in \Lambda$ and $n \geq 0$
\begin{equation} \label{chop-terms}
\mathcal{K}(\Phi(x,n)) = \sum_{m \geq n} \Delta_{x,n;m}^{\Lambda}(\mathcal{K}(\Phi(x,n)))
\end{equation}
where we used (\ref{telescope}). Now inserting (\ref{chop-terms}) into (\ref{apply-K}), we find that
\begin{equation} \label{KPhi-balled}
\mathcal{K}(V) = \sum_{x \in \Lambda} \sum_{n \geq 0} \sum_{m \geq n} \Delta_{x,n; m}^{\Lambda}(\mathcal{K}(\Phi(x,n))) 
= \sum_{x \in \Lambda} \sum_{m \geq 0} \Psi(x,m) 
\end{equation}
where we have re-ordered the sums on $m$ and $n$ and set
\begin{equation} \label{def-psi-x-m}
\Psi(x,m) = \sum_{n=0}^m \Delta_{x, n; m}^{\Lambda}( \mathcal{K}(\Phi(x,n))) \, .
\end{equation}

It is easy to check that $\mathcal{K}(V)$, as written in (\ref{KPhi-balled}), is an anchored interaction as in Definition~\ref{def:ball-int}. 
In fact, since $\mathcal{K}$ commutes with the involution, we have that
\begin{equation}
\Psi(x,m)^* = \Psi(x,m) \in \mathcal{A}_{b_x^{\Lambda}(m)} \quad \mbox{for all } x \in \Lambda \mbox{ and } m \geq 0 \, .
\end{equation} 

Moreover, if $\Phi$ satisfies (\ref{support_property}), then so too does $\Psi$. In fact, if for some $x \in \Lambda$ and $m \geq 0$, we have that $\Psi(x,m) \neq 0$, then
there is $0 \leq n \leq m$ for which $\Delta^{\Lambda}_{x,n;m}(\mathcal{K}(\Phi(x,n))) \neq 0$.
If $n=m$, then $0 \neq \Delta^{\Lambda}_{x,m;m}(\mathcal{K}(\Phi(x,m))) = \Pi^{\Lambda}_{b_x^{\Lambda}(m)}( \mathcal{K}(\Phi(x,m)))$
and therefore, $\Phi(x,m) \neq 0$. In this case, since $\Phi$ satisfies (\ref{support_property}), we know that there
are points $y,z \in b_x^{\Lambda}(m)$ for which $d(y,z)>m-1$. Otherwise, there is $0 \leq n <m$ for which
$\Delta^{\Lambda}_{x,n;m}(\mathcal{K}(\Phi(x,n))) \neq 0$ and thus 
$\Delta^{\Lambda}_{x,n;m} = \Pi^{\Lambda}_{b_x^{\Lambda}(m)} - \Pi^{\Lambda}_{b_x^{\Lambda}(m-1)} \not\equiv 0$. In this case,
it must be that $b_x^{\Lambda}(m) \setminus b_x^{\Lambda}(m-1) \neq \emptyset$. Thus with $y =x$ and $z \in b_x^{\Lambda}(m) \setminus b_x^{\Lambda}(m-1)$,
we find $y,z \in b_x^{\Lambda}(m)$ for which $d(y,z)>m-1$.

We will now show that if $\Phi$ and $\mathcal{K}$ have appropriate decay, then so too does $\Psi$.
Let us now assume that the linear mapping $\mathcal{K}: \mathcal{A}_{\Lambda} \to \mathcal{A}_{\Lambda}$ 
is locally bounded and quasi-local. More precisely, we assume that:
\begin{enumerate}
	\item[(i)]{\it $\mathcal{K}$ is locally bounded:} There is a number $p \geq 0$ and $B< \infty$ for which
	\begin{equation} \label{K-lb}
	\| \mathcal{K}(A) \| \leq B |X|^p \| A \| .
	\end{equation}
	for all $A \in \mathcal{A}_X$ with $X \subset \Lambda$. Here $p$ is called the order of the local bound for $\mathcal{K}$.
	\item[(ii)] {\it $\mathcal{K}$ is quasi-local:} There is a number $q \geq 0$ and a non-increasing function $G:[0, \infty) \to (0, \infty)$ with
	$\lim_{r \to \infty} G(r) = 0$ for which
	\begin{equation} \label{K-ql}
	\| [ \mathcal{K}(A), B ] \| \leq |X|^q \| A \| \| B \| G(d(X,Y))
	\end{equation}
	for all $A \in \mathcal{A}_X$, $B \in \mathcal{A}_Y,$ and $X,Y \subset \Lambda$. Here $q$ is called the order of the quasi-locality bound for $\mathcal{K}$ and $G$ is the associated
	decay function. 
	
\end{enumerate}

Let $F$ be an $F$-function on $(\Gamma, d)$ in the sense described in Section~\ref{sec:Fnorm}.
For any $r \geq 0$, we will say that $\Phi$ has an $r$-th moment which is bounded by $F$ if
there is a number $\| \Phi \|_{r,F} < \infty$ for which given any $y, z \in \Lambda$,
\begin{equation} \label{fin-mth-mom}
\sum_{x \in \Lambda} \sum_{\stackrel{n \geq 0:}{y,z \in b_x^{\Lambda}(n)}} | b_x^{\Lambda}(n)|^r \| \Phi(x,n) \| \leq \| \Phi \|_{r,F} F(d(y,z)).
\end{equation}

The following result is the analogue of \cite[Theorem 5.13]{nachtergaele:2019} for anchored interactions.
\begin{thm} \label{thm:trans-int-bd} Let $\mathcal{K} : \mathcal{A}_{\Lambda} \to \mathcal{A}_{\Lambda}$ be a linear map which is 
	locally bounded and quasi-local, and $F$ be an $F$-function on $(\Gamma, d)$.
	Set $r= \max(p,q)$ with $p$ and $q$, respectively, the orders of the local
	bound and quasi-local estimate for $\mathcal{K}$, and let $\Phi\in \cB_F^r$ be an anchored interaction on $\Lambda$, satisfying (\ref{support_property}). In this case, the terms of transformed interaction $\Psi$ defined as in (\ref{KPhi-balled})-(\ref{def-psi-x-m}) satisfy the following bound: for any $y,z \in \Lambda$,
	\begin{eqnarray} \label{Psi-Fnorm-est}
	\sum_{x \in \Lambda} \sum_{\stackrel{m \geq 0:}{y,z \in b_x^{\Lambda}(m)}} \| \Psi(x,m) \| & \leq & B \| \Phi \|_{r,F}F(d(y,z)) + 4 \kappa d(y,z)^v H(d(y,z)/2) \nonumber \\
	& \mbox{ } & \quad + 4 \sum_{w \in \{y,z\}} \sum_{\stackrel{x \in \Lambda:}{d(x,w)>d(y,z)}} H(d(x,w)) \, .
	\end{eqnarray}
	Here the function $H:[0, \infty) \to [0, \infty)$ is given by
	\begin{equation} \label{def:trans-H}
	H(R) = G(0) \| \Phi \|_{r,F} \lfloor R/2 \rfloor F( \lfloor R/2 \rfloor - 1 ) + \| \Phi \|_{r,F} F(0) \sum_{k \geq \lfloor R/2 \rfloor } G(k)
	\end{equation}
\end{thm} 
\begin{proof}
	We begin this argument as in the proof of \cite[Theorem 5.13]{nachtergaele:2019}. For each $x \in \Lambda$ and $m \geq 0$, the estimate
	\begin{eqnarray}
	\| \Psi(x,m) \| & \leq & \sum_{n=0}^m \| \Delta^{\Lambda}_{x,n;m}( \mathcal{K}(\Phi(x,n))) \| \nonumber \\
	& \leq & \| \mathcal{K}( \Phi(x,m)) \| + \sum_{n=0}^{m-1} \| \Delta^{\Lambda}_{x, n; m}( \mathcal{K}(\Phi(x,n))) \| \nonumber \\
	& \leq & B |b_x^{\Lambda}(m)|^p \| \Phi(x,m) \| + 4 \sum_{n=0}^{m-1} |b_x^{\Lambda}(n)|^q \| \Phi(x,n) \| G(m-n-1) 
	\end{eqnarray}
	follows using (\ref{def-psi-x-m}), the form of the local decompositions, see (\ref{def:Delta}), the local bound (\ref{K-lb}), and inserted 
	the quasi-local bound (\ref{K-ql}) into the general estimate in Lemma~\ref{lem:qlm_loc_est}, see (\ref{Delta_bd}). From this bound, it is
	clear that
	\begin{eqnarray}
	\sum_{x \in \Lambda} \sum_{\stackrel{m \geq 0:}{y,z \in b_x^{\Lambda}(m)}} \| \Psi(x,m) \| & \leq & 
	B \sum_{x \in \Lambda} \sum_{\stackrel{m \geq 0:}{y,z \in b_x^{\Lambda}(m)}} |b_x^{\Lambda}(m)|^p \| \Phi(x,m) \|  \nonumber \\
	& \mbox{ } & \quad + 4 \sum_{x \in \Lambda} \sum_{\stackrel{m \geq 0:}{y,z \in b_x^{\Lambda}(m)}}  \sum_{n=0}^{m-1} |b_x^{\Lambda}(n)|^q \| \Phi(x,n) \| G(m-n-1) 
	\end{eqnarray}
	Using the bound on the $r$-th moment of $\Phi$, the first term on the right-hand-side of (\ref{Psi-Fnorm-est}) is clear. 
	
	For the second term above, we re-write
	\begin{equation}
	\sum_{x \in \Lambda} \sum_{\stackrel{m \geq 0:}{y,z \in b_x^{\Lambda}(m)}}  \sum_{n=0}^{m-1} |b_x^{\Lambda}(n)|^q \| \Phi(x,n) \| G(m-n-1) =
	\sum_{x \in \Lambda} \sum_{k \geq 0} G(k)  \sum_{\stackrel{n \geq 0:}{y,z \in b_x^{\Lambda}(n+k+1)}} |b_x^{\Lambda}(n)|^q \| \Phi(x,n) \| \, .
	\end{equation}
	
	Now the argument diverges slightly from the proof of \cite[Theorem 5.13]{nachtergaele:2019}.
	Here, to further estimate, we will split the sum on $x \in \Lambda$. 
	Before doing so, for each $x \in \Lambda$, let us denote by $m_0(x)$ the
	smallest integer $m \geq 0$ for which $y,z \in b_x^{\Lambda}(m)$. Two observations readily follow.
	First, for any $x \in \Lambda$, $d(y,z) \leq d(y,x) + d(x,z) \leq 2 m_0(x)$. Next, for 
	each fixed $x \in \Lambda$, 
	\begin{equation}
	\sum_{k \geq 0} G(k)  \sum_{\stackrel{n \geq 0:}{y,z \in b_x^{\Lambda}(n+k+1)}} |b_x^{\Lambda}(n)|^q \| \Phi(x,n) \| = \sum_{k \geq 0} G(k)  \sum_{n \geq |m_0(x) -k-1|_+} |b_x^{\Lambda}(n)|^q \| \Phi(x,n) \|
	\end{equation} 
	where we have used the notation (\ref{non-neg-part}).  
	
	Let us now split $\Lambda$ by writing
	\begin{equation}
	\Lambda_{y,z} = \left\{ x \in \Lambda : \max\left(d(x,y), d(x,z)\right) \leq d(y,z) \right\} = b_y^{\Lambda}(d(y,z)) \cap b_z^{\Lambda}(d(y,z)) 
	\end{equation}
	and setting $\Lambda = \Lambda_{y,z} \cup \Lambda_{y,z}^c$, a disjoint union. Note here that $\Lambda_{y,z}^c = \Lambda \setminus \Lambda_{y,z}$.  Now, for any $x \in \Lambda_{y,z}$, we estimate
	\begin{eqnarray}
	\sum_{k \geq 0} G(k)  \sum_{n \geq |m_0(x) -k-1|_+} |b_x^{\Lambda}(n)|^q \| \Phi(x,n) \| & \leq & \sum_{k \geq 0} G(k)  \sum_{n \geq |d(y,z)/2 -k-1|_+} |b_x^{\Lambda}(n)|^r \| \Phi(x,n) \| \nonumber \\
	& \leq & H(d(y,z)/2)  
	\end{eqnarray}
	where we have used that $d(y,z)/2 \leq m_0(x)$ and Lemma~\ref{lem:sum-dec-est}. Note that in this application of 
	Lemma~\ref{lem:sum-dec-est} we have taken $F(n) = |b_x^{\Lambda}(n)|^r \| \Phi(x,n) \|$ and used the analogue of (\ref{int-bd-con}). Since the right-hand-side
	above is independent of $x \in \Lambda_{y,z}$ and $| \Lambda_{y,z}| \leq |b_y^{\Lambda}(d(y,z))| \leq \kappa d(y,z)^{\nu}$ by $\nu$-regularity of
	$(\Gamma, d)$, see (\ref{nu-reg}), we have now obtained the second term on the right-hand-side of (\ref{Psi-Fnorm-est}).

	Finally, for each $x \in \Lambda_{y,z}^c$ there is $w \in \{y,z\}$ for which $d(x,w)>d(y,z)$. In this case,
	\begin{equation}
	\sum_{k \geq 0} G(k)  \sum_{n \geq |m_0(x) -k-1|_+} |b_x^{\Lambda}(n)|^q \| \Phi(x,n) \|  \leq  \sum_{k \geq 0} G(k)  \sum_{n \geq |d(x,y) -k-1|_+} |b_x^{\Lambda}(n)|^r \| \Phi(x,n) \| 
	\end{equation}
	since $d(x,w) \leq m_0(x)$ for each $w \in \{y,z\}$. Another application of Lemma~\ref{lem:sum-dec-est} completes the bound claimed in (\ref{Psi-Fnorm-est}).
\end{proof}

In applications of Theorem~\ref{thm:trans-int-bd}, it is common to know more detailed properties of the function $F$ which bounds
the decay of the anchored interaction $\Phi$, see (\ref{fin-mth-mom}), as well as the function $G$ which bounds the 
quasi-locality of $\mathcal{K}$, see (\ref{K-ql}).  The corollary that follows demonstrates a useful form
of this estimate which holds whenever both $F$ and $G$ are members of the same decay class; here we
refer specifically to the decay classes described in Definition~\ref{def:dec-class}. 
\begin{cor} \label{cor:trans-int-F-dec}
	Under the assumptions of Theorem~\ref{thm:trans-int-bd}, suppose further that there exist positive numbers $\eta$, $\xi$, and $\theta$ for which
	both $F$ and $G$, the decay functions associate to $\Phi$ and $\mathcal{K}$ through (\ref{fin-mth-mom}) and (\ref{K-ql}) respectively, are in 
	decay class $(\eta, \xi, \theta)$. In this case, for each $0< \eta' < \eta$, there is an $F$-function $F_{\Psi}^{\eta'}$ on $(\Gamma, d)$ for which
	\begin{equation} \label{trans-int-F-bd}
	\sum_{x \in \Lambda} \sum_{\stackrel{m \geq 0:}{y,z \in b_x^{\Lambda}(m)}} \| \Psi(x,m) \| \leq F_{\Psi}^{\eta'}(d(y,z))
	\end{equation}
	and moreover, for any $\zeta > \nu +1$, there are positive numbers $C_1$, $C_2$, $d$, and $a'$ ,satisfying $C_1 \geq C_2e^{- \eta' f_{\xi}(a' d^{\theta})}$,
	for which one may take $F_{\Psi}^{\eta'}$ with the form $F_{\Psi}^{\eta'} = F_{\Psi, 0} \cdot F_{\Psi, \eta'}^{\rm dec}$ where:
	\begin{equation} \label{F-factors}
	F_{\Psi, 0}(r) = \frac{1}{(1+r)^{\zeta}} \quad \mbox{and} \quad F_{\Psi, \eta'}^{\rm dec}(r) = \left\{ \begin{array}{cl} C_1 & \mbox{if } 0 \leq r \leq d \\
	C_2 e^{- \eta' f_{\xi}(a' r^{\theta})} & \mbox{if } r>d. \end{array} \right.
	\end{equation}
\end{cor}
\begin{proof}
	We prove this corollary in three steps. First, we argue that 
	\begin{equation} \label{trans-G-bd}
	\sum_{x \in \Lambda} \sum_{\stackrel{m \geq 0:}{y,z \in b_x^{\Lambda}(m)}} \| \Psi(x,m) \| \leq G_{\Psi}(d(y,z))
	\end{equation}
	for some function $G_{\Psi}$ in decay class $(\eta, \xi, \theta)$. Then we show that the estimate above implies the family of
	bounds in (\ref{trans-int-F-bd}) with functions $F_{\Psi}^{\eta'}$ having the form described in (\ref{F-factors}). Finally, we argue that
	each of the functions $F_{\Psi}^{\eta'}$, as above, are indeed $F$-functions on $(\Gamma, d)$.
	
	A direct application of Theorem~\ref{thm:trans-int-bd} shows that (\ref{trans-G-bd}) holds for the 
	function $G_{\Psi}$ defined by the right-hand-side of (\ref{Psi-Fnorm-est}). Note that the function
	$H$, as defined in (\ref{def:trans-H}), is clearly a member of the decay class $(\eta, \xi, \theta)$, and thus so too is
	\begin{equation}
	x \mapsto B \| \Phi \|_{r,F} F(x) + 4 \kappa x^{\nu} H(x/2) \, ,
	\end{equation}
	here we use, for example, the comments made in Remark~\ref{rem:dec-class}. To conclude that this
	$G_{\Psi}$ is in the appropriate decay class, we need only confirm that this is true for the final term
	on the right-hand-side of (\ref{Psi-Fnorm-est}). Since $H$ is in the appropriate decay class, 
	for any $0 < \eta' < \eta$ and each choice of $\zeta > \nu +1$, there are positive numbers 
	$C_1$, $C_2$, $a$, and $d$, with $C_1 \geq C_2 e^{- \eta' f_{\xi}(a d^{\theta})}$,
	for which
	\begin{equation}
	(1 + r)^{\zeta} H(r) \leq \left\{ \begin{array}{cl} C_1 & \mbox{if } 0 \leq r \leq d, \\
	C_2 e^{- \eta' f_{\xi}(a r^{\theta})} & \mbox{if } r > d. \end{array} \right.
	\end{equation}
	In this case, we have that for $0 \leq d(y,z) \leq d$,
	\begin{equation}
	4 \sum_{w \in \{y,z\}} \sum_{\stackrel{x \in \Lambda:}{d(x,w)>d(y,z)}} H(d(x,w)) \leq 8 C_1 \max_{x \in \Lambda} \sum_{y \in \Lambda} \frac{1}{(1+d(x,y))^{\zeta}} \, .
	\end{equation}
	When $d(y,z)>d$, we also have
	\begin{equation}
	4 \sum_{w \in \{y,z\}} \sum_{\stackrel{x \in \Lambda:}{d(x,w)>d(y,z)}} H(d(x,w)) \leq 8 C_2 e^{- \eta' f_{\xi}(a d(y,z)^{\theta})} \max_{x \in \Lambda} \sum_{y \in \Lambda} 
	\frac{1}{(1+d(x,y))^{\zeta}} \, .
	\end{equation}
	This completes the proof that $G_{\Psi}$ is in decay class $(\eta, \xi, \theta)$. 
	
	Now, for any choice of $\zeta > \nu +1$, we may write 
	\begin{equation}
	G_{\Psi}(r) = \frac{1}{(1+r)^{\zeta}} (1+r)^{\zeta} G_{\Psi}(r) \quad \mbox{for all } r \geq 0 \, .
	\end{equation}
	Since $G_{\Psi}$ is in decay class $(\eta, \xi, \theta)$, arguing as above we see that the family of
	bounds claimed in (\ref{trans-int-F-bd})-(\ref{F-factors}) holds.
	
	Turning to the final point, let $F_{\Psi}^{\eta'} = F_{\Psi,0} F_{\Psi, \eta'}^{\rm dec}$ be as in (\ref{F-factors}).
	As is proven, e.g. in Proposition A.2 of \cite{nachtergaele:2019}, the function
	\begin{equation}
	\tilde{F}(r) = \left\{ \begin{array}{cl} C_1 & 0 \leq r \leq d, \\ \frac{C_2}{(1+r)^{\zeta}} e^{- \eta' f_{\xi}(a'r^{\theta})} & r >d. \end{array} \right.
	\end{equation}
	is an $F$-function on $(\Gamma, d)$. One also readily checks that
	\begin{equation}
	\frac{1}{(1+d)^{ \zeta}} \tilde{F}(r) \leq F_{\Psi}^{\eta'}(r) \leq \tilde{F}(r) \quad \mbox{for all } r \geq 0 \, ,
	\end{equation}
	and thus $F_{\Psi}^{\eta'}$ is an $F$-function on $(\Gamma, d)$ as well.
\end{proof}

%
%

\section{Indistinguishability of Matrix Product States}\label{app:MPS}

We consider ground state indistinguishability for a frustration-free quantum spin chains with matrix product (MPS) ground states and open boundary conditions. We show that for such models there is an exponential decay function $\Omega$ for which the indistinguishability radius can be taken as the distance from the site to the volume boundary, i.e. $r_x^\Omega(\Lambda) = d(x,\partial \Lambda)$. After setting notation and reviewing key properties of MPS, we prove ground state indistinguishability for models with a unique infinite volume ground state. Afterwards, we turn to an example of discrete symmetry breaking where the local ground state space is spanned by several distinct MPS, and discuss how in the thermodynamic limit the ground states become orthogonal.

\subsection{MPS Indistinguishability with a Unique Ground State}\label{app:unique_MPS}
We associate the same on-site Hilbert space to each site, i.e. $\cH_x=\bC^d$ for each $x\in\bZ$, and assume the MPS is translation invariant and has a primitive transfer operator $\bbE:M_k\to M_k$. Here, $k$ is the bond dimension of the MPS and $M_k$ is the set of $k\times k$ matrices. For a fixed an orthonormal basis $\{\ket{i} : i=1,\ldots d\}\subset\cH_x$, these assumptions imply there is a set of matrices $\{v_i\}_{i=1}^d\subset M_k$ generating the MPS, and a density matrix $\rho\in M_k$ for which $\bE$ (in isometric form) satisfies:
\begin{equation}\label{def:transferE}
\bE(B) := \sum_{i=1}^d v_i^*Bv_i, \quad \bE(\idty) = \idty, \quad \bE^t(\rho) = \rho.
\end{equation}
The primitive assumption guarantees $\rho$ is invertible, and that there is a $\lambda \in[0,1)$ and $c>0$ for which
\begin{equation}\label{eq:E_convergence}
\|\bE^n-\bE^\infty\| \leq c\lambda^n
\end{equation}
where $\bE^\infty(B) := \Tr(\rho B)\idty$.

It is well known that for the resulting MPS2
\begin{equation}
M_k \ni B \; \mapsto \; \Gamma_{[l,r]}(B) := \sum_{i_l,\ldots,i_r=1}^d\Tr(Bv_{i_l}\cdots v_{i_r})\ket{i_l\ldots i_r}\in \cH_{[l,r]},
\end{equation}
there is a (non-unique) finite-range, frustration-free interaction on $\cA_{\bZ}^{\rm loc}$ so that the ground state space for each of the corresponding local Hamiltonians is $\ker H_{[l,r]} = \ran \Gamma_{[l,r]}$. Moreover, the assumptions guarantee that this model has a unique ground state in the thermodynamic limit, $\omega:\cA_\bZ \to \bC$, defined by
\be\label{eq:infinite_state}
\omega(A) = \Tr(\rho \bE_A(\idty)),  \quad \forall \, A\in\cA_{\bZ}^{\rm loc}
\ee
where for each $A\in \cA_{[a,b]}$ and all $B\in M_k$:
\be\label{eq:E_A}
\bE_A(B) := V_{[a,b]}^* A\otimes B V_{[a,b]} \quad \text{with}\quad
V_{[a,b]} := \sum_{i_a, \ldots, i_b=1}^d \ket{i_a\ldots i_b}\otimes v_{i_a}\cdots v_{i_b}.
\ee
For any finite interval $[a,b]$, it is easy to verify using \eqref{def:transferE} that
\be
V_{[a,b]}^*V_{[a,b]} = \bE^{b-a+1}(\idty)= \idty,
\ee
and moreover, that the state $\omega$ is consistent with the identification $A\mapsto A':=\idty^{\otimes n}\otimes A \otimes \idty^{\otimes m}\in\cA_{[a-n,b+m]}$. For the latter, one finds $\omega(A')=\omega(A)$ by first verifying $\bE_{A'}=\bE^n\circ \bE_A \circ \bE^m$ and then using \eqref{def:transferE}-\eqref{eq:infinite_state}.

We now turn to the indistinguishability of the ground states. Given two intervals $[a,b]\subset[l,r]$, let $P_{[l,r]}$ denote the orthogonal projection onto $\ran (\Gamma_{[l,r]})$ and notice that for any $A\in \cA_{[a,b]}$,
\be\label{eq:MPS_LTQO1}
\|P_{[l,r]}AP_{[l,r]}-\omega(A)P_{[l,r]}\| = 
\sup_{\Gamma_{[l,r]}(B),\Gamma_{[l,r]}(C)\neq 0}\frac{\left|\braket{\Gamma_{[l,r]}(B)}{(A-\omega(A)\idty)\,\Gamma_{[l,r]}(C)}\right|}{\|\Gamma_{[l,r]}(B)\|\|\Gamma_{[l,r]}(C)\|}.
\ee

The first result we provide, which will be used to bound the RHS above, was first proved in \cite[Lemma 5.2]{fannes:1992}. We provide the proof here as it outlines the basic techniques needed to establish ground state indistinguishability.

\begin{lemma}\label{lem:MPS_Exp_Val}
	Fix finite intervals $[a,b]\subseteq[l,r]$ and let $\Gamma_{[l,r]}:M_k \to \cH_{[l,r]}$ be a translation invariant MPS with primative transfer operator $\bE$ as in \eqref{def:transferE}-\eqref{eq:E_convergence}. Then for any $A\in\cA_{[a,b]}$,
	\begin{equation}\label{eq:MPS_expectation}
	\left|\braket{\Gamma_{[l,r]}(B)}{A\,\Gamma_{[l,r]}(C)}-\omega(A)\braket{B}{C}_{\rho}\right| \leq c\left(\Tr(\rho^{-1})\lambda^{a-l}+\lambda^{r-b}\right)\|A\|\|B\|_{\rho}\|C\|_{\rho}
	\end{equation} 
	where $\braket{B}{C}_\rho := \Tr(\rho B^*C)$ is the inner product on $M_k$ induced by $\rho$, and $\omega$ is as in \eqref{eq:infinite_state}.
\end{lemma}

\begin{proof}
Fix any orthonormal basis $\caB$ of $\bC^k$, and consider the LHS of \eqref{eq:MPS_expectation}. Using this orthonormal basis to rewrite the trace and applying \eqref{def:transferE} and \eqref{eq:E_A}, the first inner product can be rewritten as 
\bea
\braket{\Gamma_{[l,r]}(B)}{A\Gamma_{[l,r]}(C)}
	 & = &
\sum_{\substack{i_{l}, \ldots, i_r \\ j_{l}, \ldots, j_r}}
\Tr(v_{i_{r}}^*\ldots v_{i_{l}}^*B^*)\Tr(Cv_{j_l}\ldots v_{j_{r}}) \bra{i_{l}\ldots i_r} A \ket{j_{l}\ldots j_r}\nonumber\\
& = &
\sum_{\alpha,\beta\in\cB}\sum_{\substack{i_{l}, \ldots, i_r \\ j_{l}, \ldots, j_r}}
\braket{\alpha}{v_{i_{r}}^*\ldots v_{i_{l}}^*B^*\alpha}\braket{\beta}{Cv_{j_l}\ldots v_{j_{r}}\beta} \bra{i_{l}\ldots i_r} A \ket{j_{l}\ldots j_r}\nonumber\\
& = & \sum_{\alpha,\beta\in\cB}\braket{\alpha}{\bbE^{r-b}\circ\bbE_A\circ\bbE^{a-l}\left(B^*\ket{\alpha}\bra{\beta}C\right)\beta}.\label{MPS_decomp}
\eea
Here, we use that $A$ is supported on $[a,b]$, and choose the convention that $\bE^0$ is the identity operator on $M_k$. Letting $\bE^\infty$ be as in \eqref{eq:E_convergence}, linearity implies that
\be\label{E_decomp}
\bbE^{r-b}\circ\bbE_A\circ\bbE^{a-l} = \bbE^{\infty}\circ\bbE_A\circ\bbE^{\infty}+(\bbE^{r-b}-\bE^\infty)\circ\bbE_A\circ\bbE^{\infty}+\bbE^{r-b}\circ\bbE_A\circ(\bbE^{a-l}-\bE^\infty).
\ee

Inserting the definition of $\bE^\infty$ and using the orthonormality of $\cB$, one finds that the first term in this decomposition corresponds to the matrix inner product from \eqref{eq:MPS_expectation}, i.e.
\be\label{eq:MPS1}
\omega(A)\braket{B}{C}_\rho = \sum_{\alpha,\beta\in\cB} \braket{\alpha}{\bbE^{\infty}\circ\bbE_A\circ\bbE^{\infty}\left(B^*\ket{\alpha}\bra{\beta}C\right)\beta}.
\ee

For the second term, we can similarly rewrite the summation as
\[
\sum_{\alpha,\beta\in\cB}\braket{\alpha}{(\bbE^{r-b}-\bE^\infty)\circ\bbE_A\circ\bbE^{\infty}\left(B^*\ket{\alpha}\bra{\beta}C\right)\beta} = \Tr(C\rho B^*(\bbE^{r-b}-\bE^\infty)\circ\bbE_A(\idty)).
\]
Notice that $\|\bE_A\| \leq \|A\|$ since $V_{[a,b]}$ is an isomertry. As a result, using \eqref{eq:E_convergence} and applying Holder's inequality proves:
\begin{align}
| \Tr(C\rho B^*(\bbE^{r-b}-\bE^\infty)\circ\bbE_A(\idty))| & \leq \|(\bbE^{r-b}-\bE^\infty)\circ\bbE_A(\idty)\| \|C\rho B^*\|_1 \nonumber\\
& \leq c\lambda^{r-b} \|A\|\|C\rho^{1/2}\|_2\|\rho^{1/2} B^*\|_2 \nonumber\\
& =  c\lambda^{r-b} \|A\|\|C\|_\rho\|B\|_\rho \label{eq:MPS2}
\end{align}
where $\|\cdot\|_k$ for $k=1,2$ denotes the usual trace class and Hilbert-Schmidt norms, respectively. 

For the final term in \eqref{E_decomp}, we again use \eqref{eq:E_convergence} to bound
\[
\left|\braket{\alpha}{\bbE^{r-b}\circ\bbE_A\circ(\bbE^{a-l}-\bE^\infty)\left(B^*\ket{\alpha}\bra{\beta}C\right)\beta}\right| \leq 
c\lambda^{a-l}\|A\|\|B^*\ket{\alpha}\|\|C\ket{\beta}\|
\]
for any $\alpha,\beta\in\cB$. Choose $\cB$ to be any orthonormal basis that diagonalizes $\rho$, i.e. $\rho\ket{\alpha} = \rho_\alpha\ket{\alpha}$. Summing over $\alpha\in\cB$ and applying Cauchy-Schwarz yields,
\[
\sum_{\alpha\in\cB}\|B^*\ket{\alpha}\| = \sum_{\alpha}\rho_{\alpha}^{-1/2}\|B^*\rho^{1/2}\ket{\alpha}\| \leq \sqrt{\Tr(\rho^{-1})}\|B\|_{\rho}.
\]
The analogous bound holds when summing over $\beta$. As a consequence,
\begin{equation}
\sum_{\alpha,\beta\in\cB}\left|\braket{\alpha}{\bbE^{r-b}\circ\bbE_A\circ(\bbE^{a-l}-\bE^\infty)\left(B^*\ket{\alpha}\bra{\beta}C\right)\beta}\right| \leq c\Tr(\rho^{-1})\lambda^{a-l}\|A\|\|B\|_\rho\|C\|_\rho.\label{eq:MPS3}
\end{equation}

Thus, inserting \eqref{E_decomp} into \eqref{MPS_decomp}, the bound in \eqref{eq:MPS_expectation} follows from combining \eqref{eq:MPS1}-\eqref{eq:MPS3}.
\end{proof}

In the situation that $A=\idty$, the operator inside the summation of \eqref{MPS_decomp} becomes $\bE^{r-l+1}=\bE^{\infty}+(\bE^{r-l+1}-\bE^\infty)$. In this case, the argument from Lemma~\ref{lem:MPS_Exp_Val} simplifies, and using an estimate similar to \eqref{eq:MPS3} one can prove
\be\label{eq:MPS_IP}
\left| \braket{\Gamma_{[l,r]}(B)}{\Gamma_{[l,r]}(C)}-\braket{B}{C}_\rho\right| \leq c\Tr(\rho^{-1})\lambda^{r-l+1}\|B\|_\rho\|C\|_\rho,
\ee
see, e.g. \cite[Lemma 5.2]{fannes:1992}. Choosing $B=C$, this bound implies that $\Gamma_{[l,r]}$ is injective for sufficiently large intervals. Specifically, for any $0\neq B\in M_k$ one has
\be \label{eq:MPS_norm_bd}
1-c\Tr(\rho^{-1})\lambda^{r-l+1} \leq\frac{\|\Gamma_{[l,r]}(B)\|^2}{\|B\|_\rho^2} \leq 1+c\Tr(\rho^{-1})\lambda^{r-l+1}.
\ee

We are now ready to prove the lower bound on the indistinguishability radius for models with MPS ground states and a unique thermodynamic limit.
\begin{thm}\label{thm:MPS_LTQO}	Fix finite intervals $[a,b]\subseteq[l,r]$ and let $\Gamma_{[l,r]}:M_k \to \cH_{[l,r]}$ be a translation invariant MPS with primative transfer operator $\bE$ as in \eqref{def:transferE}-\eqref{eq:E_convergence}. Then for any $A\in\cA_{[a,b]}$,
	\begin{equation}\label{eq:MPS_LTQO}
	\|P_{[l,r]}AP_{[l,r]}-\omega(A)P_{[l,r]}\|\leq C(r-l+1)\left[\Tr(\rho^{-1})(\lambda^{r-l+1}+\lambda^{a-l})+\lambda^{r-b}\right]\|A\|
	\end{equation} 
	where $\omega$ is as in \eqref{eq:infinite_state} and $C(n) := c(1-c\Tr(\rho^{-1})\lambda^{n})^{-1}$. As a result, one has $r_{x}^\Omega(\Lambda) \geq d(x,\partial\Lambda)$ for all sites $x$ in an interval $\Lambda$ when choosing
	\be\label{eq:MPS_LTQO_fn}
	\Omega(n) := 2C(n)(2\Tr(\rho^{-1})+1)\lambda^n.
	\ee
\end{thm}
Before proving the result, we note that if we set $X=[a,b]$ and $\Lambda =[l,r]$, then $d(X,\partial\Lambda)=\min\{l-a,r-b\}$, and so one could further bound \eqref{eq:MPS_LTQO} by
\be\label{decay_fun_bd}
\|P_{\Lambda}AP_{\Lambda}-\omega(A)P_{\Lambda}\|\leq C(|\Lambda|)(2\Tr(\rho^{-1})+1)\lambda^{d(X,\partial\Lambda)}\|A\|,
\ee
which motivates the choice for the LTQO function. The extra factor of two in \eqref{eq:MPS_LTQO_fn} comes from replacing the infinite state $\omega$ with the finite ground state functional $\omega_\Lambda(A)=\Tr(AP_\Lambda)/\Tr(P_\Lambda)$ used to define the indistinguishability radius, see \eqref{LTQO_length}.

\begin{proof}
Fix nonzero matrices $B,C\in M_k$. Applying \eqref{eq:MPS_expectation} and \eqref{eq:MPS_IP}, one obtains the bound
\begin{align}
|\braket{\Gamma_{[l,r]}(B)}{(A-\omega(A)\idty)\Gamma_{[l,r]}(C)}| \leq & |\braket{\Gamma_{[l,r]}(B)}{A\Gamma_{[l,r]}(C)}-\omega(A)\braket{B}{C}_\rho| \\
&+ |\omega(A)|\left|\braket{\Gamma_{[l,r]}(B)}{\Gamma_{[l,r]}(C)}-\braket{B}{C}_\rho\right|\nonumber\\
\leq & c\left[\Tr(\rho^{-1})(\lambda^{r-l+1}+\lambda^{a-l})+\lambda^{r-b}\right]\|A\|\|B\|_\rho\|C\|_{\rho} \label{eq:LTQO1}
\end{align}
Therefore, \eqref{eq:MPS_LTQO} follows from combining this with \eqref{eq:MPS_LTQO1} as inverting the bounds in \eqref{eq:MPS_norm_bd} implies
\be\label{eq:LTQO2}
\frac{\|B\|_\rho}{\|\Gamma_{[l,r]}(B)\|} \leq \frac{1}{\sqrt{1-c\Tr(\rho^{-1})\lambda^{r-l+1}}}.
\ee

To determine the indistinguishability radius, for any site $x\in\Lambda$ one trivially has that $b_x^\Lambda(n) = [x-n,x+n]$ for all $n\leq d(x,\partial\Lambda)$. Thus, if $A\in \cA_{[x-m,x+m]}$ for $m\leq n$, \eqref{eq:MPS_LTQO} implies
\[
\|P_{b_x^\Lambda(n)}AP_{b_x^{\Lambda}(n)}-\omega(A)P_{b_x^{\Lambda}(n)}\| \leq C(n-m)(2\Tr(\rho^{-1})+1)\lambda^{n-m}\|A\|,
\]
where we trivially use $\lambda^{2n+1}\leq \lambda^{n-m}$. Defining the state $\omega_\Lambda(A):=\Tr(AP_\Lambda)/\Tr(P_\Lambda)$, the indistinguishability radius is given by bounding
\[
\|P_{b_x^\Lambda(n)}AP_{b_x^{\Lambda}(n)}-\omega_\Lambda(A)P_{b_x^{\Lambda}(n)}\|\leq \|P_{b_x^\Lambda(n)}AP_{b_x^{\Lambda}(n)}-\omega(A)P_{b_x^{\Lambda}(n)}\|+|\omega_\Lambda(A)-\omega(A)|.
\]
Thus, the claimed bound on $r_x^\Omega(\Lambda)$ follows from estimating the second quantity above. 

For $\Lambda$ sufficiently large, injectivity implies $\omega_\Lambda(A) = k^{-2}\sum_{i=1}^{k^2}\braket{\Gamma_\Lambda(B_i)}{A\Gamma_\Lambda(B_i)}$ for an orthonormal basis $\{\Gamma_\Lambda(B_i)\}$ of $\ran(\Gamma_\Lambda)$. For any normalized state $\Gamma_\Lambda(B)$, the arguments from \eqref{eq:LTQO1}-\eqref{eq:LTQO2} apply with $B=C$ to produce
\begin{align*}
|\braket{\Gamma_\Lambda(B)}{A\Gamma_\Lambda(B)}-\omega(A)|\leq& C(n-m)(2\Tr(\rho^{-1})+1)\lambda^{n-m}\|A\|,
\end{align*} 
where we use \eqref{decay_fun_bd}, that $C(r)$ and $\lambda^r$ are decreasing functions, and $n-k\leq d(b_x(k),\partial\Lambda)\leq |\Lambda|$.
Thus, the result follows from the bound
\[
|\omega_\Lambda(A)-\omega(A)| \leq \frac{1}{k^2}\sum_{i=1}^{k^2}|\braket{\Gamma_\Lambda(B_i)}{A\Gamma_\Lambda(B_i)}-\omega(A)| \leq C(n-m)(2\Tr(\rho^{-1})+1)\lambda^{n-m}\|A\|.
\]
\end{proof}

\subsection{Indistinguishability with Multiple MPS Ground States}\label{app:N_MPS} We now turn our attention to the situation of a frustration-free model whose ground states are spanned by several distinct MPS. Specifically, we assume there are $n\geq 2$ matrix product states for which the corresponding infinite volume ground states are unique, i.e.
\[
\omega^i \neq \omega^j \quad\text{for}\quad i\neq j,
\]
and the ground state space of the local Hamiltonians is $\ker H_\Lambda = \sum_{i=1}^n \ran(\Gamma_\Lambda^i)$. Each of these matrix product states individually satisfy the conditions of the previous section, namely \eqref{def:transferE}-\eqref{eq:E_convergence}, and so each of the infinite states $\omega^i$ is of the form \eqref{eq:infinite_state}, see also \eqref{eq:E_A}. As outlined in Section~\ref{sec:TL_G}, the correct indistinguishability condition in this situation is: 
\be
\|P_{[a-n,b+n]}^{i}AP_{[a-n,b+n]}^{j}-\delta_{ij}\omega_i(A)P_{[a-n,b+n]}^{i}\| \leq \|A\|\Omega(n)
\ee
for all $A\in \cA_{[a,b]}$ where $P_{\Lambda}^i$ is the orthogonal projection onto $\ran(\Gamma_\Lambda^i)$. Applying Theorem~\ref{thm:MPS_LTQO}, such a bound clearly holds when $i=j$. Therefore, we need only show that in the case that $i\neq j$, 
\be\label{MPS_orthog}
\sup_{\Gamma^i_\Lambda(B), \, \Gamma^j_\Lambda(C)\neq 0} \frac{\left|\braket{\Gamma^i_\Lambda(B)}{A\Gamma^j_\Lambda(C)}\right|}{\|\Gamma^i_\Lambda(B)\|\|\Gamma^j_\Lambda(C)\|} \leq \|A\|\Omega(n),
\ee
where we set $\Lambda = [a-n,b+n]$. This is the content of \cite[Lemma 6]{nachtergaele:1996}. We briefly outline this argument.

To simplify notation, consider two distinct infinite states $\omega^{1}$, $\omega^2:\cA_{\bZ}\to\bC$ defined via matrix product states on the same quantum spin system. Let $\{v_i\}_{i=1}^d \subset M_k$ and $\{w_i\}_{i=1}^d \subset M_l$ be the set of matrices defining $\Gamma_\Lambda^1$ and $\Gamma_\Lambda^2$ with respect to the same orthonormal basis for $\cH_x=\bC^d$. Arguing similarly as in \eqref{MPS_decomp}, for any $A\in\cA_{[a,b]}$, one has
\[
\braket{\Gamma^1_{\Lambda}(B)}{A\Gamma^2_{\Lambda}(C)} = \sum_{\substack{\alpha\in\cB_k \\ \beta\in\cB_L}}
\braket{\alpha}{\bbF^{n}\circ\bbF_A\circ\bbF^{n}\left(B^*\ket{\alpha}\bra{\beta}C\right)\beta}
\]
were $\cB_k$, and $\cB_l$ are orthonormal bases for the respective virtual spin spaces and, with respect to the isometries from \eqref{eq:E_A}, the transfer operators $\bbF_A$ and $\bbF$ on $M_{k\times l}$ are given by
\begin{align}
\bbF_A(B) = & V_{[a,b]}^*A\otimes B W_{[a,b]}\\
\bbF(B) = &\sum_{i=1}^d v_i^* B w_i = V^* (\idty_d\otimes B) W
\end{align}
where $V:=\sum_{i} \ket{i}\otimes v_i$ and $W:=\sum_{i} \ket{i}\otimes w_i$ are isometries for a single site. 

Using a modified version of the argument used to obtain \eqref{eq:MPS3}, see also \eqref{eq:LTQO2}, one finds that \eqref{MPS_orthog} follows from showing that
	\begin{equation}
	\|F\| :=\sup_{0\neq B\in M_{k\times l}}\frac{\|\bbF(B)\|_{\rho_2}}{\|B\|_{\rho_2}} <1
	\end{equation}
where $\|B\|_{\rho_2}^2 := \Tr(\rho_2 B^*B)$ is the inner product induced by the density matrix $\rho_2$ associated with $\omega^2$. It is easy to see that $\|\bbF\|\leq 1$ as for any $B,C\in M_{k\times l}$
\[
|\Tr(\rho_2C^*\bbF(B))|^2 \leq \Tr(\rho_2C^*V^*VC)\Tr(\rho_2W^*\idty\otimes B^*B W) = \Tr(\rho_2C^*C)\Tr(\rho_2B^*B), 
\]
where we use that $V^*V=\idty$, and $W^*\idty\otimes B^*B W =\sum_{i}w_i^*B^*Bw_i$ with \eqref{def:transferE}. Strict equality follows from showing that $\|F\|=1$ implies $\omega^1=\omega^2$. The details of this argument, which can be found at the end of the proof of \cite[Lemma 6]{nachtergaele:1996}, are left to the reader.

\section*{Acknowledgments}

All three authors wish to thank the Departments of Mathematics of the University of Arizona and the University of California, Davis, for extending their kind  hospitality to us and for the stimulating atmosphere they offered during several visits back and forth over the years it took to complete this project. Based upon work supported by the National Science Foundation under grant DMS-1813149 (BN) and the DFG under EXC-2111–390814868 (AY). 
 
\providecommand{\bysame}{\leavevmode\hbox to3em{\hrulefill}\thinspace}
\providecommand{\MR}{\relax\ifhmode\unskip\space\fi MR }
\providecommand{\MRhref}[2]{%
  \href{http://www.ams.org/mathscinet-getitem?mr=#1}{#2}
}
\providecommand{\href}[2]{#2}

\end{document}